\documentclass[11pt]{article} 

\usepackage[utf8]{inputenc}
\usepackage[pdftex,left=1in,top=1in,bottom=1in,right=1in]{geometry}
\usepackage{amsmath, amssymb, dsfont, mathrsfs,amsthm}
\usepackage{braket}
\usepackage{graphicx}
\usepackage{caption}
\usepackage{bm}
\usepackage{xcolor}
\usepackage[framemethod=tikz]{mdframed} 
\usepackage[colorlinks=true,citecolor=blue,linkcolor=blue]{hyperref}
\usepackage[nameinlink, noabbrev, capitalize]{cleveref}
\usepackage{physics}
\usepackage{verbatim} 
\usepackage{enumitem}
\hypersetup{
    colorlinks,
    linkcolor={red!50!black},
    citecolor={blue!50!black},
    urlcolor={blue!80!black}
}

\newcommand{\instalper}{\sigma}

\mathchardef\mhyphen="2D
\newcommand{\expfromddd}[2]{{\mathsf{Exp}}_{{#1}, {#2}}}
\newcommand{\expfromdddall}[2]{{\mathsf{Exp}}_{{#1}}}

\newcommand{\expnot}[4]{\expfromddd{{\mathcal{D}_{#1}^{(#2, #3)}}}{#4}}

\newcommand{\failmes}{\top}
\newcommand{\morelocalcosetcount}{a(\lambda)}

\newcommand{\fclass}{\mathcal{F}_{clas}}
\newcommand{\fullkey}{\mathsf{FULL-KEY}}
\newcommand{\typekey}{\mathsf{TYPE}}
\newcommand{\punckey}{\mathsf{PUNC-KEY}}
\newcommand{\granlen}{a(\lambda)}

\newcommand{\C}{\mathds{C}}
\newcommand{\N}{\mathds{N}}

\newcommand{\F}{\mathds{F}}

\newcommand{\R}{\mathbb{R}}
\newcommand{\E}{\mathds{E}}

\newcommand{\eps}{\varepsilon}

\newcommand{\ibegame}{\mathsf{IBE-IND}}
\newcommand{\puncibegame}{\mathsf{PUN-IBE-IND}}
\newcommand{\puncfegame}{\mathsf{PUN-FE-IND}}
\newcommand{\unclonfegame}{\mathsf{FEAntiPiracy}}
\newcommand{\strongunclonfegame}{\mathsf{FEStrongAntiPiracy}}
\newcommand{\unlearngame}{\mathsf{LearningGame}}
\newcommand{\antipirgame}{\mathsf{AntiPiracyGame}}
\newcommand{\unclonepkegame}{\mathsf{PKEAntiPiracy}}
\newcommand{\unclonsiggame}{\mathsf{SignatureAntiPiracy}}
\newcommand{\strongunclonpkegame}{\mathsf{PKEStrongAntiPiracy}}
\newcommand{\unclonprfgame}{\mathsf{PRFAntiPiracy}}

\newcommand{\hiddentriggame}{\mathsf{HiddenTriggerExp}}
\newcommand{\nife}{\mathsf{FEAntiPiracyNI}}

\newcommand{\pke}{\mathsf{PKE}}

\newcommand{\dss}{\mathsf{DS}}

\newcommand{\hyb}{\mathsf{Hyb}}

\newcommand{\samp}{\leftarrow}
\newcommand{\reg}{\mathsf{R}}
\newcommand{\regi}[1]{\reg_{\mathsf{#1}}}

\newcommand{\io}{i\mathcal{O}}

\newcommand{\idset}{\mathsf{ID}}

\newtheorem{theorem}{Theorem}
 
 \newtheorem{claim}{Claim}
 \newtheorem{lemma}{Lemma}
 
 \newtheorem{corollary}{Corollary}
 
 \newtheorem{definition}{Definition}
 
 \newtheorem{remark}{Remark}

\newcommand{\poly}{\mathsf{poly}}

\newcommand{\Dec}{\mathsf{Dec}}

\newcommand{\Ver}{\mathsf{Ver}}

\newcommand{\zo}{\{0, 1\}}

\newcommand{\negl}{\mathsf{negl}}
\newcommand{\subexp}{\mathsf{subexp}}

\newcommand{\localcosetcount}{3\cdot\lambda^3}

\newcommand{\trd}[2]{\norm{#1 - #2}_{Tr}} 

\newcommand{\statdist}[2]{|#1 - #2|}

\newcommand{\circsize}{Q(\lambda)}
\newcommand{\idlen}{L}

\newcommand{\moeflip}{\mathsf{MoE-MultChal}}
\newcommand{\moe}{\mathsf{MoE}}
\newcommand{\moecoll}{\mathsf{MoE-Coll}}
\newcommand{\moecollsel}{\mathsf{MoE-Coll-Sel}}

\newcommand{\cosettcount}{c_L(\lambda)}
\newcommand{\cosetgenparam}{1^{L(\lambda) + \lambda}}

\newcommand{\Hilbert}{\mathcal{H}} 

\newcommand{\keygen}{\mathsf{KeyGen}}
\newcommand{\qkeygen}{\mathsf{QKeyGen}}
\newcommand{\enc}{\mathsf{Enc}}
\newcommand{\dec}{\mathsf{Dec}}

\newcommand{\Eval}{\mathsf{Eval}}

\newcommand{\qunivcla}{\mathsf{U}_{quantum}} 

\newcommand{\adve}{\mathcal{A}}

\newcommand{\constmoe}{C_{\mathsf{MoE}}}
\newcommand{\constmoeflip}{C_{\mathsf{MoE.MultChal}}}
\newcommand{\constmoecoll}{C_{\mathsf{MoE.Coll}}}
\newcommand{\constmoecollsel}{C_{\mathsf{MoE.Coll.Sel}}}
\newcommand{\moecollpunckey}{\mathsf{Moe-Coll-PuncKey}}

\newcommand{\true}{1}
\newcommand{\false}{0}

\newcommand{\Sign}{\mathsf{Sign}}

\newcommand{\fe}{\mathsf{FE}}

\newcommand{\shiftd}{\Delta_{\mathsf{Shift}}}

\let\originalleft\left
\let\originalright\right
\renewcommand{\left}{\mathopen{}\mathclose\bgroup\originalleft}
\renewcommand{\right}{\aftergroup\egroup\originalright}

\title{Unclonable Cryptography with Unbounded Collusions and Impossibility of Hyperefficient Shadow Tomography}

 \author{Alper \c{C}akan\\ Carnegie Mellon University\\ \texttt{acakan@cs.cmu.edu} \and Vipul Goyal\\NTT Research \& Carnegie Mellon University\\  \texttt{vipul@vipulgoyal.org}}

\date{}

\begin{document}
	
\maketitle

\begin{abstract}
Quantum no-cloning theorem gives rise to the intriguing possibility of quantum copy protection where we encode a program or functionality in a quantum state such that a user in possession of $k$ copies cannot create $k+1$ copies, for any $k$. Introduced by Aaronson (CCC'09) over a decade ago, copy protection has proven to be notoriously hard to achieve.  Previous work has been able to achieve copy-protection for various functionalities only in restricted models: (i) in the bounded collusion setting where $k \to k+1$ security is achieved for a-priori fixed collusion bound $k$ (in the plain model with the same computational assumptions as ours, by Liu, Liu, Qian, Zhandry [TCC'22]), or, (ii) only $k \to 2k$ security is achieved (relative to a structured quantum oracle, by Aaronson [CCC'09]).

In this work, we give the first \emph{unbounded} collusion-resistant (i.e. multiple-copy secure) copy-protection schemes, answering the long-standing open question of constructing such schemes, raised by multiple previous works starting with Aaronson (CCC'09). 

More specifically, we obtain the following results.
\begin{itemize}
    \item We construct (i) public-key encryption, (ii) public-key functional encryption, (iii) signature and (iv) pseudorandom function schemes whose keys are copy-protected against {unbounded collusions} in the plain model (i.e. without any idealized oracles), assuming (post-quantum) subexponentially secure $\io$ and LWE. 
    
    \item We show that any unlearnable functionality can be copy-protected against unbounded collusions, relative to a classical oracle.

        \item As a corollary of our results, we rule out the existence of \emph{hyperefficient quantum shadow tomography}, 
    \begin{itemize}
        \item even given non-black-box access to the measurements, assuming subexponentially secure $\io$ and LWE, or,
        \item unconditionally relative to a quantumly accessible \emph{classical} oracle,
    \end{itemize}
    and hence answer an open question by Aaronson (STOC'18). 
  
\end{itemize}

We obtain our results through a novel technique which uses identity-based encryption to construct multiple copy secure copy-protection schemes from $\mathsf{1\mhyphen copy} \to \mathsf{2\mhyphen copy}$ secure schemes. We believe our technique is of independent interest.

Along the way, we also obtain the following results.
\begin{itemize}
 \item We define and prove the security of new collusion-resistant monogamy-of-entanglement games for \emph{coset states}.
    \item We construct a classical puncturable functional encryption scheme whose master secret key can be punctured at all functions $f$ such that $f(m_0) \neq f(m_1)$. This might also be of independent interest.

\end{itemize}
\end{abstract}
\newpage
\tableofcontents
\newpage
\section{Introduction}\label{sec:intro}
The no-cloning principle, a fundamental implication of quantum mechanics, shows that arbitrary unknown quantum states cannot be copied. This simple principle allows us to imagine applications that are classically impossible. Indeed, it has found a wide range of applications in cryptography, starting with the work of Wiener \cite{Wie83} where he puts forward the notion of \emph{quantum money}, where we imagine that there is a bank producing quantum states, called \emph{banknotes}, that are secure against counterfeiting: any (malicious) user in possession of $k$ banknotes for any $k$ cannot produce $k+1$ \emph{authentic} banknotes. The interesting notion of quantum banknotes (i.e., unclonable authenticatable quantum states) also led Aaronson \cite{Aar09} to pose the following question:
\begin{quote}
\begin{center}
    Can we use quantum information to \emph{copy-protect functionalities/programs}, where user(s) in possession of some number of copies of a program $P$ cannot produce more working copies?
\end{center}
\end{quote}
In more detail, we want to achieve the following. A vendor encodes a functionality\footnote{For example, a proprietary software or a decryption program/key of an encryption system that is used to distribute encrypted content} into a quantum state, and a user in possession of such a state can use it to evaluate the functionality any number of times, and we want to achieve $a \to b$ copy-protection: any malicious user(s) in possession of $a$ such copies of the program cannot produce $b$ working copies. Similar to quantum money, this is an impossible feat in a classical world since classical information can be readily copied any amount of times. Therefore, in a classical world, once you are given a single working copy of the program, you can make any number of copies of it. 

Perhaps surprisingly, \cite{Aar09} showed copy-protection using quantum information is indeed possible: relative to a \emph{structured\footnote{The oracle used in this construction takes as input a function and a value, evaluates the function on the value, or takes as input a function and outputs a Haar random state associated with it.} quantum oracle}, any \emph{unlearnable} program can be copy-protected in a way that is  $k \to k + r$ secure (for any [polynomial] $k$ and some $r > k$). That is, in the construction of \cite{Aar09}, the adversary is prevented from doubling their number of working copies. Later, Aaronson et al. \cite{ALLZZCONF20} showed that relative to a \emph{classical} structured oracle (that depends on the program being copy protected) model, any unlearnable program can be copy-protected, but this time only in the $\mathsf{1\mhyphen copy} \to \mathsf{2\mhyphen copy}$ setting.

\paragraph{Copy-Protecting Decryption Keys, PRFs and Signing Keys} In a related line of work, Georgiou and Zhandry \cite{GZ20} started the study of \emph{single-decryptor encryption}, that is, copy-protection for decryption functionality (i.e. secret keys) of a public-key encryption (PKE) scheme where an adversary tries to create $k + 1$ working decryption keys given only $k$ copy-protected keys. More formally, in this model, a \emph{pirate} adversary obtains the classical public key and $k$ copy-protected quantum secret keys of the scheme. Then, it produces $k + 1$ \emph{freeloader} adversaries that are possibly entangled but not communicating\footnote{If they were allowed to communicate, one freeloader could hold the secret key and all the other freeloaders would simply send their challenge ciphertexts to him to decrypt and send back the result.}, and these freeloaders are presented with classical challenge ciphertexts. We require that they cannot all succeed in decrypting simultaneously. \cite{GZ20} also gave a $1 \to 2$ secure copy-protection scheme relative to a structured oracle. Later, Coladangelo et al. \cite{CLLZ21} showed how to construct a $1 \to 2$ copy-protected public-key encryption using \emph{coset states}, this time in the plain model, assuming quantum hardness of LWE, (post-quantum) subexponentially secure indistinguishability obfuscation and one-way functions. They also construct $1 \to 2$ copy-protection schemes for pseudorandom functions (PRF), based on the same assumptions. Liu et al. \cite{LLQZ22} constructed \emph{bounded} collusion-resistant PKE and PRF schemes, by showing through an elegant proof that the $k$-way parallel repetitions of the schemes of \cite{CLLZ21} are bounded $k \to k+1$ copy-protection secure. Further, they also construct a \emph{bounded} $k \to k + 1$ copy-protection secure scheme for the signing keys of a signature scheme. However, for all schemes of \cite{LLQZ22}, the collusion-bound $k$ is fixed during setup, the sizes of the schemes grow (linearly) with the bound $k$ and the copy-protected key generation is stateful.

\paragraph{Collusion-Resistant Copy-Protection} Unfortunately, none of the previous work satisfy the most-general notion of unbounded collusion-resistant copy-protection where we require $k \to k + 1$ security for all polynomials $k$ (that is not known and hence the size of the scheme does not depend on it). 

In particular, all schemes of \cite{ALLZZCONF20}, \cite{CLLZ21} and \cite{LLQZ22} can easily be broken when the adversary obtains multiple copies. Any $2$ users (in case of the first two works) or $k + 1$ users (for the fixed $k$ value, in case of \cite{LLQZ22})\footnote{We re-emphasize that the size of the scheme (e.g. ciphertext and public-key sizes) grows with the set $k$ value, so it cannot be set arbitrarily large.} with copy-protected keys can create an \emph{anonymous classical} program/string (which can be copied/distributed any number of times) that can be used to decrypt any ciphertext in case of encryption schemes, or evaluate/sign any input in case of general programs, PRFs and signatures. The only other scheme, that of \cite{Aar09}, is only $k \to 2k$ secure rather than $k \to k + 1$, and more importantly, since it relies on structured quantum oracles, it cannot even be heuristically instantiated since we do not have any (even candidate) constructions of general-purpose quantum circuit obfuscation. In fact, \cite{LLQZ22} argues that even any extension of the scheme of \cite{Aar09} would require such obfuscation, since it uses Haar random states and there is evidence that these states cannot be \emph{classically verified} (\cite{LLQZ22, Kre21}).

We believe that the security guarantees of the previous work (\cite{ALLZZCONF20}, \cite{CLLZ21}, \cite{LLQZ22}) are very unrealistic in the age of the Internet: the users can actually mount the anonymous attacks described above through classical channels, by simply measuring their key and sending the classical measurement result to other parties or posting it online! 
\paragraph{Computational Complexity of Shadow Tomography} Lastly, aside from theoretical interest in the unbounded collusion setting in and of itself, we note that it is a theoretically important problem also due to its intimate connection to the computational complexity of another important problem, \emph{shadow tomography} \cite{Aar18} (see \cref{sec:ourresults} and \cref{sec:hypertom}).
 
The above state of affairs leaves open the following natural question also raised explicitly in several previous works \cite{Aar09, AC12, CLLZ21, LLQZ22}:
 \begin{quote}
\begin{center}
    Can we use quantum information to construct unbounded collusion-resistant copy-protection schemes?
\end{center}
\end{quote}

In this work, we answer the above question positively, in the plain model, with computational assumptions matching the previous work.

\subsection{Our Results}\label{sec:ourresults}
In this work, we resolve the long-standing open problem of constructing fully collusion-resistant copy-protection schemes by constructing such schemes for public-key encryption, public-key functional encryption, signatures and pseudorandom functions, all in the plain model.

\paragraph{Copy-Protecting Decryption Keys (\cref{sec:pke} and \cref{sec:mainfe})} We construct encryption schemes where the secret keys are copy-protected.
\begin{theorem}
    Assuming post-quantum subexponentially secure indistinguishability obfuscation and subexponentially secure LWE, there exists a public-key encryption scheme with fully collusion-resistant copy-protected secret keys.
\end{theorem}

Our computational assumptions above match\footnote{\label{footnote:note1} More specifically, our assumptions exactly match the assumptions made by  \cite{LLQZ22}, but \cite{CLLZ21} assumes polynomially secure LWE whereas we assume subexponentially secure LWE. We emphasize that \cite{CLLZ21} still assume subexponentially secure $\io$ and subexponentially secure one-way functions.} the assumptions made by \cite{CLLZ21} to achieve $1 \to 2$ copy-protection and those made by \cite{LLQZ22} to achieve $k \to k + 1$ bounded collusion-resistant copy-protected public-key encryption schemes.
 
\begin{theorem}
    Assuming post-quantum subexponentially secure indistinguishability obfuscation and subexponentially secure LWE, there exists a public-key \emph{functional} encryption scheme with fully collusion-resistant copy-protected secret keys.
\end{theorem}

Prior to our work, the only construction of functional encryption with copy-protected secret keys (given by Kitagawa and Nishimaki \cite{KN22}) was in the $1 \to 2$  copy-protection setting, based on assumptions same as ours, and in a weaker security model where no key queries were allowed after seeing the challenge ciphertext (see \cref{sec:relatedwork} for more details). Furthermore, on top of matching the assumptions previous work used for constructing copy-protected public-key encryption, the $\io$ assumption we make for our copy-protected FE scheme can be considered necessary since functional encryption is known to be equivalent to indistinguishability obfuscation (up to subexponential security loss) \cite{BV18}.

Since functional encryption can be used to construct identity-based encryption \cite{Sha85} and attribute-based encryption \cite{SW05, GPSW06} in a straightforward manner, our work also gives the first identity-based encryption and attribute-based encryption schemes with collusion-resistant copy-protected secret keys. Through copy-protected identity-based encryption, we can also obtain \emph{unclonable identity cards}, first suggested by \cite{Aar09}.
\paragraph{Copy-Protecting PRF and Signature Keys (\cref{sec:cpprf} and \cref{sec:unclondigsig})}
We also construct copy-protection schemes for a family of pseudorandom functions (PRF) and signing keys of a signature scheme. 
\begin{theorem}
    Assuming post-quantum subexponentially secure indistinguishability obfuscation and subexponentially secure LWE, there exists a PRF and a signature scheme with fully collusion-resistant copy-protected keys.
\end{theorem}

    \paragraph{Copy-Protecting All Unlearnable Functionalities}
We also show how to copy-protect any unlearnable functionality, relative to a classical oracle.
\begin{theorem}\label{thm:introallunlearnable}
    Assuming post-quantum subexponentially secure one-way functions\footnote{We can also achieve this result unconditionally, if we do not insist that the classical oracle is efficient.}, for any unlearnable functionality, there exists a fully collusion-resistant copy-protection scheme relative to an efficient \emph{classical} oracle. 
\end{theorem}
This supersedes\footnote{Note that \cref{thm:introallunlearnable} and similar results of \cite{Aar09, ALLZZ20} cannot be securely instantiated in the plain model for all unlearnable functionalities, since \cite{AP21} proves that there exists an unlearnable functionality that cannot be copy-protected in the plain model.} both \cite{Aar09}, which uses a structured quantum oracle and only satisfies $k \to 2k$ copy-protection, and \cite{ALLZZCONF20} which uses a structured classical oracle but only satisfies $1 \to 2$ copy-protection.

\paragraph{Impossibility of Hyperefficient Shadow Tomography (\cref{sec:hypertom})} Shadow tomography, introduced by Aaronson \cite{Aar18}, is the following task: Given many copies of a mixed state $\rho$ and a list of binary measurements $\{E_1, \dots, E_M\}$, estimate the acceptance probabilities $\Tr(E_i\rho)$ of measurements $E_i$ within additive error $\eps$, for all measurements $i \in [M]$. This task has important ramifications for quantum information theory, since it means that we can learn many properties of a quantum state without needing to do a full tomography of it, which necessarily requires exponentially many copies of the state \cite{o2016efficient}. It has also found many applications in cryptography, such as (i) \cite{Aar18} who shows that unconditional copy-protection is not possible (ii) \cite{mala} who shows that unconditional PKE cannot exist even if we allow public-keys to be quantum and (iii) \cite{dak} who shows that unconditional one-way state generators cannot exist. Lastly, shadow tomography also has connections to the question of \emph{classical vs. quantum advice}, and the related complexity classes $\mathsf{BQP/poly}$ and $\mathsf{BQP/qpoly}$. Note that in general, and in particular in all of these applications, the measurement set is indexed by all possible strings in some support and $M$ is exponential in the security parameter or in the number of qubits. In fact, the case $M = \poly(\lambda)$ can be trivially solved in polynomial ($M/\eps^2$) time with polynomially many copies, by estimating each $\Tr(E_i\rho)$ for $i \in [M]$ simply by actually performing the measurements $E_i$ multiple times on separate copies.

\cite{Aar18} showed that shadow tomography can be performed in a sample-efficient manner; using $\poly(n, \log M, \frac{1}{\eps})$ copies of an $n$-qubit state $\rho$, however, their scheme is not \emph{computationally efficient}, with time complexity $\Tilde{O}(M)$\footnote{As noted above, $M$ is exponential in the security parameter or in the number of qubits.}. In light of above, they posed the following as an open question: is \emph{hyperefficient} shadow tomography possible? That is, is it possible to perform shadow tomography with time complexity $\poly(n, \log M, \frac{1}{\eps})$? Note that in this case, we ask that the set of measurements $\{E_i\}_{i \in M}$ be implemented by a uniform quantum algorithm $E$ that on input $i, \tau$, applies the measurement $E_i$ to the state $\tau$. We will be given this quantum circuit $E$ as input and we are asked to output a quantum circuit $C$ such that $C(i)$ estimates $\Tr(E_i\rho)$ for all $i$.\footnote{Without these assumptions, even reading the descriptions of all measurements or outputting all the estimates would take $\Omega(M)$ time.}

Previously, hyperefficient shadow tomography was ruled out only relative to quantum oracles \cite{Aar18, aaronson2007quantum, Kre21}, where we only get oracle access to the measurement circuit $E$. Through a generic attack on copy-protection schemes using shadow tomography given by \cite{Aar18, sattath2022uncloneable}, a corollary of our results is the impossibility of hyperefficient shadow tomography, answering the open question of \cite{Aar18}.
\begin{corollary}
    Assuming post-quantum subexponentially secure indistinguishability obfuscation and LWE, there does not exist a hyperefficient shadow tomography algorithm.
\end{corollary}
\begin{corollary}
    Assuming post-quantum subexponentially secure one-way functions\footnote{Similar to before, we can achieve this result unconditionally if we do not insist on efficient oracles}, relative to an efficient \emph{classical} oracle, there cannot exist a hyperefficient shadow tomography algorithm.
\end{corollary}    
We note that making computational assumptions is necessary, since, hyperefficient shadow tomography is possible given access to $\mathsf{PP}$ oracle.\footnote{We thank an anonymous reviewer for pointing out this remark.}

\paragraph{Technical Contributions and Additional Results} An important contribution of our work is a novel technique to construct collusion-resistant copy-protection schemes which relies on using identity-based encryption (IBE). We use this technique in all of our constructions and we believe it to be of independent interest. Our technique could be considered an analogue of the technique of using digital signatures to construct full-fledged (i.e. collusion-resistant) quantum money from single banknote schemes \cite{LAF09, FGH12, AC12}. We also define and prove the security of new collusion-resistant \emph{monogamy-of-entanglement} games \cite{CLLZ21, CV22} for coset states to prove the security of our schemes. See \cref{sec:coset} for details.
 
Finally, using the techniques we employ to prove the security of our functional encryption scheme, we also give a construction of a classical functional encryption scheme where the master secret key can be punctured such that the resulting master key allows issuing keys only for functions $f$ that satisfy $f(m_0) = f(m_1)$. This allows us to remove the interaction/key queries after the challenge ciphertext in the usual functional encryption security game (\cref{defn:fe}), since the adversary can issue their own keys using the punctured master secret key. This might also be of independent interest. See \cref{sec:puncfecons} for details.
\begin{theorem}
    Assuming subexponentially secure indistinguishability obfuscation and one-way functions, there exists a functional encryption scheme whose master secret key can be punctured at all functions $f$ such that $f(m_0) \neq f(m_1)$.
\end{theorem}

\section{Technical Overview}
\subsection{Public-Key Encryption with Copy-Protected Secret Keys}\label{sec:intropke}
Let us first describe our security model, which is the same as previous work \cite{Aar09, GZ20, CLLZ21, LLQZ22}. We consider a public-key encryption scheme with classical ciphertexts, a classical public-key and an additional (quantum) algorithm $\mathsf{QKeyGen}$. The  copy-protected key generation algorithm $\mathsf{QKeyGen}$, on input the classical secret key, outputs a reusable quantum state that can be used to decrypt any number of times. For security, we will require that a user with $k$ copy-protected secret keys cannot create $k + 1$ keys. More formally, in an \emph{anti-piracy game} (\cref{defn:regularpkeantipir}) for public-key encryption, we have an adversary, called a \emph{pirate}. This adversary is given the public key $pk$, and then for any (polynomial) number of rounds, it queries for quantum copy-protected secret keys. After it is done, it outputs pairs of challenge messages $(m_\ell^0, m_\ell^1)_{\ell \in [k + 1]}$ and $k + 1$ (possibly entangled) \emph{freeloader} adversaries, where $k$ is the number of copy-protected keys it has queried. Then, the challenger samples challenge bits $b_\ell$, and presents each freeloader with $\mathsf{Enc}(pk, m^{b_\ell}_\ell)$. The freeloaders output their predictions $b'_\ell$, and the adversary wins if $b'_\ell = b_\ell$ for all $\ell \in [k + 1]$. We require that no efficient adversary can win with probability better than $1/2+\negl(\lambda)$. The baseline success probability is $1/2$, since the pirate adversary can output $k$ of its keys to the first $k$ freeloaders, and let the last freeloader randomly guess the challenge bit $b_{k + 1}$. 

\subsubsection*{$1\to2$ Copy-Protection Secure Construction of Coladangelo et al. \cite{CLLZ21}}
As a warm-up, we will recall the $1\to2$ copy-protection secure construction based on coset states, given by \cite{CLLZ21}, which also forms the base of our construction.

A coset state \cite{CLLZ21, VZ21} is a state of the form $\sum_{a \in A}(-1)^{\langle s', a \rangle}\ket{a + s} =: \ket{A_{s, s'}}$ where $A \subseteq \F_2^n$ is a subspace and $s, s' \in \F_2^n$. \cite{CLLZ21, CV22} showed that coset states satisfy a property called \emph{strong monogamy-of-entanglement (MoE)}, which is as follows. Consider the following game between an adversary tuple $\adve_0, \adve_1, \adve_2$ and a challenger. Challenger uniformly at random samples a subspace $A \subseteq \F_2^n$ of dimension $n/2$ and elements $s, s' \in \F_2^n$, and submits $\ket{A_{s,s'}}$ and the obfuscated programs\footnote{Here, we overload the notation to let $A+s$ also denote the program that takes as input a vector $v$ and outputs $1$ if $v \in A + s$, and $0$ if not, and similarly for $A^\perp+s$.} $\io(A+s), \io(A^\perp + s')$ to the adversary $\adve_0$. Then, the adversary $\adve_0$ outputs two (entangled) registers $\reg_1, \reg_2$, for $\adve_1, \adve_2$. Then, $\adve_1, \adve_2$ receive their registers and also the description of the subspace $A$ (but not the vectors $s, s'$ of course). Finally, $\adve_1$ is required to output a vector in $A + s$ and $\adve_2$ is required to output a vector in $A^\perp + s'$. Strong MoE property says that no efficient adversary can win this game with non-negligible probability. In a variation used implicitly by \cite{CLLZ21} and later formalized in a different context by \cite{CGLR23}, we present $\adve_0$  with multiple, say $c$ many, independent coset states (called a \emph{coset state tuple}) and the corresponding membership checking programs, and require that $\adve_1, \adve_2$ each output vectors in $A_i + s_i$ or $A^\perp_i + s'_i$ for all $i \in [c]$, depending on random challenge strings $r^1, r^2 \in \zo^{c}$ presented to them. By a reduction to the original version, it can be shown that no efficient adversary can win this game with non-negligible probability (\cref{defn:strmoe}). We call this variation the multi-challenge version.

Now, we move onto the copy-protected public-key encryption construction of \cite{CLLZ21}. During setup, we sample a coset tuple $(A_i, s_i, s'_i)_{i \in [c]}$. The coset state tuple $\ket{A_{i, s_i,s_i'}}_{i \in [c(\lambda)]}$ becomes the copy-protected quantum secret key, and we output $pk = (\io(A_i+s), \io(A_i^\perp + s_i'))_{i \in [c(\lambda)]}$  as the public key. Finally, to encrypt a message $m$, we sample a random string $r$ and an indistinguishability obfuscation $\mathsf{OP} \samp \io(\mathsf{PCt}_{pk, r, m})$, where $\mathsf{PCt}_{pk, r, m}$ is a program that takes is input vectors $(v_i)_{i \in [c]}$ and, checks if they are in correct cosets with respect to $r$. That is, we require $v \in A_i + s_i$ if the $i$-th bit of $r$ is $0$ and $v \in A^\perp_i + s'_i$ if it is $1$. The program $\mathsf{PCt}_{pk, r, m}$ outputs the message $m$ if and only if the vectors pass the test. We output $(\mathsf{OPCt}, r)$ as the ciphertext. To decrpyt a message, we simply apply QFT (quantum Fourier transform) to our coset state tuple at indices where $(r)_i = 1$. Then, it is easy to see that running $\mathsf{OPCt}$ coherently on our key and measuring the result gives us $m$ with probability $1$\footnote{By Gentle Measurement Lemma (\cref{lem:gentlemes}), this also means that we can revert the quantum key back to its original state after decrypting a ciphertext.}.

On a high level, the security follows by multi-challenge MoE game, since the two freeloaders, to decrypt their ciphertexts, must be querying the programs $\mathsf{PCt}^{(1)}, \mathsf{PCt}^{(2)}$ at the correct vectors with respect to $r_1, r_2$ respectively, which is exactly the challenge in the MoE game. The proof is more involved since (i) $\io$ is used rather than ideal oracles and (ii) the freeloaders can be entangled. We discuss this further in the upcoming sections.

\subsubsection*{Challenges for Collusion-Resistant Copy-Protection}
First, we note that the construction of \cite{CLLZ21} is trivially insecure when the adversary is given two copies of the secret key: The adversary can measure one copy of the state $\ket{A_{i, s_i, s_i'}}$ in the computational basis and the other copy in the Hadamard basis, thus obtaining vectors $v_i \in A_i + s_i$ and $w_i \in A^\perp_i + s'_i$ for all $i \in [c]$. Using these vectors, one can decrypt any ciphertext and since these vectors are classical information, the pirate adversary can indeed produce any number of working secret keys. Thus, the scheme only satisfies $1 \to 2$ unclonability.

One natural solution, argued by \cite{LLQZ22}, is to try and employ quantum states that already possess a collusion-resistant unclonability guarantee, such as Haar random states or their computational neighbor, pseudorandom states. This is indeed the approach employed by \cite{Aar09} to achieve $k \to 2k$ copy-protection relative to a structured quantum oracle. However, the problem is that there is no known way of verifying such states or employing these states to construct a copy-protection scheme without the use of quantum oracles, and there is evidence that this is an inherent property of such states \cite{LLQZ22, Kre21}.

Another natural solution, used by \cite{LLQZ22}, is to \emph{independently} sample a new coset state tuple $\ket{A^{\textcolor{blue}{(j)}}_{i, s_i, s_i'}}_{i \in [c(\lambda)]}$ whenever a copy-protected secret key is requested rather than giving out the same key state multiple times. In this case, the ciphertext program also takes as input the index $j$ of the key the decryption procedure is using, and verifies the input vectors with respect to that coset tuple. Therefore, they need to include the corresponding obfuscated membership checking programs for each possible key in the public-key, since otherwise the ciphertexts would not be decryptable by that key. Therefore, we can only have $k$ different key states for a fixed $k$ chosen during setup (which is when $pk$ is created). 
Therefore, the construction of \cite{LLQZ22} only achieves $k \to k + 1$ copy-protection where the collusion-bound $k$ needs to be known at the time of setup, and the size of the scheme (public key, ciphertexts) grows with $k$, since the scheme basically consists of $k$ independent instances of the $1 \to 2$ secure scheme of \cite{CLLZ21}. Furthermore, similar to the scheme of \cite{CLLZ21}, this scheme becomes trivially insecure once given $k + 1$ keys, since we will have obtained one of the coset state tuples twice.

\subsubsection*{Our Solution: Pseudorandom Coset States and Identity-Based Encryption}
As discussed above, if we are sampling independent coset states for each copy-protected key query, we need to have an a-priori bound on the number of different keys. In the unbounded setting, since there are exponentially many cosets, it is not possible to verify all possible cosets using a polynomial size public key $pk$.

Our solution to this is to \emph{compress} the public-key by using pseudorandom coset states rather than truly random ones. We sample a PRF key $K$ and include it in the classical secret key. Then, whenever we need to sample a copy-protected quantum secret key using our classical secret key, we sample a random identity string $id$ from $\zo^\lambda$ and then sample a coset state tuple using the randomness $F(K, id)$. Our public-key will be an obfuscated program $\mathsf{OPMem}_{K}$ (with PRF key $K$ embedded) that takes in an $id$, some vectors $(v_i)_{i \in [c]}$ and a \emph{basis} $r$, and verifies the vectors $(v_i)_{i \in [c]}$ with respect to $r$ and the coset tuple associated with $id$. We now have a polynomial size public-key that allows us to verify any possible (honest) coset state tuple.

A high level intuition for security is as follows, where for now we assume we use ideal oracles instead of $\io$. By PRF security, the adversary's view is indistinguishable from having obtained $k$ independent coset state tuples since for any efficient adversary that obtains any (polynomial) number of quantum secret keys, they will all have unique identity strings with overwhelming probability. Note that we still need to argue that one cannot produce $k + 1$ working keys from $k$ independent coset state tuples, which we discuss how to argue in \cref{sec:pkeintroproof}. 

However, in reality, we are using $\io$ and not ideal oracles. Now, the first problem is that, the coset state tuples that the adversary obtains during key query phase are no longer pseudorandom, since the adversary does not only have query access to the PRF but rather has the PRF key $K$ inside $pk$. A standard solution when using PRFs and indistinguishability obfuscation is to puncture the PRF key at some inputs. Let $id_1, \dots, id_k$ be the identity strings of the $k$ copy-protected keys obtained by the adversary. We can try to puncture the PRF key at $id_1, \dots, id_k$, but this would make the size of our public-key dependent on $k$. A much more important problem is that the adversary is not required to run $\mathsf{PCt}$ on only one of $id_i$, and in fact, $\mathsf{PCt}$ might be leaking\footnote{Since we are not using black-box obfuscation for $\mathsf{PCt}$.} $m$. Or, the adversary somehow might be obtaining the hidden message $m$ by running it on some unrelated identity $id$ and vectors that pass the verification of $\mathsf{PMem}$ for $id$\footnote{Since we are not using black-box obfuscation for $\mathsf{PMem}$.}. The latter is because the adversary has access to $K$ in some form (i.e. inside $\mathsf{PMem}$), therefore, it might be somehow obtaining $F(K, id)$ for some $id$. To rule this possibility out, we would need to puncture the PRF key at all strings in $\zo^\lambda$!

To solve this problem and to puncture the PRF key only at few points, we first want to make sure that the adversary can obtain the hidden message $m$ only by running $\mathsf{PCt}$ on an identity string associated with one of the copy-protected keys it did obtain. To ensure this, we use the following approach based on identity-based encryption (IBE)(\cref{defn:ibe}). When $\mathsf{PCt}$ is queried on some $id$ and some vectors $(u_j)_j$, after verifying that the vectors are in the correct cosets with respect to $id$ and $r$, the program $\mathsf{PCt}$ outputs an $\mathsf{IBE}$ encryption of $m$ under the identity $id$, rather than $m$ in the clear. We will also change our copy-protected key generation algorithm to output the $\mathsf{IBE}$ secret key associated with $id$. Now, we will be able to argue that if an adversary is able to decrypt a ciphertext and obtain $m$, then it must have obtained $\mathsf{IBE.Enc}(pk, id_i, m)$ for some $id_i$. This is because by the security of $\mathsf{IBE}$, the adversary cannot decrypt ciphertexts under identities other than $id_1, \dots, id_k$ - the only identities for which it has obtained the IBE secret keys. Above in turn means that the adversary must have run $\mathsf{PCt}$ on $id_i$ and the correct vectors for the coset tuple associated with $id_i$. In essence, we are forcing the adversary the clone one of the original copy-protected secret keys rather than coming up with a new key. Hence, we will eventually reduce to the MoE security of the coset state tuple associated with some $id_i$. Now, we need to only puncture the PRF key at (at most) $k$ points! This is still too many.\footnote{Remember that when obfuscating a program using $\io$, all programs that we will move between must be of the same size. Thus, if we are puncturing the PRF key at $k$ points, our initial obfuscated public-key program needs to be padded to a size that depends on $k$.} However, we observe the following: the adversary obtains $k$ secret keys $sk_{id_1}, \dots, sk_{id_k}$ of the IBE scheme while there are $k + 1$ freeloaders. Hence, by pigeonhole principle, two of the $k + 1$ freeloaders must be using the same key $sk_{id_i}$ for some $i \in [k]$, and hence, the same coset state tuple - the one associated with $id_i$. As a result, we will only need to puncture the PRF key $K$ at $id_i$. See \cref{sec:pkecons} for the full scheme.
\subsection{Proving Security}\label{sec:pkeintroproof}
In this section, we give a high-level overview of the security proof of our public-key encryption construction. Our goal is to reduce the security of our scheme to the monogamy-of-entanglement game (see \cref{sec:intropke} and \cref{defn:strmoecoll}), which we will do so by extracting coset vectors $(v_i)_{i \in [c]}$ from the freeloader adversaries. On a high level, our proof uses ideas from \cite{CLLZ21, ALLZZCONF20} for simultaneous extraction from entangled adversaries. The security proof of our functional encryption construction follows similarly and we refer the reader to \cref{sec:mainfe} for details. 

Note that in general, applying an extraction (which is essentially a measurement) on one of the freeloader adversaries might irreversibly damage the other ones since they are entangled. We will first make the testing of the freeloaders projective, which will allow us to argue that we can extract vectors from entangled adversaries since (i) repeating a projective measurement always gives the same outcome and does not change the state, (ii) acting (e.g. extracting) on some part of a state, informally, does not change the behaviour of projective measurements on the other part \emph{too much} (\cref{thm:simulproj}).
Now, let us briefly discuss \emph{projective implementations}, introduced by Zhandry \cite{Z20}. Let $\mathcal{E} = \{\mathcal{E}_1, \mathcal{E}_0 = I - \mathcal{E}\}$ be a binary POVM. \cite{Z20} shows that there is a \emph{projective} measurement (indexed by a finite subset of $\R_{0 \leq \cdot \leq 1}$) denoted $\mathsf{PI}(\mathcal{E})$ such that the following procedure has the same output distribution as applying $\mathcal{E}$ to $\rho$, for any state $\rho$.\footnote{We can equivalently say that the expected value of $\mathsf{PI}(\mathcal{E})\cdot\rho$ is $\Tr[\mathcal{E}_1\rho]$}
\begin{enumerate}
    \item Apply $\mathsf{PI}(\mathcal{E})$ to $\rho$ obtain a value $p \in [0, 1]$.
    \item Output $1$ with probability $p$.
\end{enumerate}
Essentially, the projective implementation estimates the probability that $\mathcal{E}$ \emph{accepts} $\rho$, and does so through a projective measurement. Note that $\mathsf{PI}(\mathcal{E})$ in general is inefficient, however, it can be approximated efficiently \cite{Z20}. We will ignore this issue in this section - see \cref{sec:pkeproof} for details.

In our anti-piracy game (\cref{defn:regularpkeantipir}), we assume that the pirate adversary outputs each freeloader as $(U, \sigma)$ where $U$ is a unitary and $\sigma$ is some quantum state. We interpret this as a quantum circuit\footnote{Note that while $U$ is a unitary, this definition is enough to capture general quantum circuits since the adversary can also include empty workspace qubits inside $\sigma$, along with some quantum information obtained from the copy-protected keys.} (with some hardwired quantum state) that takes in a challenge ciphertext and outputs a prediction $b'$. The challenger executes the freeloader using an appropriate universal quantum circuit. Now, let $\mathcal{D}$ be a ciphertext distribution and let $(U_i, \reg_i)$ be a freeloader output by the pirate adversary (where $\reg_i$ denotes the register containing the quantum part), and consider the following measurement on $\reg_i$.
\begin{enumerate}
    \item Sample $b \samp \zo$.
    \item Sample $ct \samp \mathcal{D}(m^b_i)$.
    \item Execute $U(ct_i, \reg_i)$, measure the first qubit of the output registers in computational basis to obtain $b'$.
    \item Output $1$ if $b' = b$.
\end{enumerate}
When we set $\mathcal{D}$ to be the honest ciphertext distribution where we encrypt $m$ as $\mathsf{PKE.Enc}(pk, m)$, we see that the above measurement exactly corresponds to the testing of the freeloader in the anti-piracy game. Now, consider a modified game (parameterized by some inverse polynomial $\gamma(\lambda)$) where instead of performing this measurement directly, the challenger performs its projective implementation $\mathsf{PI}_{\mathcal{D}}$, and the adversary is said to win if the output is $> 1/2+\gamma(\lambda)$ for all $k + 1$ freeloaders. Essentially, we are estimating the success probabilities of the freeloaders and comparing it to the baseline. Note that since $\mathsf{PI}_{\mathcal{D}}$ is projective, once we apply it and obtain a value $p$, the post-measurement state will again give $p$ when its tested again for $\mathcal{D}$. \cite{CLLZ21} proves that this modified game is stronger: it implies the security of the original anti-piracy game. Hence, we will prove security with respect to this stronger game.\footnote{There is a caveat here that we need to prove security with respect to this game for all inverse polynomial $\gamma(\lambda)$ so that it implies security with respect to the original game.}

Now, we move onto a sketch of the security proof of our scheme. The idea is to test the freeloaders with respect to multiple challenge ciphertext distributions to pinpoint two freeloaders that use the same coset state tuple, and then extracting coset vectors from them and violating its $1 \to 2$ MoE security. Let us assume that an adversary wins the (modified) anti-piracy game (with probability $1/p(\lambda)$ where $p(\cdot)$ is a polynomial), meaning that applying $\mathsf{PI}_{\mathcal{D}}$ yields $> 1/2+\gamma(\lambda)$ for all $k + 1$ freeloaders simultaneously with probability $> 1/p(\lambda)$. 
 We define ciphertext distributions $\mathcal{D}_j$, for all $j \in \{0, 1, \dots, 2^\lambda\}$, representing all possible identity strings in $\zo^\lambda$ (plus, the dummy upper bound $2^\lambda$). We define $\mathcal{D}_j$ so that an encryption of a message $m$ is $(\io(\mathsf{PCt}^j), r)$ where $\mathsf{PCt}^j$ is the program that works as the honest ciphertext program if the input $id$ satisfies $id \geq j$, and otherwise it replaces its hardcoded message $m$ with $\top$ at the beginning. Observe that $\mathcal{D}_0$ corresponds to the honest ciphertext distribution, since $id < 0$ is never satisfied. Similarly, $\mathcal{D}_{2^\lambda}$ corresponds to the dummy ciphertext distribution where the message is not actually contained in the ciphertext. 
 Now, consider the following thought experiment. We apply the measurements $\mathsf{PI}_{\mathcal{D}_i}$ sequentially from $j = 0$ to $j = 2^{\lambda}$, to all $k + 1$ freeloaders. Let $q_{\ell,j}$ denote the outcomes for each freeloader $\ell \in [k + 1]$. Intuitively, a non-negligible jump/gap between $q_{\ell,j}$ and $q_{\ell, j+1}$ for $j \in \{0,\dots,2^\lambda - 1\}$ will mean that the freeloader $\ell$ is querying the ciphertext program at some vectors that are correct for the coset tuple associated with $j$. Since $\mathcal{D}_0$ is the honest ciphertext distribution of this scheme, the step $j = 0$ corresponds to the original security game and hence we get $q_{\ell, \textcolor{red}{0}} > 1/2 + \gamma$ for all $\ell \in [k + 1]$ by assumption. We will also have $q_{\ell, \textcolor{red}{2^\lambda}} \leq 1/2$  for all $\ell \in [k + 1]$ since the step $j =2^\lambda$ corresponds to the ciphertext distribution $\mathcal{D}_{2^\lambda}$ that does not actually contain the message, and therefore no freeloader\footnote{Here, we are talking about any freeloader program/state, not necessarily the initial ones, since the state of the freeloaders has changed since we already applied the previous tests $\mathsf{PI}_{\mathcal{D}_j}$ for $j = 0, \dots, 2^\lambda - 1$} can succeed with probability better than $1/2$ against $\mathcal{D}_{2^\lambda}$. Previous works (\cite{AC12, LLQZ22}) use a pigeonhole principle to reduce $k \to k+1$ security to $1\to 2$ unclonability security, where they conclude that two freeloaders must have a large gap between $|q_{\ell, j} - q_{\ell, j + 1}|$ at the same jump point $j$, meaning that they are utilizing the same coset state tuple (\cite{LLQZ22}) or two quantum money banknotes must come from the same initial banknote (\cite{AC12}); where they randomly guess this critical index and place the $1\to 2$ challenge there. However, the problem in our case is that the possible jump points $j_{\ell}$ are in $\{0, 1, \dots, 2^\lambda - 1\}$, whereas we only have $k + 1$ freeloaders. This creates a multitude of problems: (i) we cannot conclude that there will be a non-negligible jump since the average step between $q_{\ell,0}$ and $q_{\ell, 2^\lambda}$ is $\gamma/2^\lambda$, which is negligible, (ii) even if there is a non-negligible jump, we cannot apply the pigeonhole principle to guarantee that there is a pair of freeloaders $\ell, \ell'$ that have the jump index $j_{\ell} = j_{\ell'}$ since we have $2^\lambda$ slots for $k + 1$ freeloaders. Further, note that even if both of the previous concerns worked out and the two freeloaders' non-negligible jump indices coincide, we cannot actually test the freeloaders with respect to all $\mathcal{D}_j$ to find it or randomly guess it since there are exponentially many possibilities. However, a careful reader might guess that thanks to the $\mathsf{IBE}$ security, the challenge ciphertext distributions above actually \emph{collapse} around $k$ points: $j = id_1, \dots, id_k$, the identity strings of the secret keys obtained by the adversary. That is, we claim that jumps can only happen at indices $j$ that correspond to some $id_i$. The reason is that, informally, the difference between $\mathcal{D}_j$ and $\mathcal{D}_{j+1}$ only occurs when the obfuscated ciphertext program is evaluated at $id = j$, in which case the output is $\mathsf{IBE}$ encryptions of $m$ and $\top$ respectively, both under the identity $j$. However, if $j$ is not one of $id_i$, then the different outputs of these programs will be $\mathsf{IBE}$ ciphertexts that are indistinguishable to the adversary, by the security of $\mathsf{IBE}$. Therefore, no freeloader can detect this change, and there cannot be a jump between $q_{\ell, \textcolor{red}{j}} $ and $q_{\ell, \textcolor{red}{j+1}} $. This (i) allows us to conclude that for each freeloader there must be a $\gamma/k$ jump (which is non-negligible) at one of $j = id_1, \dots, id_k$, and (ii) since the jump points are now all in $\{id_1, \dots, id_k\}$, we can apply a pigeonhole argument to say that there is two freeloaders have the same jump point since there are $k + 1$ freeloaders with $k$ jump slots. However, note that the ciphertext programs are only $\io$ programs and not ideal oracles, therefore, the above argument is only informal and needs to be proven. Overall, while the above intuitions are the crux of our technique, formalizing these requires care and the full proof delicately intertwines all these observations, whilst also dealing with further technical problems. We refer the reader to \cref{sec:pkeproof} for the full proof. We also need some new results on \emph{collusion-resistant MoE for pseudorandom coset states}, which we prove in \cref{sec:coset}.

\subsection{Public-Key Functional Encryption with Copy-Protected Functional Keys}
In the setting of functional encryption, we now have functional keys, where a functional key for a function $f$ allows one to obtain $f(m)$ given the encryption $\mathsf{Enc}(m)$, and nothing else. Similar to PKE (\cref{sec:intropke}), for functional encryption with copy-protected keys, we require that an adversary that obtains $k$ copy-protected (functional) keys cannot create $k + 1$ working keys (for any functions). We also allow the adversary to obtain classical functional keys. See \cref{sec:fedefnmain} for details of our model. We move onto our construction. The starting point is our public-key encryption scheme. To generate a quantum copy-protected key for a function $f$, we sample a random $id$ as before, but now we generate the coset tuple using the randomness $F(K, id || f)$ rather than $F(K, id)$. Basically, the coset states are now associated with both the function $f$ and a random $id$. We note that the random identity is still required, since we allow the adversary to query for multiple copy-protected keys for the same function $f$. We also change our ciphertexts so that they now output an $\mathsf{IBE}$ encryption of $f(m)$ under the identity $id || f$, rather than outputting an encryption of $m$. We refer the reader to \cref{sec:unclonfecons} for the full construction.

\subsubsection*{Proving Security: Building and Using Puncturable Functional Encryption}
Proof of security for our FE scheme will be similar to the proof of our PKE scheme (\cref{sec:pkeintroproof}). In particular, we will now identify identity string-function pairs (associated with the functional keys) with elements of $\{0, 1,\dots, 2^{\lambda}\cdot2^{\lambda}\}$, and have ciphertext distributions $\mathcal{D}_i$ for all such elements. While we have $2^{2\lambda}$ jump points, similar to PKE, we can argue that they can occur only at $k$ points: $j = id_1 || f_1, \dots, id_k || f_k$, where $f_1, \dots, f_k$ are the functions the pirate adversary has queried in copy-protected mode and $id_1, \dots, id_k$ are the associated identity strings in the same order. The reason is that, for other values, either (i) the adversary will not have the IBE secret key for the identity $id || g$ (meaning that it has not queried for the function $g$), or (ii) we will have $g(m_0) = g(m_1)$ (if it has queried for the function $g$ in the classical mode). Thus, $\mathcal{D}_{i}$ and $\mathcal{D}_{i+1}$ will be indistinguishable at points other than ones listed above; either by IBE security or since the $\mathsf{PCt}$ will output $g(m_0) = g(m_1)$ in both distributions.

There is one caveat left. As discussed before, in our copy-protection security proofs, we crucially rely on \emph{projective implementations} \cite{Z20} to estimate the success of the freeloader adversaries for the task where they are given an encryption of $m^b$ with random $b \samp \zo$ and they output a prediction $b'$ for it. This allows us to simultaneously extract vectors from two entangled freeloader adversaries. While projective implementations are in general inefficient, \cite{Z20} also gives an efficient algorithm (called \emph{approximated projective implementation}) that approximates it well, using a technique similar to the celebrated witness-preserving $\mathsf{QMA}$ amplification result of \cite{MW05}. Crucially, we note that above decryption process between the challenger and freeloader, for which we estimate the success probability, is non-interactive (i.e., not single round). However, in a copy-protected functional encryption security game, the freeloader adversaries will be allowed to query for more functional keys after they receive their challenge ciphertexts, for any polynomial number of rounds. Therefore, we will not able to use the approximated projective implementation as-is to estimate the success probability of a freeloader adversary for functional encryption. While one solution might be to try and generalize approximate projective implementations to interactive procedures, given that the original technique of \cite{MW05} also only applies to $\mathsf{QMA}$ (which is single round), this might be a challenging task. 

We side-step the issue above using a classical solution. We define a variation of our scheme where the challenger gives the freeloader adversaries a \emph{punctured master secret key} $pmsk$ along with their challange ciphertext. This punctured key has the challenge messages $m^0, m^1$ chosen by the adversary hardcoded, and it takes in a function $f$ and outputs the secret key for $f$ if $f(m^0) = f(m^1)$. Then, since the freeloader adversaries can simulate (using this punctured key $pmsk$) themselves any key queries that they want to make after seeing the challenge ciphertext, we remove the interaction between the freeloaders and the challenger. As a result, we are again able to use approximate projective implementations in our technique.

The only remaining challenge is making sure that our functional encryption construction is still secure when the adversaries obtain this punctured master secret key, which is an obfuscated program that contains the master secret key $msk$. While $pmsk$ will only answer the queries on functions that the adversary was allowed to query for anyways, the problem is that we are using indistinguishability obfuscation rather than black-box obfuscation to compute $psmk$. To resolve this issue, we upgrade our FE construction to use an identity-based encryption scheme with puncturable master secret keys (\cref{sec:puncibemain}). In such a scheme, we are able to produce a master secret key that can issue identity keys for any identity other than the identity it was punctured at. When we are proving the security of our functional encryption scheme, we will construct hybrids corresponding to all possible $id || f$. Moving between each hybrid, we only need to rely on the security of $\mathsf{IBE}$ at this identity. Therefore, in our security proof, we will not only use a puncturing argument inside our obfuscated ciphertext program $\mathsf{PCt}$, but we will also puncture the $\mathsf{IBE}$ master secret key inside $pmsk$ at $id || f$. Thus, we will be able to rely on the security of $\mathsf{IBE}$ even when the adversary has $pmsk$. We refer the reader to \cref{sec:feproof} for the full proof.

\subsection{PRFs and Signature Schemes with Copy-Protected Secret Keys}
Let us first describe the setting. In the case of PRFs, we imagine a quantum key generation algorithm that, given the PRF key $K$, can generate copy-protected keys that can be used to evaluate the PRF $F(K, \cdot)$ any number of times. For copy-protection, we require that given $k$ such keys, the adversary cannot create $k + 1$ freeloaders that can distinguish $F(K, x)$ versus a random string from the co-domain of $F$, given uniformly at random $x$.\footnote{Note that $x$ being randomized and being revealated after the splitting is required, since otherwise the pirate can evaluate the PRF before splitting into freeloaders, and it can simply give the classical evaluation results to the freeloaders.} See \cref{sec:cpprf} for more details. In the case of signatures, we have copy-protected re-usable signing keys that can sign any message. Similar to above, given $k$ such keys, pirate outputs $k+1$ freeloaders, and we ask the freeloaders to sign random messages.\footnote{As in the case of PRFs, known/deterministic messages can be signed before the split by the pirate, hence, random challenge messages are required.} See \cref{sec:cpprf} and \cref{sec:unclondigsig} for more details.

Our signature scheme will be the same as our PRF scheme, where the signature on $m$ will be the PRF evaluation $F(K, m)$, with the difference from the PRF scheme being that we will also have a verification key. Similar to the signature scheme construction of Sahai and Waters \cite{SW14}, the verification key will be an obfuscated program that verifies a message-signature pair $(m, \sigma)$ by checking $f(\sigma) = f(F(K, m))$ where $f$ is a one-way function. Due to these similarities, we only discuss our signature scheme here.

 In our signature scheme, a copy-protected signing key will consist of two parts: (i) a coset state tuple generated using the randomness $F(\textcolor{red}{K''}, id)$ for random $id$, similar to our PKE scheme; (ii) an obfuscated signing program $\mathsf{PSign}_K$. The program $\mathsf{PSign}_K$ will take as input a message $m$, along with $id$ and vectors $(v_i)_{i \in [c]}$. Similar to the ciphertext programs $\mathsf{PCt}$ in our PKE construction, $\mathsf{PSign}_K$ will verify that the vectors $(v_i)_{i \in [c]}$ are in the correct cosets $A_i + s_i$ or $A^\perp_i + s'_i$, depending on the $i$-th bit of $m$, where the tuple $(A_i, s_i, s'_i)$ is associated with $id$. Informally, since the challenge messages $m^1, m^2$ are random, similar to $r^1, r^2$ in the PKE case, we will be able to violate the monogamy-of-entanglement game, given freeloader adversaries that can sign these message - arriving at a contradiction. However, since we are using $\io$ and not black-box obfuscation to obfuscate $\mathsf{PSign}_K$, some information about $K$ might leaking, allowing the adversary to sign messages without querying the program with correct vectors. Similar to \cite{CLLZ21, LLQZ22}, we use the \emph{hidden trigger technique} of \cite{SW14} to solve this issue and reduce the security of our signature scheme to that of our copy-protected PKE scheme.

Hidden triggers, introduced by \cite{SW14} to construct deniable encryption, is a sparse set of inputs that can be efficiently sampled and are pseudorandom, even given a program that \emph{uses} these inputs. In the case of \cite{CLLZ21, LLQZ22}, their set of hidden trigger inputs are special encodings of the ciphertexts of their copy-protected PKE scheme. Using this technique, they embed a separate thread in $\mathsf{PSign}_K$ that detects if the message $m$ is a trigger input, and in that case, executes the embedded ciphertext program $\mathsf{PCt}$ (which is a PKE encryption of $F(K,m)$) in this input instead of normal execution. This allows them to reduce the task of finding the signature $F(K, m)$ for a message $m$ to the task of decrypting a PKE ciphertext encrypting $F(K,m)$ (hence reducing security to their PKE scheme), by undetectably replacing the random challenge messages to be signed with hidden triggers.

In our case, two new issues arise. First, as discussed, to achieve collusion-resistance, our PKE ciphertext programs crucially output $\mathsf{IBE}$ ciphertexts upon successful coset vector verification, meaning that they are randomized programs, which makes it more challenging to encode them as hidden triggers. We solve this issue as follows. Inside the ciphertext program, we expand the hidden signature $F(K, m)$ using a PRG (since there are various length requirements on the input-output size to be able to use the hidden trigger technique), and we use part of the expanded string as a PRF key to supply randomness to $\mathsf{IBE.Encrypt}$.

Secondly, the previous work \cite{CLLZ21, LLQZ22} crucially rely on puncturing the PRF key $K$ at all the challenge points $m_1, \dots, m_{k+1}$, to replace these challenge messages with hidden trigger inputs and utilize the hidden thread in $\mathsf{PSign}$. However, in our case, we would have to puncture PRF key at $k + 1$ points since we have $k + 1$ freeloaders/challenges, where $k$ is not a-priori bounded - this is not possible since the sizes of the punctured key and the obfuscated programs would need to grow with $k$. We solve this issue by making our hidden trigger inputs \emph{publicly generatable}, that is, by arguing that hidden trigger inputs are indistinguishable from uniform strings even given a program that generates these inputs (which needs to include the PRF key $K$). This allows us to only prove that a single challenge message is indistinguishable from a single hidden trigger input, and then we simply rely on the hybrid lemma to conclude the same result for any number of challenge messages, since the trigger inputs can now be generated by the adversary itself during the hybrid lemma argument. We use a new \emph{prefix-puncturing} (\cref{defn:prepuncpf}) argument for the PRF key $K$ to achieve publicly-generatable hidden triggers for our scheme. The full proof is technical, we refer the reader to  \cref{sec:digsigcpproof} and \cref{sec:prfcpproof} for the full proofs. We believe our technique might be of independent interest and might find applications in classical cryptography.

\subsection{Impossibility of Hyperefficient Shadow Tomography}
As discussed in the introduction, an important corollary of our result is the impossibility of hyperefficient shadow tomography. Suppose a shadow tomography procedure exists. We now describe a generic attack on copy-protection, given by \cite{Aar18} and adapted to the case of copy-protecting decryption keys by \cite{sattath2022uncloneable}, that uses shadow tomography. Let $s(\lambda)$ be the size of the ciphertexts of a public-key encryption scheme $\pke$ with collusion-resistant copy-protected secret keys, for $1$-bit messages. Then, define the set of measurements $\{E_{ct}\}_{ct \in \zo^{s(\lambda)}}$ as follows: $E_{ct}$ is the binary measurement $\pke.\Dec(\cdot, ct)$. That is, given a state (which will be a copy-protected secret key in our case), $E_{ct}$ is the binary measurement implemented by running $\pke.\dec$ on $\rho$ and accepting if it outputs $1$. Then, it is easy to see that once we obtain the estimates of the acceptance probabilities of all $E_{ct}$ for the state $\rho$ where $\rho$ is the copy-protected secret key, when we are given a ciphertext $ct$, we can simply use this estimate to tell if $ct$ is an encryption of $1$ or $0$, since $E_{ct}$ would \emph{accept} $\rho$ if $ct$ is an encryption of $1$, and reject it otherwise. Since these estimates are classical values, given some number of keys we can perform shadow tomography and then we can create any number of decryption programs.

The attack above is used by \cite{Aar18, sattath2022uncloneable} to conclude that \emph{unconditional} collusion-resistant copy-protection is impossible, since \cite{Aar18} gives a shadow tomography procedure that uses polynomially many copies of a state $\rho$, however, the procedure takes exponential time. Now, the question is, does there exist a \emph{hyperefficient} shadow tomography procedure? We observe that the measurement set $\{E_{ct}\}_{ct \in \zo^{s(\lambda)}}$ above is actually implemented by a uniform algorithm: $\pke.\Dec(\cdot, \cdot)$. Hence, if there exists a hyperefficient shadow tomography (\cref{defn:hypertom}) procedure, it would output (a classical description of) a quantum circuit that can estimate all $E_{ct}$, given time and number of copies that are both $\poly\left(|\rho|, \log (2^{s(\lambda)})\right) = \poly(\lambda)$. However, this would break the security of our collusion resistant copy-protected PKE scheme (\cref{sec:pke}), since we can query for sufficiently many keys, perform the shadow tomography and freely distribute the resulting classical information. Thus, we conclude that hyperefficient shadow tomography is not possible. We refer the reader to \cref{sec:hypertom} for more details.

\subsection{Related Works}\label{sec:relatedwork}
\paragraph{Copy-Protection} See \cref{sec:intro} for an overview of work on copy-protecting general functionalities \cite{Aar09, ALLZZCONF20}, secret keys of a PKE scheme \cite{GZ20, CLLZ21, LLQZ22}, PRF keys \cite{CLLZ21, LLQZ22}, and signing keys \cite{LLQZ22}. Kitagawa and Nishimaki \cite{KN22} defined functional encryption with copy-protected functional keys in the weaker $\mathsf{1\mhyphen copy} \to \mathsf{2\mhyphen copy}$ model where the adversary can only obtain one copy-protected functional key and the freeloaders cannot query for more functional keys after receiving their challenge ciphertexts. They showed how to construct secure schemes in this model from any public-key encryption scheme with copy-protected secret keys, using $\io$. Coladangelo, Majenz, and Poremba \cite{cmp20} and Ananth et al. \cite{AKL22} showed how to construct copy-protection for point functions and compute-and-compare functions in the quantum random oracle model. 
\paragraph{Secure Leasing}
Ananth and La Placa \cite{AP21} introduced \emph{secure software leasing}, which is a weaker version of copy-protection where the adversaries are only prevented from creating two copies of their program that can both be run using the \emph{honest} evaluation algorithm. \cite{AP21} also show that even this weaker notion is impossible to achieve for all unlearnable programs \emph{in the plain model}\footnote{Note that the results of \cite{Aar09} and \cite{ALLZZCONF20} do not contradict this since they are in the oracle model.}, based on some standard assumptions. Various work also define a variant where we require that the adversary cannot produce at the same time a working copy (now allowed to be run with any algorithm) and a valid \emph{deletion certificate} for a program that passes the verification of the software vendor. Note that both of these variants are implied by copy-protection \cite{KN22, CGLR23}: we can always let the deletion certificate to be the copy-protected program itself and the deletion certificate verification procedure can simply test the returned program on various inputs. Various works \cite{AP21, ALLZZCONF20, KNY21, KN22, BGG23} construct secure leasing schemes for various primitives such as functional encryption, PRFs, indistinguishability obfuscation, based on various assumptions. As discussed above, since our schemes are unclonable, they also give publicly verifiable securely leasable schemes for the same primitives, such as functional encryption.
\paragraph{Functional Encryption} \cite{CGJS15} construct delegatable functional encryption from hierarchical identity-based encryption (HIBE) and indistinguishability obfuscation where the ciphertext is an obfuscated program that outputs a HIBE ciphertext, similar to our FE construction. Bitansky and Vaikuntanathan \cite{BV18}, Kitagawa et al. \cite{Kitagawa2022} and Yang et al. \cite{YALXY19} construct what they call \emph{puncturable functional encryption}, however, their definitions are completely different from ours (and each other) and are incomparable to our model. In the first two, they construct symmetric-key functional encryption whose secret keys can be punctured at a message or a tag. The goal is to construct indistinguishability obfuscation and succinctness is an important property for their functional encryption schemes. In \cite{YALXY19}, they construct a scheme where a functional key can be punctured at a ciphertext. Different from both works, in our classical functional encryption scheme with puncturable master key (\cref{sec:mainfe}), we will have a \emph{public-key} scheme whose master secret key can be punctured at all functions that are not \emph{differentiating} $m_0, m_1$.

\subsection{Organization}
In \cref{sec:piti}, we recall some technical results and also show some new ones for projective implementations, which are needed in our copy-protection security proofs. 

In \cref{sec:coset}, we show some new collusion-resistant monogamy-of-entanglement results for coset states, which are again needed in our proofs. 

In \cref{sec:puncibemain}, we show how to construct a puncturable identity-based encryption scheme (which is needed for our copy-protected FE scheme) and a puncturable functional encryption scheme (which uses techniques similar to our copy-protected FE scheme and might be of independent interest).

In \cref{sec:pke}, \cref{sec:mainfe}, \cref{sec:unclondigsig}, \cref{sec:cpprf} we give our copy-protected public-key encryption, public-key functional encryption, signature and PRF schemes with collusion-resistant copy-protected keys, respectively, along with their security proofs.

In \cref{sec:allunlearn}, we give collusion-resistant copy-protection schemes for all unlearnable functionalities.

In \cref{sec:hypertom}, we show the impossibility of hyperefficient shadow tomography.

\section{Preliminaries}\label{sec:prelim}
\subsection{Notation}
All of our assumptions (e.g. existence of one-way functions) will be implicitly post-quantum.

We write $\lambda$ to denote the security parameter. We write $\poly(\cdot)$ to denote a polynomial function. We write $f(\lambda) \leq \negl(\lambda)$ or $f(\lambda) < \negl(\lambda)$ and say that $f(\cdot)$ is negligible if for any polynomial $p(\cdot)$, there exits $\lambda_0$ such that $f(\lambda) < \frac{1}{p(\lambda)}$ for all $\lambda > \lambda_0$. We will write $\subexp(\cdot)$ to mean a subexponential function, meaning, $f(n) = 2^{n^c}$ for some constant $0 < c < 1$ and all sufficiently large $n$.

We say that an algorithm is efficient if it is quantum polynomial time (QPT), that is, there exists a uniform family of polynomial size quantum circuits that computes it. Unless otherwise stated, we will consider non-uniform QPT adversaries. We use the term \emph{subexponentially secure} to mean either that the advantage of any \emph{QPT} or \emph{subexponential} time adversary is $\subexp(-\lambda)$, the distinction will be clear from context. In our constructions, we will rely on the subexponential security of the underlying primitives for specific subexponential functions, such as $2^{-\lambda^c}$-security. However (unless otherwise specified) this is equivalent to assuming subexponential security for any subexponential function, since we can scale the security parameter by a polynomial.

We write $\statdist{X}{Y}$ to denote the total variation distance between two classical random variables and we write $\trd{\rho}{\sigma}$ to denote the trace distance between two quantum random variables (i.e. density matrices) $\rho, \sigma$. For a sequence of (classical or quantum) random variables $X = \{X_\lambda\}_{\lambda}, Y = \{Y_\lambda\}_\lambda$, we write $X \approx_\eps Y$ to mean $\statdist{X}{Y} < \eps$ or $\trd{X}{Y} < \eps(\lambda)$; and we write $X \approx^c_\eps Y$ to mean $\left|\Pr[\adve(1^\lambda, X) = 1] - \Pr[\adve(1^\lambda, Y) = 1]\right| <  \eps(\lambda)$ for any appropriately (will be clear from context) bounded (i.e computational) adversary $\adve$. Both are only for all sufficiently large $\lambda$. In both cases we omit $\eps$ when $\eps = \negl(\lambda)$ and we will omit specifying the adversarial constraint when the constraint is that the adversary runs in polynomial time. 

We will write $\mathcal{M}$ to denote a message space (e.g., $\zo^{m(\lambda)}$).

For a string $x$, we will write $(x)_i$ to denote the $i$-th character/bit.

We assume that the reader is familiar with the basics of quantum information theory. We will use the quantum register model, where a register is an object that has a quantum state that evolves when we act  on it. We will usually write $\reg$ to denote a quantum register and $\Hilbert$ to denote a Hilbert space. We refer the reader to \cite{NC10} and \cite{tqi} for a comprehensive review of quantum information theory.

 Wherever we use indistinguishability obfuscation $\io$, we assume that the obfuscated circuits are appropriately padded.

\subsection{Puncturable Pseudorandom Functions}\label{sec:prf}
In this section, we introduce puncturable pseudorandom functions.
\begin{definition}[\cite{SW14}]\label{defn:puncprf}
    A puncturable pseudorandom function (PRF) is a family of functions $\{F: \zo^{c(
\lambda)} \times \zo^{m(\lambda)} \to \zo^{n(\lambda)}\}_{\lambda \in \N^+}$ with the following efficient algorithms.
    \begin{itemize}
        \item $F.\mathsf{Setup}(1^\lambda):$ Takes in a security parameter and outputs a key in $\zo^{c(\lambda)}$.
        \item $F(K, x):$\footnote{We overload the notation and write $F$ to both denote  the function itself and the evaluation algorithm.} Takes in a key and an input, outputs an evaluation of the PRF.
        \item $F.\mathsf{Puncture}(K, S):$ Takes as input a key and a set $S \subseteq \zo^{m(\lambda)}$, outputs a punctured key.
    \end{itemize}
    
    We require the following.
    \paragraph{Correctness.} For all efficient distributions $\mathcal{D}(1^\lambda)$ over the power set $2^{\zo^{m(\lambda)}}$, we require
    \begin{equation*}
        \Pr[\forall x \not\in S~~F(K_S, x) = F(K, x): \begin{array}{c}
              S \samp \mathcal{D}(1^\lambda) \\
              K \samp \keygen(1^\lambda) \\
              K_S \samp \mathsf{Puncture}(K, S)
        \end{array}] = 1.
    \end{equation*}
    \paragraph{Puncturing Security} We require that any stateful QPT adversary $\adve$ wins the following game with probability at most $1/2 + \negl(\lambda)$.
    \begin{enumerate}
        \item $\adve$ outputs a set $S$.
        \item The challenger samples $K \samp \keygen(1^\lambda)$ and $ K_S \samp \mathsf{Puncture}(K, S)$
        \item The challenger samples $b \samp \zo$. If $b = 0$, the challenger submits $K_S, \{F(K, x)\}_{x \in S}$ to the adversary. Otherwise, it submits $K_S, \{y_s\}_{s \in S}$ to  the adversary where $y_s \samp \zo^{n(\lambda)}$ for all $s \in S$.
        \item The adversary outputs a guess $b'$ and we say that the adversary has won if $b' = b$.
    \end{enumerate}
\end{definition}

\begin{definition}[Injective PRF \cite{SW14}]\label{defn:injectprf}
    A PRF family $F$ is said to be \emph{statistically injective} with failure probability $\eps(\lambda)$ if, with probability $1 - \eps(\lambda)$ over the sampling of the key $K$, the function $F(K, \cdot)$ is injective.
\end{definition}
\begin{definition}[Extracting PRF \cite{SW14}]\label{defn:extractprf}
    A PRF family $F$ with input space $\zo^{m(\lambda)}$ and output space $\zo^{n(\lambda)}$ is said to be extracting with error $\eps(\lambda)$ for min-entropy $k(\lambda)$ if for any distribution of $X$ over $\zo^{m(\lambda)}$, we have $(K, F(K, X)) \approx_{\eps(\lambda)} (K, U)$ where $K \samp F.\mathsf{Setup}(1^\lambda)$ and $U$ is sampled uniformly at random from $\zo^{n(\lambda)}$.
\end{definition}
\begin{theorem}[\cite{SW14, GGM86, Z12}]\label{thm:puncprfexists}
Let $n(\cdot), m(\cdot), e(\lambda), k(\lambda)$ be efficiently computable functions.
\begin{itemize}
    \item If (post-quantum) one-way functions exist, then there exists a (post-quantum) puncturable PRF with input space $\zo^{m(\lambda)}$ and output space $\zo^{n(\lambda)}$.

\item If we assume subexponentially-secure (post-quantum) one-way functions exist, then for any $c > 0$, there exists a (post-quantum) $2^{-\lambda^c}$-secure\footnote{While the original results are for negligible security against polynomial time adversaries, it is easy to see that they carry over to subexponential security. Further, by scaling the security parameter by a polynomial and simple input/output conversions, subexponentially secure (for any exponent $c'$) one-way functions is sufficient to construct for any $c$ a puncturable PRF that is $2^{-\lambda^c}$-secure.} puncturable PRF against subexponential time adversaries with input space $\zo^{m(\lambda)}$ and output space $\zo^{n(\lambda)}$.

\item If (post-quantum) one-way functions exist, then there exists a puncturable extracting PRF with error $2^{-e(\lambda)}$ for min-entropy $k(\lambda)$, with input space $\zo^{m(\lambda)}$ and output space $\zo^{n(\lambda)}$, if $m(\lambda) \geq k(\lambda) \geq n(\lambda) + 2\cdot e(\lambda) + 2$. The same result follows for the subexponential case as above.

\item  If (post-quantum) one-way functions exist, then there exists a puncturable statistically injective PRF with error $2^{-e(\lambda)}$, with input space $\zo^{m(\lambda)}$ and output space $\zo^{n(\lambda)}$, if $n(\lambda) \geq 2\cdot m(\lambda) + e(\lambda)$. The same result follows for the subexponential case as above.
\end{itemize}
\end{theorem}

We also introduce PRFs with prefix puncturing, similar to puncturable PRFs and prefix constrained PRF keys \cite{boneh2013constrained}\footnote{In a prefix constrained PRF key, one requires that given the constrained key, any input $x$ that starts with the prefix can be evaluated, and all other PRF output values remain pseudorandom. In our setting we will require the opposite: for any input that starts with the prefix, the output will remain pseudorandom, while other inputs can be evaluated using the punctured key.}.
\begin{definition}\label{defn:prepuncpf}
    A prefix puncturable pseudorandom function (PRF) is a PRF $\{F: \zo^{c(
\lambda)} \times \zo^{m(\lambda)} \to \zo^{n(\lambda)}\}_{\lambda \in \N^+}$ with the following additional algorithm.
    \begin{itemize}
        \item $F.\mathsf{Puncture}(K, pre):$ Takes as input a key and a prefix $pre$ of length at most $m(\lambda)$, outputs a punctured key.
    \end{itemize}
    
    We require the following.
    \paragraph{Correctness.} For all efficient distributions $\mathcal{D}(1^\lambda)$ over the set $\bigcup_{\ell \leq m(\lambda)}{\zo^{\ell}} \otimes \zo^{m(\lambda)}$, we require
    \begin{equation*}
        \Pr[pre\text{~is a prefix of~}x \bigvee F(K\{pre||\cdot\}, x) = F(K, x): \begin{array}{c}
              (pre, x) \samp \mathcal{D}(1^\lambda) \\
              K \samp \keygen(1^\lambda) \\
              K\{pre||\cdot\} \samp \mathsf{Puncture}(K, pre)
        \end{array}] = 1.
    \end{equation*}
    \paragraph{Puncturing Security} We require that any stateful QPT adversary $\adve$ wins the following game with probability at most $1/2 + \negl(\lambda)$.
    \begin{enumerate}
        \item $\adve$ outputs a prefix $pre$  of length at most $m(\lambda)$ and a string $x \in \zo^{m(\lambda)}$ such that $pre$ is a prefix of $x$.
        \item The challenger samples $K \samp \keygen(1^\lambda)$ and $ K\{pre || \cdot\} \samp \mathsf{Puncture}(K, pre)$.
        \item The challenger samples $b \samp \zo$. If $b = 0$, the challenger submits $K\{pre || \cdot\}, F(K, x)$ to the adversary. Otherwise, it submits $K\{pre || \cdot\}, y$ to  the adversary where $y \samp \zo^{n(\lambda)}$.
        \item The adversary outputs a guess $b'$ and we say that the adversary has won if $b' = b$.
    \end{enumerate}
\end{definition}
\begin{theorem}\label{thm:prepuncprfexists}
Let $n(\cdot), m(\cdot), e(\lambda), k(\lambda)$ be efficiently computable functions.
 If (post-quantum) one-way functions exist, then there exists a prefix puncturable extracting PRF with error $2^{-e(\lambda)}$ for min-entropy $k(\lambda)$, with input space $\zo^{m(\lambda)}$ and output space $\zo^{n(\lambda)}$, if $m(\lambda) \geq k(\lambda) \geq n(\lambda) + 2\cdot e(\lambda) + 2$. The same result follows for the subexponential case as above.
\end{theorem}
The above theorem follows in two steps. First, we can obtain a prefix puncturable PRF using the GGM construction \cite{GGM86} (which is post-quantum secure \cite{Z12}): we partially open the evaluation tree on the key $K$ according to $pre = b_1\dots b_\ell$, and then output the keys for leaves $\overline{b_1}, b_1\overline{b_2}, b_1b_2\overline{b_3}, \dots, b_1\cdots b_{\ell- 1}\overline{b_\ell}$. Then, by an application of \cite[Theorem~3]{SW14} we can make it extracting.
\subsection{Indistinguishability Obfuscation}
In this section, we introduce indistinguishability obfuscation.
\begin{definition}
    An indistinguishability obfuscation scheme $\io$ for a class of circuits $\mathcal{C} = \{\mathcal{C}_\lambda\}_\lambda$ satisfies the following.
    \paragraph{Correctness.} For all $\lambda, C \in \mathcal{C}_\lambda$ and inputs $x$,
    $\Pr[\Tilde{C}(x) = C(x): \Tilde{C} \samp \io(1^\lambda, C)] = 1$.

\end{definition}
    \paragraph{Security.} Let $\mathcal{B}$ be any QPT algorithm that outputs two circuits $C_0, C_1 \in \mathcal{C}$ of the same size, along with auxiliary information, such that $\Pr[\forall x ~ C_0(x)=C_1(x) : (C_0, C_1, \regi{aux}) \samp \mathcal{B}(1^\lambda)] \geq 1 - \negl(\lambda)$. Then, for any QPT adversary $\mathcal{A}$,
    \begin{align*}
      \bigg|&\Pr[\adve(\io(1^\lambda, C_0), \regi{aux}) = 1 :  (C_0, C_1, \regi{aux}) \samp \mathcal{B}(1^\lambda)] -\\ &\Pr[\adve(\io(1^\lambda, C_1), \regi{aux}) = 1 : (C_0, C_1, \regi{aux}) \samp \mathcal{B}(1^\lambda)]\bigg| \leq \negl(\lambda).  
    \end{align*}

\subsection{Functional Encryption}

In this section we introduce the basic definitions of functional encryption schemes.

\begin{definition}[Functional encryption]\label{defn:fe}
    A functional encryption scheme for a class of functions $\mathfrak{F}$ consists of the following algorithms that satisfy the correctness and security guarantees below.
    \begin{itemize}
        \item $\mathsf{Setup}(1^\lambda)$: Outputs a master secret key $msk$ and a public key $pk$.
        \item $\mathsf{KeyGen}(msk, f)$: Takes in the master secret key and a function $f \in \mathfrak{F}$, outputs a functional key for $f$.
        \item $\mathsf{Enc}(pk, m)$: Takes in the public key and a message $m$, outputs an encryption of $m$.
        \item $\mathsf{Dec}(sk, ct)$: Takes in a functional key $sk$ and a ciphertext, outputs a message or $\perp$.
    \end{itemize}
    \paragraph{Correctness}
    For all functions $f \in \mathfrak{F}$ and all messages $m$, we require the following.
    \begin{equation*}
        \Pr[\mathsf{Dec}(sk, ct) = f(m) : \begin{array}{c}
             msk, pk \samp \mathsf{Setup}(1) \\
             sk \samp \mathsf{KeyGen}(msk, f) \\
             ct \samp \mathsf{Enc}(pk, m)
        \end{array}] = 1.
    \end{equation*}
    \paragraph{Adaptive indistinguishability security}
    Consider the following game between a challenger and an adversary $\adve$.
    \paragraph{$\underline{\mathsf{FE-IND}(\lambda, \adve)}$}
    \begin{enumerate}
        \item Challenger samples the keys $msk, pk \samp \mathsf{Setup}(1)$.
        \item The adversary receives $pk$. It makes polynomially many queries by sending functions $f \in \mathcal{F}$ and receiving the corresponding functional key $sk_f \samp \mathsf{KeyGen}(msk, f)$.
        \item The adversary outputs challenge messages $m_0, m_1$.
        \item The challenger samples a challenge bit $b \samp \zo$ and prepares $ct \samp \mathsf{Enc}(pk, m_b)$.
        \item The adversary receives $ct$, and it makes polynomially many functional key queries.
        \item The adversary outputs a guess $b'$.
        \item The challenger checks if $f(m_0) = f(m_1)$ for all $f$ queried by the adversary. If not, it outputs $0$ and terminates.
        \item The challenger outputs $1$ if $b' = b$.
    \end{enumerate}
    We require that for any QPT adversary $\adve$,
    $$
    \Pr[\mathsf{FE-IND}(\lambda, \adve) = 1] \leq \frac{1}{2} + \negl(\lambda). 
    $$
    
    If the adversary outputs the challenge messages before the keys are sampled, we call it \emph{selective indistinguishability security}.
\end{definition}

\subsection{Quantum Information Theory}
In this section, we present various technical lemmas regarding quantum information theory.
\begin{lemma}[Almost As Good As New Lemma \cite{aarlemma}, verbatim]\label{lem:asgoodasnew}
    Let $\rho$ be a mixed state acting on $\C^{d}$. Let $U$ be a unitary and $(\Pi_0, \Pi_1 = I - \Pi_0)$ be projectors all acting on $\C^{d} \otimes \C^{d'}$. We interpret $(U, \Pi_0, \Pi_1)$ as a measurement performed by appending an ancillary system of dimension $d'$ in the state $\ketbra{0}{0}$, applying $U$ and then performing the projective measurement $\Pi_0, \Pi_1$ on the larger system. Assuming that the outcome corresponding to $\Pi_0$ has probability $1 - \eps$, we have
    \begin{equation*}
        \trd{\rho}{\rho'} \leq \sqrt{\eps}
    \end{equation*}
    where $\rho'$ is the state after performing the measurement, undoing the unitary $U$ and tracing out the ancillary system.
\end{lemma}

We sometimes also use the following related result.
\begin{lemma}[Gentle Measurement Lemma \cite{Wildelec}]\label{lem:gentlemes}
    Let $E$ be a POVM element and $\rho$ be a state of appropriate dimension. Suppose the outcome $E$ has a high probability of occuring, that is, $\Tr{E\rho} \geq 1 - \eps$. Then, if we apply a canonical implementation, $\sqrt{E}$, of this measurement, the post-measurement state conditioned on this outcome is close to the original state:
    \begin{equation*}
        \trd{\rho}{\frac{\sqrt{E}\rho\sqrt{E}}{\Tr{E\rho}}} \leq \sqrt{\eps}.
   \end{equation*}
\end{lemma}

\begin{theorem}[Implementation Independence of Measurements on Bipartite States]\label{thm:impindep}
    Let $\Lambda = \{M_i\}_{i \in \mathcal{I}}, \Lambda' = \{E_i\}_{i \in \mathcal{I}}$ be two general measurements whose POVMs are equivalent, that is, $M_i^{\dagger}M_i = E_i^{\dagger}E_i$ for all $i \in \mathcal{I}$.

    Let $\rho$ be any bipartite state whose first register has the appropriate dimension for $\Lambda, \Lambda'$. Then, the post-measurement state of the second register conditioned on any outcome $i \in \mathcal{I}$ is the same when either $\Lambda$ or $\Lambda'$ is applied to the first register of $\rho$. That is, \begin{equation*}
        (\Tr\otimes I)\frac{(M_i\otimes I)\rho(M_i^\dagger\otimes I)}{\Tr{(M_i\otimes I)\rho(M_i^\dagger\otimes I)}} =  (\Tr\otimes I)\frac{(E_i\otimes I)\rho(E_i^\dagger\otimes I)}{\Tr{(E_i\otimes I)\rho(E_i^\dagger\otimes I)}}
    \end{equation*}
\end{theorem}
\begin{proof}
See \cref{appn:impindep}.
\end{proof}

\begin{lemma}[Trace Distance Conditioned on Measurement Outcome]\label{lem:postmesdistlem}
        Let $M = \{M_i\}_{i \in \mathcal{I}}$ be a general measurement and $\rho, \sigma$ be two states of appropriate dimension such that $\trd{\rho}{\sigma} \leq \eps$.  Let $p_i$ denote probability of outcome $i$ when $M$ is applied to $\rho$, that is $p_i = \Tr{M_i\rho}$. Then,
        \begin{equation*}
            \trd{\rho'}{\sigma'} \leq \frac{3\eps}{2p_i}
        \end{equation*}
        where $\rho', \sigma'$ are post-measurement states conditioned on outcome $i$ when the measurement $M$ is applied to $\rho, \sigma$, respectively. That is, $\rho' = \frac{M_i\rho M_i^\dagger}{\Tr{M_i\rho M_i^\dagger}}$ and $\sigma' = \frac{M_i\sigma M_i^\dagger}{\Tr{M_i\sigma M_i^\dagger}}$.
\end{lemma}
\begin{proof}
    See \cref{appn:postmesdistlem}.
\end{proof}

\begin{theorem}\label{thm:simulproj}
    Let $\rho$ be a bipartite state and $\Lambda = \{\Pi_1, \dots \}, \Lambda' = \{\Pi'_1, \dots \}$ be two projective measurements over each of these registers, respectively. Suppose \begin{equation*}
        \Tr{\Pi_1\otimes \Pi'_1 \rho} \geq 1 - \eps.
    \end{equation*}
    Let $M = \{M_i\}_{i \in \mathcal{I}}$ be a general measurement over the first register and fix any $i \in \mathcal{I}$. Let $\tau$ denote the post-measurement state of the second register after applying the measurement $M$ on the first register of $\rho$ and conditioned on obtaining outcome $i$. Let $p_i$ denote probability of outcome $i$, that is $p_i = \Tr{(M_i \otimes I)\rho(M_i^\dagger \otimes I)}$. Then, \begin{equation*}
                \Tr{\Pi'_1 \tau} \geq 1 - \frac{3\sqrt{\eps}}{2p_i}.
    \end{equation*}
\end{theorem}
\begin{proof}
See \cref{appn:simulproj}.
\end{proof}

\begin{theorem}[Quantum Union Bound for Commuting Projectors]\label{thm:qub}
    Let $\Pi_1, \dots, \Pi_n$ be a set of commuting projectors. Then, for any state $\rho$ of appropriate dimension,
    \begin{equation*}
        \Tr[(I - \Pi_1\dots\Pi_n) \rho] \leq \sum_{i \in [n]}\Tr[(I - \Pi_i)\rho].
    \end{equation*}
\end{theorem}
\begin{proof}
    While this is a folklore result, we give a proof in \cref{appn:qubproof} for completeness.
\end{proof}

\begin{definition}[Query Algorithm]
    Let $\mathcal{O}$ be a function. A query algorithm $\adve$ with oracle access to $\mathcal{O}$ is defined by an evolution unitary $U_{\adve}$ of $\adve$, and we also define the oracle unitary $U_{\mathcal{O}}$ as $U_{\mathcal{O}}: \ket{w,x,y} \to \ket{w,x,\mathcal{O}(x)\oplus y}$, with the registers ordered as the working register of $\adve$, the query input register and the query output register. $\adve$ is executed by applying $U_\adve$ and then $U_{\mathcal{O}}$ in sequence, e.g., the final state of the algorithm is $(U_{\adve}U_{\mathcal{O}})^T \ket{\psi}$ for an algorithm with initial state $\ket{\psi}$ that makes $T$ queries. If the algorithm has classical output, the output is obtained by measuring (a part of) the working register at the end.
\end{definition}

\begin{theorem}[\cite{bbbv97}]\label{bbbv}

Let $\adve$ be a quantum algorithm making queries to an oracle $\mathcal{O}$. Let $\ket{\psi_t} = \sum_{w,x,y}\alpha_{w,x,y}\ket{w,x,y}$ denote the joint state of the working register, the query input register and the query output register of the algorithm right before the $t$-th query. For a subset $S$ of the domain of $\mathcal{O}$, let $q_S(\ket{\psi_t}) = \sum_{x \in S} |\alpha_{w,x,y}|^2$ and $q_S = \sum_{t} q_s(\ket{\psi_t})$, and call $q_S$ the query weight of $S$.
Let $\mathcal{O}'$ be another oracle whose output differs from $\mathcal{O}$ only on points $x \in S$. Then, if $\adve$ makes $T$ queries to the oracle $\mathcal{O}$ and $S$ is a subset such that $q_S \leq \eps^2/T$, we have $\trd{\ketbra{\psi}{\psi}}{\ketbra{\psi'}{\psi'}} \leq \eps$ where $\ket{\psi},\ket{\psi'}$ denote the final state of the algorithm $\adve$ when given access to the oracles $\mathcal{O}, \mathcal{O}'$ respectively.
\end{theorem}

\subsection{Compute-and-Compare Obfuscation}
In this section, we introduce compute-and-compare obfuscation.
\begin{definition}[Compute-and-compare program]
     Let $f: \zo^{a(\lambda)} \to \zo^{b(\lambda)}$ be a function, $y \in \zo^{b(\lambda)}$ be a target value and $z$ a hidden message. The following program $P$, described by $(f, y, z)$, is called a \emph{compute-and-compare program.}

    \paragraph{$P(x):$}
    Compute $f(x)$ and compare it to $y$. If they are equal, output $z$. Otherwise, output $\perp$.
\end{definition}

We say that a distribution $\mathcal{D}$ of such programs (along with quantum auxiliary information $\regi{aux}$) is sub-exponentially unpredictable if for any QPT adversary, given the auxiliary information $\regi{aux}$ and the description of $f$, the adversary can predict the target value $y$ with at most subexponential probability.

\begin{definition}\label{defn:ccobf}
A compute-and-compare obfuscation scheme for a class of distributions  consists of efficient algorithms $\mathsf{CCObf.Obf}$ and $\mathsf{CCObf.Sim}$ that satisfy the following. Consider any distribution $\mathcal{D}$ over compute-and-compare programs, along with quantum auxiliary input, in this class.

\paragraph{Correctness.}
For any function $(f, y, z)$ in the support of $\mathcal{D}$,
$\Pr[\forall x ~ D'(x) = D(x) : D' \samp \mathsf{CCObf.Obf}(f, y, z)] \geq 1 - \negl(\lambda)$.

\paragraph{Security}
$(\mathsf{CCObf.Obf}(f, y, z), \regi{aux}) \approx (\mathsf{CCObf.Sim}(1^\lambda, |f|, |y|, |z|), \regi{aux})$ where $(f, y, z), \regi{aux} \samp \mathcal{D}(1^\lambda)$.
\end{definition}

\begin{theorem}[\cite{CLLZ21, WZ17}]\label{thm:ccobf}
    Assuming the existence of post-quantum $\io$ and LWE, there exists compute-and-compare obfuscation for any class of sub-exponentially unpredictable distributions.

    Assuming the existence of subexponentially secure $\io$ and LWE against subexponential time quantum adversaries, there exists subexponentially secure compute-and-compare obfuscation against subexponential time adversaries for any class of sub-exponentially unpredictable distributions.\footnote{The original result is only for polynomial hardness against QPT adversaries, but it is easy to see that it also holds in the subexponential setting.}
\end{theorem}

\section{Projective and Threshold Implementations}\label{sec:piti}
In this section, we introduce the notion of projective and threshold implementations \cite{Z20, ALLZZCONF20}, along with their efficient versions; which are tools we use in our security proofs. Then, we recall some properties from previous work and also show some new technical results that will be needed in the security proofs of our schemes.
\begin{definition}[Shift Distance \cite{Z20}]
Let $\mathcal{D}_0, \mathcal{D}_1$ be two distributions over $\R_{\geq 0}$. The shift distance with parameter $\eps \geq 0$ between $\mathcal{D}_0, \mathcal{D}_1$, denoted $\shiftd^\eps(\mathcal{D}_0, \mathcal{D}_1)$, is defined to be
\begin{equation*}
\inf\left\{\delta\in \R_{\geq 0}:  \forall x \in \R_{\geq 0} \Pr_{a \samp \mathcal{D}_0}[ a \leq x] \leq \Pr_{a \samp \mathcal{D}_1}[a \leq x + \eps] + \delta. \right\}
\end{equation*}

We define the shift distance between two measurements $\mathcal{M}_0, \mathcal{M}_1$ over the same space $\Hilbert$ to be
$$
\shiftd^\eps(\mathcal{M}_0, \mathcal{M}_1) = \sup_{\ket{\psi} \in \Hilbert} \shiftd(\mathcal{M}_0\ket{\psi}, \mathcal{M}_1\ket{\psi}).
$$
\end{definition}
\begin{definition}[$(\eps,\delta)$-Almost Projective \cite{Z20}]
    Let $\Lambda$ be a measurement with index set $\mathcal{I} \subseteq \R$. $\Lambda$ is called \emph{$(\eps,\delta)$-almost projective} if the following is satisfied for all states $\rho$ of appropriate dimension. Apply $\Lambda$ to $\rho$ to obtain an outcome $x$ and then apply $\Lambda$ again to the post-measurement state to obtain an outcome $x'$. Then, $\Pr[|x - x'| \leq \eps] \geq 1 -\delta$.
\end{definition}

\begin{theorem}[Projective Implementation \cite{Z20}]\label{thm:piprop}
Let $\mathcal{E} = \{\mathcal{E}_1, \mathcal{E}_0 = I - \mathcal{E}_1\}$ be a binary POVM. Then, there exists a projective measurement, called \emph{projective implementation} of $\mathcal{E}$ and denoted $\mathsf{PI}(\mathcal{E})$, indexed by a finite set consisting of elements in $[0, 1]$ and it satisfies the following. 
For any state $\rho$ of appropriate dimension, the following experiment has the same distribution as the outcome of applying $\mathcal{E}$ to $\rho$.
\begin{enumerate}
    \item Apply $\mathsf{PI}(\mathcal{E})$ to $\rho$ to obtain an outcome $p$.
    \item Output $1$ with probability $p$ and $0$ otherwise.
\end{enumerate}

Since $\mathsf{PI}(\mathcal{E})$ is projective, if the outcome of applying it to a state is $p$, then applying it again to the post-measurement state gives outcome $p$ with probability $1$.
\end{theorem}

Below, we will consider measurements that are defined as mixtures of projective measurements. For a collection of binary projective measurements $\mathcal{P} = \{P_i, I-P_i\}_{i \in \mathcal{I}}$ and a distribution $\mathcal{D}$ over $\mathcal{I}$, we will write $\mathcal{P}_{\mathcal{D}}$ to denote the measurement where we sample $i \samp \mathcal{D}$ and apply the projective measurement $\{P_i, I-P_i\}$. In general, projective implementation of a mixture of projective measurements can be of exponential size, but it can be efficiently approximated.

\begin{theorem}[Approximate Projective Implementation \cite{Z20}]\label{thm:apiprop}
Let $\mathcal{P} = \{P_i, I-P_i\}_{i \in \mathcal{I}}$ be a collection of binary projective measurements with index set $\mathcal{I}$ and $\mathcal{D}$ be a distribution over $\mathcal{I}$. Suppose we can efficiently implement the measurement $$\Lambda = \left\{\sum_{i \in \mathcal{I}} \ketbra{i}{i}\otimes P_i, I - {\sum_{i \in \mathcal{I}} \ketbra{i}{i}\otimes P_i}\right\}.$$
Then, for $0 < \eps, \delta \leq 1$, there exists a measurement, called \emph{approximate projective implementation} of $\mathcal{P}_{\mathcal{D}}$ and denoted $\mathsf{API}^{\eps,\delta}_{\mathcal{P}, \mathcal{D}}$, that satisfies the following.
\begin{itemize}
    \item $\mathsf{API}^{\eps,\delta}_{\mathcal{P}, \mathcal{D}}$ is $(\eps,\delta)$-almost projective.
    \item $\shiftd^\eps(\mathsf{API}^{\eps,\delta}_{\mathcal{P}, \mathcal{D}}, \mathsf{PI}(\mathcal{P}_{\mathcal{D}})) \leq \delta$.
    \item Expected run time of $\mathsf{API}^{\eps,\delta}_{\mathcal{P}, \mathcal{D}}$ is polynomial in $1/\eps, \log(1/\delta)$ and the runtimes of $\{P_i, I - P_i\}$, $\mathcal{D}$ and the procedure mapping $i$ to $\{P_i, I - P_i\}$.
 \end{itemize}

\end{theorem}

\begin{theorem}[\protect{\cite[Theorem~6.5]{Z20}}]
        Let $\mathcal{D}_b$ for $b \in \zo$ be efficient distributions over the same support with classical output and $\rho$ be an efficiently constructible state. Let $\mathcal{P}$ be a collection of projective measurements indexed by the support of $\mathcal{D}_b$, and consider the mixture of measurements $\mathcal{P}_{\mathcal{D}_b}$ where we sample a measurement according to $\mathcal{D}_b$ and apply it.
    Suppose $\mathcal{D}_0 \approx \mathcal{D}_1$. Then, for any inverse polynomial $\eps$,
    \begin{equation*}
        \shiftd^{\eps}(\mathsf{PI}(\mathcal{P}_{\mathcal{D}_0})\cdot\rho, 
\mathsf{PI}(\mathcal{P}_{\mathcal{D}_1})\cdot\rho) \leq \negl(\lambda).
 \end{equation*}
\end{theorem}
We give the following generalization which differs from the previous theorem in a couple of aspects. First, we consider measurements over multiple registers. Second, we allow the measured state and the measurement to be correlated, which will be needed in our copy-protection proofs. Finally, we give more fine-grained results in terms of adversary's advantage and runtime, which will again be needed in our proofs.
\begin{theorem}\label{thm:distti}
    Let $\lambda$ denote the security parameter and let $k(\lambda)$ be a polynomial, $\eps(\lambda)$ an inverse polynomial and $\delta(\lambda)$ be an inverse exponential. 
    
    Let $\mathcal{S}^b$ and $\{\mathcal{B}^b_\ell\}_{\ell \in [k(\lambda)]}$ for each $b \in \zo$ be efficient distributions as follows.  $\mathcal{S}^b$ outputs a $k$-partite state and a classical string $pp$, while $\mathcal{B}^b_\ell$ take in $pp$ and are classical. For each $\ell \in [k(\lambda)]$, consider the output distribution of the following experiment, denoted by $(\mathcal{S}^b, \mathcal{B}^b_\ell)$.
    \begin{enumerate}
        \item $\rho, pp \samp \mathcal{S}^b(1^\lambda)$.
        \item Sample $s \samp \mathcal{B}^b_{\ell}(pp)$.
        \item Output $(\rho, s, pp)$.
    \end{enumerate}

Let $\mathcal{P}_\ell$ for each $\ell \in [k]$ be a collection of binary projective measurements indexed by output space of $\mathcal{B}_\ell^b$. For each fixed value of $pp$, consider the mixture of measurements, denoted $\mathcal{P}_{\ell, \mathcal{B}^b_\ell(pp)}$, where we sample a measurement $s$ from $\mathcal{P}_\ell$ as $s = \mathcal{B}^b_\ell(pp; r)$ where $r \samp \mathcal{R}$ and apply it. Suppose we can efficiently apply the above measurement for arbitrary given superpositions of $r$ values. Let $\mathsf{API}^{\eps, \delta}(\mathcal{P}_{\ell, \mathcal{B}^b_\ell(pp)})$ denote the approximate projective implementation of this mixture and let $\Vec{p}_b$ be a tuple consisting of the outcomes of the following experiment.

    \begin{enumerate}
        \item $\rho, pp \samp \mathcal{S}^b(1^\lambda)$.
        \item Apply $\otimes_{\ell\in[k(\lambda)]}\mathsf{API}^{\eps, \delta}(\mathcal{P}_{\ell, \mathcal{B}^b_\ell(pp)})$ on $\rho$.
    \end{enumerate}

Then,
\begin{itemize}
    \item Suppose $(\mathcal{S}^0, \mathcal{B}_\ell^0) \approx (\mathcal{S}^1, \mathcal{B}_\ell^1)$ for each $\ell \in [k]$. Then, 
    \begin{equation*}
        \statdist{\Vec{p}_0}{\Vec{p}_1} \leq \negl(\lambda).    
 \end{equation*}
 \item Suppose $(\mathcal{S}^0, \mathcal{B}_\ell^0) \approx^c_{\nu(\lambda)} (\mathcal{S}^1, \mathcal{B}_\ell^1)$ for all $(\frac{k(\lambda)}{\mu^2(\lambda)}\cdot\poly(\lambda))$-time adversaries for each $\ell \in [k]$ for some $\nu, \mu$ satisfying $\nu(\lambda) < \mu^2(\lambda)\poly(\lambda)$. Then, 
    \begin{equation*}
        \statdist{\Vec{p}_0}{\Vec{p}_1} \leq \mu(\lambda).    
 \end{equation*}
\end{itemize}

\end{theorem}
\begin{proof}
See \cref{appn:proofdistti}.
\end{proof}

Now, we reproduce the results of \cite{ALLZZCONF20} regarding \emph{threshold implementations}.
\begin{theorem}[Threshold Implementation \cite{ALLZZCONF20}]\label{thm:singleatiprop}
    Consider the following measurement, denoted $\mathsf{ATI}^{\eps,\delta}_{\mathcal{P},\mathcal{D}, \eta}$, associated with a collection of projective measurements $\mathcal{P}$, a distribution $\mathcal{D}$ over the index set of $\mathcal{P}$ and a threshold value $\eta \in [0, 1]$, applied to a state $\rho$.

    \begin{enumerate}
        \item Apply $\mathsf{API}^{\eps,\delta}_{\mathcal{P}, \mathcal{D}}$ to $\rho$, let $p$ be the outcome.
        \item Outcome $1$ if and only if $p \geq \eta$.
    \end{enumerate}

We denote by $\Tr[\mathsf{ATI}^{\eps,\delta}_{\mathcal{P},\mathcal{D},\eta}\cdot\rho]$ the probability that the outcome above is $1$. If $\mathsf{API}^{\eps,\delta}_{\mathcal{P},\mathcal{D}}$ is replaced with $\mathsf{PI}({\mathcal{P}_{\mathcal{D}}})$, then we denote the resulting measurement as $\mathsf{TI}_\eta(\mathcal{P}_\mathcal{D})$ and write $\Tr[\mathsf{TI}_\eta(\mathcal{P}_\mathcal{D})\cdot\rho]$ to denote the probability that the outcome is $1$.

We then have the following.
\begin{itemize}
\item For any state $\rho$, 
\begin{equation*}
   \Tr[\mathsf{ATI}^{\eps,\delta}_{\mathcal{P},\mathcal{D}, \eta - \eps}\cdot\rho] \geq \Tr[\mathsf{TI}_\eta(\mathcal{P}_\mathcal{D})\cdot\rho] - \delta.
\end{equation*}
\item For any state $\rho$, 
\begin{equation*}
   \Tr[\mathsf{TI}_{\eta-\eps}(\mathcal{P}_\mathcal{D})\cdot\rho] \geq \Tr[\mathsf{ATI}^{\eps,\delta}_{\mathcal{P},\mathcal{D}, \eta}\cdot\rho] - \delta.
\end{equation*}
 \item $\mathsf{ATI}^{\eps,\delta}_{\mathcal{P},\mathcal{D}, \eta}$ is efficient whenever $\mathsf{API}^{\eps,\delta}_{\mathcal{P},\mathcal{D}}$ is.
    \item $\mathsf{TI}_{\mathcal{P},\mathcal{D},\eta}$ is a projection and the collapsed state conditioned on outcome $1$ is a mixture of eigenvectors of $\mathcal{D}$ with eigenvalue $\geq \eta$.
\end{itemize}
\end{theorem}
We give some further generalizations below.
\begin{theorem}\label{thm:multiatiprop}
 For any $k \in \mathbb{N}$, let $\mathcal{P}_\ell, \mathcal{D}_\ell$ be a collection of projective measurements and a distribution on the index set of this collection, respectively, and $\eta_\ell \in [0, 1]$ be threshold values for all $\ell \in [k]$. Write $\Tr[\left(\bigotimes_{\ell \in [k]} \mathsf{ATI}^{\eps,\delta}_{\mathcal{P}_\ell,\mathcal{D}_\ell,\eta_\ell}\right)\cdot \rho]$ to denote the probability that the outcome of the joint measurement $\bigotimes_{\ell \in [k]} \mathsf{ATI}^{\eps,\delta}_{\mathcal{P}_\ell,\mathcal{D}_\ell,\eta_\ell}$ applied on $\rho$ is all $1$, and similarly for $\mathsf{TI}$.

 Then, we have the following.
\begin{itemize}

\item  \cite[Corollary~3]{ALLZZ20}  For any $k$-partite state $\rho$, 
\begin{equation*}
   \Tr[\left(\bigotimes_{\ell \in [k]} \mathsf{ATI}^{\eps,\delta}_{\mathcal{P}_\ell, \mathcal{D}_\ell, \eta_\ell - \eps}\right)\rho] \geq \Tr[\left(\bigotimes_{\ell \in [k]} \mathsf{TI}_{\eta_\ell}({\mathcal{P}_\ell}_{\mathcal{D}_\ell})\right)\rho] - k\cdot \delta.
\end{equation*}
\item  \cite[Corollary~3]{ALLZZ20} For any $k$-partite state $\rho$, let $\rho'$ be the collapsed state obtained after applying $\bigotimes_{\ell \in [k]} \mathsf{ATI}^{\eps,\delta}_{\mathcal{P}_\ell,\mathcal{D}_\ell,\eta_\ell}$ to $\rho$ and obtaining the outcome $1$. Then, 
\begin{equation*}
   \Tr[\left(\bigotimes_{\ell \in [k]}\mathsf{TI}_{\eta_\ell - 2\eps}({\mathcal{P}_\ell}_{\mathcal{D}_\ell})\right)\rho'] \geq 1 - 2k\cdot \delta.
\end{equation*}

\item  For any $k$-partite state $\rho$, let $\rho'$ be the collapsed state obtained after applying $\bigotimes_{\ell \in [k]} \mathsf{ATI}^{\eps,\delta}_{\mathcal{P}_\ell, \mathcal{D}_\ell,\eta}$ to $\rho$ and obtaining the outcome $1$. Then, 
\begin{equation*}
   \Tr[\left(\bigotimes_{\ell \in [k]} \mathsf{ATI}^{\eps,\delta}_{\mathcal{P}_\ell, \mathcal{D}_\ell, \eta_\ell - 3\eps}\right)\cdot\rho'] \geq 1 - 3k\cdot \delta.
\end{equation*}

\item  For any $k$-partite state $\rho$, 
\begin{equation*}
  \Tr[\left(\bigotimes_{\ell \in [k]} \mathsf{TI}_{\eta_\ell - \eps}({\mathcal{P}_\ell}_{\mathcal{D}_\ell})\right)\rho] \geq  \Tr[\left(\bigotimes_{\ell \in [k]} \mathsf{ATI}^{\eps,\delta}_{\mathcal{P}_\ell, \mathcal{D}_\ell, \eta_\ell}\right)\rho] - k\cdot \delta.
\end{equation*}
   
\end{itemize}
\end{theorem}
\begin{proof}
See \cref{appn:multiatiprop}.
\end{proof}

\begin{theorem}\label{thm:multiapiprop}
For any $k \in \mathbb{N}$, let $\mathcal{P}_\ell, \mathcal{D}_\ell$ be a collection of projective measurements and a distribution on the index set of this collection, respectively. Let $\rho$ be any $k$-partite state of appropriate dimension. Consider the measurement outcome $\vec{p}$ and the post-measurement state $\rho'$ obtained by applying $\left(\bigotimes_{\ell \in [k]} \mathsf{API}^{\eps,\delta}_{\mathcal{P}_\ell,\mathcal{D}_\ell}\right)$ to a state $\rho$. Let $\vec{p'}$ be the measurement outcome obtained by applying the measurement $\left(\bigotimes_{\ell \in [k]} \mathsf{PI}({\mathcal{P}_\ell}_{\mathcal{D}_\ell})\right)$ to $\rho'$. Then,
\begin{align*}
    \Pr[\forall \ell \in [k]~~(\vec{p'})_\ell \leq (\vec{p})_\ell + 2\eps] \geq 1 - 2\cdot k \cdot \delta \\
    \Pr[\forall \ell \in [k]~~(\vec{p'})_\ell \geq (\vec{p})_\ell - 2\eps] \geq 1 - 2\cdot k \cdot \delta.
\end{align*}
\end{theorem}

\begin{proof}
See \cref{appn:multiatiprop}.
\end{proof}

\begin{theorem}\label{thm:multiapismallproof}
For any $k \in \mathbb{N}$, let $\mathcal{P}_\ell, \mathcal{D}_\ell$ be a collection of projective measurements and a distribution on the index set of this collection, respectively. Let $\rho$ be any $k$-partite state of appropriate dimension. Consider the measurement outcome $\vec{p}$ obtained by applying $\left(\bigotimes_{\ell \in [k]} \mathsf{API}^{\eps,\delta}_{\mathcal{P}_\ell,\mathcal{D}_\ell}\right)$ to a state $\rho$. Let $\vec{p'}$ be the measurement outcome obtained by applying the measurement $\left(\bigotimes_{\ell \in [k]} \mathsf{PI}({\mathcal{P}_\ell}_{\mathcal{D}_\ell})\right)$ to $\rho$. Then,
\begin{align*}
    \Pr[\forall \ell \in [k]~~(\vec{p'})_\ell \leq \eta_\ell + \eps] \geq \Pr[\forall \ell \in [k]~~(\vec{p})_\ell \leq \eta_\ell] - k \cdot \delta \\
    \Pr[\forall \ell \in [k]~~(\vec{p})_\ell \leq \eta_\ell + \eps] \geq \Pr[\forall \ell \in [k]~~(\vec{p'})_\ell \leq \eta_\ell] - k \cdot \delta.
\end{align*}
\end{theorem}
\begin{proof}
See \cref{appn:multiatiprop}.
\end{proof}

\section{Coset States}\label{sec:coset}
In this section, we start by giving the definition of coset states  \cite{CLLZ21}, which are the states we use in our constructions, and then state the monogamy-of-entanglement property they satisfy. Then, we define two new security games for coset states that streamlines our proofs later on and show secure constructions for these games.
\begin{definition}[Coset States \cite{CLLZ21}]
Let $A$ be a subspace of $\F^n_2$ and $s, s'$ be vectors in $\F^n_2$. We define the coset state associated with $A, s, s'$, denoted $\ket{A_{s, s'}}$, to be
\begin{equation*}
    \ket{A_{s, s'}} = \sum_{a \in A} \frac{1}{\sqrt{|A|}} (-1)^{\langle s', a \rangle} \ket{a + s}.
\end{equation*}
\end{definition}
We will write $A + s$ to denote both the coset $A + s$ and the program that takes as input a vector $v \in \F_2^n$ and outputs $\true$ if and only if $v \in A + s$, and $0$ otherwise. The distinction will be clear from context.
\begin{theorem}[\cite{CLLZ21}]
    Consider a subspace $A \subseteq \F^n_2$ and vectors $s, s' \in \F^n_2$.
    \begin{enumerate}
    \item There exists an efficient quantum algorithm that outputs $\ket{A_{s, s'}}$ given $s, s'$ and the description of $A$.
        \item $H^{\otimes n}\ket{A_{s, s'}} = \ket{(A^{\perp})_{s', s}}$.
        \item We define the canonical element of a coset $A + v$ to be the lexicographically smallest element in the coset and denote it $\mathsf{Can}_A(v)$. There exists an efficient classical algorithm that, on input the description of $A$ and a vector $v$, outputs $\mathsf{Can}_A(v)$.
    \end{enumerate}
\end{theorem}
Coset states satisfy a natural monogamy-of-entanglement (\emph{MoE}) property where any adversary can win the following game with only negligible probability: We present the adversary with a coset state, and it is required to split the state into two (possibly entangled) registers that can be used to simultaneously output vectors in the cosets $A + s$ and $A^{\perp} + s'$ respectively.

\begin{theorem}[Monogamy-of-Entanglement Property for Coset States \cite{CLLZ21, CV22}]\label{defn:strmoeorig}
Consider the following game between the challenger and an adversary tuple $\adve = (\adve_0, \adve_1, \adve_2)$.
\paragraph{\underline{$\moe(\lambda, \adve)$}}
\begin{enumerate}
        \item  Sample uniformly at random a subspace $A$ of $\F_2^\lambda$ of dimension $\frac{\lambda}{2}$ and two elements $s, s' \samp \F_2^\lambda$. 
        \item Sample $\mathsf{OP}^0 \samp \io(A + s)$ and $\mathsf{OP}^1 \samp \io(A^{\perp} + s')$.
    \item Submit $\ket{A_{s, s'}}, \mathsf{OP}^0 , \mathsf{OP}^1$ to $\adve_0$.
    \item $\adve$ outputs two (possibly entangled) registers $\reg_1, \reg_2$.
    \item For $j \in \{1, 2\}$, run $v_{j} \samp \adve_j(\reg_j, A)$.
    \item Output $1$ if and only if $v_1 \in A + s$ and $v_2 \in A^{\perp} + s'$.
\end{enumerate}

Assuming the existence of $\io$ and one-way functions, then for any QPT adversary tuple $\adve$, \begin{equation*}
    \Pr[\moe(\lambda, \adve) = 1] \leq \negl(\lambda).
\end{equation*}

If we assume the existence of subexponentially-secure $\io$ and one-way functions, then there exists a constant $\constmoe > 0$ such that for any QPT adversary tuple 
 \begin{equation*}
    \Pr[\moe(\lambda, \adve) = 1] \leq 2^{-\lambda^{\constmoe}}
\end{equation*}
for all sufficiently large $\lambda$.
\end{theorem}

In the previous constructions of unclonable primitives \cite{CLLZ21, LLQZ22}, and also in our constructions, the security of the unclonable schemes rely on requiring the \emph{freeloader} adversaries to output a vector from either $A + s$ or $A^\perp + s'$, depending on a random challenge bit presented to them. However, the \emph{pirate} (i.e. splitting) adversary can always guess this challenge bit and measure the coset state accordingly before splitting into freeloaders, and it would be right for both freeloaders with probability $(1/2)^2$. Therefore, to achieve negligible or subexponential security, we amplify the security by using multiple coset states. This variant of the game is implicitly used in \cite{CLLZ21, LLQZ22} and a similar amplification theorem for a related game is also formally proven in \cite{CGLR23}; both for the case of uniformly random challenge strings. We generalize the amplification result to any unpredictable distribution of challenge strings.

\begin{definition}[Unpredictable Distribution]
    Let $\mathcal{D} = \{\mathcal{D}_\lambda\}_\lambda$ be a family of distributions over $\zo^{a(\lambda)}$ where $a(\cdot)$ is some polynomial and we write $\mathcal{D}(x)$ to mean $\Pr_{x' \samp \mathcal{D}}[x' = x]$. Then, $\mathcal{D}$ is said to be (statistically) \emph{unpredictable} if $\max_{x \in \zo^{a(\lambda)}} \{\mathcal{D}(x)\} \leq \negl(\lambda)$.
\end{definition}
\begin{definition}\label{defn:cosetgen}
Define $\mathsf{CosetGen}$ to be the following algorithm, where we have the default parameter values $\morelocalcosetcount = \localcosetcount$ and $\kappa(\lambda) = {\lambda^{\left\lceil3/\constmoe\right\rceil}}$.

\begin{mdframed}
        {\bf $\underline{\mathsf{CosetGen}(1^\lambda, \morelocalcosetcount = \localcosetcount, \kappa(\lambda) =   {\lambda^{\left\lceil3/\constmoe\right\rceil}})}$}
    \begin{enumerate}
    \item For $i \in [\morelocalcosetcount]$, sample uniformly at random a subspace $A_i$ of $\F_2^{\kappa(\lambda)}$ of dimension $\kappa(\lambda)/{2}$ and two elements $s_i, s_i' \samp \F_2^{\kappa(\lambda)}$. 
    \item Output $(A_i, s_i, s_i')_{i \in [\morelocalcosetcount]}$.
\end{enumerate}
    \end{mdframed}

We call the output of $\mathsf{CosetGen}$ a coset tuple.
\end{definition}
\begin{theorem}[Monogamy-of-Entanglement Property for Coset States - Multi-Challenge Version]\label{defn:strmoe}
Let $\mathcal{D}$ be a distribution over $\zo^{\morelocalcosetcount}$ that is unpredictable.
Consider the following game between the challenger and an adversary tuple $\adve = (\adve_0, \adve_1, \adve_2)$.
\paragraph{\underline{$\moeflip(\lambda, \adve)$}}
\begin{enumerate}
    \item $(A_i, s_i, s_i')_{i \in [\morelocalcosetcount]} \samp \mathsf{CosetGen}(1^\lambda, \morelocalcosetcount, \kappa(\lambda))$.
     \item For $i \in [\morelocalcosetcount]$, 
    \begin{enumerate}[label*=\arabic*.]
    \item Sample $\mathsf{OP}^0_i \samp \io(A_i + s_i)$.
    \item Sample $\mathsf{OP}^1_i \samp \io(A^{\perp}_i + s'_i)$.
    \end{enumerate}
    \item Submit $\left\{\ket{A_{i, s_i, s'_i}}\right\}_{i \in [\morelocalcosetcount]}, (\mathsf{OP}_i^0, \mathsf{OP}_i^1)_{i \in [\morelocalcosetcount]}$ to $\adve_0$.
    \item $\adve$ outputs two (possibly entangled) registers $\reg_1, \reg_2$.
    \item Challenger samples $r_1 \samp \mathcal{D}$ and $r_2 \samp \mathcal{D}$.
    \item For $\ell \in \{1, 2\}$, run $(v_{\ell,i})_{i \in [\morelocalcosetcount]} \samp \adve_j(\reg_j, r_j, (A_i)_{i \in [\morelocalcosetcount]})$.
    \item For $\ell \in \{1, 2\}$ and all $i \in [\morelocalcosetcount]$, check if $v_{\ell ,i} \in A_i + s_i$ if $(r_\ell )_i = 0$ and if $v_{\ell ,i} \in A^{\perp}_i + s'_i$ if $(r_\ell )_i = 1$. Output $1$ if and only if all the checks pass. Otherwise, output $0$.
\end{enumerate}

Assuming the existence of $\io$ and one-way functions, and setting $\kappa(\lambda) = \lambda$, then for any QPT adversary tuple $\adve$, \begin{equation*}
    \Pr[\moeflip(\lambda, \adve) = 1] \leq \negl(\lambda).
\end{equation*}

If we assume the existence of subexponentially-secure $\io$ and one-way functions, and set $\mathcal{D}$ to be a distribution that is subexponentially unpredictable, then there exists a constant $\constmoeflip$ such that for any QPT adversary tuple 
 \begin{equation*}
    \Pr[\moeflip(\lambda, \adve) = 1] \leq 2^{-\lambda^{\constmoeflip}}
\end{equation*}
for all sufficiently large $\lambda$.
By setting $\mathcal{D}$ to be $2^{-\localcosetcount}$ unpredictable and $\kappa(\lambda) ={\lambda^{\left\lceil3/\constmoe\right\rceil}}$, there exists such $\constmoeflip > 2$.
\end{theorem}
\begin{proof}
Suppose there exists an QPT adversary tuple $\adve = (\adve_0, \adve_1, \adve_2)$ that wins $\moeflip$ above with probability $\eps(\lambda)$, that is, $\Pr[\moeflip(\lambda,\adve) = 1] \geq \eps(\lambda)$. 

Consider the following adversary $\adve' = (\adve'_0, \adve'_1, \adve'_2)$ for $\moe$. 

\paragraph{\underline{$\adve^{'}_0$}\\}
On input a state $\rho$ and the obfuscated programs $\mathsf{OP}^{0*}, \mathsf{OP}^{1*}$, sample $r_1 \samp \mathcal{D}$ and $r_2 \samp \mathcal{D}$. If $r_1 = r_2$, then abort. Let $j^*$ be an index where $r_1$ and $r_2$ differ, that is, $(r_1)_{j^*} \neq (r_2)_{j^*}$. 
For all $j \in [\morelocalcosetcount] \setminus \{j^*\}$, sample a subspace $A_j$, elements $s_j, s_j' \samp \F^{\kappa(\lambda)}_2$, then set $\rho_j = \ket{A_{j, s_j, s'_j}}$ and sample $\mathsf{OP}^0_j \samp \io(A_j + s_j)$ and $\mathsf{OP}^1_j \samp \io(A^{\perp}_j + s'_j,)$. 
    Then, run $\adve((\rho_j, \mathsf{OP}^0_j, \mathsf{OP}^1_j)_{j \in [\morelocalcosetcount]})$ to obtain a bipartite state $\sigma$.
    Finally, output 
    \begin{equation*}
((\sigma[1], (A_j)_{j \in [\morelocalcosetcount] \setminus \{j^*\}}, j^*, r_1), (\sigma[2], (A_j)_{j \in [\morelocalcosetcount] \setminus \{j^*\}}, j^*, r_2)).        
    \end{equation*}
if $(r_1)_{j^*} = 0$ and $(r_2)_{j^*} = 1$, or 
\begin{equation*}
((\sigma[2], (A_j)_{j \in [\morelocalcosetcount] \setminus \{j^*\}}, j^*, r_2), (\sigma[1], (A_j)_{j \in [\morelocalcosetcount] \setminus \{j^*\}}, j^*, r_1)).        
    \end{equation*}
    if $(r_1)_{j^*} = 1$ and $(r_2)_{j^*} = 0$.
    \paragraph{\underline{$\adve'_\ell$ for $\ell \in \{1, 2\}$}\\}
    $\adve'_\ell$ runs $\adve_\ell$ on its own input and the subspace description $A$ it obtains from the challenger of $\moeflip$. Note that $\adve'_\ell$ can correctly rearrange the input order when passing it to $\adve_\ell$ since it knows $j^*$. Finally, it outputs the $j^*$-th vector in the output of $\adve_\ell$.

It is easy to see that the adversary $\adve'$ wins whenever it does not abort and the vectors output by $\adve$ are correct. Let $p_{s_1, s_2}$ denote the probability of $\adve$ winning $\moeflip$ conditioned on $r_1 = s_1$ and $r_2 = s_2$. Then, we have 
\begin{equation*}
    \eps(\lambda) = \sum_{s_1, s_2 \in \zo^{\morelocalcosetcount}} \mathcal{D}(s_1)\cdot \mathcal{D}(s_2) p_{s_1,s_2}
\end{equation*}
and
\begin{align*}
    \Pr[\moe(\lambda, \adve')] &= \sum_{s_1 \neq s_2 \in \zo^{\morelocalcosetcount}} \mathcal{D}(s_1)\cdot \mathcal{D}(s_2)\cdot p_{s_1,s_2} 
    \\&=  \eps(\lambda) -  \sum_{s \in \zo^{\morelocalcosetcount}}(\mathcal{D}(s))^2 \cdot p_{s,s}
    \\&\geq \eps(\lambda) -  \sum_{s \in \zo^{\morelocalcosetcount}}(\mathcal{D}(s))^2
    \\&\geq \eps(\lambda) -  \left[\sum_{s \in \zo^{\morelocalcosetcount}}(\mathcal{D}(s))\right]\cdot\max_{s \in \zo^{\morelocalcosetcount}}\{\mathcal{D}(s)\}
    \\&\geq \eps(\lambda) -  \max_{s \in \zo^{\morelocalcosetcount}}\{\mathcal{D}(s)\}.
\end{align*}

For the case of negligible security, we will have $\eps(\lambda) > \frac{1}{\poly(\lambda)}$ and $\max_{s \in \zo^{\morelocalcosetcount}}\{\mathcal{D}(s)\} = \negl(\lambda)$ since $\mathcal{D}$ is unpredictable, hence we get $\Pr[\moe(\lambda, \adve')] > \frac{1}{\poly(\lambda)}$, which is a contradiction by \cref{defn:strmoeorig}. The subexponential cases follow by similar calculations.
\end{proof}

Finally, we introduce another variant of the game that is useful for our unbounded collusion secure constructions. In this game, the adversary queries multiple coset state tuples (which are associated with \emph{identity strings}) that are generated pseudorandomly, and it is allowed to choose the coset state tuple for which it wants to \textit{break} the monogamy-of-entanglement property. The adversary is also presented with an (obfuscated) program that allows it to make membership queries for any coset tuple by specifying its identity.

\begin{theorem}[Monogamy-of-Entanglement Property for Coset States - Collusion-Resistant Version]\label{defn:strmoecoll}
Let $\idlen(\lambda)$ be a polynomial, denoting the length of the identity strings. Define $\cosettcount = 3\cdot(\idlen(\lambda) + \lambda)^3$. Let $\mathcal{D}$ be a distribution over $\zo^{\cosettcount}$. Consider the following game between the challenger and an adversary tuple $\adve = (\adve_0, \adve_1, \adve_2)$.
\paragraph{\underline{$\moecoll(\lambda, \idlen(\lambda), \adve)$}}
\begin{enumerate}
    \item The challenger initializes the list $\idset = [~]$. 
    
    \item The challenger samples a PRF key $K \samp F.\mathsf{KeyGen}(1^{\lambda})$. 

    \item The challenger samples $\mathsf{OPMem} \samp \io(\mathsf{PMem}_{K})$, where $\mathsf{PMem}_{K}$ is the following program.
    
\begin{mdframed}
        {\bf $\underline{\mathsf{PMem}_{K}(id, u_1, \dots, u_{\cosettcount}, r)}$}
        
        {\bf Hardcoded: $K$}
        \begin{enumerate}[label=\arabic*.]
            \item $(A_i, s_i, s_i')_{i \in [\cosettcount]} \samp \mathsf{CosetGen}(\cosetgenparam; F(K, id))$.
            \item For each $i \in [\cosettcount]$, check if $u_i \in A_i + s_i$ if $(r)_i = 0$ and check if $u_i \in A^{\perp}_i + s'_i$ if $(r)_i = 1$. If any of the checks fail, output $0$ and terminate.
            \item Output $1$.
        \end{enumerate}

    \end{mdframed}

    \item The challenger submits $\mathsf{OPMem}$ to the adversary.

    \item \textbf{\underline{Query Phase 1:}} For polynomially many rounds, the adversary makes queries as follows. The adversary submits an identity string $id \in \zo^{\idlen(\lambda)}$ to the challenger. Then, the challenger adds $id$ to the list $\idset$, samples  $(A_i, s_i, s_i')_{i \in [\cosettcount]} \samp \mathsf{CosetGen}(\cosetgenparam; F(K, id))$ and submits the state $\left\{\ket{A_{i, s_i, s'_i}}\right\}_{i \in [\cosettcount]}$ to the adversary.

    \item \textbf{\underline{Splitting Phase:}} The adversary $\mathcal{A}_0$ outputs outputs an identity string $id^* \in \zo^{\idlen(\lambda)}$ and a \emph{bipartite} register $\reg$.

    \item The challenger samples $(A^*_i, s^*_i, s^{'*}_i)_{i \in [\cosettcount]} \samp \mathsf{CosetGen}(\cosetgenparam; F(K, id^*))$.

    \item \textbf{\underline{Query Phase 2:}} For $\ell \in \{1, 2\}$, each adversary $\adve_\ell$ is given $R[\ell]$ and $(A_i^*)_{i \in [\cosettcount]}$. For polynomially many rounds, each adversary makes queries to the challenger as follows. $\adve_\ell$ submits an identity string $id$ to the challenger. If $id \neq id^*$, the challenger  samples  $(A_i, s_i, s_i')_{i \in [\cosettcount]} \samp \mathsf{CosetGen}(\cosetgenparam; F(K, id))$ and submits the state $\left\{\ket{A_{i, s_i, s'_i}}\right\}_{i \in [\cosettcount]}$ to the adversary $\adve_\ell$.
    \item \textbf{\underline{Challenge Phase:}} The challenger samples $r_1 \samp \mathcal{D}$ and $r_2 \samp \mathcal{D}$.
    \item For $\ell \in \{1, 2\}$, each adversary $\adve_\ell$  is given $r_\ell$ and it outputs a tuple of vectors $(v_{\ell, i})_{i \in [{\cosettcount}]}$.

    \item The challenger, for all $\ell \in \{1, 2\}$ and for all $i \in [\cosettcount]$, checks if $v_{\ell, i} \in A^*_{i} + s^*_i$ if $(r_\ell)_i = 0$ and checks if $v_{\ell, i} \in {(A^*)}^\perp_{i} + s^{'*}_i$ if $(r_\ell)_i = 1$.
    
    If all the checks above pass and $id^*$ appears in $\idset$ at most once, the challenger outputs $1$. Otherwise, it outputs $0$.
    \end{enumerate}

    Similarly, we define $\moecollsel(\lambda, \idlen(\lambda), \adve)$ to be the selective version of the above game where the adversary outputs the chosen identity $id^*$ at the beginning of the game.

Assuming the existence of $\io$ and one-way functions, then for any polynomial $\idlen(\lambda)$, for any unpredictable distribution $\mathcal{D}$ and for any QPT adversary tuple $\adve$, \begin{equation*}
    \Pr[\moecollsel(\lambda, \idlen(\lambda), \adve) = 1] \leq \negl(\lambda).
\end{equation*}

If we assume the existence of subexponentially-secure $\io$ and one-way functions, and set $\mathcal{D}$ to be the uniform distribution, then for any polynomial $\idlen(\lambda)$ there exists constants $\constmoecollsel, \constmoecoll > 0$ such that for any QPT adversary tuple 
 \begin{align*}
 &\Pr[\moecollsel(\lambda, \idlen(\lambda), \adve) = 1] \leq 2^{-\lambda^{\constmoecollsel}}\\
    &\Pr[\moecoll(\lambda, \idlen(\lambda), \adve) = 1] \leq 2^{-\lambda^{\constmoecoll}}
\end{align*}
for all sufficiently large $\lambda$.
\end{theorem}
\begin{proof}

    We will give the full security proof for the \emph{adaptive} case. For the adaptive case, we will rely on \textit{complexity leveraging}, i.e. basically guessing the challenge identity $id^*$. The security of the selective game $\moecollsel$ follows from the same arguments, with the difference being that we do not use complexity leveraging, since the $id^*$ is selected by the adversary at the beginning (hence no need to guess it), and therefore we can simply puncture the PRF key at that point when preparing $\mathsf{OPMem}$.
    
    We now prove adaptive security. Let $\io$ be a $2^{-\lambda^{c_{\mathsf{iO}}}}$-secure indistinguishability obfuscation scheme and $F$ be a $2^{-\lambda^{c_{\mathsf{PRF}}}}$-secure puncturable PRF family with input length $\idlen(\lambda)$ and output length same as the size of the randomness used by $\mathsf{CosetGen}$, where $c_{\mathsf{iO}}, c_{\mathsf{PRF}}$ are some constants satisfying $\lambda^{{c_{\mathsf{PRF}}}} > (\lambda + \idlen(\lambda))^3$ and $\lambda^{{c_{\mathsf{iO}}}} > (\lambda + \idlen(\lambda))^3$. Note that such a PRF exists assuming subexponentially secure one-way functions (\cref{thm:puncprfexists}). 
    
We first define a stronger game as follows.
\paragraph{\underline{$\mathsf{Moe-Coll-PuncKey}$}}
\begin{enumerate}
    \item The challenger initializes the list $\idset = [~]$.
    
    \item The challenger samples a PRF key $K \samp F.\mathsf{KeyGen}(1^{\lambda})$. 
    
    \item The challenger samples $\mathsf{OPMem} \samp \io(\mathsf{PMem}_{K})$, where $\mathsf{PMem}_{K}$ is the following program.
    
\begin{mdframed}
        {\bf $\underline{\mathsf{PMem}_{K}(id, u_1, \dots, u_{\cosettcount}, r)}$}
        
        {\bf Hardcoded: $K$}
        \begin{enumerate}[label=\arabic*.]
            \item $(A_i, s_i, s_i')_{i \in [\cosettcount]} \samp \mathsf{CosetGen}(\cosetgenparam; F(K, id))$.
            \item For each $i \in [\cosettcount]$, check if $u_i \in A_i + s_i$ if $(r)_i = 0$ and check if $u_i \in A^{\perp}_i + s'_i$ if $(r)_i = 1$. If any of the checks fail, output $0$ and terminate.
            \item Output $1$.
        \end{enumerate}

    \end{mdframed}

    \item The challenger submits $\mathsf{OPMem}$ to the adversary.

    \item For polynomially many rounds, the adversary makes queries as follows. The adversary submits an identity string $id$ to the challenger and a query type, either $\mathsf{CLASSICAL}$ or $\mathsf{STATE}$. Then, the challenger samples $(A_i, s_i, s_i')_{i \in [\cosettcount]} \samp \mathsf{CosetGen}(\cosetgenparam; F(K, id))$.
    
    If the type is $\mathsf{CLASSICAL}$, the challenger adds $id$ to the list $\idset$ \emph{twice} and submits $(A_i, s_i, s_i')_{i \in [\cosettcount]}$ to the adversary.
    
    If the type is $\mathsf{STATE}$, the challenger adds $id$ to the list $\idset$ \emph{once} and submits the state $\left\{\ket{A_{i, s_i, s'_i}}\right\}_{i \in [\cosettcount]}$ to the adversary.

    \item The adversary $\mathcal{A}_0$ outputs outputs an identity string $id^* \in \zo^{\idlen(\lambda)}$.
    
    \item  The challenger computes $(A^*_i, s^*_i, s^{'*}_i)_{i \in {\cosettcount}} = \mathsf{CosetGen}(\cosetgenparam; F(K, id^*))$.

    \item \label{item:moecollstop} {The challenger computes $K\{id^*\} \samp F.\mathsf{Punc}(K, id^*)$ and submits it to the adversary}.

    \item The adversary outputs a \emph{bipartite} register $\reg$.

        \item  The challenger samples $r_1 \samp \mathcal{D}$ and $r_2 \samp \mathcal{D}$.

    \item   For $\ell \in \{1, 2\}$, each adversary $\adve_\ell$ is given $R[\ell], (A_i^*)_{i \in [\cosettcount]}, r_\ell$ and ${K\{id^*\}}$, and it outputs a tuple of vectors $(v_{\ell, i})_{i \in [{\cosettcount}]}$.

    \item The challenger, for all $\ell \in \{1, 2\}$ and for all $i \in [\cosettcount]$, checks if $v_{\ell, i} \in A^*_{i} + s^*_i$ if $(r_\ell)_i = 0$ and checks if $v_{\ell, i} \in {(A^*)}^\perp_{i} + s^{'*}_i$ if $(r_\ell)_i = 1$.
    
    If all the checks above pass and $id^*$ appears in $\idset$ at most once, the challenger outputs $1$. Otherwise, it outputs $0$.

\end{enumerate}

It is easy to see that the security in the stronger game implies security in the original game, since the adversaries $\adve_1, \adve_2$ can simulate their coset queries simply by evaluating the PRF using $K\{id^*\}$. Note that in the original game, they are not allowed to query for $id^*$ after the split, therefore $K\{id^*\}$ rather than $K$ is sufficient.

Now suppose for a contradiction that there exists an adversary $(\adve_0, \adve_1, \adve_2)$ that wins the stronger security game with probability $2^{-\lambda}$. We define a tuple of efficient algorithms $(\adve'_0, \adve'_1, \adve'_2)$ as follows.

\paragraph{\underline{$\adve'_0$}\\}
Sample $id' \samp \zo^{\idlen(\lambda)}$.
Simulate both $\adve_0$ and the challenger of the stronger game above, up to (including) \cref{item:moecollstop}. If $id^* = id'$, output the output of $\adve_0$, along with two copies of $(K\{id^*\}, (A_i^*)_{i \in [\cosettcount]})$, one for each $\adve'_\ell$. Otherwise, output $(\perp, \perp)$.

\paragraph{\underline{$\adve'_\ell$ for $\ell \in \{1, 2\}$}\\}
If the input is $\perp$, output $\perp$ and terminate. Otherwise, simulate the rest of the challenger and $\adve_\ell$.
\\

Observe that the probability that both $\adve'_\ell$ simultaneously output the \emph{correct} vectors is at least $2^{-\lambda}\cdot 2^{-\idlen(\lambda)}$, since $\adve$ outputs the correct vectors with probability $2^{-\lambda}$ by assumption and we have $id' = id^*$ with probability $2^{-\idlen(\lambda)}$ independently. Now, we will modify the algorithms $\adve'$ through a sequence of steps to finally obtain an adversary that wins $\moeflip$ with probability $2^{-2\cdot(\lambda + L(\lambda))}$, which is a contradiction by \cref{defn:strmoe}. Throughout rest of the proof, we will assume $id^* = id'$, which is indeed required to win the game.

We define $\adve_0''$ by modifying $\adve_0'$ so that it now samples $\mathsf{OPMem}$ as follows. It  computes  $z = F(K, id')$ at the beginning of the game and $(A^*_i, s^*_i, s^{'*}_i)_{i \in {\cosettcount}} = \mathsf{CosetGen}(\cosetgenparam; F(K, id^*))$. Then, we first compute for each $i \in [\cosettcount]$, $\mathsf{OP}_i^{{*}0} \samp \io(A^{*}_i + s^{*}_i)$ and $\mathsf{OP}_i^{{*}1} \samp \io(A^{{*}\perp}_i + s_i^{'{*}})$ (i.e., as in \cref{defn:strmoe}) using $(A^*_i, s^*_i, s^{'*}_i)_{i \in {\cosettcount}}$. Then, it samples $\mathsf{OPMem} \samp \io(\mathsf{PMem'}_{K\{id'\}, id', (\mathsf{OP}_i^{*0}, \mathsf{OP}_i^{*1})_{i\in[\cosettcount]}})$.
    \begin{mdframed}
        {\bf $\underline{\mathsf{PMem'}_{K\{id'\}, id', (\mathsf{OP}_i^{*0}, \mathsf{OP}_i^{*1})_{i\in[\cosettcount]}}(id, u_1, \dots, u_{\cosettcount}, r)}$}
        
        {\bf Hardcoded: $\textcolor{red}{K\{id'\}, id', (\mathsf{OP}_i^{*0}, \mathsf{OP}_i^{*1})_{i\in[\cosettcount]}}$}
        \begin{enumerate}[label=\arabic*.]
        \item \textcolor{red}{If $id = id'$, execute the following.}
        \begin{enumerate}[label*=\arabic*.]
            \item \textcolor{red}{For each $i \in [\cosettcount]$, check if $\mathsf{OP}_i^{*0}(u_i) = 1$ if $(r)_i = 0$ and check if $\mathsf{OP}_i^{*1}(u_i) = 1$ if $(r)_i = 1$.}
            \item \textcolor{red}{If all the checks pass, output $1$ and terminate. Otherwise, output $\false$ and terminate.}
        \end{enumerate}
        
          \item $(A_i, s_i, s_i')_{i \in [\cosettcount]} \samp \mathsf{CosetGen}(\cosetgenparam; F_1(\textcolor{red}{K\{id'\}}, id))$.
            \item For each $i \in [\cosettcount]$, check if $u_i \in A_i + s_i$ if $(r)_i = 0$ and check if $u_i \in A^{\perp}_i + s'_i$ if $(r)_i = 1$. If any of the checks fail, output $0$ and terminate.
            \item Output $1$.
        \end{enumerate}
    \end{mdframed}
Further, $\adve''_0$ answers any query made by $\adve$ for $id'$ using $z$.
By correctness of the obfuscations $\mathsf{OP}_i^{{*}0}, \mathsf{OP}_i^{{*}1}$, and by security of the obfuscation of $\mathsf{PMem}'$, we get that the modified adversary outputs the correct vectors with probability at least $2^{-\lambda - \idlen(\lambda)} - 2^{-(\lambda + \idlen(\lambda))^3}$.

Observe that above, we never evaluate the PRF at $id'$ except at the first step, where we compute $z$, and the adversary only gets the punctured key $K\{id'\}$\footnote{Remember that we have $id' = id^*$.}. We define $\adve''_0$ so that it now samples $z$ uniformly at random. Then, by above and by puncturable PRF security (\cref{defn:puncprf}), the adversary outputs the correct vectors with probability at least $2^{-\lambda - \idlen(\lambda)} - 2\cdot2^{-(\lambda + \idlen(\lambda))^3}$. Note that selective security is sufficient since the adversary picks the puncturing point $id'$ before the PRF key is sampled.

Finally, we construct an adversary $\adve_0'''$ for $\moeflip$ with security parameter $L(\lambda) + \lambda$ as follows. It simulates $\adve_0''$, but instead of answering the queries related to $id'$ itself, it uses the coset state tuple it obtains from its challenger. Note that since we require that $id'$ is queried at most once to win, the single copy obtained from the challenger is sufficient.  $\adve_\ell'''$ is constructed similarly, where they use the subspace descriptions and the challenge string $r_\ell$ submitted to them by the challenger. This simulates the game above perfectly, since we place the coset tuple obtained from the challenger in place of the coset tuple associated with $id^*$, which has the same distribution since $z$ is random. Also note that the adversary outputs the correct vectors for this coset tuple, since its choice is $id^*$. Therefore $\adve'''$ wins $\moeflip$ with probability $2^{-\lambda - \idlen(\lambda)} - 2\cdot2^{-(\lambda + \idlen(\lambda))^3} > 2^{-2\cdot(\lambda + \idlen(\lambda))} > 2^{-(\lambda + \idlen(\lambda))^2}$, which is a contradiction (\cref{defn:strmoe}).
\end{proof}

\section{Identity-Based and Functional Encryption with Puncturable Master Secret Key}\label{sec:puncibemain}
In this section, we give definitions for public-key identity-based encryption, along with its variant where we can puncture the master secret key \cite{CHEN2019450}. We also define functional encryption whose master secret key can be punctured simultaneously at all functions such that $f(m_0) \neq f(m_1)$. Then, we show how to construct such schemes in the plain model.

\subsection{Definitions}
We first give the definition of usual public-key identity-based encryption.
\begin{definition}\label{defn:ibe}
 An identity-based encryption scheme with message space $\mathcal{M}$ and identity space $\mathcal{ID}$ consists of the following algorithms that satisfy the correctness guarantee below.

 \begin{itemize}
        \item $\mathsf{Setup}(1^\lambda):$ Takes a security parameter, $\lambda$; outputs a public key $pk$ and a master secret key $msk$.
        \item $\mathsf{KeyGen}(msk, id):$ Takes the master secret key and an identity $id \in \mathcal{ID}$, outputs a secret key for the identity $id$.
    \item $\mathsf{Enc}(pk, id, m):$ Takes the public key $pk$, an identity $id$ and a message $m \in \mathcal{M}$, outputs an encryption of $m$ under the identity $id$.
    \item $\mathsf{Dec}(sk, ct):$ Takes a secret key and a ciphertext, outputs either a message or $\perp$.
 \end{itemize}

\paragraph{Correctness} For all messages $m \in \mathcal{M}$ and identities $id, \in \mathcal{ID}$, we require \begin{equation*}
    \Pr[\mathsf{IBE.Dec}(sk, ct) = m : \begin{array}{c}
         pk, msk \samp \mathsf{IBE.Setup}(1^\lambda)  \\
        sk \samp \mathsf{IBE.KeyGen}(msk, id)\\
         ct \samp \mathsf{Enc}(pk, id, m)
    \end{array}] = 1.
\end{equation*}
\end{definition}

We define the following security notion for identity-based encryption.

\begin{definition}[Adaptive Indistinguishability-Based Security for Identity-Based Encryption]
    Consider the following game between the challenger and an adversary $\adve$.
    \paragraph{$\underline{\ibegame(\lambda, \adve)}$}
    \begin{enumerate}
        \item The challenger runs $(pk, msk) \samp \mathsf{IBE.Setup}(1^\lambda)$ and then it submits $pk$ to the adversary. It also initializes the set $\idset$.

        \item \textbf{Query Phase 1:} For multiple rounds, the adversary adaptively submits an identity string $id \in \idset$. For each query, the challenger samples $sk \samp \mathsf{IBE.KeyGen}(msk, id)$ and submits $sk$ to the adversary. It also adds $id$ to $\idset$.

        \item The adversary outputs an identity $id^*$ and a pair of messages $m_0, m_1$.

        \item The challenger samples $b \samp \zo$ and $ct \samp \mathsf{IBE.Enc}(pk, id^*, m_b)$. It submits $ct$ to the adversary.

        \item \textbf{Query Phase 2:} For multiple rounds, the adversary adaptively submits an identity string $id \in \idset$. For each query, the challenger samples $sk \samp \mathsf{IBE.KeyGen}(msk, id)$ and submits $sk$ to the adversary. It also adds $id$ to $\idset$.

        \item The adversary outputs a guess $b' \in \zo$.

        \item The challenger outputs $1$ if and only if $b' = b$ and $id^* \not\in \idset$.
    \end{enumerate}
We say that an identity-based encryption scheme $\mathsf{IBE}$ satisfies \emph{adaptive indistinguishability-based security} if for any QPT adversary $\adve$ 
\begin{equation*}
    \Pr[\ibegame(\lambda, \adve) = 1] \leq 1/2 + \negl(\lambda).
\end{equation*}

    If the adversary outputs $id^*$ before $\mathsf{IBE.Setup}$ is run, we call it \emph{selective indistinguishability-based security}.
\end{definition}

We can also define an even weaker variant where the adversary cannot query for specific identities, but is only given the keys for randomly sampled identities. While this weaker variant would be sufficient for our unclonable PKE construction (\cref{sec:pkecons}), since we do not know of any simpler constructions (compared to the selectively secure construction below), we do not pursue this further. 

We now move onto identity-based encryption with puncturable master secret keys, introduced by Chen, Zhang, Deng, Chang \cite{CHEN2019450}. This is defined to be an identity-based encryption scheme where the master secret key can be punctured at an identity so that the resulting key can be used to issue secret keys for any identity except for the punctured identity. \cite{CHEN2019450} also give a construction based on hierarchical identity-based encryption.

Our definition is simpler than that of \cite{CHEN2019450}, which allows the adversary to adaptively query for different secret keys before selecting the identity at which the master secret key will be punctured. Our construction below can also be made secure with respect to their definition by employing an adaptively secure puncturable PRF, however, our simplified definition suffices for our unclonable primitive constructions.

\begin{definition}
    Identity-based encryption with puncturable master secret key is an identity-based encryption scheme (\cref{defn:ibe}) with the following additional algorithms and correctness guarantees.
    
\begin{itemize}
    \item $\mathsf{Punc}(msk, id)$: Takes as input the master secret key $msk$ and an identity $id$, outputs a master secret key that is punctured at $id$.

\end{itemize}

\paragraph{Punctured Key Correctness} For all messages $m \in \mathcal{M}$ and identities $id, id' \in \mathcal{ID}$ such that $id \neq id'$, \begin{equation*}
    \Pr[\mathsf{IBE.Dec}(sk, ct) = m : \begin{array}{c}
         pk, msk \samp \mathsf{IBE.Setup}(1^\lambda)  \\
        msk' \samp \mathsf{IBE.Punc}(msk, id') \\
        sk \samp \mathsf{IBE.KeyGen}(msk', id)\\
         ct \samp \mathsf{IBE.Enc}(pk, id, m)
    \end{array}] = 1.
\end{equation*}
\end{definition}

We also define a stronger version of punctured key correctness, where we require that there be no difference between sampling a secret key for an identity using the actual master secret key versus using a punctured master secret key.

\begin{definition}[Strong Punctured Key Correctness]\label{defn:strongpuncibe}
    For all identities $id, id' \in \mathcal{ID}$ such that $id \neq id'$, we require
   \begin{equation*}
   (psk, msk, pk) \equiv (sk, msk, pk)
   \end{equation*}
   where \begin{align*}
   pk, msk &\samp \mathsf{IBE.Setup}(1^\lambda)  \\
   sk &\samp \mathsf{IBE.KeyGen}(msk, id) \\
   msk' &\samp \mathsf{IBE.Punc}(msk, id') \\
   psk &\samp \mathsf{IBE.KeyGen}(msk', id).
   \end{align*}
\end{definition}

\begin{definition}[Puncturable Master Secret Key Security for Identity-Based Encryption]
\label{defn:puncibe}
    Consider the following game between the challenger and an adversary $\adve$.
    \paragraph{$\underline{\puncibegame(\lambda, \adve)}$}
    \begin{enumerate}
    \item The adversary outputs an identity $id^*$.
        \item The challenger runs $(pk, msk) \samp \mathsf{IBE.Setup}(1^\lambda)$ and then $msk^* \samp \mathsf{IBE.Punc}(msk, id^*)$. Then, it submits $pk, msk^*$ to the adversary.

        \item  The adversary outputs a pair of messages $m_0, m_1$.

        \item The challenger samples a challenge bit $b \samp \zo$ and computes $ct \samp \mathsf{IBE.Enc}(pk, id^*, m_b)$. Then, it submits $ct$ to the adversary.

        \item The adversary outputs a guess $b' \in \zo$.

        \item The challenger outputs $1$ if and only if $b' = b$.
    \end{enumerate}
We say that an identity-based encryption scheme $\mathsf{IBE}$ satisfies puncturable master secret key security if any QPT adversary $\adve$,
\begin{equation*}
    \Pr[\puncibegame(\lambda, \adve) = 1] \leq \frac{1}{2} + \negl(\lambda).
\end{equation*}

\end{definition}

It is easy to see that puncturable master secret key security with strong punctured key correctness implies indistinguishability-based security. We formalize this below.

\begin{theorem}
    Let $\mathsf{IBE}$ be an identity-based encryption scheme that satisfies [polynomial, subexponential] puncturable master secret key security (\cref{defn:puncibe}). Then, it also satisfies [polynomial, subexponential] selective indistinguishability-based security  (\cref{defn:ibe}). 
\end{theorem}
\begin{proof}
    Suppose for a contradiction that there exist a QPT adversary $\adve$ that wins the selective indistinguishability-based security game against $\mathsf{IBE}$ with probability $\eps(\lambda)$. We claim that the adversary $\adve'$ described below wins the puncturable master secret key security game with $\eps(\lambda)$.

    $\adve'$ runs $\adve$ to obtain $id^*$ and outputs it. Then, it receives $pk, msk^*$ from the challenger. Then, $\adve'$ simulates the first query phase as follows. It runs $\adve$, and whenever it queries for an identity string $id$, $\adve'$ computes $sk \samp \mathsf{IBE.KeyGen}(msk^*, id)$ and gives $sk$ to $\adve$. When $\adve$ yields $m_0, m_1$, then $\adve$ outputs these values. $\adve'$ simulates the second query phase similarly using $msk^*$ after receiving the challenge ciphertext from the challenger. Finally, when $\adve$ outputs it guess $b'$, the adversary $\adve'$ also outputs it.

    Since $id^* \not\in \idset$, i.e., since the adversary $\adve$ never queries for $id^*$, hence by strong punctured key correctness of $\mathsf{IBE}$, there is no difference between sampling the secret keys using the punctured master secret key (as above) or the actual master secret key (as in the original selective indistinguishability-based security game). Hence, $\adve'$ and the challenger above perfectly simulate the selective indistinguishability-based security game played by $\adve$. Therefore, $\adve'$ wins the puncturable master secret key game with probability $\eps(\lambda)$.

    Plugging in $\eps(\lambda) = 1/\poly(\lambda)$ or $\eps(\lambda) > \subexp(\lambda)$ completes the proof.
\end{proof}

Finally, we define functional encryption where the master secret key can be punctured at all functions $f$  such that $f(m_0) = f(m_1)$. Note that previous works \cite{BV18, YALXY19, Kitagawa2022} define their own versions of puncturable functional encryption which is different than ours, see \cref{sec:relatedwork}. However, throughout the paper, we will write \emph{puncturable functional encryption} to mean our definition.

\begin{definition}[Puncturable Functional Encryption]
    Puncturable functional encryption is a functional encryption scheme (\cref{defn:fe}) with the following additional algorithms and correctness guarantees.
    
\begin{itemize}
    \item $\mathsf{Punc}(msk, m_0, m_1)$: Takes as input the master secret key $msk$ and outputs a master secret key that is \emph{punctured}.

\end{itemize}

\paragraph{Punctured Key Correctness} For all messages $m, m_0, m_1 \in \mathcal{M}$ and functions $f \in \mathfrak{F}$ such that $f(m_0) \neq f(m_1)$, \begin{equation*}
    \Pr[\mathsf{FE.Dec}(sk, ct) = f(m) : \begin{array}{c}
         pk, msk \samp \mathsf{FE.Setup}(1^\lambda)  \\
        msk' \samp \mathsf{FE.Punc}(msk, m_0, m_1) \\
        sk \samp \mathsf{FE.KeyGen}(msk', f)\\
         ct \samp \mathsf{FE.Enc}(pk, m)
    \end{array}] = 1.
\end{equation*}
\end{definition}

Similar to IBE, we also define a stronger version of punctured key correctness.

\begin{definition}[Strong Punctured Key Correctness]\label{defn:strongpuncfe}
    For all messages $m_0, m_1 \in \mathcal{M}$ and functions $f \in \mathfrak{F}$ such that $f(m_0) \neq f(m_1)$, we require
   \begin{equation*}
   (psk, msk, pk) \equiv (sk, msk, pk)
   \end{equation*}
   where \begin{align*}
   pk, msk &\samp \mathsf{FE.Setup}(1^\lambda)  \\
   sk &\samp \mathsf{FE.KeyGen}(msk, f) \\
   msk' &\samp \mathsf{FE.Punc}(msk, m_0, m_1) \\
   psk &\samp \mathsf{FE.KeyGen}(msk', f).
   \end{align*}
\end{definition}

\begin{definition}[Puncturable Functional Encryption Security]
\label{defn:puncfe}
    Consider the following game between the challenger and an adversary $\adve$.
    \paragraph{$\underline{\puncfegame(\lambda, \adve)}$}
     \begin{enumerate}
        \item Challenger samples the keys $msk, pk \samp \mathsf{FE.Setup}(1^\lambda)$.
        \item The adversary receives $pk$. It makes polynomially many queries by sending a function $f \in \mathcal{F}$ and receiving the corresponding functional key $sk_f \samp \mathsf{FE.KeyGen}(msk, f)$.
        \item The adversary outputs challenge messages $m_0, m_1$. 
        \item The challenger checks if $f(m_0) = f(m_1)$ for all $f \in \mathfrak{F}$ that was queried by the adversary. If this condition is not satisfied, it outputs $0$ and terminates.
        \item The challenger samples $msk' \samp \mathsf{FE.Punc}(msk, m_0, m_1)$.
        \item The challenger samples a challenge bit $b \samp \zo$ and prepares $ct \samp \mathsf{Enc}(pk, m_b)$.
        \item The adversary receives $msk', ct$ and outputs a guess $b'$.
        \item The challenger outputs $1$ if $b' = b$.
    \end{enumerate}
   
We say that an functional encryption scheme $\mathsf{FE}$ satisfies puncturable master secret key security if any QPT adversary $\adve$,
\begin{equation*}
    \Pr[\puncfegame(\lambda, \adve) = 1] \leq \frac{1}{2} + \negl(\lambda).
\end{equation*}

\end{definition}

\subsection{Puncturable Identity-Based Encryption Construction}\label{sec:consibe}
In this section, we show how to construct an identity-based encryption with puncturable master secret key from indistinguishability obfuscation and public-key encryption, which in turn can be constructed from indistinguishability obfuscation and one-way functions \cite{SW14, Z12}. 

    We obtain our construction through a small addition to the PKE-to-IBE transformation of \cite{CHEN2019450}. In their construction of an IBE scheme, the master secret key is a (puncturable) PRF key that is used to spin up a fresh instance of a PKE scheme for each identity value. In our puncturable IBE scheme, the puncturing algorithm is simply the puncturing algorithm for the PRF family. We present the full scheme below for completeness.

    Assume the existence of following schemes.
\begin{itemize}
    \item $\io$, an indistinguishability obfuscation scheme,
    \item $\mathsf{PKE}$, a public-key encryption scheme with message space $\mathcal{M}$,
    \item $F$, a puncturable PRF family with input space $\mathcal{ID}$ and output length same as the size of the randomness used by $\mathsf{PKE.KeyGen}$,
\end{itemize}

    \paragraph{$\underline{\mathsf{IBE.Setup}(1^{\lambda})}$}
\begin{enumerate}
    \item Sample a PRF key $K \samp F.\mathsf{KeyGen}(1^{\lambda})$.
    \item Sample $\mathsf{OPKeyGen} \samp \io(\mathsf{PKeyGen}_{K}, 1^{\lambda})$, where $\mathsf{PKeyGen}_{K}$ is the following program.
    
\begin{mdframed}
        {\bf $\underline{\mathsf{PKeyGen}_{K}(id)}$}
        
        {\bf Hardcoded: $K$}
        \begin{enumerate}[label=\arabic*.]
            \item Sample $ipk, isk \samp \mathsf{PKE.KeyGen}(1^\lambda; F(K, id))$.
            \item Output $ipk$.
        \end{enumerate}

    \end{mdframed}
    \item Output $(\mathsf{OPKeyGen}, K)$.
    \end{enumerate}

 \paragraph{$\underline{\mathsf{IBE.KeyGen}(msk, id)}$}
 \begin{enumerate}
     \item Parse $K = msk$.
     \item Compute $(ipk, isk) \samp \mathsf{PKE.KeyGen}(1^\lambda; F(K, id))$.
     \item Output $isk$.
 \end{enumerate}

\paragraph{$\underline{\mathsf{IBE.Punc}(msk, id)}$}
 \begin{enumerate}
     \item Parse $K = msk$.
     \item Compute $K' \samp F.\mathsf{Punc}(K, id)$.
     \item Output $K'$.
 \end{enumerate}
 
 \paragraph{$\underline{\mathsf{IBE.Enc}(pk, id, m)}$}
 \begin{enumerate}
     \item Parse $\mathsf{OPKeyGen} = pk$.
     \item Compute $ipk \samp \mathsf{OPKeyGen}(id)$.
     \item Output $\mathsf{PKE.Enc}(ipk, m)$.
 \end{enumerate}

\paragraph{$\underline{\mathsf{IBE.Dec}(sk, ct)}$}

Same as $\mathsf{PKE.Dec}$.

\begin{theorem}
    $\mathsf{IBE}$ satisfies both correctness (\cref{defn:strongpuncibe}) and strong punctured master secret key correctness (\cref{defn:puncibe}).
\end{theorem}
\begin{proof}
Correctness is easy to see, by correctness of $\mathsf{PKE}$ and $\mathsf{IBE}$. 

We move onto strong punctured master secret key correctness. Consider any $id \neq id'$.
    By punctured key correctness of $F$, we have that $F(K\{id'\}, id) = F(K, id)$ with probability 1 over the choice of the key and sampling of the punctured key. The result follows.
\end{proof}

\begin{theorem}\label{thm:ibepunc}
    $\mathsf{IBE}$ satisfies puncturable master secret key security (\cref{defn:puncibe}).
\end{theorem}

Since public-key encryption can be based on $\io$ and one-way functions \cite{SW14, Z12}, and puncturable PRFs can be based on one-way functions also (\cref{thm:puncprfexists}), we get the following corollary.
\begin{corollary}
    Assuming the existence of indistinguishability obfuscation and one-way functions, there exist an identity-based encryption scheme with puncturable master secret keys, for any message length and identity length that is polynomial in the security parameter.
\end{corollary}

We prove \cref{thm:ibepunc} in \cref{sec:proofibe}. It is easy to see that our proof also generalizes to the subexponential security case. Hence, we get the following.

\begin{corollary}\label{thm:subexpibe}
    Assuming the existence of subexponentially secure indistinguishability obfuscation and one-way functions, there exist a subexponentially secure identity-based encryption scheme with puncturable master secret keys.
\end{corollary}

\subsection{Proof of Security for Puncturable IBE}\label{sec:proofibe}
In this section, we prove \cref{thm:ibepunc}. Suppose for a contradiction that there exists a QPT adversary $\adve$ that wins the puncturable master secret key security game $\puncibegame$ (\cref{defn:puncibe}) with non-negligible probability. We prove security through a series of hybrids, each of which is constructed by modifying the previous hybrid.

\paragraph{$\hyb_0$:} The original game $\puncibegame(\lambda, \adve)$.
\paragraph{$\hyb_1$:} The challenger computes $z^* = F(K, id^*)$ and $K\{id^*\} \samp F.\mathsf{Punc}(K, id^*)$ after the adversary has submitted $id^*$. Then, it computes $ipk^*, isk^* \samp \mathsf{PKE.KeyGen}(1^\lambda; z^*)$. Finally, instead of sampling the public key $pk$ as before, it now computes it as $\mathsf{OPKeyGen} \samp \io(\mathsf{PKeyGen}'_{K\{id^*\}, ipk^*, id^*}, 1^{\lambda})$, where $\mathsf{PKeyGen}'_{K\{id^*\}, ipk^*, id^*}$ is the following program.
    
\begin{mdframed}
        {\bf $\underline{\mathsf{PKeyGen}'_{K\{id^*\}, ipk^*, id^*}(id)}$}
        
        {\bf Hardcoded: $K\{id^*\}, \textcolor{red}{ipk^*, id^*}$}
        \begin{enumerate}[label=\arabic*.]
            \item \textcolor{red}{If $id = id^*$, output $ipk^*$ and terminate.}
            \item Sample $ipk, isk \samp \mathsf{PKE.KeyGen}(1^\lambda; F(K, id))$.
            \item Output $ipk$.
        \end{enumerate}

    \end{mdframed}

\paragraph{$\hyb_2$:} The challenger now samples $z^*$ uniformly at random from the output space of $F$ instead of computing it as  $z^* = F(K, id^*)$.

\begin{claim}
    $\hyb_0 \approx \hyb_1$.
\end{claim}
\begin{proof}
By  correctness of the punctured PRF keys, we have that $\mathsf{PKeyGen}'_{K\{id^*\}, ipk^*, id^*}$ and $\mathsf{PKeyGen}_{K}$ have the same functionality. The result follows by the security of $\io$.
\end{proof}
\begin{claim}
    $\hyb_1 \approx \hyb_2$.
\end{claim}
\begin{proof}
    Observe that in $\hyb_1$, the adversary has only access to (efficient functions of) PRF evaluations at points other than $id^*$ and the punctured PRF key $K\{id^*\}$, rather than the full key $K$. Therefore, the result follows from the puncturable PRF security (\cref{defn:puncprf}).
\end{proof}

By above, we get that $\adve$ wins in $\hyb_2$ with non-negligible probability. We claim that the adversary $\adve'$ described below wins the public-key encryption security game against $\mathsf{PKE}$ with non-negligible probability.

$\adve'$ runs $\adve$ to obtain $id^*$. Then, it samples a PRF key $K$ for $F$ and computes $K\{id^*\} \samp F.\mathsf{Punc}(K, id^*)$. When it receives the public key $pk$ from the challenger, $\adve'$ computes  $\mathsf{OPKeyGen} \samp \io(\mathsf{PKeyGen}_{K\{id^*\}, pk, id^*}, 1^{\lambda})$, where $\mathsf{PKeyGen}_{K\{id^*\}, pk, id^*}$ is the following program.
    
\begin{mdframed}
{\bf $\underline{\mathsf{PKeyGen}'_{K\{id^*\}, pk, id^*}(id)}$}

{\bf Hardcoded: $K\{id^*\}, {pk, id^*}$}
\begin{enumerate}[label=\arabic*.]
    \item {If $id = id^*$, output $pk$ and terminate.}
    \item Sample $ipk, isk \samp \mathsf{PKE.KeyGen}(1^\lambda; F(K, id))$.
    \item Output $ipk$.
\end{enumerate}

\end{mdframed}
Then, $\adve'$ runs $\adve$ on $\mathsf{OPKeyGen}$ and $K\{id^*\}$ to obtain $m_0, m_1$, which it submits to the challenger. Finally, when $\adve'$ obtains the challenge ciphertext from the challenger, it runs $\adve$ on it to obtain the guess $b'$, which it submits to the challenger.

It is easy to see that $\adve'$ wins the public-key encryption security game with the same probability as $\adve$ wins in $\hyb_2$, which is non-negligible. This is a contradiction to the security of $\mathsf{PKE}$.

\subsection{Puncturable Functional Encryption Construction}\label{sec:puncfecons}
As an application of our puncturable IBE scheme and as a warm-up to our copy-protected functional encryption scheme, we show how to construct puncturable functional encryption. In our copy-protected functional encryption scheme (\cref{sec:mainfe}), we use a punctured master secret key in a non-black-box way to remove interaction from the post-challenge-ciphertext phase of the security game.

    Now we move onto our construction, which will be similar to the delegatable functional encryption scheme of \cite{CGJS15}. Assume the existence of following schemes and we will construct a puncturable functional encryption for the class of functions $\mathfrak{F}$ defined as all circuits that are of size at most $\circsize$, where $\circsize$ is a fixed polynomial.
\begin{itemize}
    \item $\io$, $2^{-\lambda - Q(\lambda)}$-secure indistinguishability obfuscation scheme,
    \item $\mathsf{IBE}$, a $2^{-\lambda - Q(\lambda)}$-secure public-key identity-based encryption scheme for identity space $\zo^{Q(\lambda)}$ with puncturable master secret keys (\cref{defn:puncibe}) and deterministic identity key generation satisfying strong punctured key correctness (\cref{defn:strongpuncibe})
    \item $F$, a a $2^{-\lambda - Q(\lambda)}$-secure puncturable PRF family with input space $\zo^{Q(\lambda)}$ and output length same as the size of the randomness used by $\mathsf{IBE.Enc}$,
\end{itemize}

    \paragraph{$\underline{\mathsf{FE.Setup}(1^{\lambda})}$}
\begin{enumerate}
\item Sample $pk, imsk \samp \mathsf{IBE.Setup}(1^\lambda)$.
\item Output $pk, (imsk, \fullkey)$.
    \end{enumerate}

 \paragraph{$\underline{\mathsf{FE.KeyGen}(msk', f)}$}
 \begin{enumerate}
 \item Parse $(msk'', \typekey) = msk'$.
\item If $\typekey = \punckey$, output $msk''(f)$ and terminate.
   \item Sample $sk \samp \mathsf{IBE.KeyGen}(msk'', f)$.
   \item Output $(sk, f)$.
 \end{enumerate}

\paragraph{$\underline{\mathsf{FE.Punc}(msk, m_0, m_1)}$}
 \begin{enumerate}
 \item Parse $(imsk, \typekey) = msk$. Terminate if $\typekey \neq \fullkey$.
 \item Sample $\mathsf{OPKey} \samp \io(1^\lambda, \mathsf{PKey}_{imsk,  m_0, m_1})$ where $\mathsf{PKey}_{imsk,  m_0, m_1}$ is the following program.
 
\begin{mdframed}
{\bf $\underline{\mathsf{PKey}_{imsk, m_0, m_1}(f)}$}

{\bf Hardcoded: $imsk, m_0, m_1$}
\begin{enumerate}[label=\arabic*.]
    \item If $f(m_0) \neq f(m_1)$, output $\perp$ and terminate.
    \item Compute $sk = \mathsf{IBE.KeyGen}(imsk, f)$.
    \item Output $(sk, f)$.
\end{enumerate}

\end{mdframed}
     
     \item Output $\mathsf{OPKey}$.
 \end{enumerate}
 
 \paragraph{$\underline{\mathsf{FE.Enc}(pk, m)}$}
 \begin{enumerate}
     \item Sample $K \samp F.\mathsf{Setup}(1^\lambda).$
     \item Sample $\mathsf{OPCt} \samp \io(\mathsf{PCt}_{m, pk, K})$ where $\mathsf{PCt}_{m, pk, K}$ is the following program.
     \begin{mdframed}
{\bf $\underline{\mathsf{PCt}_{m, pk, K}(f)}$}

{\bf Hardcoded: $m, pk, K$}
\begin{enumerate}[label=\arabic*.]
    \item Compute $a = f(m)$.
    \item Compute $ct = \mathsf{IBE.Enc}(pk, f, a; F(K, f))$.
    \item Output $ct$.
\end{enumerate}
\end{mdframed}
     \item Output $\mathsf{OPCt}$.
 \end{enumerate}

\paragraph{$\underline{\mathsf{FE.Dec}(sk, ct)}$}
\begin{enumerate}
    \item Parse $(sk', f) = sk$.
    \item Parse $\mathsf{OPCt} = ct$.
    \item $ct' = \mathsf{OPCt}(f)$.
    \item Output $\mathsf{IBE.Dec}(sk', ct')$.
\end{enumerate}

\begin{theorem}
    $\mathsf{FE}$ satisfies both correctness (\cref{defn:strongpuncfe}) and strong punctured master secret key correctness (\cref{defn:puncfe}).
\end{theorem}
\begin{proof}
    Follows in a straightforward manner from the correctness of the underlying primitives and the fact that $\mathsf{IBE.KeyGen}$ is deterministic.
\end{proof}
\begin{theorem}\label{thm:fepunc}
    $\mathsf{FE}$ satisfies puncturable master secret key security (\cref{defn:puncfe}).
\end{theorem}

Since we construct in \cref{sec:consibe} a puncturable IBE with the properties required by $\mathsf{FE}$ based on $\io$ and one-way functions, we get the following corollary.
\begin{corollary}
    Assuming the existence of subexponentially secure indistinguishability obfuscation and one-way functions, there exist a puncturable functional encryption scheme.
\end{corollary}

\subsection{Proof of Security for Puncturable Functional Encryption}
In this section, we prove \cref{thm:fepunc}. Our proof will be similar to security proof of the delegatable functional encryption scheme in \cite{CGJS15}.  Throughout the proof, we will interpret the functions $f \in \mathfrak{F}$, which are represented by circuits of size $Q(\lambda)$, as numbers in $\{0, 1, \dots, 2^{Q} - 1\}$.

Suppose for a contradiction there exists a QPT adversary $\adve$ that wins the puncturable functional encryption game with non-negligible advantage, that is, $\Pr[\puncfegame(\lambda, \adve) = 1] \geq 1/2+1/p(\lambda)$ for some polynomial $p(\cdot)$ and for infinitely many values of $\lambda > 0$. We will prove security through a series of hybrids, each of which is obtained by modifying the previous one, starting with $\hyb_0$.

We define $\hyb_0$ to be the same as the original security game $\puncfegame(\lambda, \adve)$.

\paragraph{\underline{$\hyb_t$ for $t \in \{0, 1, \dots, 2^Q\}$}:} We now compute the challenge ciphertext (encryption of $m_b$) as $\mathsf{OPCt} \samp \io(\mathsf{PCt}^{(t)})$.
    \begin{mdframed}
{\bf $\underline{\mathsf{PCt}^{(t)}(f)}$}

{\bf Hardcoded: $\textcolor{red}{m_0, m_1, b}, pk, K$}
\begin{enumerate}[label=\arabic*.]
    \item \textcolor{red}{If $f < t$, set $a = f(m_{1-b})$. Otherwise, set $a = f(m_b)$.}
    \item Compute $ct = \mathsf{IBE.Enc}(pk, f, a; F(K, f))$.
    \item Output $ct$.
\end{enumerate}
\end{mdframed}

Define the event $E_t$ to be the event that the pair of challenge messages $m_0, m_1$ output by the adversary $\adve$ are such that $t(m_0) = t(m_1)$.

\begin{lemma}
    $\statdist{(\hyb_t)_{|E_t}}{(\hyb_t)_{|E_t}} < 2^{-\lambda - Q(\lambda)}$.
\end{lemma}
\begin{proof}
    Observe the programs $\mathsf{PCt}^{(t)}$ and $\mathsf{PCt}^{(t+1)}$ differ only on input $f = t$, in which case the first program compute $a = t(m_b)$ whereas the second program computes $a = t(m_{1-b})$. However, conditioned on $E_t$, we have $t(m_b) = t(m_{1-b})$. Hence, the programs have the same functionality and the result follows from the security of $\io$.
\end{proof}

We will also prove $\statdist{(\hyb_t)_{|\overline{E_t}}}{(\hyb_t)_{|\overline{E_t}}} < 2^{-\lambda/2 - Q(\lambda)}$. This combined with the above lemma gives $\statdist{\hyb_t}{\hyb_{t+1}} < {2^{-\lambda/2 - Q(\lambda)}}$ through triangle inequality. Crucially, note that the probability of the event $E_t$ is the same in both hybrids since we only change the way we compute the challenge ciphertext (which the adversary sees after choosing $m_0, m_1$).

To prove $\statdist{(\hyb_t)_{|\overline{E_t}}}{(\hyb_t)_{|\overline{E_t}}} < 2^{-\lambda/2 - Q(\lambda)}$, we define a sequence of intermediary hybrids.

\paragraph{\underline{$\hyb_{t, 1}$ for $t \in \{0, 1, \dots, 2^Q\}$}:} We first compute $ct^* \samp \mathsf{IBE.Enc}(pk, t, f(m_b); F(K, t))$ and $K\{t\} \samp F.\mathsf{Punc}(K, t)$. Then,  we now compute the challenge ciphertext as $\mathsf{OPCt} \samp \io(\mathsf{PCt}^{(t, 1)})$.
    \begin{mdframed}
{\bf $\underline{\mathsf{PCt}^{(t, 1)}(f)}$}

{\bf Hardcoded: $\textcolor{red}{ct^*, K\{t\}}, m_0, m_1, b, pk$}
\begin{enumerate}[label=\arabic*.]
    \item \textcolor{red}{If $f = t$, output $ct^*$ and terminate.}
    \item {If $f < t$, set $a = f(m_{1-b})$. Otherwise, set $a = f(m_b)$.}
    \item Compute $ct = \mathsf{IBE.Enc}(pk, f, a; F(\textcolor{red}{K\{t\}}, f))$.
    \item Output $ct$.
\end{enumerate}
\end{mdframed}

\paragraph{\underline{$\hyb_{t, 2}$ for $t \in \{0, 1, \dots, 2^Q\}$}:} We now sample $ct^* \samp \mathsf{IBE.Enc}(pk, t, f(m_b); z)$ where $z$ is sampled uniformly at random from the output space of $F$.

\paragraph{\underline{$\hyb_{t, 3}$ for $t \in \{0, 1, \dots, 2^Q\}$}}:
We now change the way we sample punctured master secret key as follows. At the beginning of the game, we compute $imsk' \samp \mathsf{IBE.Punc}(imsk, t)$. We now output $\mathsf{OPKey} \samp \mathsf{PKey}^{(t)}$.
 
\begin{mdframed}
{\bf $\underline{\mathsf{PKey}^{(t)}(f)}$}

{\bf Hardcoded: $\textcolor{red}{imsk'}, m_0, m_1$}
\begin{enumerate}[label=\arabic*.]
    \item If $f(m_0) \neq f(m_1)$, output $\perp$ and terminate.
    \item Compute $sk = \mathsf{IBE.KeyGen}(\textcolor{red}{imsk'}, f)$.
    \item Output $(sk, f)$.
\end{enumerate}

\end{mdframed}

\paragraph{\underline{$\hyb_{t, 4}$ for $t \in \{0, 1, \dots, 2^Q\}$}}:
    Same as $\hyb_{t,3}$ but we now compute $ct^*$ as $ct^* \samp \mathsf{IBE.Enc}(pk, t, \textcolor{red}{f(m_{1-b})}; z)$.

 \paragraph{\underline{$\hyb_{t, 5}$ for $t \in \{0, 1, \dots, 2^Q\}$}}:
Same as $\textcolor{blue}{\hyb_{t,2}}$ but we compute $ct^*$ as $ct^* \samp \mathsf{IBE.Enc}(pk, t, \textcolor{red}{ f(m_{1-b})}; F(K, t))$.

 \paragraph{\underline{$\hyb_{t, 6}$ for $t \in \{0, 1, \dots, 2^Q\}$}}:
Same as $\textcolor{blue}{\hyb_{t, 1}}$ but we compute $ct^*$ as $ct^* \samp \mathsf{IBE.Enc}(pk, t, \textcolor{red}{ f(m_{1-b})}; F(K, t))$.

\begin{lemma}\label{lem:puncfehyb01}
    $\statdist{(\hyb_{t})_{|\overline{E_t}}}{(\hyb_{t, 1})_{|\overline{E_t}}} < 2^{-\lambda - Q(\lambda)}$ for all $t \in \{0, 1, \dots, 2^Q\}$.
\end{lemma}
\begin{proof}
By the punctured key correctness of $F$, the programs $\mathsf{PCt}^{(t)}$ and $\mathsf{PCt}^{(t,1)}$ have the same functionality. The results follows by the security of $\io$.
\end{proof}

\begin{lemma}\label{lem:puncfehyb12}
    $\statdist{(\hyb_{t,1})_{|\overline{E_t}}}{(\hyb_{t, 2})_{|\overline{E_t}}} < 2^{-\lambda - Q(\lambda)}$ for all $t \in \{0, 1, \dots, 2^Q\}$.
\end{lemma}
\begin{proof}
    Observe that the adversary only has the punctured key $K\{t\}$. The resut follows by the security of the puncturable PRF $F$.
\end{proof}

\begin{lemma}\label{lem:puncfehyb23}
    $\statdist{(\hyb_{t,2})_{|\overline{E_t}}}{(\hyb_{t, 3})_{|\overline{E_t}}} < 2^{-\lambda - Q(\lambda)}\cdot \poly(\lambda)$ for all $t \in \{0, 1, \dots, 2^Q\}$.
\end{lemma}
\begin{proof}
    Since we are conditioned on the event $\overline{E_t}$, we have $t(m_0) \neq t(m_1)$.  Hence, by strong punctured key correctness of $\mathsf{IBE}$, all the programs $\mathsf{PKey}_P$ in these hybrids have the same functionality. Result follows from the security of $\io$.
\end{proof}

\begin{lemma}
    $\statdist{(\hyb_{t,3})_{|\overline{E_t}}}{(\hyb_{t, 4})_{|\overline{E_t}}} < 2^{-\lambda - Q(\lambda)}$ for all $t \in \{0, 1, \dots, 2^Q\}$.
\end{lemma}
\begin{proof}
    Since we are conditioned on the event $\overline{E_t}$, we have $t(m_0) \neq t(m_1)$. Further, to win, for any function $f \in \mathfrak{F}$, if $f$ is queried by the adversary, then $f(m_0) = f(m_1)$. Combining these, we get that $t$ was never queried by the adversary. Therefore, the adversary never gets a secret key for the identity $t$.  In particular, all the identity secret keys obtained by the adversary (as a result of functional key queries) can instead be obtained using $imsk'$, the IBE master secret key punctured at $t$, due to the strong punctured key correctness of IBE. Further, the punctured IBE master secret key obtained by the adversary (which is inside the punctured FE master secret key) is also punctured at $t$. Finally, observe that $ct^*$ is an IBE encryption (sampled using true randomness $z$) under the identity $t$. Hence, the result follows by puncturable IBE security (\cref{defn:puncibe}). 
\end{proof}

\begin{lemma}
    $\statdist{(\hyb_{t,4})_{|\overline{E_t}}}{(\hyb_{t, 5})_{|\overline{E_t}}} < 2^{-\lambda - Q(\lambda)}\cdot \poly(\lambda)$ for all $t \in \{0, 1, \dots, 2^Q\}$.
\end{lemma}
\begin{proof}
    Same argument as \cref{lem:puncfehyb23}.
\end{proof}

\begin{lemma}
    $\statdist{(\hyb_{t,5})_{|\overline{E_t}}}{(\hyb_{t,6})_{|\overline{E_t}}} < 2^{-\lambda - Q(\lambda)}$ for all $t \in \{0, 1, \dots, 2^Q\}$.
\end{lemma}
\begin{proof}
    Same argument as \cref{lem:puncfehyb12}.
\end{proof}

\begin{lemma}
    $\statdist{(\hyb_{t,6})_{|\overline{E_t}}}{(\hyb_{t+1})_{|\overline{E_t}}} < 2^{-\lambda - Q(\lambda)}$ for all $t \in \{0, 1, \dots, 2^Q\}$.
\end{lemma}
\begin{proof}
    Same argument as \cref{lem:puncfehyb01}.
\end{proof}

Combining the above lemmata, we get $\statdist{(\hyb_t)_{|\overline{E_t}}}{(\hyb_t)_{|\overline{E_t}}} < 2^{-\lambda/2 - Q(\lambda)}$. Then, as argued before, this gives $\statdist{\hyb_t}{\hyb_{t+1}} < 2^{2^{-\lambda/2 - Q(\lambda)}}$, finally yielding $\statdist{\hyb_0}{\hyb_{2^Q}} < {2^{-\lambda/2}}$. Hence, $\Pr[\hyb_{0} = 1] \geq 1/2 + 1/p(\lambda)$ by assumption and therefore $\Pr[\hyb_{2^Q} = 1] \geq 1/2 + 1/(2\cdot p(\lambda))$. However, this is clearly a contradiction, since the challenge messages in these hybrids are $m_b$ and $m_{1-b}$ respectively (while we still compare the adversary's guess to $b$).

\section{Public-Key Encryption with Copy-Protected Secret Keys}\label{sec:pke}
In this section, we define public-key encryption with copy-protected secret keys. Then, we give our construction based on coset states and prove it secure.
\subsection{Definitions}
\begin{definition}[Public-key Encryption with Copy-Protected Secret Keys]
A public-key encryption scheme with copy-protected secret keys consists of the following efficient algorithms.

\begin{itemize}
\item $\keygen(1^\lambda)$: Takes in the security parameter, output a classical secret key $sk$ and a public key $pk$.
    \item $\qkeygen(sk)$: Takes as input the classical secret key and outputs a quantum secret key.

    \item $\enc(pk, m)$: Takes in the public key and a message $m \in \mathcal{M}$, outputs and encryption of $m$.
    
    \item $\dec(\regi{dec}, ct)$: Takes in a quantum secret key and a ciphertext, outputs a message or $\perp$.
\end{itemize}

We require correctness\footnote{While our schemes satisfy perfect correctness, i.e., correctness with probability $1$, some work relax the definition to $1 - \negl(\lambda)$.} and \emph{CPA security}.
\paragraph{Correctness} For all messages $m \in \mathcal{M}$, \begin{equation*}
    \Pr[\dec(\regi{dec}, ct) = m : \begin{array}{c}
         pk, sk \samp \mathsf{Setup}(1^\lambda)  \\
         \regi{dec} \samp \qkeygen(sk) \\
         ct \samp \mathsf{Enc}(pk, m)
    \end{array}] = 1.
\end{equation*}

\paragraph{CPA Security} For any stateful QPT adversary $\adve$,
    \begin{equation*}
    \Pr[\adve(ct) = b : \begin{array}{c}
         pk, sk \samp \mathsf{Setup}(1^\lambda)  \\
         m_0, m_1 \samp \adve(pk, 1^\lambda)\\
         b \samp \zo \\
         ct \samp \mathsf{Enc}(pk, m_b)
    \end{array}] \leq \frac{1}{2} + \negl(\lambda).
    \end{equation*}
\end{definition}
As observed by \cite{CLLZ21}, correctness of the scheme along with Almost As Good As New Lemma (\cref{lem:asgoodasnew}) means that we can implement decryption in a way such that the quantum secret key is not disturbed. Thus, we can reuse the key to decrypt any number of times.

Following prior work, we will use two different security notions, regular anti-piracy and strong anti-piracy. The former will be the natural security notion while the latter definition is easier to work with when proving security. Both of our definitions follow \cite{CLLZ21, LLQZ22}, with the strengthening that we allow unbounded\footnote{Still polynomial since the adversary is QPT.} number of key queries and we also allow the adversary to choose different challenge messages for each freeloader.

Now, we move onto the first definition. In this definition, the pirate (or \emph{splitting}) adversary queries for copy-protected keys for any number of rounds. Then, if it has queried for $k$ keys, it outputs $k+1$ freeloaders, which are unitaries along with hardwired quantum states. More precisely, it outputs a $(k+1)$-partite (possibly entangled) register $\regi{adv}$ and unitaries $U_\ell$. Then, the challenger presents these freeloaders with challenge ciphertexts, and the adversary wins if all freeloaders correctly predict the challenges. Below, we write $\qunivcla$ to denote the quantum universal circuit $\qunivcla((U, \rho), x)$ that takes in a unitary ${U}$ and a state $\rho$, and simulates the \emph{induced} quantum circuit on input $x$ (i.e. computes ${U}(\rho, x)$), and finally measures the first output qubit in the computational basis. We note that the freeloaders being unitaries is not restrictive and actually captures general quantum circuits since the hardwired quantum state $(\regi{adv})_\ell$ can include\footnote{It will also include some quantum information that the pirate adversary has produced from the copy-protected keys.} workspace qubits initialized to zeroes.

\begin{definition}[CPA-Style Regular $\gamma$-Anti-Piracy Security]\label{defn:regularpkeantipir}
Let $\pke$ be a public key encryption scheme with copy-protected secret keys. Consider the following game between the challenger and an adversary $\adve$.
\paragraph{$\underline{\unclonepkegame(\lambda, \adve)}$}
\begin{enumerate}
    \item The challenger runs $sk, pk \samp \pke.\mathsf{Setup}(1^\lambda)$ and submits $pk$ to the adversary.
    \item For multiple rounds, $\adve$ makes quantum key queries. For each query, the challenger generates a key as $\reg \samp \pke.\qkeygen(sk)$ and submits $\reg$ to the adversary.
    \item $\adve$ outputs a $(k + 1)$-partite register $\regi{adv}$, unitaries $\{U_\ell\}_{\ell \in [k+1]}$ and challenge messages $\{m_\ell^0, m_\ell^1\}_{\ell \in [k+1]}$, where $k$ is the number of queries it made.
    \item The challenger executes the following for each $\ell \in [k + 1]$.
    \begin{enumerate}[label*=\arabic*.]
        \item $b_\ell \samp \zo$.
        \item $ct_\ell \samp \pke.\mathsf{Enc}(pk, m^{b_\ell}_\ell)$.
        \item $b'_\ell \samp \qunivcla(U_\ell, \regi{adv}[\ell], ct_\ell)$.
        \item Check if $b'_\ell = b_\ell$.
    \end{enumerate}
    \item The challenger outputs 1 if and only if all the checks pass.
\end{enumerate}
    We say that $\pke$ satisfies $\gamma$-anti-piracy security if for any QPT adversary $\adve$,
    \begin{equation*}
    \Pr[\unclonepkegame(\lambda, \adve) = 1] \leq \frac{1}{2} + \gamma(\lambda) + \negl(\lambda).
    \end{equation*}

    We ignore writing $\gamma$ when $\gamma = 0$.
\end{definition}
Note that we can also define a version where the freeloader adversaries try to guess the whole message $m_\ell$ where $m_\ell \samp \mathcal{M}$; and we require negligible probability of success. It is not known if this version is not implied by CPA security, see \cite{CLLZ21}. However, our construction will satisfy both notions.
Before moving onto the stronger definition, we need the following notation.
\begin{definition}[Decryptor Testing]\label{defn:goodecr}
    In the anti-piracy game between the challenger and an adversary, fix $\ell \in [k + 1]$, some values $m_\ell^0, m_\ell^1$ of the challenge messages, a freeloader unitary $U_\ell$, and some value $st$ of a classical state maintained by the challenger (which will be defined later). Let $\mathcal{D}$ be an efficient ciphertext distribution that can depend on $st$. That is, $\mathcal{D}^{st}(m; r)$ is an efficient classical algorithm where $m \in \mathcal{M}$, $r \in \mathcal{R}$ and $\mathcal{R}$ is a random coin set. 
    
    Consider the following mixture $\mathcal{P}$ of binary projective measurements, induced by $\mathcal{D}$ and $m_\ell^0, m_\ell^1, U_\ell, st$, applied on a state $\rho$.
    \begin{enumerate}
        \item Sample $b \samp \zo$.
        \item Sample $r \samp \mathcal{R}$.
        \item Run $ct \samp \mathcal{D}^{st}(m^b_\ell; r)$.
        \item Execute $U_\ell$ on $(\rho, ct)$, and measure the first qubit of the output registers, let $b'$ be the output.
        \item Output 1 if $b' = b$. Otherwise, output $0$.
    \end{enumerate}

    Observe that we can efficiently execute the above measurement\footnote{More formally, we are actually talking about the measurement where $r, b$ are fixed} for arbitrary given superpositions of $r$ and $b$ values. Therefore, by \cref{sec:piti}, there exists both exact and approximated projective and threshold implementations for $\mathcal{P}$. We write $\mathsf{PI}_{\ell, \mathcal{D}}$ and $\mathsf{API}^{ \eps, \delta}_{\ell, \mathcal{D}}$ to denote the projective implementation and approximate projective implementation of $\mathcal{P}$, respectively. Similarly, let $\mathsf{TI}_{\ell, \mathcal{D}, \eta}$ and $\mathsf{ATI}^{\eps, \delta}_{\ell, \mathcal{D}, \eta}$ denote the threshold and efficient approximate threshold implementations of $\mathcal{P}$ for a threshold value $\eta$. 
    
    While the fixed values $m_\ell^0, m_\ell^1, U_\ell, st$ are omitted from the notation, they will be clear from the context. Unless otherwise specified, we will write $\mathcal{D}$ to denote the honest ciphertext distribution, that is, we encrypt $m$ as \begin{equation*}
        ct \samp \mathsf{PKE.Enc}(pk, m)
    \end{equation*}
    where $pk$ is part of $st$.
\end{definition}

\begin{definition}[CPA-Style Strong $\gamma$-Anti-Piracy]\label{defn:strongantipir}
    Let $\pke$ be a public key encryption scheme with copy-protected secret keys. Consider the following game between the challenger and an adversary $\adve$.
    \paragraph{$\underline{\strongunclonpkegame(\lambda, \gamma(\lambda), \adve)}$}
\begin{enumerate}
    \item The challenger runs $sk, pk \samp \pke.\mathsf{Setup}(1^\lambda)$ and submits $pk$ to the adversary.
    \item For multiple rounds, $\adve$ makes quantum key queries. For each query, the challenger generates a key as $\reg \samp \pke.\qkeygen(sk)$ and submits $\reg$ to the adversary.
    \item $\adve$ outputs a $(k + 1)$-partite register $\regi{adv}$, unitaries $\{U_\ell\}_{\ell \in [k+1]}$ and challenge messages $\{m_\ell^0, m_\ell^1\}_{\ell \in [k+1]}$, where $k$ is the number of queries it made.
    \item The challenger applies the test
    \begin{equation*}
    \bigotimes_{\ell \in [k+1]} \mathsf{TI}_{\ell, \mathcal{D}, 1/2+\gamma}
    \end{equation*}
    to $\regi{adv}$ and outputs $1$ if and only if the measurement result is all $1$.
\end{enumerate}
We say that $\pke$ satisfies strong $\gamma$-anti-piracy security if for any QPT adversary $\adve$,
    \begin{equation*}
       \Pr[\strongunclonpkegame(\lambda, \gamma(\lambda), \adve)] \leq \negl(\lambda).
    \end{equation*}
\end{definition}

We also have the following relationship between the various security definitions for public-key encryption.
\begin{theorem}[\cite{CLLZ21}]
Suppose a public key encryption scheme with copy-protected keys satisfies CPA-style strong $\gamma$-anti-piracy (\cref{defn:strongantipir}). Then, it also satisfies CPA-style regular $\gamma$-anti-piracy (\cref{defn:regularpkeantipir}).
\end{theorem}
While \cite{CLLZ21} proves the above for only $1 \to 2$ anti-piracy security, it can be generalized to unbounded collusion setting - see proof of \cref{thm:festrongtoregular}.

\begin{theorem}[\cite{CLLZ21}]
    Suppose a public key encryption scheme with copy-protected keys satisfies CPA-style regular $\gamma$-anti-piracy (\cref{defn:regularpkeantipir}) for any inverse polynomial $\gamma$. Then, it also satisfies regular CPA security and regular $\gamma$-anti-piracy for $\gamma = 0$.
\end{theorem}
This is simply due to the definition of $\negl(\lambda)$ and $\gamma$-anti-piracy.

\subsection{Construction}\label{sec:pkecons}
In this section, we present our construction. Assume the existence of following primitives where we set $\nu(\lambda) = 2^{-6\lambda}\cdot2^{-8\lambda^{0.3\constmoecoll}}$.
\begin{itemize}
    \item $\io$, indistinguishability obfuscation scheme that is $\nu(\lambda)$-secure against $2^{5\lambda}\cdot2^{8\lambda^{0.3\constmoecoll}}$-time adversaries,
    \item $\mathsf{IBE}$, identity-based encryption scheme for the identity space $\mathcal{ID} = \zo^\lambda$ (\cref{defn:ibe}) that is $\nu(\lambda)$-secure against $2^{5\lambda}\cdot2^{8\lambda^{0.3\constmoecoll}}$-time adversaries,
    \item $F_1$, puncturable PRF family with input length $\lambda$ and output length same as the size of the randomness used by $\mathsf{CosetGen}$ (\cref{defn:cosetgen}) that is $\nu(\lambda)$-secure against $2^{5\lambda}\cdot2^{8\lambda^{0.3\constmoecoll}}$-time adversaries,
    \item $F_2$, puncturable PRF family with input length $\lambda$ and output length same as the size of the randomness used by $\mathsf{IBE.Enc}$ that is $\nu(\lambda)$-secure against $2^{5\lambda}\cdot2^{8\lambda^{0.3\constmoecoll}}$-time adversaries,
    \item $\mathsf{CCObf}$, compute-and-compare obfuscation for $2^{-\lambda^{0.2\cdot {\constmoecoll}}}$-unpredictable distributions that is $2^{-2\lambda - 1}\cdot 2^{-2\lambda^{0.3\constmoecoll}}$-secure against $2^{3\lambda}\cdot 2^{2\lambda^{0.3\constmoecoll}}$-time adversaries,
\end{itemize}
A remark is in order regarding our assumptions. We note that all of our assumptions above can be based on any subexponential $\io$ and LWE assumption. For example, if we have an $\io$ scheme that is $2^{-\lambda^{c_1}}$-secure against $2^{\lambda^{c_2}}$-time adversaries; in our construction we implicitly initiate it with security parameter $\lambda^{c'}$ where $c' = \max\{0.3\constmoecoll/c_1,0.3\constmoecoll/c_2\}$. While this might require larger padding for obfuscated circuits, this is still within polynomial factors. The same applies for the other primitives. Thus, our assumptions can be based solely on subexponential hardness for any exponent, since we can always scale the security parameter by a polynomial factor when instantiating the underlying primitives.

Set $L(\lambda) = \lambda$ and therefore $c_L(\lambda) = 24\cdot\lambda^3$ (see \cref{defn:strmoecoll}). We also assume that all obfuscated programs in the construction and in the proof are appropriately padded.

We now give our construction for public-key encryption with copy-protected secret keys. 
\paragraph{$\underline{\mathsf{PKE.Setup}(1^{\lambda})}$}
\begin{enumerate}
    \item Sample a PRF key $K_1 \samp F_1.\mathsf{KeyGen}(1^\lambda)$.
    \item Sample $cpk, csmk \samp \mathsf{IBE.Setup}(1^\lambda)$.
    \item Sample $\mathsf{OPMem} \samp \io(\mathsf{PMem}_{K_1})$, where $\mathsf{PMem}_{K_1}$ is the following program.
    
\begin{mdframed}\label{code:pmemorig}
        {\bf $\underline{\mathsf{PMem}_{K_1}(id, u_1, \dots, u_{\cosettcount}, r)}$}
        
        {\bf Hardcoded: $K_1$}
        \begin{enumerate}[label=\arabic*.]
            \item $(A_i, s_i, s_i')_{i \in [\cosettcount]} \samp \mathsf{CosetGen}(1^{L(\lambda) + \lambda}; F_1(K_1, id))$.
            \item For each $i \in [\cosettcount]$, check if $u_i \in A_i + s_i$ if $(r)_i = 0$ and check if $u_i \in A^{\perp}_i + s'_i$ if $(r)_i = 1$. If any of the checks fail, output $0$ and terminate.
            \item Output $1$.
        \end{enumerate}

    \end{mdframed}

    \item Set $pk = (cpk, \mathsf{OPMem})$ and $sk = (cmsk, K_1)$.
    \item Output $(pk, sk)$.
\end{enumerate}

\paragraph{$\underline{\mathsf{PKE.QKeyGen}(sk)}$}
\begin{enumerate}
    \item Parse $(cmsk, K_1) = sk$.
    \item Sample $id \samp \zo^{\lambda}$.
    \item $(A_i, s_i, s_i')_{i \in [\cosettcount]} = \mathsf{CosetGen}(1^{L(\lambda) + \lambda}; F_1(K_1, id))$.
    \item $ck \samp \mathsf{IBE.KeyGen}(cmsk, id)$.
    
    \item Output $\left(\ket{A_{i, s_i, s'_i}}\right)_{i \in [\cosettcount]}, ck, id$.
\end{enumerate}

\paragraph{$\underline{\mathsf{PKE.Enc}(pk, m)}$}
\begin{enumerate}
    \item Parse $(cpk, \mathsf{OPMem}) = pk$.
    \item Sample $r \samp \zo^{\cosettcount}$.
    \item Sample a PRF key $K_2$ for $F_2$ as $K_2 \samp F_2.\mathsf{KeyGen}(1^{\lambda})$.
    \item Sample $\mathsf{OPCt} \samp \io(\mathsf{PCt}_{\mathsf{OPMem}, cpk, K_2, r, m})$, where $\mathsf{PCt}_{\mathsf{OPMem}, cpk, K_2, r, m}$ is the following program.
 \begin{mdframed}
        {\bf $\underline{\mathsf{PCt}_{\mathsf{OPMem}, cpk, K_2, r, m}(id, u_1, \dots, u_{\cosettcount})}$}
        
        {\bf Hardcoded: $\mathsf{OPMem}, cpk, K_2, r, m$}
        \begin{enumerate}[label=\arabic*.]
            \item Run $\mathsf{OPMem}(id, u_1, \dots, u_{\cosettcount}, r)$. If it outputs $0$, output $\perp$ and terminate.
            \item Output $\mathsf{IBE.Enc}(cpk, id, m; F_2(K_2, id))$.
        \end{enumerate}
    \end{mdframed}
    \item Output $(\mathsf{OPCt}, r)$.
\end{enumerate}

\paragraph{$\underline{\mathsf{PKE.Dec}(\regi{key}, ct)}$}
\begin{enumerate}
    \item Parse $((\reg_i)_{i \in [\cosettcount]}, ck, id) = \regi{key}$ and $(\mathsf{OPCt}, r) = ct$.
    \item For indices $i \in [\cosettcount]$ such that $(r)_i = 1$, apply $H^{\otimes \kappa(L(\lambda) + \lambda)}$ to $\reg_i$.
    \item Run the program $\mathsf{OPCt}$ coherently on $id$ and $(\reg_i)_{i \in [\cosettcount]}$.
    \item Measure the output register and denote the outcome by $cct$.
    \item Output $\mathsf{IBE.Dec}(ck, cct)$.
\end{enumerate}

Correctness with probability $1$ follows in a straightforward manner from the correctness of the underlying schemes. We claim that the construction is also secure.
\begin{theorem}\label{thm:pkeantipiracy}
 $\mathsf{PKE}$ satisfies strong $\gamma$-anti-piracy for any inverse polynomial $\gamma$.
\end{theorem}
When we instantiate the assumed building blocks with known constructions, we get the following corollary.
\begin{corollary}\label{thm:cppkeexists}
Assuming subexponentially secure $i\mathcal{O}$ and subexponentially secure LWE, there exists a public-key encryption scheme that satisfies anti-piracy security against unbounded collusion.
\end{corollary}
\begin{proof}
    $\mathsf{IBE}$ can be constructed based on $\io$ and one-way functions (\cref{thm:subexpibe}). $F_1$ and $F_2$  can be constructed based on one-way functions (\cref{thm:puncprfexists}). which in turn can be obtained from LWE. $\mathsf{CCObf}$ can be constructed based on $\io$ and LWE (\cref{thm:ccobf}).
\end{proof}

\subsection{Proof of Strong Anti-Piracy Security}\label{sec:pkeproof}
In this section, we will prove \cref{thm:pkeantipiracy}. We note that our construction also satisfies random challenge anti-piracy security; we only give the full proof for \cref{thm:pkeantipiracy} and the  random challenge anti-piracy security follows by mostly the same proof. 

Throughout the proof, we will interpret identity strings for $\mathsf{IBE}$, which are $\lambda$-bit strings, as integers in the set $\{0, 1, \dots, 2^\lambda - 1\}$. Without loss of generality, we assume that $\mathsf{IBE}$ can also encrypt the symbol $\failmes$, which is outside the message space $\mathcal{M}$ for $\mathsf{PKE}$.

Fix any inverse polynomial $\gamma(\lambda)$ and suppose for a contradiction that there exists an efficient adversary $\adve$ that wins the strong $\gamma$-anti-piracy game with non-negligible probability. Let $k$ denote the number of keys obtained by the adversary. 
Define $\hyb_0$ to be the original security game $\strongunclonpkegame(\lambda, \gamma(\lambda), \adve)$.

\subsubsection*{Making Key Identities Unique}
Define $\hyb_1$ by modifying $\hyb_0$ as follows. We change the way we sample the identity strings in $\mathsf{PKE.QKeyGen}$ during each quantum key query. Let the challenger record each sampled identity when answering each query, and when answering a new query, it samples uniformly at random an identity value from the set $\{1, \dots, 2^\lambda - 1\}$ \emph{that has not appeared before}\footnote{Note that this can be done on-the-go in polynomial time, with overwhelming probability, e.g. through rejection sampling}. That is, we sample unique identity strings for each quantum key. Also, we define the following notation. Let $id_{\alpha(i)}$ be the $i^{th}$ identity value sampled where $\alpha(\cdot)$ is the permutation $[k] \to [k]$ such that $0 < id_1 < \dots < id_k < 2^{\lambda}$. That is, $id_{\alpha(i)}$ is the identity string that is sampling during the $i^{th}$ query of the adversary. For simplicity of notation, we also set $id_0 = 0$ and $id_{k + 1} = 2^\lambda$.

\begin{claim}
$\statdist{\hyb_0}{\hyb_1} < \exp(-\lambda)$. 
\end{claim}
\begin{proof}
An easy calculation shows that uniformly and independently sampling from $\{0, 1, \dots, 2^\lambda - 1\}$ $k$ times gives $k$ unique values from the set $\{1, \dots, 2^\lambda - 1\}$ with probability at least
\begin{equation*}
    1 - \frac{k^2(\lambda)}{2^\lambda}.
\end{equation*} 
The result follows since $k(\cdot)$ is a polynomial.
\end{proof}

\subsubsection*{Making the Challenger Efficient}
Define $\hyb_2$ by modifying $\hyb_1$ as follows. At the end of the game, instead of applying threshold implementations $\mathsf{TI}_{\ell, \mathcal{D}, 1/2+\gamma}$, the challenger applies approximate threshold implementations $\mathsf{ATI}^{\eps, \delta}_{\ell, \mathcal{D}, 1/2+\frac{31\gamma}{32}}$ with $\eps = \frac{\gamma}{32k}$ and $\delta = 2^{-10\lambda}\cdot2^{-10\lambda^{\constmoecoll}}$. It outputs $1$ if and only if all $\mathsf{ATI}$ output $1$.

\begin{claim}
    $\Pr[\hyb_2 = 1] > \Pr[\hyb_1 = 1] - {\exp(-\lambda)}$.
\end{claim}
\begin{proof}
Let $\sigma$ be the $(k + 1)$-partite state output by the adversary. By \cref{thm:multiatiprop}, we get
\begin{equation*}
    \Tr[\left(\bigotimes_{\ell \in [k+1]}\mathsf{ATI}^{\eps, \delta}_{\ell,\mathcal{D}, 1/2+\frac{31\gamma}{32}}\right)\sigma] \geq \Tr[\left(\bigotimes_{\ell \in [k+1]}\mathsf{TI}_{\ell, \mathcal{D}, 1/2+\gamma}\right)\sigma] - (k(\lambda) + 1)\cdot \exp(-\lambda).
\end{equation*}
Observe that the trace expressions on the left-hand side and the right-hand side are the winning probabilities in $\hyb_2$ and $\hyb_1$ respectively. The result follows since $k(\lambda)$ is a polynomial.
\end{proof}

Therefore, $\adve$ wins in $\hyb_2$ with probability $\frac{1}{p(\cdot)}$ for some polynomial $p(\cdot)$ and infinitely many values of $\lambda > 0$. Note that in $\hyb_2$, now the challenger is also efficient by \cref{thm:singleatiprop} and our choice of $\eps, \delta$.

The rest of the proof will be devoted to showing that using $\adve$, we can construct an adversary that breaks the selective monogamy-of-entanglement game $\moecollsel$ (\cref{defn:strmoecoll}). We will use projective and threshold implementations for various mixtures of measurements to test the freeloaders. The public key $pk$, the identity strings $id_1, \dots, id_k$ and the permutation $\alpha$ will be part of the classical state $st$ of the challenger, in the sense of \cref{defn:goodecr}. The particular distribution on the collection of projective measurements (induced by a challenge ciphertext distribution) will vary, and it will be denoted explicitly.

\begin{definition}
    For all $j \in [k]$, let $(A_i^j, s_i^j, s_i^{'j})_{{i \in [\cosettcount]}}$ denote the tuple of subspaces and vectors sampled during the sampling of the $(\alpha^{-1}(j))$-th key. That is, it is the coset tuple associated with $id_j$.
\end{definition}

\subsubsection*{A Monogamy-of-Entanglement Type Game}
First, we define the following monogamy-of-entanglement-type game $\mathcal{G}$ for a tuple of adversaries $(\adve'_0, \adve'_1, \adve'_2)$. Observe that it will be straightforward to reduce the game $\mathcal{G}$ to $\moecollsel$ with no loss of security, since the former is the same as the latter except that it includes an independent $\mathsf{IBE}$ instance that can sampled by the reduction.
\paragraph{$\underline{\mathcal{G}(\lambda, (\adve'_0, \adve'_1, \adve'_2))}$}
\begin{enumerate}
\item The adversary outputs an index $j^* \in [k]$.
    \item The challenger executes $pk, sk \samp \mathsf{PKE.Setup}(1^\lambda)$ and submits $pk$ to $\adve'_0$.
    \item For $k$ rounds, $\adve'$ makes quantum key queries. For each query, the challenger samples a quantum key as in $\mathsf{PKE.QKeyGen}$, but by sampling the identity $id$ in a collision-free way (as in $\hyb_1$), and submits it to $\adve_0'$.
    \item The adversary outputs a \emph{bipartite} register $\regi{bip}$.
    \item For $\ell \in \{1, 2\}$, the challenger does the following.
    \begin{enumerate}[label*=\arabic*.]
        \item Sample $r_\ell \samp \zo^{\cosettcount}$.
        \item Run $\adve'_\ell$ on $\regi{bip}[\ell]$, $(A^{j^*}_i)_{i \in [\cosettcount]}$ and $r_\ell$ to obtain a tuple of vectors $(v_{\ell, i})_{i \in [{\cosettcount}]}$.
        \item For all $i \in [\cosettcount]$, check if $v_{\ell, i} \in A^{j^*}_{i} + s^{j^*}_i$ if $(r_\ell)_i = 0$ and check if $v_{\ell, i} \in {(A^{j^*})}^\perp_{i} + s^{'{j^*}}_i$ if $(r_\ell)_i = 1$.
    \end{enumerate}
    If all the checks pass, the challenger outputs $1$. Otherwise, it outputs $0$.
\end{enumerate}

Now, we construct a tuple of adversaries $(\adve'_0, \adve'_1, \adve'_2)$ for $\mathcal{G}$, starting with $\adve'_0$. Let $\mathcal{D}_j$ for $j \in \{0, \dots, k + 1\}$ be efficient ciphertext distributions, which we will define later. 
\begin{mdframed}
        {\bf $\underline{\adve_0'(pk)}$}
        \begin{enumerate}
        \item Uniformly at random sample $x, y, j^*$ such that $1 \leq x < y \leq k + 1$ and $j^* \in \{1, \dots, k\}$. Output $j^*$.
            \item Simulate $\adve$ on $pk$ by making a quantum secret key query to the challenger whenever $\adve$ makes a query, and forwarding the obtained key to it. Let $\regi{adv}$ be the $(k + 1)$-partite register (with state $\sigma$), let $(m_\ell^0, m_\ell^1)_{\ell \in [k + 1]}$ be the challenge messages and let $(U_\ell)_{\ell \in [k + 1]}$ be the unitaries output by $\adve$ at the end of the query phase.
            \item Apply $\mathsf{API}^{\eps, \delta}_{\ell, \mathcal{D}_0}$ to all registers $\regi{adv}[\ell]$ for $\ell \in [k + 1]$, let $b_{\ell, 0}$ be the measurement outcomes.
            \item Apply $\mathsf{API}^{\eps, \delta}_{\ell, \mathcal{D}_i}$ in succession for $i = 1$ to $j^*$ to $\regi{adv}[x]$, let $b_{x, i}$ be the measurement outcomes.
            \item Apply $\mathsf{API}^{\eps, \delta}_{\ell, \mathcal{D}_i}$ in succession for $i = 1$ to $j^*$ to $\regi{adv}[y]$, let $b_{y, i}$ be the measurement outcomes.
            \item Output  \begin{align*}
                (&(\regi{adv}[x], j^*, x, y, (b_{\ell, 0})_{\ell \in [k + 1]}, (b_{x,i})_{i \in [j^*]}, (b_{y,i})_{i \in [j^*]},(U_\ell)_{\ell \in [k + 1]}),\\ &(\regi{adv}[y], j^*, x, y, (b_{\ell, 0})_{\ell \in [k + 1]}, (b_{x,i})_{i \in [j^*]}, (b_{y,i})_{i \in [j^*]}, (U_\ell)_{\ell \in [k + 1]}),\\ &j^*).
            \end{align*}
            
        \end{enumerate}
    \end{mdframed}

For $j \in \{1, \dots, k\}$, define $\mathcal{D}_j$ to be the challenge ciphertext distribution where an encryption of a message $m$ is computed as follows.
\begin{enumerate}
   \item Sample $r \samp \zo^{\cosettcount}$.
    \item Sample a PRF key $K_2$ for $F_2.\mathsf{KeyGen}(1^\lambda)$.
    \item Sample $\mathsf{OPCt} \samp \io(\textcolor{red}{ \mathsf{PCt}^{(j)}_{\mathsf{OPMem}, cpk, K_2, r, m, id_j}})$
    \begin{mdframed}
        {\bf $\underline{\mathsf{PCt}^{(j)}_{\mathsf{OPMem}, cpk, K_2, r, m, id_j}(id, u_1, \dots, u_{\cosettcount})}$}
        
        {\bf Hardcoded: $\mathsf{OPMem}, cpk, K_2, r, m, \textcolor{red}{id_j}$}
        \begin{enumerate}[label=\arabic*.]
            \item Run $\mathsf{OPMem}(id, u_1, \dots, u_{\cosettcount}, r)$. If it outputs $0$, output $\perp$ and terminate.
            \item \textcolor{red}{If $id < id_j$, set $a = \failmes$. Otherwise, set $a = m$.}
            \item Output $\mathsf{IBE.Enc}(cpk, id, \textcolor{red}{a}; F_2(K_2, id))$.
        \end{enumerate}
    \end{mdframed}
    \item Output $(\mathsf{OPCt}, r)$.
\end{enumerate}
We define $\mathcal{D}_{0}$ to be the honest ciphertext distribution $\mathcal{D}$ and we define $\mathcal{D}_{k + 1}$ as follows.
\begin{enumerate}
   \item Sample $r \samp \zo^{\cosettcount}$.
    \item Sample a PRF key $K_2$ for $F_2.\mathsf{KeyGen}(1^\lambda)$.
    \item Sample $\mathsf{OPCt} \samp \io(\textcolor{red}{ \mathsf{PCt}^{(k+1)}_{\mathsf{OPMem}, cpk, K_2, r}})$
    \begin{mdframed}
        {\bf $\underline{\mathsf{PCt}^{(k+1)}_{\mathsf{OPMem}, cpk, K_2, r}(id, u_1, \dots, u_{\cosettcount})}$}
        
        {\bf Hardcoded: $\mathsf{OPMem}, cpk, K_2, r$}
        \begin{enumerate}[label=\arabic*.]
            \item Run $\mathsf{OPMem}(id, u_1, \dots, u_{\cosettcount}, r)$. If it outputs $0$, output $\perp$ and terminate.
            \item Output $\mathsf{IBE.Enc}(cpk, id, \textcolor{red}{\failmes}; F_2(K_2, id))$.
        \end{enumerate}
    \end{mdframed}
    \item Output $(\mathsf{OPCt}, r)$.
\end{enumerate}
Note that the distribution $\mathcal{D}_{k+1}$ does not actually use the message $m$.

Observe that $\adve_0'$ can indeed execute $\mathsf{API}^{\eps, \delta}_{\ell, \mathcal{D}_i}$. The identity strings $id_j$ are part of the quantum secret keys. Further, the adversary can record the order in which the identity strings are received and also their sorted version, so it can indeed index them as $id_j$. 

Finally, we claim that there exists efficient $\adve_1', \adve_2'$ such that  $(\adve'_0, \adve_1', \adve_2')$ wins $\mathcal{G}$ with probability $$\frac{1}{ 2^{0.4\cdot\lambda^{\constmoecoll}}}.$$
We will construct these adversaries in the rest of the proof, and at the end we will show that the security of $\mathcal{G}$ can be reduced to $\moecollsel$, arriving at a contradiction.

\subsubsection*{Finding a Simultaneous Jump}
In the rest of the proof, we will formalize the following fact. Observe that the adversary only obtains $k$ different identity keys for $\mathsf{IBE}$. Therefore, informally, by the security of $\mathsf{IBE}$ and by the pigeonhole principle, two of the $k+1$ freeloaders must be using $\mathsf{IBE}$ encryptions of their challenge message under the same identity $id_j$ to decode their challenge ciphertext. This will in turn mean that they must be using the same coset state tuple, and therefore we will extract coset vectors for the same identity/tuple, which will be a contradiction by the monogamy-of-entanglement property.

As a first step, we will show that the decryption success probabilities  of the two freeloaders $x, y$, when tested with respect to the distributions $\mathcal{D}_{j}$, jump at the same index $j^*$, meaning that the two freeloaders use the same identity string block $[id_{j^*}, id_{j^* + 1} - 1]$ to decrypt.

\begin{claim}\label{claim:mesresdist}

Let $\tau$ be the state of the bipartite register $\regi{adv}[x,y]$ output by $\adve'_0$ in $\mathcal{G}$, and also consider the classical values $j^*, x, y, \{b_{\ell, i}\}_{\ell, i}$ contained in the output of $\adve'_0$.
    
    Suppose we apply the measurement $\mathsf{API}^{\eps, \delta}_{x, \mathcal{D}_{j^* + 1}}\otimes\mathsf{API}^{\eps, \delta}_{y, \mathcal{D}_{j^* + 1}}$ to $\tau$ and let $b_{x,j^*+1}, b_{y,j^*+1}$ denote the measurement outcomes we obtain. Then,
   
    \begin{equation*}
     \Pr[ b_{x,j^*}  - b_{x,j^*+1} > \frac{29\gamma}{32k}  \wedge  b_{y,j^*} - b_{y,j^*+1} > \frac{29\gamma}{32k}] > \frac{1}{4p(\lambda)\cdot k^3(\lambda)}
    \end{equation*}
    
    where the probability is taken over the randomness of the challenger, the adversary $\adve_0'$ and the measurement outcomes.
\end{claim}

First we define the following notation.

\begin{definition}
    Let $\expfromddd{\mathcal{C}}{\ell}$ denote the outcome of the following experiment where $\mathcal{C}$ is a ciphertext distribution that can depend on $pp$.
\begin{enumerate}
            \item Uniformly at random sample $x, y, j^*$ such that $1 \leq x < y \leq k + 1$ and $j^* \in \{1, \dots, k\}$.
\item Execute $pk, sk \samp \mathsf{PKE.Setup}(1^\lambda)$.
\item Simulate the first step of $\adve'_0$ and the challenger of $\mathcal{G}$:
\begin{enumerate}[label*=\arabic*.]
    \item Simulate $\adve$ on $pk$ by sampling a quantum secret key (as in $\hyb_1$) whenever $\adve$ makes a query, and submitting the key to it. Let $\regi{adv}, (m_\ell^0, m_\ell^1)_{\ell \in [k + 1]}, (U_\ell)_{\ell \in [k + 1]}$ be the output of $\adve$.
\end{enumerate}
    \item Set  $pp =  (x, y, j^*, (id_j)_{j\in[k+1]}, (m_\ell^0, m_\ell^1)_{\ell \in [k + 1]}, (U_\ell)_{\ell \in [k + 1]}, pk)$.
    \item Sample $b \samp \zo$.
    \item Sample $ct \samp \mathcal{C}(pp, m_\ell^b)$.
    \item Output $\regi{adv}, (b, ct), pp$.
\end{enumerate} 
We will write $\expfromddd{\mathcal{C}}{\ell} \approx_\nu^c \expfromddd{\mathcal{C}'}{\ell}$ to denote that the advantage of any computational adversary in distinguishing the outcomes of these experiments is $\nu$.
\end{definition}

\begin{proof}[Proof of \cref{claim:mesresdist}]
Consider instead the following modified version of $\adve_0'$. We run $\mathsf{API}^{\eps, \delta}_{\ell, \mathcal{D}_i}$ in succession from $i = 0$ to $i = k + 1$ on all registers $\ell \in [k + 1]$ of $\regi{adv}$, to obtain values $b'_{\ell, i}$. While the ordering of execution between the registers does not matter, since local operations on disjoint registers commute, for convenience, assume that we run $\mathsf{API}^{\eps, \delta}_{\ell, \mathcal{D}_i}$ on all registers before moving onto $\mathsf{API}^{\eps, \delta}_{\ell, \mathcal{D}_{i+1}}$. Let $\rho_i$ denote the post-measurement state after having run $\mathsf{API}^{\eps, \delta}_{\ell, \mathcal{D}_i}$ on all sub-registers. 

First, we claim that 
\begin{equation}\label{eqn:problk}
\Pr[\forall \ell \in [k + 1]~~b'_{\ell, {k+1}} < 1/2+\frac{2\gamma}{32} \bigg| \forall \ell \in [k + 1] \forall i \in \{0, \dots, k\}~~b'_{\ell, i} = b''_{\ell, i}] \geq 1 -  (k(\lambda) + 1)\cdot\exp(-\lambda).
\end{equation}
for any fixed tuple of values $(b''_{\ell, i})_{\ell \in [k + 1], i \in \{0, \dots, k\}}$ in the joint support of $(b'_{\ell, i})_{\ell \in [k + 1], i \in \{0, \dots, k\}}$.
 To prove this, we will instead prove the more general statement that for any quantum state $\xi$ of appropriate dimension, we have $$\Pr[\forall \ell \in [k + 1] ~ x_\ell <\frac{1}{2}+ \frac{2\gamma}{32}]  \geq 1 - (k(\lambda) + 1)\cdot\exp(-\lambda).$$
 where $(x_\ell)_{\ell \in [k + 1]} \samp \left(\bigotimes_{\ell \in[k+1]}\mathsf{API}^{\eps,\delta}_{\ell,\mathcal{D}_{k+1}}\right)\cdot\xi$.
 
 Let $\iota$ be any quantum state of appropriate dimension. By \cref{thm:singleatiprop}, we have for all $\ell \in [k + 1]$
 \begin{align*}
     &\Pr[\left(\mathsf{API}^{\eps,\delta}_{\ell,\mathcal{D}_{k+1}}\right)\cdot\iota \geq  \frac{1}{2}+ \frac{2\gamma}{32}] \\\leq &\Pr[\left(\mathsf{PI}_{\ell,\mathcal{D}_{k+1}}\right)\cdot\iota \geq \frac{1}{2}+  \frac{\gamma}{32}] + \exp(-\lambda).
 \end{align*} Then, by \cref{thm:piprop}, we have that if the outcome of $\mathsf{PI}_{\ell, \mathcal{D}_{k+1}}$ is $p'$, then the post-measurement state has success probability $p'$ for the distribution $\mathcal{D}_{k+1}$. However, the challenge ciphertext sampled according to $\mathcal{D}_{k+1}$ is independent of the challenge bit $b$, hence we always have $p' \leq 1/2$. Hence,
 $$
 \Pr[\left(\mathsf{PI}_{\ell,\mathcal{D}_{k+1}}\right)\cdot\iota \geq  \frac{1}{2}+ \frac{\gamma}{32}] = 0.
 $$
 Therefore, $\Pr[\left(\mathsf{API}^{\eps,\delta}_{\ell,\mathcal{D}_{k+1}}\right)\cdot\iota \geq  \frac{1}{2}+ \frac{2\gamma}{32}] \leq \exp(-\lambda)$. Now, if we apply $\mathsf{API}^{\eps,\delta}_{\ell,\mathcal{D}_{k+1}}$ to each part $\xi[i]$, even conditioned on some outcome obtained for the other parts, we get that the result will be $\geq 1/2 + 2\gamma/32$ with probability at most $\exp(-\lambda)$, since we showed the result above for any state $\iota$. Hence, probability of obtaining an outcome $\geq 1/2 + 2\gamma/32$ for at least one part is at most $(k(\lambda) + 1)\cdot\exp(-\lambda)$.
 This gives the desired result (\cref{eqn:problk}).

Now, we claim that we have $b'_{\ell,1} \geq \frac{1}{2} + \frac{31\gamma}{32}$ for all $\ell \in [k +  1]$ with probability $1/(2p(\lambda))$. First, by assumption we have
\begin{equation*}
\Pr[\left(\bigotimes_{\ell \in [k+1]}\mathsf{ATI}^{\eps,\delta}_{\ell,\mathcal{D}, 1/2+\frac{31\gamma}{32}}\right)\sigma]  \geq 1/p(\lambda).
\end{equation*}
since this is exactly the winning condition in $\hyb_2$. While we later apply other measurements, they do not change the marginal distribution of the initial measurement since we cannot signal backwards in time.

Assume for now that  $\expfromddd{\mathcal{D}}{\ell} \approx^c \expfromddd{\mathcal{D}_1}{\ell}$ for all $\ell \in [k+1]$ and we will prove it later (\cref{claim:d0closed1}).
Then, by \cref{thm:distti} and by above we get
\begin{align}
   \label{eqn:atid1} &\Pr[\left(\bigotimes_{\ell \in [k+1]}\mathsf{ATI}^{\eps,\delta}_{\ell,\mathcal{D}_1, 1/2+\frac{31\gamma}{32}}\right)\sigma]   \\\geq
     &\Pr[\left(\bigotimes_{\ell \in [k+1]}\mathsf{ATI}^{\eps,\delta}_{\ell,\mathcal{D}, 1/2+\frac{31\gamma}{32}}\right)\sigma] - \negl(\lambda) > 1/(2\cdot p(\lambda)).
\end{align}

In \cref{thm:distti}, it is easy to see $\expfromddd{\mathcal{D}}{\ell}$ corresponds to $(\mathcal{S}, \mathcal{D})$ and $\expfromddd{\mathcal{D}_1}{\ell}$ corresponds to $(\mathcal{S}, \mathcal{D}_1)$; while the measurement results $\vec{p_0}, \vec{p_1}$ correspond to $ \bigotimes_{\ell \in [k+1]}\mathsf{API}^{\eps,\delta}_{\ell,\mathcal{D}}\instalper$ and $\bigotimes_{\ell \in [k+1]}\mathsf{API}^{\eps,\delta}_{\ell,\mathcal{D}_1}\instalper$ when we define our collection of measurements as in \cref{defn:goodecr}. That is, our measurement is executing the given state as a decryptor using $\qunivcla$ and comparing the outcome to $b$.

Finally, by combining \cref{eqn:problk} and \cref{eqn:atid1}, we get that with probability at least $1/(4\cdot p(\lambda))$, we have that $\frac{1}{2} + \frac{31\gamma}{32} \leq b'_{\ell,1}$ and $b'_{\ell, k + 1} < \frac{1}{2} + \frac{2\gamma}{32}$ for all $\ell \in [k + 1]$. Hence, we see that with probability at least $\frac{1}{4p(\lambda)}$; for all $\ell \in [k+1]$ there is $i_{\ell} \in \{1, \dots, k\}$ such that $b'_{\ell, i_{\ell} } - b'_{\ell, i_{\ell} + 1} > \frac{29\gamma}{32k}$. Then, by pigeonhole principle, there is $\ell \neq \ell'$ such that $i_{\ell} = i_{\ell'}$. 

We claim that for any fixed $x < y \in [k + 1]$ and $j^* \in [k]$, the marginal distribution (i.e., the reduced density matrix) of $\rho_{j^*}[x, y], (b'_{\ell, 0})_{\ell \in [k + 1]}, (b'_{x,i})_{i \in [j^* + 1]}, (b'_{y,i})_{i \in [j^*+1]}$ in the above experiment is the same as the distribution of $\tau, (b_{\ell, 0})_{\ell \in [k + 1]}, (b_{x,i})_{i \in [j^* + 1]}, (b_{y,i})_{i \in [j^* + 1]}$ conditioned on the fixed values of $x, y, j^*$. This follows from two arguments. First, no-signalling between disjoint registers gives that whether or not we apply measurements on the other registers does not change the marginal distributions of measurement outcomes and post-measurement states on registers $x, y$. Similarly, by the time we are applying measurements for $\mathcal{D}_{i}$ for $i \geq j^* + 1$, the measurement outcomes for $\mathcal{D}_{j^*}$ are already determined. Since it is not possible to signal backwards in time, the marginal distributions for measurement outcomes $b_{\ell, j^*}$ is not affected by whether or not we apply the measurements for $\mathcal{D}_{i}$ for $i \geq j^* + 1$.

We have already shown that with probability $1/(4p(\lambda))$, there is guaranteed to be a \emph{jump} in measurement results. Since $x, y, j^*$ are sampled independently by $\adve'_0$, they hit the correct indices $\ell, \ell'$ satisfying $i_{\ell} = i_{\ell'}$ with probability $1/k{k+1 \choose 2}$ and $j^*$ hits $i_{\ell} = i_{\ell'}$ with probability $1/k$. Therefore, we finally have
  \begin{equation*}
     \Pr[ b_{x,j^*} - b_{x,j^*+1} > \frac{29\gamma}{32k} \wedge b_{y,j^*} - b_{y,j^*+1} > \frac{29\gamma}{32k} ] > \frac{1}{4 \cdot p(\lambda )\cdot k^3(\lambda)}.
    \end{equation*}
 \end{proof}

We have shown that the freeloaders $x, y$ use the same identity string block $[id_{j^*}, id_{j^* + 1} - 1]$ to decrypt. Now we will further show that they use the exact same identity string $id_{j^*}$. To that end, we first define some intermediary challenge ciphertext distributions. Define the following for all $j \in \{0, 1, \dots, k\}$ and $\Delta \in \{0, 1, \dots, id_{j+1}-id_{j} - 1\}$. For notational convenience, also define $\mathcal{D}_j^{id_{j+1}-id_{j}, 0}$ to be $\mathcal{D}_{j+1}^{(0,  0)}$ for all $j \in \{0, 1, \dots, k\}$. Also note that $\mathcal{D}_j^{(0,0)}$ is exactly the same as $\mathcal{D}_j$ for $j \in [k]$.

\begin{itemize}
    \item{$\underline{\mathcal{D}_{j}^{(\Delta, 0)}(m)}$:} 
    \begin{enumerate}
   \item Sample $r \samp \zo^{\cosettcount}$.
    \item Sample a PRF key $K_2$ for $F_2.\mathsf{KeyGen}(1^\lambda)$.
    \item Sample $\mathsf{OPCt} \samp \io({ \mathsf{PCt}^{(j, \Delta, 0)}_{\mathsf{OPMem}, cpk, K_2, r, m, id_j + \Delta}}).$
 \begin{mdframed}
        {\bf $\underline{{ \mathsf{PCt}^{(j, \Delta, 0)}_{\mathsf{OPMem}, cpk, K_2, r, m, id_j + \Delta}}(id, u_1, \dots, u_{\cosettcount})}$}
        
        {\bf Hardcoded: ${\mathsf{OPMem}, cpk, K_2, r, m, \textcolor{red}{id_j + \Delta}}$}
        \begin{enumerate}
            \item Run $\mathsf{OPMem}(id, u_1, \dots, u_{\cosettcount}, r)$. If it outputs $0$, output $\perp$ and terminate.
            \item \textcolor{red}{If $id < id_j + \Delta$, set $a = \failmes$. Otherwise, set $a = m$.}
            \item Output $\mathsf{IBE.Enc}(cpk, id, {a}; F_2(K_2, id))$.
        \end{enumerate}
    \end{mdframed}
    \item Output $(\mathsf{OPCt}, r)$.
\end{enumerate}

        \item{$\underline{\mathcal{D}_{j}^{(\Delta, 1)}(m)}$:}  \begin{enumerate}
   \item Sample $r \samp \zo^{\cosettcount}$.
    \item Sample a PRF key $K_2$ for $F_2.\mathsf{KeyGen}(1^\lambda)$.
    \item \textcolor{red}{ $ct^* = \mathsf{IBE.Enc}(cpk, id_j + \Delta, m; F_2(K_2, id_j+\Delta)).$}
        \item \textcolor{red}{ $K_2\{id_j+\Delta\} \samp F_2.\mathsf{Punc}(K_2, id_j+\Delta)$.}
            
    \item Sample     $\mathsf{OPCt} \samp \io({ \mathsf{PCt}^{(j, \Delta, 1)}_{\mathsf{OPMem}, cpk, K_2\{id_j+\Delta\}, r, m, id_j + \Delta, ct^*}}).$
 \begin{mdframed}
        {\bf $\underline{{ \mathsf{PCt}^{(j, \Delta, 1)}_{\mathsf{OPMem}, cpk, K_2\{id_j+\Delta\}, r, m, id_j + \Delta, ct^*}}(id, u_1, \dots, u_{\cosettcount})}$}
        
        {\bf Hardcoded: ${\mathsf{OPMem}, cpk, \textcolor{red}{K_2\{id_j+\Delta\}}, r, m, id_j + \Delta, \textcolor{red}{ct^*}}$}
        \begin{enumerate}
            \item Run $\mathsf{OPMem}(id, u_1, \dots, u_{\cosettcount}, r)$. If it outputs $0$, output $\perp$ and terminate.
            \item \textcolor{red}{If $id = id_j + \Delta$, output $ct^*$ and terminate.}
            \item {If \textcolor{red}{$id < id_j + \Delta + 1$}, set $a = \failmes$. Otherwise, set $a = m$.}
            \item Output $\mathsf{IBE.Enc}(cpk, id, {a}; F_2(K_2, id))$.
        \end{enumerate}
    \end{mdframed}
    \item Output $(\mathsf{OPCt}, r)$.
\end{enumerate}

    \item{$\underline{\mathcal{D}_{j}^{(\Delta, 2)}(m)}$:} 
    \begin{enumerate}
   \item Sample $r \samp \zo^{\cosettcount}$.
    \item Sample a PRF key $K_2$ for $F_2.\mathsf{KeyGen}(1^\lambda)$.
    \item \textcolor{red}{Sample $z^*$ uniformly at random from the output space of $F_2$.}
    \item \textcolor{red}{ $ct^* = \mathsf{IBE.Enc}(cpk, id_j + \Delta, m; z^*).$}
        \item { $K_2\{id_j+\Delta\} \samp F_2.\mathsf{Punc}(K_2, id_j+\Delta)$.}
            
    \item Sample     $\mathsf{OPCt} \samp \io({ \mathsf{PCt}^{(j, \Delta, 2)}_{\mathsf{OPMem}, cpk, K_2\{id_j+\Delta\}, r, m, id_j + \Delta, ct^*}}).$
 \begin{mdframed}
        {\bf $\underline{{ \mathsf{PCt}^{(j, \Delta, 2)}_{\mathsf{OPMem}, cpk, K_2\{id_j+\Delta\}, r, m, id_j + \Delta, ct^*}}(id, u_1, \dots, u_{\cosettcount})}$}
        
        {\bf Hardcoded: ${\mathsf{OPMem}, cpk, {K_2\{id_j+\Delta\}}, r, m, id_j + \Delta, {ct^*}}$}
        \begin{enumerate}
            \item Run $\mathsf{OPMem}(id, u_1, \dots, u_{\cosettcount}, r)$. If it outputs $0$, output $\perp$ and terminate.
            \item {If $id = id_j + \Delta$, output $ct^*$ and terminate.}
            \item {If {$id < id_j + \Delta + 1$}, set $a = \failmes$. Otherwise, set $a = m$.}
            \item Output $\mathsf{IBE.Enc}(cpk, id, {a}; F_2(K_2, id))$.
        \end{enumerate}
    \end{mdframed}
    \item Output $(\mathsf{OPCt}, r)$.
\end{enumerate}

    \item{$\underline{\mathcal{D}_{j}^{(\Delta, 3)}(m)}$:} 
   \begin{enumerate}
   \item Sample $r \samp \zo^{\cosettcount}$.
    \item Sample a PRF key $K_2$ for $F_2.\mathsf{KeyGen}(1^\lambda)$.
        \item {Sample $z^*$ uniformly at random from the output space of $F_2$.}
    \item { $ct^* = \mathsf{IBE.Enc}(cpk, id_j + \Delta, m; z^*).$}
    \item { $K_2\{id_{j}+\Delta\} \samp F_2.\mathsf{Punc}(K_2, id_{j}+\Delta)$.}
    \textcolor{red}{\item Compute $(A^{*}_i,s_i^{*},s_i^{'{*}}) = F_1(K_1, id_j + \Delta).$
      \item For $i \in [\cosettcount]$, set $g_i = \mathsf{Can}_{A^{*}_i}$ if $(r)_i = 0$ and set $g_i = \mathsf{Can}_{(A^{*}_i)^{\perp}}$ if $(r)_i = 1$.
    \item For $i \in [\cosettcount]$, compute $y_i = g_i(s_i^{*})$ if $(r)_i = 0$ and $y_i =  g_i(s_i^{'{*}})$ if $(r)_i = 1$.
    \item Set $g$ to be the function $g(v_1, \dots, v_{\cosettcount}) = (g_1(v_1) || \dots || g_{\cosettcount}(v_{\cosettcount}))$.
    \item Set $y = y_1 || \dots || y_{\cosettcount}$.
    \item  $\mathsf{OCC} \samp \mathsf{CCObf.Obf}(g, y, ct^*)$.}

         \item Sample    $\mathsf{OPCt} \samp \io({ \mathsf{PCt}^{(j,\Delta,3)}_{\mathsf{OPMem}, cpk, K_2\{id_j\}, r, m, id_{j}+\Delta, \mathsf{OCC}}}).$
  \begin{mdframed}
        {\bf $\underline{{ \mathsf{PCt}^{(j,\Delta,3)}_{\mathsf{OPMem}, cpk, K_2\{id_{j}+\Delta\}, r, m, id_{j}+\Delta, \mathsf{OCC}}}(id, u_1, \dots, u_{\cosettcount})}$}
        
        {\bf Hardcoded: ${\mathsf{OPMem}, cpk, {K_2\{id_{j}+\Delta\}}, r, m, id_{j}+\Delta, \textcolor{red}{\mathsf{OCC}}}$}
        \begin{enumerate}
        \item \textcolor{red}{If $id = id_{j}+\Delta$, output the output of $\mathsf{OCC}(u_1, \dots, u_{\cosettcount})$ and terminate.}
            \item Run $\mathsf{OPMem}(id, u_1, \dots, u_{\cosettcount}, r)$. If it outputs $0$, output $\perp$ and terminate.
            \item {If {$id < id_{j}+\Delta + 1$}, set $a = \failmes$. Otherwise, set $a = m$.}
            \item Output $\mathsf{IBE.Enc}(cpk, id, {a}; F_2(K_2, id))$.
        \end{enumerate}
    \end{mdframed}
    
    \item Output $(\mathsf{OPCt}, r)$.
\end{enumerate}

\item{$\underline{\mathcal{D}_{j}^{(\Delta, 4)}(m)}$:}  Same as $\mathcal{D}_{j}^{(\Delta, 3)}$ except for the following. Replace the line
    $$ct^* = \mathsf{IBE.Enc}(cpk, id_j + \Delta, m; z^*)
    $$
    with
    $$ 
    ct^* = \mathsf{IBE.Enc}(cpk, id_j + \Delta, \textcolor{red}{\failmes}; z^*).
    $$
    
        \item{$\underline{\mathcal{D}_{j}^{(\Delta, 5)}(m)}$:} Same as $\mathcal{D}_{j}^{(\Delta, 2)}$ except for the following. Replace the line
    $$ct^* = \mathsf{IBE.Enc}(cpk, id_j + \Delta, m; z^*)
    $$
    with
    $$ 
    ct^* = \mathsf{IBE.Enc}(cpk, id_j + \Delta, \textcolor{red}{\failmes}; z^*).
    $$

\item{$\underline{\mathcal{D}_{j}^{(\Delta, 6)}(m)}:$} Same as ${\mathcal{D}_{j}^{(\Delta, 1)}}$ except for the following. Replace the line
                $$ct^* = \mathsf{IBE.Enc}(cpk, id_j + \Delta, m; F_2(K_2, id_j+\Delta))$$
                with
                $$ 
                ct^* = \mathsf{IBE.Enc}(cpk, id_j + \Delta, \textcolor{red}{\failmes}; F_2(K_2, id_j+\Delta)).
                $$
\end{itemize}

Now, we show that these distributions \emph{collapse} around $\Delta = 0$ for each $j$. Below, all our indistinguishability claims are for $2^{5\lambda}\cdot2^{8\lambda^{0.3\constmoecoll}}$-time adversaries and we set $\nu(\lambda) = 2^{-6\lambda}\cdot2^{-8\lambda^{0.3\constmoecoll}}$.
\begin{claim}\label{claim:delta0todelta1}
$\expnot{j}{\Delta}{0}{\ell}\approx^c_{\nu(\lambda)} \expnot{j}{\Delta}{1}{\ell}$
    for all $j \in \{0, 1, \dots, k\}$, $\Delta \in \{0, 1, \dots, id_{j+1}-id_{j} - 1\}$ and $\ell \in [k + 1]$.
\end{claim}
\begin{proof}
    Observe that by punctured key correctness of $F_2$ (\cref{defn:puncprf}), the different obfuscated programs ${ \mathsf{PCt}^{(j, \Delta, 0)}_{\mathsf{OPMem}, cpk, K_2, r, m, id_j + \Delta}}$ and ${ \mathsf{PCt}^{(j, \Delta, 1)}_{\mathsf{OPMem}, cpk, K_2\{id_j+\Delta\}, r, m, id_j + \Delta, ct^*}}$ in these hybrids have the same functionality. The result follows by security of $\io$ and by our choice of parameters.
\end{proof}
\begin{claim}\label{claim:delta1todelta2}
$\expnot{j}{\Delta}{1}{\ell} \approx^c_{\nu(\lambda)} \expnot{j}{\Delta}{2}{\ell}$
  for all $j \in \{0, 1, \dots, k\}$, $\Delta \in \{0, 1, \dots, id_{j+1}-id_{j} - 1\}$ and $\ell \in[k + 1]$.
\end{claim}
\begin{proof}
    The result follows by selective puncturing security of $F_2$ (\cref{defn:puncprf}) and our choice of parameters.
\end{proof}
\begin{claim}\label{claim:delta2todelta3}
    $\expnot{j}{\Delta}{2}{\ell}\approx^c_{\nu(\lambda)} \expnot{j}{\Delta}{3}{\ell}$
    for all $j \in \{0, 1, \dots, k\}$, $\Delta \in \{0, 1, \dots, id_{j+1}-id_{j} - 1\}$ and $\ell \in [k + 1]$.
\end{claim}
\begin{proof}
       Observe that the obfuscated ciphertext programs $\mathsf{PCt}$ in these hybrids have the same functionality by correctness of $\mathsf{CCObf}$, since a vector $w$ is in $A_i^* + s_i^*$ if and only if $\mathsf{Can}_{A_i^{*}}(w) =\mathsf{Can}_{A_i^{*}}(s_i^*)$ and similarly for $(A^*)_i^\perp + s^{'*}_i$. Then, the claim follows by the security of $\io$.
\end{proof}
\begin{claim}
$\expnot{j}{\Delta}{3}{\ell} \approx^c_{\nu(\lambda)} \expnot{j}{\Delta}{4}{\ell}$
     if
    \begin{itemize}
    \item  $j \in \{1, \dots, k\}$ and $\Delta \in \{1, \dots, id_{j+1}-id_{j} - 1\}$, or
    \item  $j = 0$ and $\Delta \in \{0, 1, \dots, id_{j+1}-id_{j}-1\}$
    \end{itemize}
    and for all $\ell \in [k + 1]$.
\end{claim}
\begin{proof}
    Observe that in these hybrids, the randomness used to invoke $\mathsf{IBE.Enc}$ to compute $ct^*$ is uniformly and independently sampled. Further, the adversary only has the $\mathsf{IBE}$ keys for the identities $id_1, id_2, \dots, id_k$, all of which are different from the identity $id_j + \Delta$ under which $ct^*$ is encrypted. Hence, by IBE security (\cref{defn:ibe}), the result follows.
\end{proof}
\begin{claim}
    $\expnot{j}{\Delta}{4}{\ell} \approx^c_{\nu(\lambda)} \expnot{j}{\Delta}{5}{\ell}$ for all $j \in \{0, 1, \dots, k\}$, $\Delta \in \{0, 1, \dots, id_{j+1}-id_{j} - 1\}$ and $\ell \in [k + 1]$.
\end{claim}
\begin{proof}
    Essentially the same argument as in \cref{claim:delta2todelta3} yields the result.
\end{proof}

\begin{claim}
    $\expnot{j}{\Delta}{5}{\ell} \approx^c_{\nu(\lambda)} \expnot{j}{\Delta}{6}{\ell}$ for all $j \in \{0, 1, \dots, k\}$, $\Delta \in \{0, 1, \dots, id_{j+1}-id_{j} - 1\}$ and $\ell \in [k + 1]$.
\end{claim}
\begin{proof}
    Essentially the same argument as in \cref{claim:delta1todelta2} yields the result.
\end{proof}

\begin{claim}
$\expnot{j}{\Delta}{6}{\ell} \approx^c_{\nu(\lambda)} \expnot{j}{\Delta+1}{0}{\ell}$
    for all $j \in \{0, 1, \dots, k\}$, $\Delta \in \{0, 1, \dots, id_{j+1}-id_{j} - 1\}$ and $\ell \in [k + 1]$.
\end{claim}
\begin{proof}
        Essentially the same argument as in \cref{claim:delta0todelta1} yields the result.
\end{proof}

\begin{claim}\label{claim:d0closed1}
For all $\ell \in [k+1]$, we have 
\begin{itemize}
    \item $\expfromddd{\mathcal{D}_0}{\ell} \approx^c_{\nu(\lambda)}  \expfromddd{\mathcal{D}_1}{\ell} $
    \item $\expnot{j}{0}{4}{\ell} \approx^c_{\nu(\lambda)} \expfromddd{\mathcal{D}_{j+1}}{\ell}$ for all $j \in \{0, 1, \dots, k\}$ 
    \item $\expfromddd{\mathcal{D}_j}{\ell} \approx^c_{\nu(\lambda)} \expnot{j}{0}{3}{\ell}$ for all $j \in \{0, 1, \dots, k\}$ 
\end{itemize}
where $\nu(\lambda) = 2^{-5\lambda}\cdot2^{-8\lambda^{0.3\constmoecoll}}$.
\end{claim}
\begin{proof}
It is easy to see that $\mathcal{D}_{0} \approx^c_{\nu(\lambda)} \mathcal{D}_{0}^{(0,0)}$ and  $\mathcal{D}_{k+1} \approx^c_{\nu(\lambda)} \mathcal{D}_{k+1}^{(0, 0)}$ by the security of $\io$.

Rest follows by a simple calculation using the above results.
\end{proof}

\begin{definition}
    We will write $\mathcal{D}'$ to denote $\mathcal{D}_{j^*}^{(0,3)}$ and $\mathcal{D}''$ to denote $\mathcal{D}_{j^*}^{(0,4)}$ where $j^*$ is as output by $\adve_0'$.
\end{definition}

Finally, we show that the success probabilities for both freeloaders jump exactly at $j^*$.
\begin{claim}\label{claim:dprimdprim}
Let $\tau$ be the bipartite state output by $\adve'_0$ in $\mathcal{G}$. Let $p'_x, p'_y$ be the outcome of applying $\mathsf{PI}_{x, \mathcal{D}'}\otimes \mathsf{PI}_{y, \mathcal{D}'}$ to $\tau$. Similarly, let $p''_x, p''_y$ be the outcome of applying $\mathsf{PI}_{x, \mathcal{D}''}\otimes \mathsf{PI}_{y, \mathcal{D}''}$ to $\tau$. Then,
\begin{itemize}
    \item $\Pr[p'_x > b_{x,j^*} - \frac{3\gamma}{32k} \wedge p'_y > b_{y,j^*} - \frac{3\gamma}{32k}] \geq 1 - 2^{-2\lambda}\cdot 2^{-4\lambda^{0.3\constmoecoll}}$.
    \item $\Pr[b_{x,j^*} - p''_x > \frac{28\gamma}{32k} \wedge b_{y,j^*} - p''_y  > \frac{28\gamma}{32k}] > \frac{1}{q(\lambda)}$ for some polynomial $q(\cdot)$.
\end{itemize}
\end{claim}

\begin{proof}
 Let $(a'_x, a'_y)$ be the outcome of applying $\mathsf{API}^{\eps,\delta}_{x, \mathcal{D}_{j^*}}\otimes \mathsf{API}^{\eps,\delta}_{y,\mathcal{D}_{j^*}}$ to $\tau$.
 Then, by \cref{thm:multiapiprop},  \cref{thm:multiatiprop} and by definition of $b_{x,j^*}, b_{y,j^*}$, we have 
    $$\Pr[a'_x > b_{x,j^*} - \frac{3\gamma}{32k} \wedge a'_y > b_{y,j^*} - \frac{3\gamma}{32k}] \geq 1 - \poly(\lambda)\cdot \delta(\lambda).$$

Then, since $\expfromddd{\mathcal{D}_{j^*}}{\ell} \approx^c_\nu \expfromddd{\mathcal{D}'}{\ell}$ against $2^{5\lambda}\cdot2^{8\lambda^{0.3\constmoecoll}}$-time adversaries where $\nu(\lambda) = 2^{-5\lambda}\cdot2^{-8\lambda^{0.3\constmoecoll}}$, we get 
$$\Pr[p'_x > b_{x,j^*} - \frac{3\gamma}{32k} \wedge p'_y > b_{y,j^*} - \frac{3\gamma}{32k}] \geq 1 -  2^{-2\lambda}\cdot 2^{-4\lambda^{0.3\constmoecoll}}.$$
 by \cref{thm:distti}. 

 See the proof of \cref{claim:mesresdist} for a remark on how to invoke \cref{thm:distti}. Note that here, we are applying the measurements to $\tau$ rather than to the state $\sigma$. However, since the procedure that gives $\tau$ from $\sigma$ is an efficient procedure that only uses $pp$, the indistinguishability between $\mathcal{D}_{j^*}$ and $\mathcal{D}'$ given  $\sigma$ still applies when we are instead given $\tau$, hence \cref{thm:distti} indeed applies.

We now move onto the second claim. By \cref{claim:mesresdist}, we have that $$\Pr[b_{x,j^*} - b_{x,j^*+1} > \frac{29\gamma}{32k} \wedge b_{y,j^*} - b_{y,j^*+1} > \frac{29\gamma}{32k}]$$
is non-negligible. 
Further, we have $\expfromddd{\mathcal{D}_{j^*+1}}{\ell} \approx \expfromddd{\mathcal{D}''}{\ell}$. By \cref{thm:distti} and \cref{thm:multiapismallproof}, we get that $$\Pr[b_{x,j^*} - p''_x  > \frac{28\gamma}{32k} \wedge b_{y,j^*} - p''_y > \frac{28\gamma}{32k}]$$ is non-negligible. Similar to above, the indistinguishability of $\mathcal{D}_{j^*+1}$ and $\mathcal{D}''$ given $\sigma$ still applies when we are given the state $\tau$ instead. Therefore, \cref{thm:distti} indeed applies.
\end{proof}

\subsubsection*{Extracting MoE Vectors}
We have shown that the two freeloaders use the same identity $j^*$ to decrypt. Now we will show that we can exract MoE vectors from these two freeloaders \emph{simultaneously}. That is, we extract MoE vectors from one of the freeloaders even conditioned on successful extraction from the other one.\footnote{Note that this is not a direct consequence of extraction from a single freeloader due to entanglement.}
\begin{claim}\label{claim:pkeextract}
    There exist efficient $\adve_1', \adve_2'$ such that  $(\adve'_0, \adve_1', \adve_2')$ wins $\mathcal{G}$ with probability $\frac{1}{2^{0.4\cdot\lambda^{\constmoecoll}}}.$
\end{claim}
\begin{proof}

For a challenge ciphertext distribution $\mathcal{C}$, let $\expfromddd{\mathcal{C}}{x}'$ denote the outcome of the following experiment.

\begin{enumerate}
    \item Execute $pk, sk \samp \mathsf{PKE.Setup}(1^\lambda)$.
\item Simulate $\adve'_0$ and the challenger of $\mathcal{G}$:
\begin{enumerate}[label*=\arabic*.]
    \item Simulate $\adve$ on $pk$ by sampling a quantum secret key (as in $\hyb_1$) whenever $\adve$ makes a query, and submitting the key to it. Let $\regi{adv}, (m_\ell^0, m_\ell^1)_{\ell \in [k + 1]},(U_\ell)_{\ell \in [k + 1]}$ be the output of $\adve$.
            \item Uniformly at random sample $x, y, j^*$ such that $1 \leq x < y \leq k + 1$ and $j^* \in \{1, \dots, k\}$.
                        \item Apply $\mathsf{API}^{\eps, \delta}_{\ell, \mathcal{D}_0}$ to all registers $\regi{adv}[\ell]$ for $\ell \in [k + 1]$, let $b_{\ell, 0}$ be the measurement outcomes.
            \item Apply $\mathsf{API}^{\eps, \delta}_{\ell, \mathcal{D}_i}$ in succession for $i = 1$ to $j^*$ to $\regi{adv}[x]$, let $b_{x, i}$ be the measurement outcomes.
            \item Apply $\mathsf{API}^{\eps, \delta}_{\ell, \mathcal{D}_i}$ in succession for $i = 1$ to $j^*$ to $\regi{adv}[y]$, let $b_{y, i}$ be the measurement outcomes.
\end{enumerate}
    \item Set  $pp =  (x, y, j^*, (id_j)_{j\in[k+1]}, (m_\ell^0, m_\ell^1)_{\ell \in [k + 1]}, (U_\ell)_{\ell \in [k + 1]}, pk)$.
    \item Sample $b \samp \zo$.
    \item Sample $ct \samp \mathcal{C}(pp, m_\ell^b)$.
    \item Output $\regi{adv}[x], (b, ct), pp$.
\end{enumerate}

It is easy to see that $\expfromddd{\mathcal{C}_0}{\ell} \approx_\nu^c \expfromddd{\mathcal{C}_1}{\ell}$ implies $\expfromddd{\mathcal{C}_0}{x}' \approx_\nu^c \expfromddd{\mathcal{C}_1}{x}'$ since they only differ in their auxiliary states and we can efficiently obtain the auxiliary state of the latter using the auxiliary state of the former.

By \cref{claim:dprimdprim}, we have
\begin{enumerate}
     \item \label{item:pidprime} $\Pr[\mathsf{PI}_{x, \mathcal{D}'}\cdot\tau[1] \leq b_{x,j^*} - \frac{3\gamma}{32k}] \leq 2^{-2\lambda}\cdot 2^{-4\lambda^{0.3\constmoecoll}}$, and
    \item  \label{item:pidprimeprime} $\Pr[\mathsf{PI}_{x, \mathcal{D}''}\cdot\tau[1]<b_{x,j^*} - \frac{28\gamma}{32k}]$ is non-negligible.
\end{enumerate}

Suppose for a contradiction that $\expfromddd{\mathcal{D}'}{x}' \approx^c \expfromddd{\mathcal{D}''}{x}'$. Then, by  \cref{thm:singleatiprop} and \cref{thm:distti}, \cref{item:pidprimeprime} implies that
$$\Pr[\mathsf{PI}_{x, \mathcal{D}'}\cdot\tau[1] < b_{x,j^*} - \frac{26\gamma}{32k}] $$
is non-negligible. This is a contradiction to \cref{item:pidprime}, therefore, $\expfromddd{\mathcal{D}'}{x}'\not\approx^c \expfromddd{\mathcal{D}''}{x}'$. We define the distribution $\mathcal{D}_{sim}$ by modifying $\mathcal{D}'$ as follows: We replace the line 
    $$\mathsf{OCC} \samp \mathsf{CCObf.Obf}(g, y, ct^*)$$
    with 
    $$\mathsf{OCC} \samp \mathsf{CCObf.Sim}(1^\lambda, |g|, |y|, |ct^*|).$$
Since $\expfromddd{\mathcal{D}'}{x}' \not\approx^c \expfromddd{\mathcal{D}''}{x}'$, we have either $\expfromddd{\mathcal{D}'}{x}'\not\approx^c \expfromddd{\mathcal{D}_{sim}}{x}'$ or $\expfromddd{\mathcal{D}''}{x}' \not\approx^c \expfromddd{\mathcal{D}_{sim}}{x}'$. We will only discuss the first case but the second case follows from the same argument.

Now, we will give a distribution $\mathcal{B}$ over compute-and-compare programs (with quantum auxiliary information) and an adversary $\adve_{CC}$ that breaks the security of $\mathsf{CCObf}$ for this distribution. This in turn will mean by \cref{defn:ccobf} that there is an adversary that can predict the \emph{target} value of these programs, given the description of the compute part of the program and the auxiliary information.

We first define the distribution $\mathcal{B}$. 
\paragraph{$\underline{\mathcal{B}(1^\lambda)}$}
\begin{enumerate}
\item Execute $pk, sk \samp \mathsf{PKE.Setup}(1^\lambda)$.
\item Simulate $\adve'_0$ and the challenger of $\mathcal{G}$:
\begin{enumerate}[label*=\arabic*.]
    \item Simulate $\adve$ on $pk$ by sampling a quantum secret key (as in $\hyb_1$) whenever $\adve$ makes a query, and submitting the key to it. Let $\regi{adv}, (m_\ell^0, m_\ell^1)_{\ell \in [k + 1]},(U_\ell)_{\ell \in [k + 1]}$ be the output of $\adve$.
            \item Uniformly at random sample $x, y, j^*$ such that $1 \leq x < y \leq k + 1$ and $j^* \in \{1, \dots, k\}$.
                        \item Apply $\mathsf{API}^{\eps, \delta}_{\ell, \mathcal{D}_0}$ to all registers $\regi{adv}[\ell]$ for $\ell \in [k + 1]$, let $b_{\ell, 0}$ be the measurement outcomes.
            \item Apply $\mathsf{API}^{\eps, \delta}_{\ell, \mathcal{D}_i}$ in succession for $i = 1$ to $j^*$ to $\regi{adv}[x]$, let $b_{x, i}$ be the measurement outcomes.
            \item Apply $\mathsf{API}^{\eps, \delta}_{\ell, \mathcal{D}_i}$ in succession for $i = 1$ to $j^*$ to $\regi{adv}[y]$, let $b_{y, i}$ be the measurement outcomes.
\end{enumerate}
\item Set  $pp =  (x, y, j^*, (id_j)_{j\in[k+1]}, (m_\ell^0, m_\ell^1)_{\ell \in [k + 1]},(U_\ell)_{\ell \in [k + 1]}, pk)$.
\item Sample $b \samp \zo$.
\item Simulate the first steps of $\mathcal{D}'$ on $m_x^b$:
\begin{enumerate}[label*=\arabic*.]
    \item Sample $r \samp \zo^{\cosettcount}$.
        \item {Sample $z^*$ uniformly at random the output space of $F_2$.}
    \item { $ct^* = \mathsf{IBE.Enc}(cpk, id_j, m^b_x; z^*).$}
    \item Compute $(A^{*}_i,s_i^{*},s_i^{'{*}}) = F_1(K_1, id_{j^*}).$
      \item For $i \in [\cosettcount]$, set $g_i = \mathsf{Can}_{A^{*}_i}$ if $(r)_i = 0$ and set $g_i = \mathsf{Can}_{(A^{*}_i)^{\perp}}$ if $(r)_i = 1$.
    \item For $i \in [\cosettcount]$, compute $y_i = g_i(s_i^{*})$ if $(r)_i = 0$ and $y_i =  g_i(s_i^{'{*}})$ if $(r)_i = 1$.
    \item Set $g$ to be the function $g(v_1, \dots, v_{\cosettcount}) = (g_1(v_1) || \dots || g_{\cosettcount}(v_{\cosettcount}))$.
    \item Set $y = y_1 || \dots || y_{\cosettcount}$.
\end{enumerate}
    \item Output $(g, y, ct^*)$ as the compute-and-compare program and $$(\regi{adv}[x], pp, r, m^b_x, id_{j^*}, b)$$ as the auxiliary information.
\end{enumerate}

We define the adversary $\adve_{CC}$ as follows. Let $\mathcal{A}_{dist}$ be an adversary that distinguishes $\expfromddd{\mathcal{D}'}{x} \not\approx^c \expfromddd{\mathcal{D}_{sim}}{x}$.

\paragraph{$\underline{\mathcal{A}_{CC}(P, \regi{aux})}$}
\begin{enumerate}
\item Parse $(\reg, pp, r, m^b_x, id_{j^*}, b) = \regi{aux}$.
\item Parse $(x, y, j^*, (id_j)_{j\in[k+1]}, (m_\ell^0, m_\ell^1)_{\ell \in [k + 1]},(U_\ell)_{\ell \in [k + 1]}, pk) = pp$.
\item Sample a PRF key $K_2$ for $F_2.\mathsf{KeyGen}(1^\lambda)$.
\item Sample { $K_2\{id_{j^*}\} \samp F_2.\mathsf{Punc}(K_2, id_{j^*})$.}
    \item  Sample    $\mathsf{OPCt} \samp \io(\mathsf{PCt}).$
  \begin{mdframed}
        {\bf $\underline{\mathsf{PCt}(id, u_1, \dots, u_{\cosettcount})}$}
        
        {\bf Hardcoded: ${\mathsf{OPMem}, cpk, {K_2\{id_{j^*}\}}, r, m^b_x, id_{j^*}, P}$}
        \begin{enumerate}
        \item {If $id = id_{j^*}$, output the output of $P(u_1, \dots, u_{\cosettcount})$ and terminate.}
            \item Run $\mathsf{OPMem}(id, u_1, \dots, u_{\cosettcount}, r)$. If it outputs $0$, output $\perp$ and terminate.
            \item {If {$id < id_{j^*} + 1$}, set $a = \failmes$. Otherwise, set $a = m^b_x$.}
            \item Output $\mathsf{IBE.Enc}(cpk, id, {a}; F_2(K_2, id))$.
        \end{enumerate}
    \end{mdframed}
    \item Set $ct = (\mathsf{OPCt}, r)$.
    \item Output $\adve_{dist}(\reg, (ct, b), pp)$.
\end{enumerate}

It is easy to see that $\adve_{CC}(\mathsf{CCObf.Obf}(g, y, ct^*), \regi{aux})$ corresponds to $\adve_{dist}(\expfromddd{\mathcal{D}'}{x})$ while $\adve_{CC}(\allowbreak\mathsf{CCObf.Sim}(1^\lambda, |g|, |y|, |ct^*|), \regi{aux}))$ corresponds to $\adve_{dist}(\expfromddd{\mathcal{D}_{sim}}{x})$ where $(g, y, ct^*) \samp \mathcal{B}(1^\lambda)$.
Hence, since $\adve_{dist}$ distinguishes $\expfromddd{\mathcal{D}'}{x} \not\approx^c \expfromddd{\mathcal{D}_{sim}}{x}$, by \cref{thm:ccobf} there exists an adversary $\mathcal{M}_1$ that can extract vectors $u_i$ such that $g((u_i)_{i \in \cosettcount}) = y$, using the quantum auxiliary information defined above and the description of $g$. Note that $g$ can be computed efficiently given $r$ and $(A^{*}_i)_{i \in [\cosettcount]}$, which are indeed provided to $\adve'_1$ in $\mathcal{G}$. Similarly, $\adve_1'$ can compute $\regi{aux}$ from its input provided by $\adve_0'$. Therefore, we set $\adve'_1$ to be the adversary that computes the input as above and simulates $\mathcal{M}_1$. It is easy to see that $\adve'_1$ outputs correct vectors in the game $\mathcal{G}$ with probability at least $2^{-\lambda^{0.2\cdot {\constmoecoll}}}$, since $\mathsf{CCObf}$ is a compute-and-compare obfuscation scheme for $2^{-\lambda^{0.2\cdot {\constmoecoll}}}$-unpredictable distributions.

Now, we will argue that we can simultaneously extract MoE vectors from the second register, that is, we can extract even conditioned on a successful extraction from the first register. Let $\xi$ denote the  post-measurement state of the input state of $\adve'_2$, conditioned on $\adve'_1$ succeeding. 
First, define  $\expfromddd{\mathcal{C}}{y}''$ as follows.
\begin{enumerate}
   \item Simulate $\mathcal{B}(1^\lambda)$.
   \item Run $\adve'_1$ on $(\regi{adv}[x], pp, r, m^b_x, id_{j^*}, b)$ and $g$ to obtain vectors $(u_i)_{i \in \cosettcount}$.
   \item Check if $\mathsf{OPMem}(id_{j^*}, u_1, \dots, u_{\cosettcount}, r)$. \textcolor{blue}{If the output is $\false$, output $\perp$ and terminate}.
   \item Sample $b \samp \zo$.
    \item Sample $ct \samp \mathcal{C}(pp, m_\ell^b)$.
    \item Output $\regi{adv}[y], (b, ct), pp$.
\end{enumerate}
Observe that the state of the register $\regi{adv}[y]$ output above is $\xi$ (when the experiment outcome is not $\perp$). 
We claim that $\xi$ satisfies
\begin{enumerate}
\item \label{item:pidprime2} $\Pr[\mathsf{PI}_{y, \mathcal{D}'}\cdot\xi \leq b_{y,j^*} - \frac{3\gamma}{32k}] \leq \frac{3}{2} \cdot \sqrt{2^{-2\lambda}\cdot 2^{-4\lambda^{0.3\constmoecoll}}}\cdot 2^{\lambda^{0.2\cdot {\constmoecoll}}}.$
    \item  \label{item:pidprimeprime2} $\Pr[\mathsf{PI}_{y, \mathcal{D}_{''}}\cdot\xi < b_{y,j^*} - \frac{28\gamma}{32k}] \geq 2^{-\lambda^{0.3\cdot {\constmoecoll}}}.$
\end{enumerate}

This first claim follows from \cref{claim:dprimdprim}, \cref{thm:simulproj}, and the fact that extraction on the first register succeeds with probability $\geq 2^{-\lambda^{0.2\cdot {\constmoecoll}}}$.
We argue the second claim as follows. Let $E$ denote the event of successful extraction on the first register, and let $G$ denote the event that applying $\mathsf{PI}_{y, \mathcal{D}_{''}}$ on the second register yields a value $< b_{y,j^*} - \frac{28\gamma}{32k}$. The probability above corresponds to $\Pr[G | E]$, which equals $\frac{\Pr[E | G]\cdot \Pr[G]}{\Pr[E]} \geq \Pr[E | G]\cdot \Pr[G] \geq \Pr[E | G]\cdot\frac{1}{\poly(\lambda)}$. However, observe that we can first apply the measurement $\mathsf{PI}_{y, \mathcal{D}_{''}}$ on the second register, and then try to extract on the first register. Observe that a \emph{gap} still exists on the first register after this measurement on the second register and conditioning on the outcome $G$, by \cref{claim:dprimdprim} and \cref{thm:simulproj}, since $\Pr[G] > 1/\poly(\lambda)$. Hence, similar to the extraction argument above, we get that $\Pr[E | G] > 2^{-\lambda^{0.2\cdot {\constmoecoll}}}$, which proves our claim.

Now, suppose for a contradiction that $\expfromddd{\mathcal{D}'}{y}'' \approx^c_{\nu} \expfromddd{\mathcal{D}''}{y}''$ against $2^{3\lambda}\cdot 2^{2\lambda^{0.3\constmoecoll}}$-time adversaries where $\nu = 2^{-2\lambda - 1}\cdot 2^{-2\lambda^{0.3\constmoecoll}}$. Then, by \cref{thm:distti}, we get that \cref{item:pidprime2} implies \begin{equation*}
    \Pr[\mathsf{PI}_{y, \mathcal{D}_{''}}\cdot\xi \leq b_{y,j^*} - \frac{3\gamma}{32k}] \leq 2\cdot 2^{-\lambda}\cdot 2^{-\lambda^{0.3\cdot {\constmoecoll}}}.
\end{equation*}
which is a contradiction to \cref{item:pidprimeprime2}. Hence, $\expfromddd{\mathcal{D}'}{y}'' \not\approx^c_{\nu} \expfromddd{\mathcal{D}''}{y}''$. Then, using the same extraction argument we used for the first register, by the security of $\mathsf{CCObf}$, we get that there exists an adversary $\adve_2'$ such that it outputs the correct coset vectors with probability at least $ 2^{-\lambda^{0.2\cdot {\constmoecoll}}}$ \emph{conditioned on $\adve'_1$ outputting correct coset vectors}.  This shows that $(\adve'_0, \adve_1', \adve_2')$ wins $\mathcal{G}$ with probability ${2^{-0.4\cdot\lambda^{\constmoecoll}}}.$
\end{proof}

We have shown that there is an adversary $(\adve_0', \adve_1', \adve_2')$ that wins the game $\mathcal{G}$ with probability $${1}/{2^{0.4\cdot\lambda^{\constmoecoll}}}.$$
Finally, we show that we can construct an adversary $(\adve_0'', \adve_1'', \adve_2'')$ that can win $\moecollsel$.

\begin{claim}\label{claim:pkereducetome}
There exists efficient $\adve'' = (\adve_0'', \adve_1'', \adve_2'')$ such that
       \begin{equation*}
           \Pr[\moecollsel(\lambda, \idlen(\lambda), \adve'') = 1] \geq {2^{-0.4\cdot\lambda^{\constmoecoll}}}.
       \end{equation*}
\end{claim}
\begin{proof}
    $\adve_0''$ simulates both the challenger of $\mathcal{G}$ and the adversary $\adve'_0$ as follows. It first samples the random collision-free identity strings $id_1 < \dots < id_k$, and the random index $j^*$ (which it outputs to its challenger); and also $cpk, csmk \samp \mathsf{IBE.Setup}(1^\lambda)$. Then, it sets $pk = (cpk, \mathsf{OPMem})$ where it obtains $\mathsf{OPMem}$ from its challenger. Then, whenever $\adve'_0$ makes a key query for $i \in [k]$, it queries its own challenger for the coset state associated with $id_{\alpha(i)}$. It also samples $ck$ as in $\mathsf{PKE.QKeyGen}$ using $cmsk$, and submits the coset state, $id_{\alpha(i)}$ and $ck$ to $\adve_0'$. Finally, when $\adve'_0$ yields a bipartite register, $\adve_0''$ outputs it.

    We define $\adve_1''$ so that it simulates $\adve_1'$ and make no queries during the second query phase. $\adve_2''$ is defined similarly for $\adve_2'$.
    It is easy to see that $\adve''$ playing $\moecollsel$ perfectly simulates $\mathcal{G}$ as played by $\adve'$, hence $\adve''$ wins with probability ${2^{-0.4\cdot\lambda^{\constmoecoll}}}$.
\end{proof}
This completes the security proof, since the above is a contradiction to \cref{defn:strmoecoll}.
\section{Public-Key Functional Encryption with Copy-Protected Functional Keys}\label{sec:mainfe}
In this section, we formally define functional encryption with copy-protected functional keys. Then, we give a construction based on coset states and prove it secure.

We note that Kitagawa and Nishimaki \cite{KN22} define a simpler model of functional encryption with copy-protected functional keys and give a secure construction with respect to their model. In their model, the adversary can query for any number of functional keys, but only one can be in copy-protected mode. In turn, the adversary only outputs two freeloaders (i.e., only $1 \to 2$ copy-protection is considered). Further, the freeloader adversaries are not allowed to query for more keys after getting their challenge ciphertexts; that is, the adversaries are not fully adaptive.

\subsection{Definitions}\label{sec:fedefnmain}

 An informal overview of our security model is as follows. The piracy adversary will be allowed to adaptively query for classical (i.e., not copy-protected) and copy-protected (i.e. quantum) functional keys. At the end of this first query phase, the adversary will produce a pair of challenge messages $m^0, m^1$ and $k + 1$ registers (\textit{freeloaders}) where $k$ is the number of copy-protected keys obtained by it. After this split, the challenger presents the freeloaders each with a challenge ciphertext. Finally, after receiving the challenge ciphertexts, freeloaders can query for more functional keys, and they output their guess at the end. 

We will also require the following for the challenge message pair $m^0, m^1$ and the functions queried. First, we require that $f(m^0) = f(m^1)$ for all functions $f$ queried by the pirate in the \emph{classical mode}. This is required since, otherwise, the pirate can give all the freeloaders the classical key $sk_f$, and they can decrypt their challenge ciphertexts with this key to distinguish $\mathsf{Enc}(m_0)$ vs $\mathsf{Enc}(m_1)$. Second, for the same reason as above, we require that a freeloader can query a key for $f$ only if $f(m^0) = f(m^1)$. Note that these requirements are the same as the classical FE game (\cref{defn:fe}). Importantly, we will not require anything for functional keys that were obtained in the \emph{copy-protected mode} by the pirate adversary before the split. Thus, our security guarantee will allow $k$ out of the $k+1$ freeloaders to possibly use these copy-protected functional keys to decrypt their challenge ciphertexts. However, it should not be possible for all $k + 1$ registers to use these copy-protected keys simultaneously. 

We also define our model so that copy-protected functional keys are generated given only a classical functional key, without any extra information. Therefore, we do not need to separately require that a copy-protected key for $f$ allows no more than obtaining $f(m)$ given $\mathsf{Enc}(m)$, which is already implied by the regular functional encryption security.

\begin{definition}[Public-key Functional Encryption with Copy-Protected Secret Keys]
A public-key functional encryption scheme with copy-protected secret keys is a public-key functional encryption scheme (\cref{defn:fe}) with the following additional algorithm and guarantee.

\begin{itemize}
    \item $\qkeygen(fk)$: Takes as input a classical functional key, outputs a quantum secret key.
\end{itemize}

We require correctness\footnote{While our schemes satisfy perfect correctness, i.e., correctness with probability $1$, some work relax the definition to $1 - \negl(\lambda)$.} for the quantum functional keys.
\paragraph{Correctness} For all messages $m \in \mathcal{M}$, \begin{equation*}
    \Pr[\mathsf{Dec}(\regi{dec}, ct) = f(m) : \begin{array}{c}
         pk, msk \samp \mathsf{Setup}(1^\lambda)  \\
         sk_f \samp \keygen(msk, f) \\
         \reg_f \samp \qkeygen(sk_f) \\
         ct \samp \mathsf{Enc}(pk, m)
    \end{array}] = 1.
\end{equation*}
\end{definition}
As discussed in \cref{sec:pke}, correctness of the scheme along with \cref{lem:asgoodasnew} means that we can implement decryption in a way such that the quantum functional key is not disturbed. Thus, we can reuse the key to decrypt any number of times.

Similar to public-key encryption, we give a CPA-style anti-piracy security definition.

Let $\mathfrak{F} = \{\mathfrak{F}_\lambda\}_\lambda$ be a family of functions. We define anti-piracy security for $\mathfrak{F}$ as follows.
\begin{definition}[CPA-Style Regular Anti-Piracy Security for Functional Encryption]\label{defn:unclonfe}
Consider the following game between the challenger and an adversary $\adve$.
\paragraph{$\underline{\unclonfegame(\lambda, \adve)}$}
\begin{enumerate}
    \item The challenger runs $msk, pk \samp \fe.\mathsf{Setup}(1^\lambda)$ and submits $pk$ to the adversary. It also initializes the set $\fclass = \emptyset$.
    \item \textbf{Query Phase 1:} For multiple rounds, the adversary adaptively submits a function $f \in \mathfrak{F}$ and a query type, either $\mathsf{CLASSICAL}$ or $\mathsf{PROTECTED}$. For each $f$, the challenger does the following. It first computes $sk_f \samp \mathsf{FE}.\mathsf{KeyGen}(msk, f)$. 
    
    Then, if the query type is $\mathsf{CLASSICAL}$, it adds $f$ to $\fclass$ and submits $sk_f$ to the adversary. 
    
    Otherwise, it computes $\reg_f \samp \fe.\qkeygen(sk_f)$ and submits $\reg_f$ to the adversary.
    \item \textbf{Split Phase:}  The adversary outputs a pair of challenge messages $m^0, m^1$ and a $(k + 1)$-partite register $\regi{adv}$ (where $k$ is the number of queries of the type $\mathsf{PROTECTED}$), each part of the register being an interactive freeloader adversary that will be executed using a universal circuit. The challenger checks if  $f(m^0) = f(m^1)$ for all $f \in \fclass$. If not, it outputs $0$ and terminates.
    \item \textbf{Challenge Phase:} For each $\ell \in [k + 1]$, the challenger samples $b_\ell \samp \zo$, then computes $ct_\ell \samp \fe.\mathsf{Enc}(pk, m^{b_\ell})$ and sends $ct_\ell$ to the $\ell$-th freeloader.
    \item \textbf{Query Phase 2:} The challenger interacts with each of the $k + 1$ freeloaders, using a universal circuit, for multiple rounds as follows. The freeloader $\ell \in [k + 1]$ adaptively submits a function $f \in \mathfrak{F}$ and a query type, either $\mathsf{CLASSICAL}$ or $\mathsf{PROTECTED}$. For each query $f$, the challenger answers with $\perp$ if $f(m^0) \neq f(m^1)$. Otherwise, it does the following. It first computes $sk_f \samp \mathsf{FE}.\mathsf{KeyGen}(msk, f)$. 
    If the query type is $\mathsf{CLASSICAL}$, the challenger submits $sk_f$ to the adversary $\ell$. If the query type is $\mathsf{PROTECTED}$, the challenger computes $\reg_f \samp \fe.\qkeygen(sk_f)$ and submits $\reg_f$ to the adversary.

    \item For $\ell \in [k + 1]$, the challenger submits $ct_\ell$ to the $\ell$-th freeloader to obtain a guess $b'_\ell$. Then, it checks if $b'_\ell = b_\ell$ for all $\ell \in [k + 1]$. It outputs $1$ if and only if all the check pass.
\end{enumerate}
    We say that a public key functional encryption scheme $\fe$ with copy-protected secret keys satisfies $\gamma$-anti-piracy security if for any QPT adversary $\adve$,
    \begin{equation*}
    \Pr[\unclonfegame(\lambda, \adve) = 1] \leq \frac{1}{2} + \gamma(\lambda) + \negl(\lambda).
    \end{equation*}
    We omit indicating $\gamma$ explicitly when $\gamma = 0$.
\end{definition}

We make some remarks about this definition. First, notice that if a construction satisfies $\gamma$-anti-piracy for any inverse polynomial $\gamma$, then it also satisfies it for $\gamma = 0$, simply because of the added $\negl(\lambda)$ term above. Second, note that ($\gamma=0-$)anti-piracy security trivially implies regular functional encryption security: an adversary for the latter corresponds to an adversary for the former that only makes queries of the type $\mathsf{CLASSICAL}$.

\subsection{Construction}\label{sec:unclonfecons}
In section, we give our construction of a functional encryption scheme with copy-protected keys for the class of functions $\mathfrak{F}$ defined as all circuits that are of size at most $\circsize$, where $\circsize$ is any fixed polynomial. The construction is highly similar to our public-key encryption construction. The main difference is that a functional key for a function $f$ will consist of an $\mathsf{IBE}$ key for $id || f$ where $id$ is a random string.

Assume the existence of following primitives where we set $\nu(\lambda) = 2^{-5\lambda-\circsize}\cdot2^{-8\lambda^{0.3\constmoecoll}}$.
\begin{itemize}
    \item $\io$, indistinguishability obfuscation scheme that is $\nu(\lambda)$-secure against $2^{5\lambda}\cdot2^{8\lambda^{0.3\constmoecoll}}$-time adversaries,
    \item $\mathsf{IBE}$, identity-based encryption scheme with puncturable master secret keys (\cref{defn:puncibe}) and deterministic $\mathsf{KeyGen}$ that satisfies strong punctured key correctness (\cref{defn:strongpuncibe}), for the identity space $\mathcal{ID} = \zo^{\circsize + \lambda}$ that is $\nu(\lambda)$-secure against $2^{5\lambda}\cdot2^{8\lambda^{0.3\constmoecoll}}$-time adversaries,
    \item $F_1$, puncturable PRF family with input length $\circsize + \lambda$ and output length same as the size of the randomness used by $\mathsf{CosetGen}$ (\cref{defn:cosetgen}) that is $\nu(\lambda)$-secure against $2^{5\lambda}\cdot2^{8\lambda^{0.3\constmoecoll}}$-time adversaries,
    \item $F_2$, puncturable PRF family with input length $\circsize + \lambda$ and output length same as the size of the randomness used by $\mathsf{IBE.Enc}$ that is $\nu(\lambda)$-secure against $2^{5\lambda}\cdot2^{8\lambda^{0.3\constmoecoll}}$-time adversaries,
    \item $\mathsf{CCObf}$, compute-and-compare obfuscation for $2^{-\lambda^{0.2\cdot {\constmoecoll}}}$-unpredictable distributions that is $2^{-2\lambda - 1}\cdot 2^{-2\lambda^{0.3\constmoecoll}}$-secure against $2^{3\lambda}\cdot 2^{2\lambda^{0.3\constmoecoll}}$-time adversaries,
\end{itemize}

Similar to our public-key encryption scheme, while we assume exponential security of the above primitives for specific exponents, these assumptions can be based only on subexponential hardness for some exponent, since we can always scale the security parameter by a polynomial factor.

Also, set $L(\lambda) = Q(\lambda) + \lambda$ and hence $\cosettcount = 3\cdot( Q(\lambda) + 2\lambda)^3$.

We now give our construction. Below, assume that all programs that are obfuscated are appropriately padded. 
\paragraph{$\underline{\mathsf{FE.Setup}(1^{\lambda})}$}
\begin{enumerate}
    \item Sample a PRF key $K_1 \samp F_1.\mathsf{KeyGen}(1^\lambda)$.
    \item Sample $cpk, csmk \samp \mathsf{IBE.Setup}(1^{\lambda})$.
    \item Sample $\mathsf{OPMem} \samp \io(\mathsf{PMem}_{K_1})$, where $\mathsf{PMem}_{K_1}$ is the following program.
    
\begin{mdframed}\label{code:fepmemorig}
        {\bf $\underline{\mathsf{PMem}_{K_1}(id || f, u_1, \dots, u_{\cosettcount}, r)}$}
        
        {\bf Hardcoded: $K_1$}
        \begin{enumerate}[label=\arabic*.]
            \item $(A_i, s_i, s_i')_{i \in [\cosettcount]} \samp \mathsf{CosetGen}(\cosetgenparam; F_1(K_1, id || f))$.
            \item For each $i \in [\cosettcount]$, check if $u_i \in A_i + s_i$ if $(r)_i = 0$ and check if $u_i \in A^{\perp}_i + s'_i$ if $(r)_i = 1$. If any of the checks fail, output $0$ and terminate.
            \item Output $1$.
        \end{enumerate}

    \end{mdframed}

    \item Set $pk = (cpk, \mathsf{OPMem})$, $msk = (cmsk, K_1)$.
    \item Output $(pk, msk)$.
\end{enumerate}

\paragraph{$\underline{\mathsf{FE.KeyGen}(msk, f)}$}
\begin{enumerate}
    \item Parse $(cmsk, K_1) = msk$.
    \item Sample $id \samp \zo^{\lambda}$.
    \item Sample $ck \samp \mathsf{IBE.KeyGen}(cmsk, id || f)$.
    \item $(A_i, s_i, s_i')_{i \in [\cosettcount]} = \mathsf{CosetGen}(\cosetgenparam; F_1(K_1, id || f))$.
    \item Output $(ck, id, f, (A_i, s_i, s_i')_{i \in [\cosettcount]})$.
\end{enumerate}

\paragraph{$\underline{\mathsf{FE.QKeyGen}(fk)}$}
\begin{enumerate}
    \item Parse $(ck, id, f, (A_i, s_i, s_i')_{i \in [\cosettcount]}) = fk$.
    
    \item Output $\left(\ket{A_{i, s_i, s'_i}}\right)_{i \in [\cosettcount]}, ck, id, f$.
\end{enumerate}

\paragraph{$\underline{\mathsf{FE.Enc}(pk, m)}$}
\begin{enumerate}
    \item Parse $(cpk, \mathsf{OPMem}) = pk$.
    \item Sample $r \samp \zo^{\cosettcount}$.
    \item Sample a PRF key $K_2$ for $F_2$ as $K_2 \samp F_2.\mathsf{KeyGen}(1^{\lambda})$.
    \item Sample $\mathsf{OPCt} \samp \io(\mathsf{PCt}_{\mathsf{OPMem}, cpk, K_2, r, m})$, where $\mathsf{PCt}_{\mathsf{OPMem}, cpk, K_2, r, m}$ is the following program.
 \begin{mdframed}
        {\bf $\underline{\mathsf{PCt}_{\mathsf{OPMem}, cpk, K_2, r, m}(id || f, u_1, \dots, u_{\cosettcount})}$}
        
        {\bf Hardcoded: $\mathsf{OPMem}, cpk, K_2, r, m$}
        \begin{enumerate}[label=\arabic*.]
            \item Run $\mathsf{OPMem}(id || f, u_1, \dots, u_{\cosettcount}, r)$. If it outputs $0$, output $\perp$ and terminate.
            \item Output $\mathsf{IBE.Enc}(cpk, id || f, f(m); F_2(K_2, id || f))$.
        \end{enumerate}
    \end{mdframed}
    \item Output $(\mathsf{OPCt}, r)$.
\end{enumerate}

\paragraph{$\underline{\mathsf{FE.Dec}(\regi{key}, ct)}$}
\begin{enumerate}
    \item Parse $((\reg_i)_{i \in [\cosettcount]}, ck, id, f) = \regi{key}$ and $(\mathsf{OPCt}, r) = ct$.
    \item For indices $i \in [\cosettcount]$ such that $(r)_i = 1$, apply $H^{\otimes\kappa(L(\lambda)+\lambda)}$ to $\reg_i$.
    \item Run the program $\mathsf{OPCt}$ coherently on $id, f$ and $(\reg_i)_{i \in [\cosettcount]}$.
    \item Measure the output register and denote the outcome by $cct$.
    \item Output $\mathsf{IBE.Dec}(ck, cct)$.
\end{enumerate}

Correctness with probability $1$ follows in a straightforward manner from the correctness of the underlying schemes. We claim that the construction is also secure.
\begin{theorem}\label{thm:feantipiracy}
 $\mathsf{FE}$ satisfies $\gamma$-anti-piracy (\cref{defn:unclonfe}) for any inverse polynomial $\gamma$.
\end{theorem}
When we instantiate the assumed primitives with known constructions, we get the following corollary.
\begin{corollary}
Assuming subexponentially secure $i\mathcal{O}$ and subexponentially secure LWE, there exists a public-key functional encryption scheme that satisfies anti-piracy security against unbounded collusion.
\end{corollary}
\subsection{Proof of Anti-Piracy}\label{sec:feproof}
Proof will closely follow the strong anti-piracy security proof for our public-key encryption construction in \cref{sec:pkeproof}, which crucially relies on projective implementations to simultaneously extract vectors from all registers to obtain a reduction to the monogamy-of-entanglement game. However, since the freeloaders in the functional encryption security game are interactive as opposed to the ones in regular public-key encryption, we cannot use projective implementations directly. Therefore, we first make the post-challenge-ciphertext steps non-interactive by providing the freeloaders with a punctured master secret key $pmsk$ that lets them issue their own functional keys, as long as $f(m^0) = f(m^1)$. 

We give the following two definitions, specific to our construction $\fe$. Recall that we also assume that $\mathsf{IBE.KeyGen}$ is deterministic, which is true for the construction we give in \cref{sec:consibe}.

We start with the post-challenge-non-interactive \emph{regular} anti-piracy definition. It is defined similar to \cref{defn:regularpkeantipir}.
\begin{definition}[CPA-Style Post-Challenge-Ciphertext-Non-interactive Anti-Piracy Security for $\fe$]\label{defn:feni}
    Consider the following game between the challenger and an adversary $\adve$.

    \paragraph{$\underline{\nife(\lambda, \adve)}$}
    \begin{enumerate}
        \item The challenger runs $msk, pk \samp \fe.\mathsf{Setup}(1^\lambda)$ and submits $pk$ to the adversary. It also initializes the set $\fclass$. It parses $(cmsk, K_1) = msk$.
        \item \textbf{Query Phase 1:} For multiple rounds, the adversary adaptively submits a function $f \in \mathfrak{F}$ and a query type, either $\mathsf{CLASSICAL}$ or $\mathsf{PROTECTED}$. For each $f$, the challenger does the following. It first computes $sk_f \samp \mathsf{FE}.\mathsf{KeyGen}(msk, f)$. 
        
        Then, if the query type is $\mathsf{CLASSICAL}$, it adds $f$ to $\fclass$ and submits $sk_f$ to the adversary. 
        
        Otherwise, it computes $\reg_f \samp \fe.\qkeygen(sk_f)$ and submits $\reg_f$ to the adversary. 

        \item The adversary outputs a pair of challenge messages $m^0, m^1$, a $(k + 1)$-partite register $\regi{adv}$ (where $k$ is the number of queries of the type $\mathsf{PROTECTED}$) and freeloader unitaries $\{U_{\ell}\}_{\ell \in [k + 1]}$.
\item The challenger checks if  $f(m^0) = f(m^1)$ for all $f \in \fclass$. If not, it outputs $0$ and terminates. 

Otherwise, the challenger computes $pmsk \samp \io(\mathsf{PKey}_{cmsk, K_1})$.

\begin{mdframed}
        {\bf $\underline{\mathsf{PKey}_{cmsk, K_1}(id || f)}$}
        
        {\bf Hardcoded: $cmsk, K_1, m^0, m^1$}
        \begin{enumerate}[label=\arabic*.]
          \item Check if $f(m^0) = f(m^1)$. If not, output $\perp$ and terminate.
        \item Compute $ck = \mathsf{IBE.KeyGen}(cmsk, id || f)$.
        \item $(A_i, s_i, s_i')_{i \in [\cosettcount]} = \mathsf{CosetGen}(\cosetgenparam; F_1(K_1, id || f))$.
    \item Output $(ck, id, f, (A_i, s_i, s_i')_{i \in [\cosettcount]})$.
        \end{enumerate}

    \end{mdframed}

    \item For $\ell \in [k + 1]$, the challenger executes $b'_{\ell} \samp \qunivcla(U_\ell, \regi{adv}[\ell], ct_\ell, pmsk)$. Then, it checks if $b'_\ell = b_\ell$ and if $f(m^0) = f(m^1)$ for all $f \in \mathcal{F}_\ell$. It outputs $1$ if and only if all the check pass.
    \end{enumerate}

We say $\fe$ satisfies post-challenge-ciphertext non-interactive $\gamma$-anti-piracy security if for any QPT adversary, 
\begin{equation*}
     \Pr[\nife(\lambda, \gamma(\lambda), \adve) = 1] \leq \frac{1}{2} + \gamma(\lambda) + \negl(\lambda).
\end{equation*}
\end{definition}

We now define decryptor testing and strong anti-piracy.
\begin{definition}[Functional Encryption Decryptor Testing]\label{defn:goodecrfunc}
    In the anti-piracy game between the challenger and an adversary, fix $\ell \in [k + 1]$, some values $m^0, m^1$ of the challenge messages, a freeloader unitary $U_\ell$ and some value $st$ of a classical state of the challenger (which will be defined later). Let $\mathcal{D}$ be an efficient ciphertext and punctured master secret key distribution that can depend on $st$. That is, $\mathcal{D}^{st}(m; r)$ is an efficient classical algorithm where $m \in \mathcal{M}$, $r \in \mathcal{R}$ and $\mathcal{R}$ is a random coin set. 
    
    Consider the following mixture of binary \emph{projective measurements} $\mathcal{P}$, induced by $\mathcal{D}$ and $m^0, m^1,b U_\ell st$, applied on a state $\rho$.

    \begin{enumerate}
        \item Sample $b \samp \zo$.
        \item Sample $r \samp \mathcal{R}$.
        \item Run $ct, pmsk \samp \mathcal{D}^{st}(m^b; r)$.
        \item Execute $U_\ell$ on $(\rho, pmsk, ct)$, and measure the first qubit of the output register, let $b'$ be the output.
        \item Output 1 if $b' = b$. Otherwise, output $0$.
    \end{enumerate}

    Observe that we can efficiently execute the above measurement for arbitrary given superpositions of $r$ and $b$ values. Therefore, by \cref{sec:piti}, there exists exact and efficient approximated projective and threshold implementations for $\mathcal{P}$. We write $\mathsf{PI}_{\ell, \mathcal{D}}$ and $\mathsf{API}^{\eps, \delta}_{\ell, \mathcal{D}}$ to denote the projective implementation and approximate projective implementation of $\mathcal{P}$, respectively. Similarly, let $\mathsf{TI}_{\ell, \mathcal{D}, \eta}$ and $\mathsf{ATI}^{\eps, \delta}_{\ell, \mathcal{D}, \eta}$ denote the threshold and efficient approximate threshold implementations of $\mathcal{P}$ for a threshold value $\eta$. 
    
    The fixed values $m^0, m^1, U_\ell, st$, omitted in the notation, will be clear from the context. Unless otherwise specified, we will write $\mathcal{D}$ to denote the honest distribution where the ciphertext is sampled as \begin{equation*}
        ct \samp \fe.\mathsf{Enc}(pk, m)
    \end{equation*}
    and $pmsk$ is sampled as in \cref{defn:feni}, where $pk$ is part of $st$.
\end{definition}

\begin{definition}[CPA-Style Strong Anti-Piracy Security for $\fe$]\label{defn:strongapfe}
    Consider the following game between the challenger and an adversary $\adve$.

    \paragraph{$\underline{\strongunclonfegame(\lambda, \gamma, \adve)}$}
    \begin{enumerate}
        \item The challenger runs $msk, pk \samp \fe.\mathsf{Setup}(1^\lambda)$ and submits $pk$ to the adversary. It also initializes the set $\fclass$. It parses $(cmsk, K_1) = msk$.
        \item \textbf{Query Phase 1:} For multiple rounds, the adversary adaptively submits a function $f \in \mathfrak{F}$ and a query type, either $\mathsf{CLASSICAL}$ or $\mathsf{PROTECTED}$. For each $f$, the challenger does the following. It first computes $sk_f \samp \mathsf{FE}.\mathsf{KeyGen}(msk, f)$. 
        
        Then, if the query type is $\mathsf{CLASSICAL}$, it adds $f$ to $\fclass$ and submits $sk_f$ to the adversary. 
        
        Otherwise, it computes $\reg_f \samp \fe.\qkeygen(sk_f)$ and submits $\reg_f$ to the adversary.
 \item The adversary outputs a pair of challenge messages $m^0, m^1$, a $(k + 1)$-partite register $\regi{adv}$ (where $k$ is the number of queries of the type $\mathsf{PROTECTED}$) and freeloader unitaries $\{U_{\ell}\}_{\ell \in [k + 1]}$.

\item  The challenger checks if  $f(m^0) = f(m^1)$ for all $f \in \fclass$. If not, it outputs $0$ and terminates.

        \item The challenger applies the test
        \begin{equation*}
        \bigotimes_{\ell \in [k+1]} \mathsf{TI}_{\ell, \mathcal{D}, 1/2+\gamma}
        \end{equation*}
        to $\reg$ and outputs $1$ if and only if all the measurement results are $1$.
  
    \end{enumerate}

We say $\fe$ satisfies strong $\gamma$-anti-piracy security if for any QPT adversary, 
\begin{equation*}
     \Pr[\strongunclonfegame(\lambda, \gamma(\lambda), \adve) = 1] \leq \negl(\lambda).
\end{equation*}
\end{definition}

We first prove that the stronger definition implies the regular anti-piracy security (\cref{defn:unclonfe}).
\begin{theorem}\label{thm:festrongtoregular}
Suppose $\fe$ satisfies strong $\gamma$-anti-piracy security. Then, it also satisfies regular $\gamma$-anti-piracy security.
\end{theorem}
\begin{proof}
We first show that strong $\gamma$-anti-piracy security implies post-challenge-ciphertext non-interactive $\gamma$-anti-piracy security, by generalizing an argument made by \cite{CLLZ21} for public-key encryption.

By the properties of projective implementations (\cref{thm:piprop}); in the game $\nife$, instead of applying the freeloader unitary and comparing the output to $b'_\ell$, if we apply the corresponding projective implementation (defined in \cref{defn:goodecrfunc}) to obtain a value $p_\ell$ and output a bit $a_\ell = 1$ with probability $p_\ell$, we get the correct output distribution for all registers simultaneously\footnote{When used directly, \cref{thm:piprop} would give us this for only single register at a time. However, note that the joint distribution is also correct since the projective implementations are correct for any input state, and hence we can consider the post-measurement state of any register conditioned on measurement outcomes of the other registers, and the projective implementation will still have the correct output distribution.}. Hence, we can equivalently execute the security game $\nife$ by applying these projective implementations, obtaining some $a_\ell$, and outputting $1$ if and only if $a_\ell = 1$ for all $\ell \in [k + 1]$.

Now, note that by construction of $\mathsf{TI}$ and by the assumption that $\fe$ satisfies strong $\gamma$-anti-piracy security, we have
\begin{equation*}
    \Pr[\forall \ell \in [k + 1]~~p_{\ell} \geq \frac{1}{2} + \gamma(\lambda)] \leq \negl(\lambda).
\end{equation*}

Then, by above,
\begin{align*}
    \Pr&[\nife(\lambda, \adve) = 1] = \E[p_1\cdots p_\ell] \\= &\Pr[\forall \ell \in [k + 1]~~p_{\ell} \geq \frac{1}{2} + \gamma(\lambda)]\cdot \E[p_1\cdots p_\ell | \forall \ell \in [k + 1]~~p_{\ell} \geq \frac{1}{2} + \gamma(\lambda)] + \\
    &\Pr[\exists \ell \in [k + 1]~~p_{\ell} < \frac{1}{2} + \gamma(\lambda)]\cdot \E[p_1\cdots p_\ell | \exists \ell \in [k + 1]~~p_{\ell} < \frac{1}{2} + \gamma(\lambda)] \\
    \leq & \negl(\lambda)\cdot 1 + 1\cdot(\frac{1}{2} + \gamma(\lambda)).
\end{align*}

This completes the proof that strong $\gamma$-anti-piracy security implies post-challenge-ciphertext non-interactive $\gamma$-anti-piracy security.

Now, we show that the latter implies regular $\gamma$-anti-piracy security. 
A freeloader adversary $\mathcal{B}'$ for $\nife$ can simulate a freeloader adversary $\mathcal{B}$ for the regular $\gamma$-anti-piracy as follows. Whenever $\mathcal{B}$ makes a query for a function $f$ in Query Phase 2, $\mathcal{B}'$ samples a random identity $id$ and evaluates $pmsk$ on $id, f$. Since $f$ satisfies $f(m^0) = f(m^1)$, by correctness of $\io$, it will be able to obtain the correct key. If query is of type $\mathsf{CLASSICAL}$, then $\mathcal{B}'$ submits the obtained classical key $fk$ to $\mathcal{B}$. Otherwise, it runs $\mathsf{FE.QKeyGen}$ on $fk$ and then submits the resulting quantum key. This perfectly simulates the regular anti-piracy game, hence post-challenge-ciphertext non-interactive $\gamma$-anti-piracy security implies regular anti-piracy security.
\end{proof}

\subsubsection*{Reducing to Monogamy-of-Entanglement}
\begin{theorem}
    $\fe$ satisfies strong $\gamma$-anti-piracy security for any inverse polynomial $\gamma$. 
\end{theorem}

Combining this theorem with the results from the previous section yields that our FE construction is secure. Proof of this theorem is almost identical to the proof of security of our public-key encryption scheme (\cref{sec:pkeproof}). Therefore, we will omit the proofs of some sub-claims.

Throughout the proof, we will interpret identity strings for $\mathsf{IBE}$, which are $\circsize + \lambda$-bit strings, as integers in the set $\{0, 1, \dots, 2^{\circsize + \lambda} - 1\}$.

Fix any inverse polynomial $\gamma(\lambda)$ and suppose for a contradiction that there exists an efficient adversary $\adve$ that wins the strong $\gamma$-anti-piracy game with non-negligible probability. Let $k$ denote the number of copy-protected keys obtained by the adversary. 
Define $\hyb_0$ to be the original security game $\strongunclonfegame(\lambda, \gamma(\lambda), \adve)$. 

Define $\hyb_1$ by modifying $\hyb_0$ as follows. When generating functional keys, we sample the random identity strings $id$ in a collision-free way. Further, at the end of the game, the challenger instead applies the test $\bigotimes_{\ell \in [k+1]} \mathsf{ATI}^{\eps,\delta}_{\ell, \mathcal{D}, 1/2+\frac{31\gamma}{32}}$ instead of $\mathsf{TI}$ where we set $\eps = \frac{\gamma}{32k}$ and $\delta = 2^{-10\lambda}\cdot2^{-10\lambda^{\constmoecoll}}$.

\begin{claim}
    $\Pr[\hyb_2 = 1] > 1/p(\lambda)$ for some polynomial $p(\cdot)$ and infinitely many values of $\lambda > 0$
\end{claim}
\begin{proof}
    Follows from the same argument as in \cref{sec:pkeproof}.
\end{proof}

\begin{definition}
    For all $j \in [k]$, let $id_j || f_j$ denote the identity string and let $(A_i^j, s_i^j, s_i^{'j})_{{i \in [\cosettcount]}}$ denote the tuple of cosets sampled during the sampling of the functional key for the $(\alpha^{-1}(j))$-th query of type $\mathsf{PROTECTED}$. That is, it is the coset tuple associated with $id_j || f_j$.
\end{definition}

We now define a monogamy-of-entanglement type game $\mathcal{G}$, similar to the game defined in the PKE proof.

\paragraph{$\underline{\mathcal{G}(\lambda, (\adve'_0, \adve'_1, \adve'_2))}$}
\begin{enumerate}
        \item The challenger runs $msk, pk \samp \fe.\mathsf{Setup}(1^\lambda)$ and submits $pk$ to the adversary. It also initializes the set $\fclass$. It parses $(cmsk, K_1) = msk$.
        \item \textbf{Query Phase 1:} For multiple rounds, the adversary adaptively submits a function $f \in \mathfrak{F}$ and a query type, either $\mathsf{CLASSICAL}$ or $\mathsf{PROTECTED}$. For each $f$, the challenger does the following. It first computes $sk_f \samp \mathsf{FE}.\mathsf{KeyGen}(msk, f)$. 
        
        Then, if the query type is $\mathsf{CLASSICAL}$, it adds $f$ to $\fclass$ and submits $sk_f$ to the adversary. 
        
        Otherwise, it computes $\reg_f \samp \fe.\qkeygen(sk_f)$ and submits $\reg_f$ to the adversary.
            \item The adversary outputs a pair of challenge messages $m^0, m^1$ and an index $j^* \in [k]$ where $k$ is the number of queries it made of type $\mathsf{PROTECTED}$.

\item  The challenger checks if  $f(m^0) = f(m^1)$ for all $f \in \fclass$. If not, it outputs $0$ and terminates.

Otherwise, the challenger computes $K_1\{id_{j^*} || f_{j^*}\} \samp F_1.\mathsf{Punc}(K_1, id_{j^*} || f_{j^*})$ and submits it to the adversary.  

\item The challenger outputs a \emph{bipartite} register $\regi{bip}$.

    \item For $\ell \in \{1, 2\}$, the challenger does the following.
    \begin{enumerate}[label*=\arabic*.]
        \item Sample $r_\ell \samp \zo^{\cosettcount}$.
        \item Run $\adve'_\ell$ on $\regi{bip}[\ell]$, $(A^{j^*}_i)_{i \in [\cosettcount]}$ and $r_\ell$ to obtain a tuple of vectors $(v_{\ell, i})_{i \in [{\cosettcount}]}$.
        \item For all $i \in [\cosettcount]$, check if $v_{\ell, i} \in A^{j^*}_{i} + s^{j^*}_i$ if $(r_\ell)_i = 0$ and check if $v_{\ell, i} \in {(A^{j^*})}^\perp_{i} + s^{'{j^*}}_i$ if $(r_\ell)_i = 1$.
    \end{enumerate}
    If all the checks pass, the challenger outputs $1$. Otherwise, it outputs $0$.
\end{enumerate}

It is straightforward to reduce this game to the collusion-resistant MoE game. We construct our adversary for $\mathcal{G}$ as follows, where we will define challenge ciphertext-punctured master secret key distributions $\mathcal{D}_j$ later. Without loss of generality\footnote{In the general case, the adversary would simply sample $j^*$ from $[k']$ where $k'$ is the number of queries made by the adversary $\adve$ of type $\mathsf{PROTECTED}$ that do satisfy $f(m_0) \neq f(m_1)$.}, we will assume that all the queries made by the adversary $\adve$ of type $\mathsf{PROTECTED}$ satisfy $f(m_0) \neq f(m_1)$.
\begin{mdframed}
        {\bf $\underline{\adve_0'(pk)}$}
        \begin{enumerate}
            \item Simulate $\adve$ on $pk$ by making a functional key query to the challenger whenever $\adve$ makes a query, and forwarding the obtained key to it. Let $\regi{adv}$ be the $(k + 1)$-partite register (with state $\sigma$) and $(m^0, m^1)$ be the challenge messages output by $\adve$ at the end of the query phase.

            \item Uniformly at random sample $x, y, j^*$ such that $1 \leq x < y \leq k + 1$ and $j^* \in \{1, \dots, k\}$.

            \item Output $j^*$ to the challenger and obtain $K_1\{id_{j^*} || f_{j^*}\}$.
            
            \item Apply $\mathsf{API}^{\eps, \delta}_{\ell, \mathcal{D}_0}$ to all registers $\regi{adv}[\ell]$ for $\ell \in [k + 1]$, let $b_{\ell, 0}$ be the measurement outcomes.
            \item Apply $\mathsf{API}^{\eps, \delta}_{\ell, \mathcal{D}_i}$ in succession for $i = 1$ to $j^*$ to $\regi{adv}[x]$, let $b_{x, i}$ be the measurement outcomes.
            \item Apply $\mathsf{API}^{\eps, \delta}_{\ell, \mathcal{D}_i}$ in succession for $i = 1$ to $j^*$ to $\regi{adv}[y]$, let $b_{y, i}$ be the measurement outcomes.
            \item Output  \begin{align*}
                (&(\regi{adv}[x], j^*, x, y, (b_{\ell, 0})_{\ell \in [k + 1]}, (b_{x,i})_{i \in [j^*]}, (b_{y,i})_{i \in [j^*]}),\\ &(\regi{adv}[y], j^*, x, y, (b_{\ell, 0})_{\ell \in [k + 1]}, (b_{x,i})_{i \in [j^*]}, (b_{y,i})_{i \in [j^*]}, ),\\ &j^*).
            \end{align*}
            
        \end{enumerate}
    \end{mdframed}
    For $j \in \{1, \dots, k\}$, define $\mathcal{D}_j$ to be the following challenge ciphertext-punctured master secret key distribution.
\begin{enumerate}
   \item Sample $r \samp \zo^{\cosettcount}$.
    \item Sample a PRF key $K_2$ for $F_2$ as $K_2 \samp F_2.\mathsf{KeyGen}(1^\lambda)$.
    \item Sample $\mathsf{OPCt} \samp \io(\textcolor{red}{ \mathsf{PCt}^{(j)}_{\mathsf{OPMem}, cpk, K_2, r, m, id_j || f_j}})$
    \begin{mdframed}
        {\bf $\underline{\mathsf{PCt}^{(j)}_{\mathsf{OPMem}, cpk, K_2, r, m^b, m_{1-b}, id_j || f_j}(id || f, u_1, \dots, u_{\cosettcount})}$}
        
        {\bf Hardcoded: $\mathsf{OPMem}, cpk, K_2, r, m^b, m^{1-b}, \textcolor{red}{id_j || f_j}$}
        \begin{enumerate}[label=\arabic*.]
            \item Run $\mathsf{OPMem}(id || f, u_1, \dots, u_{\cosettcount}, r)$. If it outputs $0$, output $\perp$ and terminate.
            \item \textcolor{red}{If $id || f < id_j || f_j$, set $a = f(m^{1-b})$. Otherwise, set $a = f(m^{b})$.}
            \item Output $\mathsf{IBE.Enc}(cpk, id, \textcolor{red}{a}; F_2(K_2, id))$.
        \end{enumerate}
    \end{mdframed}
     \item Sample $pmsk \samp \io(\mathsf{PKey}_{cmsk, K_1\{id_{j^*}||f_{j^*}\}})$.

\begin{mdframed}
        {\bf $\underline{\mathsf{PKey}_{cmsk, K_1\{id_{j^*}||f_{j^*}\}}(id || f)}$}
        
        {\bf Hardcoded: $cmsk, K_1\{id_{j^*}||f_{j^*}\}, m^0, m^1$}
        \begin{enumerate}[label=\arabic*.]
          \item Check if $f(m^0) = f(m^1)$. If not, output $\perp$ and terminate.
        \item Compute $ck = \mathsf{IBE.KeyGen}(cmsk, id || f)$.
        \item $(A_i, s_i, s_i')_{i \in [\cosettcount]} = \mathsf{CosetGen}(\cosetgenparam; F_1(K_1\{id_{j^*}||f_{j^*}\}, id || f))$.
    \item Output $(ck, id, f, (A_i, s_i, s_i')_{i \in [\cosettcount]})$.
        \end{enumerate}

    \end{mdframed}
    \item Output $(\mathsf{OPCt}, r), pmsk$.
\end{enumerate}
We define $\mathcal{D}_{0}$ to be the distribution where ciphertext is computed honestly and $pmsk$ is computed as in \cref{defn:feni}. We define $\mathcal{D}_{k + 1}$ as follows.
\begin{enumerate}
   \item Sample $r \samp \zo^{\cosettcount}$.
    \item Sample a PRF key $K_2$ for $F_2.\mathsf{KeyGen}(1^\lambda)$.
    \item Sample $\mathsf{OPCt} \samp \io(\textcolor{red}{ \mathsf{PCt}^{(k+1)}_{\mathsf{OPMem}, cpk, K_2, r}})$
    \begin{mdframed}
        {\bf $\underline{\mathsf{PCt}^{(k+1)}_{\mathsf{OPMem}, cpk, K_2, r, m^0, m^1, b}(id || f, u_1, \dots, u_{\cosettcount})}$}
        
        {\bf Hardcoded: $\mathsf{OPMem}, cpk, K_2, r, m^{1-b}$}
        \begin{enumerate}[label=\arabic*.]
            \item Run $\mathsf{OPMem}(id || f, u_1, \dots, u_{\cosettcount}, r)$. If it outputs $0$, output $\perp$ and terminate.
            \item Output $\mathsf{IBE.Enc}(cpk, id, \textcolor{red}{f(m^{1-b})}; F_2(K_2, id))$.
        \end{enumerate}
    \end{mdframed}
    \item Sample $pmsk$ as in $\mathcal{D}_j$.
    \item Output $(\mathsf{OPCt}, r), pmsk$.
\end{enumerate}

\begin{definition}
    Let $\expfromddd{\mathcal{C}}{\ell}$ denote the outcome of the following experiment where $\mathcal{C}$ is a challenge ciphertext-punctured master secret key distribution that can depend on $pp$.
\begin{enumerate}
\item Execute $pk, sk \samp \mathsf{PKE.Setup}(1^\lambda)$.
\item Simulate the first steps of $\adve'_0$ and the challenger of $\mathcal{G}$:
\begin{enumerate}[label*=\arabic*.]
            \item Simulate $\adve$ on $pk$ by making a functional key query to the challenger whenever $\adve$ makes a query, and forwarding the obtained key to it. Let $\regi{adv}$ be the $(k + 1)$-partite register (with state $\sigma$) and $(m^0, m^1)$ be the challenge messages output by $\adve$ at the end of the query phase.

            \item Uniformly at random sample $x, y, j^*$ such that $1 \leq x < y \leq k + 1$ and $j^* \in \{1, \dots, k\}$.

            \item Compute $K_1\{id_{j^*} || f_{j^*}\} \samp F_1.\mathsf{Punc}(K_1, id_{j^*} || f_{j^*})$.
\end{enumerate}
    \item Set $pp =  (x, y, j^*,(id_j || f_j)_{j\in[k+1]}, m_0, m_1, pk)$.
    \item Sample $b \samp \zo$.
    \item Sample $ct, pmsk \samp \mathcal{C}(pp, m^b)$.
    \item Output $\regi{adv}, (b, ct, pmsk), pp$.
\end{enumerate} 
We will write $\expfromddd{\mathcal{C}}{\ell} \approx_\nu^c \expfromddd{\mathcal{C}'}{\ell}$ to denote that the advantage of any computational adversary in distinguishing the outcomes of these experiments is $\nu$.
\end{definition}

\begin{claim}\label{claim:femesresdist}
Let $\tau$ be the state of the bipartite register $\regi{adv}[x,y]$ output by $\adve'_0$ in $\mathcal{G}$, and also consider the classical values $j^*, x, y, \{b_{\ell, i}\}_{\ell, i}$ contained in the output of $\adve'_0$.
    
    Suppose we apply the measurement $\mathsf{API}^{\eps, \delta}_{x, \mathcal{D}_{j^* + 1}}\otimes\mathsf{API}^{\eps, \delta}_{y, \mathcal{D}_{j^* + 1}}$ to $\tau$ and let $b_{x,j^*+1}, b_{y,j^*+1}$ denote the measurement outcomes we obtain. Then,
   
    \begin{equation*}
     \Pr[ b_{x,j^*}  - b_{x,j^*+1} > \frac{29\gamma}{32k}  \wedge  b_{y,j^*} - b_{y,j^*+1} > \frac{29\gamma}{32k}] > \frac{1}{4p(\lambda)\cdot k^3(\lambda)}
    \end{equation*}
    
    where the probability is taken over the randomness of the challenger, the adversary $\adve_0'$ and the measurement outcomes.
\end{claim}
\begin{proof}
    Follows from the same argument as \cref{claim:mesresdist}. Only caveat is that we need to prove that the success probability of the freeloaders with respect to the challenge ciphertext distribution $\mathcal{D}_{k + 1}$ is $\leq 1/2$. However, this is indeed true, since $\mathcal{D}_{k + 1}$ is encoding $m^{1 - b}$ while the challenge bit is $b$.
\end{proof}

\begin{definition}
When we refer to $((id || f) + \Delta)$ as a function, we mean the following. We associate the strings in $\zo^{\lambda + Q(\lambda)}$ with numbers $\{0,1,\dots, 2^{\lambda + Q(\lambda)}\}$ in the canonical way, and we compute the sum of the numbers associated with $\Delta$ and $(id || f)$. Then, we switch back to the bit representation of this number, and take the last $Q$ bits, and we define $(id || f) + \Delta$ to be the circuit defined by this $Q$-bit string.
\end{definition}

Now, we define some intermediary distributions. Define the following for all $j \in \{0, 1, \dots, k\}$ and $\Delta \in \{0, 1, \dots, id_{j+1}||f_{j+1}-id_{j}||f_j - 1\}$. For notational convenience, also define $\mathcal{D}_j^{id_{j+1}||f_{j+1}-id_{j}||f_j, 0}$ to be $\mathcal{D}_{j+1}^{(0,  0)}$ for all $j \in \{0, 1, \dots, k\}$. Also note that $\mathcal{D}_j^{(0,0)}$ is exactly the same as $\mathcal{D}_j$ for $j \in [k]$.

\begin{itemize}
    \item{$\underline{\mathcal{D}_{j}^{(\Delta, 0)}}$:} 
    \begin{enumerate}
   \item Sample $r \samp \zo^{\cosettcount}$.
    \item Sample a PRF key $K_2$ for $F_2.\mathsf{KeyGen}(1^\lambda)$.
    \item Sample $\mathsf{OPCt} \samp \io({ \mathsf{PCt}^{(j, \Delta, 0)}_{\mathsf{OPMem}, cpk, K_2, r, m, id_j + \Delta}}).$
 \begin{mdframed}
        {\bf $\underline{{ \mathsf{PCt}^{(j, \Delta, 0)}_{\mathsf{OPMem}, cpk, K_2, r, m^b, m^{1-b}, id_j || f_j + \Delta}}(id || f, u_1, \dots, u_{\cosettcount})}$}
        
        {\bf Hardcoded: ${\mathsf{OPMem}, cpk, K_2, r, m^b, m^{1-b}, \textcolor{red}{id_j || f_j + \Delta}}$}
        \begin{enumerate}
            \item Run $\mathsf{OPMem}(id || f, u_1, \dots, u_{\cosettcount}, r)$. If it outputs $0$, output $\perp$ and terminate.
            \item \textcolor{red}{If $id || f < id_j || f_j + \Delta$, set $a = f(m^{1-b})$. Otherwise, set $a = f(m^b)$.}
            \item Output $\mathsf{IBE.Enc}(cpk, id, {a}; F_2(K_2, id))$.
        \end{enumerate}
    \end{mdframed}
    \item Sample $pmsk$ as in $\mathcal{D}_j$.
    \item Output $(\mathsf{OPCt}, r), pmsk$.
\end{enumerate}

        \item{$\underline{\mathcal{D}_{j}^{(\Delta, 1)}}$:}  \begin{enumerate}
   \item Sample $r \samp \zo^{\cosettcount}$.
    \item Sample a PRF key $K_2$ for $F_2.\mathsf{KeyGen}(1^\lambda)$.
    \item \textcolor{red}{ $ct^* = \mathsf{IBE.Enc}(cpk, id_j||f_j + \Delta, (id_j||f_j + \Delta)(m^b); F_2(K_2, id_j||f_j+\Delta)).$}
        \item \textcolor{red}{ $K_2\{id_j||f_j+\Delta\} \samp F_2.\mathsf{Punc}(K_2, id_j||f_j+\Delta)$.}
            
    \item Sample     $\mathsf{OPCt} \samp \io({ \mathsf{PCt}^{(j, \Delta, 1)}_{\mathsf{OPMem}, cpk, K_2\{id_j+\Delta\}, r, m, id_j + \Delta, ct^*}}).$
 \begin{mdframed}
        {\bf $\underline{{ \mathsf{PCt}^{(j, \Delta, 1)}_{\mathsf{OPMem}, cpk, K_2\{id_j+\Delta\}, r, m, id_j + \Delta, ct^*}}(id || f, u_1, \dots, u_{\cosettcount})}$}
        
        {\bf Hardcoded: ${\mathsf{OPMem}, cpk, \textcolor{red}{K_2\{id_j||f_j+\Delta\}}, r, m^b, m^{1-b}, id_j||f_j + \Delta, \textcolor{red}{ct^*}}$}
        \begin{enumerate}
            \item Run $\mathsf{OPMem}(id || f, u_1, \dots, u_{\cosettcount}, r)$. If it outputs $0$, output $\perp$ and terminate.
            \item \textcolor{red}{If $id||f = id_j || f_j + \Delta$, output $ct^*$ and terminate.}
            \item {If \textcolor{red}{$id < id_j || f_j + \Delta + 1$}, set $a = f(m^{1-b})$. Otherwise, set $a = f(m)$.}
            \item Output $\mathsf{IBE.Enc}(cpk, id, {a}; F_2(K_2, id))$.
        \end{enumerate}
    \end{mdframed}
    \item Sample $pmsk$ as in $\mathcal{D}_j$.
    \item Output $(\mathsf{OPCt}, r), pmsk$.
\end{enumerate}

    \item{$\underline{\mathcal{D}_{j}^{(\Delta, 2)}}$:} 
    \begin{enumerate}
   \item Sample $r \samp \zo^{\cosettcount}$.
    \item Sample a PRF key $K_2$ for $F_2.\mathsf{KeyGen}(1^\lambda)$.
    \item \textcolor{red}{Sample $z^*$ uniformly at random from the output space of $F_2$.}
    \item \textcolor{red}{ $ct^* = \mathsf{IBE.Enc}(cpk, id_j||f_j + \Delta, (id_j||f_j + \Delta)(m^b); z^*).$}
        \item { $K_2\{id_j||f_j+\Delta\} \samp F_2.\mathsf{Punc}(K_2, id_j||f_j+\Delta)$.}
            
    \item Sample     $\mathsf{OPCt} \samp \io({ \mathsf{PCt}^{(j, \Delta, 2)}_{\mathsf{OPMem}, cpk, K_2\{id_j+\Delta\}, r, m, id_j + \Delta, ct^*}}).$
 \begin{mdframed}
        {\bf $\underline{{ \mathsf{PCt}^{(j, \Delta, 2)}_{\mathsf{OPMem}, cpk, K_2\{id_j||f_j+\Delta\}, r, m^b, m^{1-b}, id_j||f_j + \Delta, ct^*}}(id || f, u_1, \dots, u_{\cosettcount})}$}
        
        {\bf Hardcoded: ${\mathsf{OPMem}, cpk, {K_2\{id_j+\Delta\}}, r, m^b, m^{1-b}, id_j||f_j + \Delta, {ct^*}}$}
        \begin{enumerate}
            \item Run $\mathsf{OPMem}(id || f, u_1, \dots, u_{\cosettcount}, r)$. If it outputs $0$, output $\perp$ and terminate.
            \item {If $id||f = id_j||f_j + \Delta$, output $ct^*$ and terminate.}
            \item {If {$id < id_j||f_j + \Delta + 1$}, set $a = f(m^{1-b})$. Otherwise, set $a = f(m)$.}
            \item Output $\mathsf{IBE.Enc}(cpk, id, {a}; F_2(K_2, id))$.
        \end{enumerate}
    \end{mdframed}
    \item Sample $pmsk$ as in $\mathcal{D}_j$.
    \item Output $(\mathsf{OPCt}, r), pmsk$.
\end{enumerate}

    \item{$\underline{\mathcal{D}_{j}^{(\Delta, 3)}}$:} 
   \begin{enumerate}
   \item Sample $r \samp \zo^{\cosettcount}$.
    \item Sample a PRF key $K_2$ for $F_2.\mathsf{KeyGen}(1^\lambda)$.
        \item {Sample $z^*$ uniformly at random from the output space of $F_2$.}
    \item { $ct^* = \mathsf{IBE.Enc}(cpk, id_j||f_j + \Delta, (id_j||f_j + \Delta)(m^b); z^*).$}
    \item { $K_2\{id_{j}||f_j+\Delta\} \samp F_2.\mathsf{Punc}(K_2, id_{j}||f_j+\Delta)$.}
    \textcolor{red}{\item Compute $(A^{*}_i,s_i^{*},s_i^{'{*}}) = F_1(K_1, id_j||f_j + \Delta).$
      \item For $i \in [\cosettcount]$, set $g_i = \mathsf{Can}_{A^{*}_i}$ if $(r)_i = 0$ and set $g_i = \mathsf{Can}_{(A^{*}_i)^{\perp}}$ if $(r)_i = 1$.
    \item For $i \in [\cosettcount]$, compute $y_i = g_i(s_i^{*})$ if $(r)_i = 0$ and $y_i =  g_i(s_i^{'{*}})$ if $(r)_i = 1$.
    \item Set $g$ to be the function $g(v_1, \dots, v_{\cosettcount}) = (g_1(v_1) || \dots || g_{\cosettcount}(v_{\cosettcount}))$.
    \item Set $y = y_1 || \dots || y_{\cosettcount}$.
    \item  $\mathsf{OCC} \samp \mathsf{CCObf.Obf}(g, y, ct^*)$.}

         \item Sample    $\mathsf{OPCt} \samp \io({ \mathsf{PCt}^{(j,\Delta,3)}_{\mathsf{OPMem}, cpk, K_2\{id_j\}, r, m, id_{j}+\Delta, \mathsf{OCC}}}).$
  \begin{mdframed}
        {\bf $\underline{{ \mathsf{PCt}^{(j,\Delta,3)}_{\mathsf{OPMem}, cpk, K_2\{id_{j}+\Delta\}, r, m, id_{j}||f_j+\Delta, \mathsf{OCC}}}(id || f, u_1, \dots, u_{\cosettcount})}$}
        
        {\bf Hardcoded: ${\mathsf{OPMem}, cpk, {K_2\{id_{j}+\Delta\}}, r, m^b, m^{1-b}, id_{j}||f_j+\Delta, \textcolor{red}{\mathsf{OCC}}}$}
        \begin{enumerate}
        \item \textcolor{red}{If $id||f = id_{j}||f_j+\Delta$, output the output of $\mathsf{OCC}(u_1, \dots, u_{\cosettcount})$ and terminate.}
            \item Run $\mathsf{OPMem}(id || f, u_1, \dots, u_{\cosettcount}, r)$. If it outputs $0$, output $\perp$ and terminate.
            \item {If {$id < id_{j}||f_j+\Delta + 1$}, set $a = f(m^{1-b})$. Otherwise, set $a = f(m)$.}
            \item Output $\mathsf{IBE.Enc}(cpk, id, {a}; F_2(K_2, id))$.
        \end{enumerate}
    \end{mdframed}
    \item Sample $pmsk$ as in $\mathcal{D}_j$.
    \item Output $(\mathsf{OPCt}, r), pmsk$.
\end{enumerate}

\item{$\underline{\mathcal{D}_{j}^{(\Delta, 4)}}$:}  Same as $\mathcal{D}_{j}^{(\Delta, 3)}$ except for the following. Replace the line
    $$ct^* = \mathsf{IBE.Enc}(cpk, id_j||f_j + \Delta, (id_j||f_j + \Delta)(m^b); z^*)
    $$
    with
    $$ 
    ct^* = \mathsf{IBE.Enc}(cpk, id_j||f_j + \Delta, \textcolor{red}{(id_j||f_j + \Delta)(m^{1-b})}; z^*).
    $$
    
        \item{$\underline{\mathcal{D}_{j}^{(\Delta, 5)}}$:} Same as $\mathcal{D}_{j}^{(\Delta, 2)}$ except for the following. Replace the line
    $$ct^* = \mathsf{IBE.Enc}(cpk, id_j||f_j + \Delta, (id_j||f_j + \Delta)(m^b); z^*)
    $$
    with
    $$ 
    ct^* = \mathsf{IBE.Enc}(cpk, id_j||f_j + \Delta, \textcolor{red}{(id_j||f_j + \Delta)(m^{1-b})}; z^*).
    $$

\item{$\underline{\mathcal{D}_{j}^{(\Delta, 6)}}:$} Same as ${\mathcal{D}_{j}^{(\Delta, 1)}}$ except for the following. Replace the line
                $$ct^* = \mathsf{IBE.Enc}(cpk, id_j||f_j + \Delta, (id_j||f_j + \Delta)(m^b); F_2(K_2, id_j||f_j+\Delta))$$
                with
                $$ 
                ct^* = \mathsf{IBE.Enc}(cpk, id_j||f_j + \Delta, \textcolor{red}{(id_j||f_j + \Delta)(m^{1-b})}; F_2(K_2, id_j||f_j+\Delta)).
                $$
\end{itemize}

Now, we show that these distributions \emph{collapse} around $\Delta = 0$ for each $j$. Below, all our indistinguishability claims are for $2^{5\lambda}\cdot2^{8\lambda^{0.3\constmoecoll}}$-time adversaries and we set $\nu(\lambda) = 2^{-5\lambda-\circsize}\cdot2^{-8\lambda^{0.3\constmoecoll}}$.
\begin{claim}\label{claim:fedelta0todelta1}
$\expnot{j}{\Delta}{0}{\ell}\approx^c_{\nu(\lambda)} \expnot{j}{\Delta}{1}{\ell}$
    for all $j \in \{0, 1, \dots, k\}$, $\Delta \in \{0, 1, \dots, id_{j+1}-id_{j} - 1\}$ and $\ell \in [k + 1]$.
\end{claim}
\begin{proof}
    Observe that by punctured key correctness of $F_2$ (\cref{defn:puncprf}), the different obfuscated programs ${ \mathsf{PCt}^{(j, \Delta, 0)}_{\mathsf{OPMem}, cpk, K_2, r, m, id_j + \Delta}}$ and ${ \mathsf{PCt}^{(j, \Delta, 1)}_{\mathsf{OPMem}, cpk, K_2\{id_j+\Delta\}, r, m, id_j + \Delta, ct^*}}$ in these hybrids have the same functionality. The result follows by security of $\io$ and by our choice of parameters.
\end{proof}
\begin{claim}\label{claim:fedelta1todelta2}
$\expnot{j}{\Delta}{1}{\ell} \approx^c_{\nu(\lambda)} \expnot{j}{\Delta}{2}{\ell}$
  for all $j \in \{0, 1, \dots, k\}$, $\Delta \in \{0, 1, \dots, id_{j+1}-id_{j} - 1\}$ and $\ell \in[k + 1]$.
\end{claim}
\begin{proof}
    The result follows by selective puncturing security of $F_2$ (\cref{defn:puncprf}) and our choice of parameters.
\end{proof}
\begin{claim}\label{claim:fedelta2todelta3}
    $\expnot{j}{\Delta}{2}{\ell}\approx^c_{\nu(\lambda)} \expnot{j}{\Delta}{3}{\ell}$
    for all $j \in \{0, 1, \dots, k\}$, $\Delta \in \{0, 1, \dots, id_{j+1}-id_{j} - 1\}$ and $\ell \in [k + 1]$.
\end{claim}
\begin{proof}
       Observe that the obfuscated ciphertext programs $\mathsf{PCt}$ in these hybrids have the same functionality by correctness of $\mathsf{CCObf}$, since a vector $w$ is in $A_i^* + s_i^*$ if and only if $\mathsf{Can}_{A_i^{*}}(w) =\mathsf{Can}_{A_i^{*}}(s_i^*)$ and similarly for $(A^*)_i^\perp + s^{'*}_i$. Then, the claim follows by the security of $\io$.
\end{proof}
\begin{claim}
$\expnot{j}{\Delta}{3}{\ell} \approx^c_{\nu(\lambda)} \expnot{j}{\Delta}{4}{\ell}$
     if
    \begin{itemize}
    \item  $j \in \{1, \dots, k\}$ and $\Delta \in \{1, \dots, id_{j+1}-id_{j} - 1\}$, or
    \item  $j = 0$ and $\Delta \in \{0, 1, \dots, id_{j+1}-id_{j}-1\}$
    \end{itemize}
    and for all $\ell \in [k + 1]$.
\end{claim}
\begin{proof}
We have two cases. First, assume that $(id_j||f_j + \Delta)(m_0) = (id_j||f_j + \Delta)(m_1)$. Then, the result easily follows. 

Otherwise, define the intermediary distributions $\mathcal{D}_{j}^{(\Delta, 3')}, \mathcal{D}_{j}^{(\Delta, 4')}$ as follows.
\paragraph{ $\underline{\mathcal{D}_{j}^{(\Delta, 3')}}$}
\begin{enumerate}
        \item Sample $(\mathsf{OPCt}, r)$ as in $\mathcal{D}_{j}^{(\Delta, 3)}$.
        \item Sample $cmsk' \samp \mathsf{IBE.Punc}(cmsk, id_j||f_j + \Delta)$.
        \item Sample $pmsk \samp \io(\mathsf{PKey}_{cmsk', K_1\{id_{j^*}||f_{j^*}\}})$.

\begin{mdframed}
        {\bf $\underline{\mathsf{PKey}_{cmsk', K_1\{id_{j^*}||f_{j^*}\}}(id || f)}$}
        
        {\bf Hardcoded: $cmsk', K_1\{id_{j^*}||f_{j^*}\}, m^0, m^1$}
        \begin{enumerate}[label=\arabic*.]
          \item Check if $f(m^0) = f(m^1)$. If not, output $\perp$ and terminate.
        \item Compute $ck = \mathsf{IBE.KeyGen}(cmsk', id || f)$.
        \item $(A_i, s_i, s_i')_{i \in [\cosettcount]} = \mathsf{CosetGen}(\cosetgenparam; F_1(K_1\{id_{j^*}||f_{j^*}\}, id || f))$.
    \item Output $(ck, id, f, (A_i, s_i, s_i')_{i \in [\cosettcount]})$.
        \end{enumerate}

    \end{mdframed}
\end{enumerate}
\paragraph{ $\underline{\mathcal{D}_{j}^{(\Delta, 4')}}$}
 Same as $\mathcal{D}_{j}^{(\Delta, 3')}$ except for the following. Replace the line
    $$ct^* = \mathsf{IBE.Enc}(cpk, id_j||f_j + \Delta, (id_j||f_j + \Delta)(m^b); z^*)
    $$
    with
    $$ 
    ct^* = \mathsf{IBE.Enc}(cpk, id_j||f_j + \Delta, \textcolor{red}{(id_j||f_j + \Delta)(m^{1-b})}; z^*).
    $$

First, we claim $\expnot{j}{\Delta}{3}{\ell} \approx^c_{\nu(\lambda)} \expnot{j}{\Delta}{3'}{\ell}$ and  $\expnot{j}{\Delta}{4}{\ell} \approx^c_{\nu(\lambda)} \expnot{j}{\Delta}{4'}{\ell}$. We will only argue the first one and the second one follows similarly.  Observe that by strong punctured key correctness of deterministic $\mathsf{IBE.KeyGen}$, the obfuscated programs $\mathsf{PKey}$ in these hybrids can behave differently only on input $(id_j||f_j + \Delta)(m_0).$. However, since we are considering the case $(id_j||f_j + \Delta)(m_0) \neq (id_j||f_j + \Delta)(m_1)$, the programs will not go past the first line and will have the exact same functionality. Then, the claim follows by security of $\io$.

Now, we claim $\expnot{j}{\Delta}{3'}{\ell} \approx^c_{\nu(\lambda)} \expnot{j}{\Delta}{4'}{\ell}$.
    Observe that in these hybrids, the randomness used to invoke $\mathsf{IBE.Enc}$ to compute $ct^*$ is uniformly and independently sampled. Further, the adversary only has the $\mathsf{IBE}$ keys for the identities $id_1 || f_1,  id_2 || f_2, \dots, id_k || f_k$, all of which are different from the identity $id_j || f_j + \Delta$ under which $ct^*$ is encrypted. Finally, the master secret key $cmsk'$ obtained by the adversary is punctured at $id_j || f_j + \Delta$. Hence, the result follows from the punctured master secret key security of $\mathsf{IBE}$.
\end{proof}
\begin{claim}
    $\expnot{j}{\Delta}{4}{\ell} \approx^c_{\nu(\lambda)} \expnot{j}{\Delta}{5}{\ell}$ for all $j \in \{0, 1, \dots, k\}$, $\Delta \in \{0, 1, \dots, id_{j+1}-id_{j} - 1\}$ and $\ell \in [k + 1]$.
\end{claim}
\begin{proof}
    Essentially the same argument as in \cref{claim:fedelta2todelta3} yields the result.
\end{proof}

\begin{claim}
    $\expnot{j}{\Delta}{5}{\ell} \approx^c_{\nu(\lambda)} \expnot{j}{\Delta}{6}{\ell}$ for all $j \in \{0, 1, \dots, k\}$, $\Delta \in \{0, 1, \dots, id_{j+1}-id_{j} - 1\}$ and $\ell \in [k + 1]$.
\end{claim}
\begin{proof}
    Essentially the same argument as in \cref{claim:fedelta1todelta2} yields the result.
\end{proof}

\begin{claim}
$\expnot{j}{\Delta}{6}{\ell} \approx^c_{\nu(\lambda)} \expnot{j}{\Delta+1}{0}{\ell}$
    for all $j \in \{0, 1, \dots, k\}$, $\Delta \in \{0, 1, \dots, id_{j+1}-id_{j} - 1\}$ and $\ell \in [k + 1]$.
\end{claim}
\begin{proof}
        Essentially the same argument as in \cref{claim:fedelta0todelta1} yields the result.
\end{proof}

\begin{claim}\label{claim:fed0closed1}
For all $\ell \in [k+1]$, we have 
\begin{itemize}
    \item $\expfromddd{\mathcal{D}_0}{\ell} \approx^c_{\nu(\lambda)}  \expfromddd{\mathcal{D}_1}{\ell} $
    \item $\expnot{j}{0}{4}{\ell} \approx^c_{\nu(\lambda)} \expfromddd{\mathcal{D}_{j+1}}{\ell}$ for all $j \in \{0, 1, \dots, k\}$ 
    \item $\expfromddd{\mathcal{D}_j}{\ell} \approx^c_{\nu(\lambda)} \expnot{j}{0}{3}{\ell}$ for all $j \in \{0, 1, \dots, k\}$ 
\end{itemize}
where $\nu(\lambda) = 2^{-5\lambda}\cdot2^{-8\lambda^{0.3\constmoecoll}}$.
\end{claim}
\begin{proof}
It is easy to see that $\mathcal{D}_{0} \approx^c_{\nu(\lambda)} \mathcal{D}_{0}^{(0,0)}$ and  $\mathcal{D}_{k+1} \approx^c_{\nu(\lambda)} \mathcal{D}_{k+1}^{(0, 0)}$ by the security of $\io$. For the former, in particular, observe that the obfuscated punctured master secret key programs have the same functionality since while $K_1$ in $\mathcal{D}_{0}$ is punctured at $id_{j^*} || f_{j^*}$, we have that $f_{j^*}(m_0) \neq f_{j^*}(m_1)$.

Rest follows by a simple calculation using the above results.
\end{proof}

\begin{definition}
    We will write $\mathcal{D}'$ to denote $\mathcal{D}_{j^*}^{(0,3)}$ and $\mathcal{D}''$ to denote $\mathcal{D}_{j^*}^{(0,4)}$ where $j^*$ is as output by $\adve_0'$.
\end{definition}

\begin{claim}\label{claim:fedprimdprim}
Let $\tau$ be the bipartite state output by $\adve'_0$ in $\mathcal{G}$. Let $p'_x, p'_y$ be the outcome of applying $\mathsf{PI}_{x, \mathcal{D}'}\otimes \mathsf{PI}_{y, \mathcal{D}'}$ to $\tau$. Similarly, let $p''_x, p''_y$ be the outcome of applying $\mathsf{PI}_{x, \mathcal{D}''}\otimes \mathsf{PI}_{y, \mathcal{D}''}$ to $\tau$. Then,
\begin{itemize}
    \item $\Pr[p'_x > b_{x,j^*} - \frac{3\gamma}{32k} \wedge p'_y > b_{y,j^*} - \frac{3\gamma}{32k}] \geq 1 - 2^{-2\lambda}\cdot 2^{-4\lambda^{0.3\constmoecoll}}$.
    \item $\Pr[b_{x,j^*} - p''_x > \frac{28\gamma}{32k} \wedge b_{y,j^*} - p''_y  > \frac{28\gamma}{32k}] > \frac{1}{q(\lambda)}$ for some polynomial $q(\cdot)$.
\end{itemize}
\end{claim}
\begin{proof}
    Follows from the same argument as in \cref{claim:dprimdprim}.
\end{proof}
\begin{claim}
    There exist efficient $\adve_1', \adve_2'$ such that  $(\adve'_0, \adve_1', \adve_2')$ wins $\mathcal{G}$ with probability $\frac{1}{2^{0.4\cdot\lambda^{\constmoecoll}}}.$
\end{claim}
\begin{proof}
    Follows from the same argument as in \cref{claim:pkeextract}.
\end{proof}
\begin{claim}\label{claim:fepkereducetome}
There exists efficient $\adve'' = (\adve_0'', \adve_1'', \adve_2'')$ such that
       \begin{equation*}
           \Pr[\moecollpunckey(\lambda, \idlen(\lambda), \adve'') = 1] \geq {2^{-0.4\cdot\lambda^{\constmoecoll}}}.
       \end{equation*}
\end{claim}
\begin{proof}
It is straightforward to reduce $\mathcal{G}$ to $\moecollpunckey$, which is proven secure in the proof of \cref{defn:strmoecoll}. See also \cref{claim:pkereducetome}.
\end{proof}
The above constitutes a contradiction by \cref{defn:strmoecoll}, therefore this completes the security proof.

\section{Signature Scheme with Copy-Protected Keys}\label{sec:unclondigsig}
In this section, we define signature schemes with copy-protected signing keys. Then, we give our construction based on coset states and prove it secure.
\subsection{Definitions}
\begin{definition}[Signature Scheme with Copy-Protected Secret Keys]\label{defn:cpds}
A signature scheme with copy-protected secret keys consists of the following efficient algorithms.

\begin{itemize}
\item $\keygen(1^\lambda)$: Takes in the security parameter, output a classical signing key $sk$ and a classical verification key $vk$.
    \item $\qkeygen(sk)$: Takes as input the classical signing key and outputs a quantum signing key.

    \item $\Sign(\regi{sk}, m)$: Takes in a quantum signing key and a message $m$, outputs a classical signature on $m$.

    \item $\Ver(vk, m, sig)$: Takes in the verification key, a message $m \in \mathcal{M}$ and a claimed signature $sig$ on $m$, outputs $1$ (accept) or $0$ (reject).
\end{itemize}

We require correctness.
\paragraph{Correctness} For all messages $m \in \mathcal{M}$, \begin{equation*}
    \Pr[\Ver({vk}, sig) = 1 : \begin{array}{c}
         sk, vk \samp \mathsf{Setup}(1^\lambda)  \\
         \regi{sk} \samp \qkeygen(sk) \\
         sig \samp \mathsf{Sign}(\regi{sk}, m)
    \end{array}] \geq 1 - \negl(\lambda).
\end{equation*}
\end{definition}
\begin{definition}[Pseudodeterministic Signatures]\label{defn:dspr}
    A signature scheme is said to be \emph{pseudodeterministic} if for any value of $sk, vk$ in the support induced by $\mathsf{KeyGen}$, for any message $m \in \mathcal{M}$, there exists a fixed signature $sig_{sk, vk, m}$ such that
        \begin{equation*}
    \Pr[sig = sig_{sk, vk, m} : \begin{array}{c}
         sk, vk \samp \mathsf{Setup}(1^\lambda)  \\
         \regi{sk} \samp \qkeygen(sk) \\
         sig \samp \mathsf{Sign}(\regi{sk}, m)
    \end{array}] \geq 1 - \negl(\lambda).
\end{equation*}
\end{definition}
As observed by \cite{LLQZ22}, a pseudodeterministic signature scheme, along with \cref{lem:asgoodasnew}, means that we can implement the signing in a way such that the quantum secret key is only negligibly disturbed. Thus, we can reuse the key to sign any polynomial number of times. Our scheme (\cref{sec:sigcons}) will be pseudodeterministic.

We now define anti-piracy security for signature schemes, similar to our PKE definition (\cref{defn:regularpkeantipir}).
\begin{definition}[Anti-Piracy Security for Signature Schemes]\label{defn:unclonsign}
Let $\dss$ be a signature scheme with copy-protected secret keys. Consider the following game between the challenger and an adversary $\adve$.
\paragraph{$\underline{\unclonsiggame(\lambda, \adve)}$}
\begin{enumerate}
    \item The challenger runs $sk, vk \samp \dss.\mathsf{Setup}(1^\lambda)$ and submits $vk$ to the adversary.
    \item For multiple rounds, $\adve$ makes quantum key queries. For each query, the challenger generates a key as $\reg \samp \dss.\qkeygen(sk)$ and submits $\reg$ to the adversary.
    \item $\adve$ outputs a $(k + 1)$-partite register $\regi{adv}$ and freeloader unitaries $\{U_{\ell}\}_{\ell \in [k + 1]}$ where $k$ is the number of queries it made.
    \item The challenger executes the following for each $\ell \in [k + 1]$.
    \begin{enumerate}[label*=\arabic*.]
        \item $m_\ell \samp \mathcal{M}$.
        \item $sig_\ell \samp \qunivcla(U_\ell, \regi{adv}[\ell], m_\ell)$.
        \item Check if $\dss.\Ver(vk, m_\ell, sig_\ell) = 1$.
    \end{enumerate}
    \item The challenger outputs 1 if and only if all the checks pass.
\end{enumerate}
    We say that $\dss$ satisfies anti-piracy security if for any QPT adversary $\adve$,
    \begin{equation*}
    \Pr[\unclonsiggame(\lambda, \adve) = 1] \leq \negl(\lambda).
    \end{equation*}
\end{definition}

\subsection{Construction}\label{sec:sigcons}
In this section, we present our construction. Assume the existence of following primitives where we set $\nu(\lambda) = 2^{-6\lambda}\cdot2^{-8\lambda^{0.3\constmoecoll}}$.
\begin{itemize}
    \item $F$, prefix puncturable extracting PRF (\cref{defn:prepuncpf})  with error $2^{-\lambda - 1}$ for min-entropy $s_2(\lambda) + s_3(\lambda)$, with input length $m(\lambda)$ and output length $n(\lambda)$, 
    \item $\io$, indistinguishability obfuscation scheme that is $\nu(\lambda)$-secure against $2^{5\lambda}\cdot2^{8\lambda^{0.3\constmoecoll}}$-time adversaries,
    \item $\mathsf{IBE}$, identity-based encryption scheme for the identity space $\mathcal{ID} = \zo^\lambda$ (\cref{defn:ibe}) that is $\nu(\lambda)$-secure against $2^{5\lambda}\cdot2^{8\lambda^{0.3\constmoecoll}}$-time adversaries,
    \item $F_1$, puncturable PRF family with input length $\lambda$ and output length same as the size of the randomness used by $\mathsf{CosetGen}$ (\cref{defn:cosetgen}), that is $\nu(\lambda)$-secure against $2^{5\lambda}\cdot2^{8\lambda^{0.3\constmoecoll}}$-time adversaries,
    \item $F_2$, puncturable PRF family with input length $\lambda$ and output length same as the size of the randomness used by $\mathsf{IBE.Enc}$ that is $\nu(\lambda)$-secure against $2^{5\lambda}\cdot2^{8\lambda^{0.3\constmoecoll}}$-time adversaries,\footnote{We also assume that $F_2$ has uniformly random keys (when not punctured), that is, the key generation algorithm $F_2.\mathsf{KeyGen}$ simply samples and outputs a uniformly random string. This is satisfied by the puncturable PRF constructions based on one-way functions we are using.}
    \item $\mathsf{CCObf}$, compute-and-compare obfuscation for $2^{-\lambda^{0.2\cdot {\constmoecoll}}}$-unpredictable distributions that is $2^{-2\lambda - 1}\cdot 2^{-2\lambda^{0.3\constmoecoll}}$-secure against $2^{3\lambda}\cdot 2^{2\lambda^{0.3\constmoecoll}}$-time adversaries,
    \item $F_3$, puncturable statistically injective PRF with error probability $2^{-\lambda}$ with input length $s_3(\lambda)$ and output length $s_2(\lambda)$,
    \item $F_4$, puncturable PRF with input length $s_2(\lambda)$ and output length $s_3(\lambda)$, 
        \item $G_1$, a pseudorandom generator with input length $n(\lambda)$ and output length $n(\lambda)$ plus the key size of the PRF $F_2$,

    \item $G_2$, a pseudorandom generator with input length $s_1(\lambda)/2$ and output length 
    $s_1(\lambda)$,
    \item $G_3$, a pseudorandom generator with input length $\lambda$ and output length 
    $2\cdot\lambda$,

    \item $f$, a subexponentially secure injective one-way function with input space $\zo^{n(\lambda)}$.
\end{itemize}

We also set the parameters from above as follows:
\begin{itemize}
\item $n(\lambda) = \lambda$,
\item $s_1(\lambda) = \cosettcount$,
\item $s_3(\lambda) - s_1(\lambda) - 2\lambda$ to be larger than the size of the obfuscations (of the program $Q$) defined in \cref{defn:hiddentrigger},
\item $s_2(\lambda) \geq 2\cdot s_3(\lambda) + \lambda$,
\item $s_2(\lambda) + s_3(\lambda) \geq n(\lambda) + 2\lambda + 4$,
\item $m(\lambda) = s_1(\lambda) + s_2(\lambda) + s_3(\lambda)$.
\end{itemize}

As in our other schemes, while some of our security assumptions above are exponential with specific exponents, all of these assumptions can be based solely on subexponential hardness for any exponent, since we can always scale the security parameter by a polynomial factor when instantiating the underlying primitives.

Set $L(\lambda) = \lambda$ and therefore $c_L(\lambda) = 24\cdot\lambda^3$ (see \cref{defn:strmoecoll}). We also assume that all obfuscated programs in the construction and in the proof are appropriately padded.

We now give our signature scheme with copy-protected signing keys, for the message space $\mathcal{M} = \zo^{m(\lambda)}$.
\paragraph{$\underline{\dss.\mathsf{Setup}(1^{\lambda})}$}
\begin{enumerate}
    \item Sample PRF keys $K \samp F.\mathsf{KeyGen}(1^\lambda)$ and $K_i \samp F_i.\mathsf{KeyGen}(1^\lambda)$ for $i \in \{1, 3, 4\}$.
    \item Sample $cpk, csmk \samp \mathsf{IBE.Setup}(1^\lambda)$.
    \item Sample $\mathsf{OPVer} \samp \io(\mathsf{PVer})$ where $\mathsf{PVer}$ is the following program.
    
\begin{mdframed}
        {\bf $\underline{\mathsf{PVer}(m, sig)}$}
        
        {\bf Hardcoded: $K, K_3, K_4$}
        \begin{enumerate}[label=\arabic*.]
        \textbf{Hidden Trigger Check}
        \item Parse $m_1 || m_2 || m_3 = m$ with $|m_i| = s_i$.
        \item Compute $m_1' || \mathsf{OQ}' || r'  = F_4(K_4, m_2) \oplus m_3$.
        \item Check if $m_1' = m_1$ and $m_2 = F_3(K_3, m_1' || \mathsf{OQ}' || r' )$. If so, treat $\mathsf{OQ}'$ as a classical circuit, output $\mathsf{OQ}'(\textsf{mode}=\textsf{verify}, sig || 0^{\cosettcount \cdot \lambda})$ and terminate.\\
        \textbf{Normal Mode}
        \item Parse $y || K_2' = G_1(F(K, m))$ with $|y| = n(\lambda)$.
        \item Output $1$ if $f(sig) = f(y)$. Otherwise, output $0$.
        \end{enumerate}

    \end{mdframed}
    
 \item Sample $\mathsf{OPMem} \samp \io(\mathsf{PMem}_{K_1})$, where $\mathsf{PMem}_{K_1}$ is the following program.
    
\begin{mdframed}
        {\bf $\underline{\mathsf{PMem}_{K_1}(id, u_1, \dots, u_{\cosettcount}, x)}$}
        
        {\bf Hardcoded: $K_1$}
        \begin{enumerate}[label=\arabic*.]
            \item $(A_i, s_i, s_i')_{i \in [\cosettcount]} \samp \mathsf{CosetGen}(1^{L(\lambda) + \lambda}; F_1(K_1, id))$.
            \item For each $i \in [\cosettcount]$, check if $u_i \in A_i + s_i$ if $(x)_i = 0$ and check if $u_i \in A^{\perp}_i + s'_i$ if $(x)_i = 1$. If any of the checks fail, output $0$ and terminate.
            \item Output $1$.
        \end{enumerate}

    \end{mdframed}
    
    \item Sample $\mathsf{OPEval} \samp \io(\mathsf{PEval})$, where $\mathsf{PEval}$ is the following program.\footnote{Note that it is also possible to put the coset generation PRF key $K_1$ directly inside $\mathsf{OPEval}$ due to the $\io$ security. However, we elect to use $\mathsf{OPMem}$ to preserve the similarities to our PKE construction.}

    \begin{mdframed}
        {\bf $\underline{\mathsf{PEval}(m, id, u_1, \dots, u_{\cosettcount})}$}
        
        {\bf Hardcoded: $\mathsf{OPMem}, cpk, K, K_3, K_4$}
        \begin{enumerate}[label=\arabic*.]
        \textbf{Hidden Trigger Check}
        \item Parse $m_1 || m_2 || m_3 = m$ with $|m_i| = s_i$.
        \item Compute $m_1'||  \mathsf{OQ}' || r'  = F_4(K_4, m_2) \oplus m_3$.
        \item Check if $m_1' = m_1$ and $m_2 = F_3(K_3, m_1'||\mathsf{OQ}'  || r' )$. If so, treat $\mathsf{OQ}'$ as a classical circuit, output $\mathsf{OQ}'(\textsf{mode}=\textsf{eval}, id, u_1, \dots, u_{\cosettcount})$ and terminate.\\
        \textbf{Normal Mode}
            \item Run $\mathsf{OPMem}(id, u_1, \dots, u_{\cosettcount}, m_1)$. If it outputs $0$, output $\perp$ and terminate.
            \item Parse $y || K_2' = G_1(F(K, m))$ with $|y| = n(\lambda)$.
            \item Output $\mathsf{IBE.Enc}(cpk, id, y; F_2(K_2', id))$.
        \end{enumerate}
    \end{mdframed}
    \item Set $vk = \mathsf{OPVer}$ and $sk = (cmsk, cpk, K_1, \mathsf{OPEval})$.
    \item Output $(vk, sk)$.
\end{enumerate}

\paragraph{$\underline{\mathsf{DS.QKeyGen}(sk)}$}
\begin{enumerate}
    \item Parse $ (cmsk, cpk, K_1, \mathsf{OPEval}) = sk$.

    \item Sample $id \samp \zo^{\lambda}$.
    \item $(A_i, s_i, s_i')_{i \in [\cosettcount]} = \mathsf{CosetGen}(1^{L(\lambda) + \lambda}; F_1(K_1, id))$.
    \item $ck \samp \mathsf{IBE.KeyGen}(cmsk, id)$.
    \item Output $\left(\ket{A_{i, s_i, s'_i}}\right)_{i \in [\cosettcount]}, ck, id,  \mathsf{OPEval}$.
\end{enumerate}

\paragraph{$\underline{\mathsf{DS.Sign}(\regi{key}, m)}$}
\begin{enumerate}
    \item Parse $((\reg_i)_{i \in [\cosettcount]}, ck, id,  \mathsf{OPEval}) = \regi{key}$.
    \item Parse $m_1 || m_2 || m_3 = m$ with $|m_i| = s_i$.
    \item For indices $i \in [\cosettcount]$ such that $(m_0)_i = 1$, apply $H^{\otimes \kappa(L(\lambda) + \lambda)}$ to $\reg_i$.
    \item Run the program $\mathsf{OPEval}$ coherently on $m, id$ and $(\reg_i)_{i \in [\cosettcount]}$.
    \item Measure the output register and denote the outcome by $cct$.
    \item Output $\mathsf{IBE.Dec}(ck, cct)$.
\end{enumerate}

\paragraph{$\underline{\mathsf{DS.Ver}(vk, m, sig)}$}
\begin{enumerate}
    \item Parse $\mathsf{OPVer} = vk$.
    \item Output $\mathsf{OPVer}(m, sig)$.
\end{enumerate}

We claim that the construction is correct and secure.
\begin{theorem}\label{thm:dscorrect}
 $\mathsf{DS}$ satisfies correctness (\cref{defn:cpds}) and psuedodeterminism (\cref{defn:dspr}), and hence reusability.   
\end{theorem}
\begin{theorem}\label{thm:dseufcma}
 $\mathsf{DS}$ satisfies selective\footnote{It can also be made by adaptively secure by complexity leveraging and slightly changing the parameters of the underlying primitives, since we are already assuming subexponentially secure primitives.} message existential unforgeability security.
\end{theorem}
\begin{theorem}\label{thm:dsantipiracy}
 $\mathsf{DS}$ satisfies anti-piracy security (\cref{defn:unclonsign}).
\end{theorem}
When we instantiate the assumed building blocks with known constructions, we get the following corollary.
\begin{corollary}\label{thm:cpdsexists}
Assuming subexponentially secure $i\mathcal{O}$ and subexponentially secure LWE, there exists a signature scheme that satisfies anti-piracy security against unbounded collusion.
\end{corollary}

\subsection{Proof of Correctness and Reusability}\label{sec:digsigcorrect}
It is easy to see that the scheme has psuedodeterministic signatures. If a message does not satisfy the hidden trigger condition, the signing procedure will output the first $n(\lambda)$ bits of $ G_1(F(K, m))$ by the correctness of the $\io$ and $\mathsf{IBE}$ schemes. In particular, we are assuming perfect correctness for the IBE scheme, which is indeed true for our instantiation (\cref{thm:subexpibe}). As a result, the IBE ciphertext $\mathsf{IBE.Enc}(cpk, id, y; F_2(K_2', id))$ output by $\mathsf{OPEval}$ will decrypt to $y$, the first $n(\lambda)$ bits of $G_1(F(K, m))$, for any $m$. If the message $m$ does satisfy the hidden trigger condition, then the output of $\mathsf{OPEval}$ is deterministically determined by $m$ (once $sk, vk$ are fixed). Then, we run $\mathsf{IBE.Dec}$ which we can assume to be deterministic since it has perfect correctness. Therefore, psuedodeterministic signatures property indeed holds for all $m \in \mathcal{M}$.

Now, we move onto correctness. The following discussion follows closely to the proof of correctness given by \cite{LLQZ22} for their scheme. First, it is easy to see that the correctness holds for any $m$ that does not satisfy the hidden trigger condition. We will show that, for any fixed message $m \in \mathcal{M}$, with overwhelming probability over the randomness of the scheme (including setup), the message $m$ will not satisfy the hidden trigger condition, thus proving correctness (\cref{defn:cpds}).

Now, fix a message $m \in \mathcal{M}$. Let $\Hat{F_4}$ denote the truncated version of $F_4$ where we only keep the first $s_1$ bits, which is also a PRF. Similarly, let $\Hat{m_3}$ denote the first $s_1$ bits of $m_3$. Observe that if $m \in \mathcal{M}$ does not satisfy correctness, then by above it satisfies the hidden trigger condition, and therefore satisfies $\Hat{F_4}(K_4, m_2) = m_1 \oplus \Hat{m_3}$ with probability $1/\poly(\lambda)$. Therefore, any sequence of messages that does not satisfy correctness gives us a (non-uniform) adversary for the PRF $\Hat{F_4}$ where we can simply distinguish an output $\Hat{F_4}(K_4, m_2)$ of the PRF from a random value by checking if it is equal to $m_1 \oplus \Hat{m_3}$, which would be satisfied by a random value only with exponentially small probability. This breaks the PRF $\Hat{F_4}$ with probability $1/\poly(\lambda)$, which is a contradiction.

\subsection{Proof of Existential Unforgeability}
We prove the security through a series of hybrids, which are binary random variables denoting the outcome of the forgery game and each one is constructed by modifying the previous one.

\paragraph{\underline{$\hyb_0$}:} The original selective message existential unforgeability security game. 
\paragraph{\underline{$\hyb_1$}:} We define the $\hyb_1$ so that, after the adversary outputs the challenge message $m^*$, the challenger checks if $m^*$ satisfies the hidden trigger condition, and it terminates if so. As argued in \cref{sec:digsigcorrect}, this can only happen with negligible probability. Hence, $\hyb_0 \approx \hyb_1$.

\paragraph{\underline{$\hyb_2$}:}
We first compute $g^* = G_1(F(K, m^*))$, and parse $y^* || K_2^* = g^*$ with $|y^*| = n(\lambda)$ and we set $z^* = f(y^*)$. Finally, we now sample $\mathsf{OPVer}$ as $\mathsf{OPVer} \samp \io(\mathsf{PVer}')$ 
    
\begin{mdframed}
        {\bf $\underline{\mathsf{PVer}'(m, sig)}$}
        
        {\bf Hardcoded: $\textcolor{red}{K\{m^*\}}, K_3, K_4, \textcolor{red}{m^*, z^*}$}
        \begin{enumerate}[label=\arabic*.]
        \textbf{Hidden Trigger Check}
        \item Parse $m_1 || m_2 || m_3 = m$ with $|m_i| = s_i$.
        \item Compute $m_1' || \mathsf{OQ}' || r'  = F_4(K_4, m_2) \oplus m_3$.
        \item Check if $m_1' = m_1$ and $m_2 = F_3(K_3, m_1' || \mathsf{OQ}' || r' )$. If so, treat $\mathsf{OQ}'$ as a classical circuit, output $\mathsf{OQ}'(\textsf{mode}=\textsf{verify}, sig || 0^{\cosettcount \cdot \lambda})$ and terminate.\\
        \textbf{Normal Mode}
        \textcolor{red}{\item If $m = m^*$: Output $1$ if $f(sig) = z^*$ and output $0$ if  $f(sig) \neq z^*$, and terminate.}
        \item Parse $y || K_2' = G_1(F(K\{m^*\}, m^*))$ with $|y| = n(\lambda)$.
        \item Output $1$ if $f(sig) = f(y)$. Otherwise, output $0$.
        \end{enumerate}

    \end{mdframed}
By the punctured key correctness of $F$, the functionality of $\mathsf{PVer}$ did not change. Thus, $\hyb_1 \approx \hyb_2$ by the security of $\io$.
 
 \paragraph{\underline{$\hyb_3$}:} We now sample $g^*$ uniformly at random. Observe that $\mathsf{PVer}$ only has the punctured key $K\{m^*\}$, and the signing oracle only answers queries for messages $m \neq m^*$, which can also be simulated using $K\{m^*\}$ rather than $K$. Thus, we have $\hyb_2 \approx \hyb_3$ by the punctured key security of $F$ and the security of the PRG $G_1$.

We claim that $\Pr[\hyb_3 = 1] \leq \negl(\lambda)$. Observe that, since $m^*$ is not a hidden trigger input, $\hyb_3 = 1$ occurs if and only if the forged signature $sig^*$ output by the adversary satisfies $f(sig^*) = z^*$, where $z^* = f(y^*)$ and $y^*$ is uniformly at random. Therefore, $\Pr[\hyb_3 = 1] \leq \negl(\lambda)$ follows by the security of the one-way function $f$.

\subsection{Proof of Anti-Piracy Security}\label{sec:digsigcpproof}
In this section, we prove \cref{thm:dsantipiracy}.

First, we show that hidden trigger inputs are indistinguishable from uniformly random challenge strings, even when the adversary gets a (obfuscated) program that allows it to generate its own hidden trigger inputs.

\begin{definition}[Hidden Trigger Inputs]\label{defn:hiddentrigger}
    Let $\mathsf{GenTrigger}_{K, K_3, K_4, \mathsf{OPMem}, cpk}$ be the following program, where the hardcoded values are as in the signature scheme construction (\cref{sec:sigcons}). The input format to the program will be clear from context.
    \begin{mdframed}
        {\bf $\underline{\mathsf{GenTrigger}_{K, K_3, K_4, \mathsf{OPMem},cpk}(r_1, r_2, r_3)}$}
        
        {\bf Hardcoded: $K, K_3, K_4, \mathsf{OPMem}, cpk$}
        \begin{enumerate}
        
            \item Parse $x_1 || x_2 || x_3 = G_2(r_1)$ with $|x_i| = s_i$.
            \item Parse $y || K_2' = G_1(F(K, x))$ with $|y| = n(\lambda)$.
            \item $\mathsf{OQ} \samp \io(Q_{cpk, \mathsf{OPMem}, x_1, K'_2, y}; r_3)$.
            \item $x_2' = F_3(K_3, x_1|| \mathsf{OQ} || G_3(r_2)  )$.
            \item $x_3' = F_4(K_4, x_2') \oplus (x_1|| \mathsf{OQ} ||G_3(r_2))$.
            \item Output $x_1 || x_2' || x_3'$.
        \end{enumerate}
        
    \end{mdframed} 
    The circuit $Q_{cpk, \mathsf{OPMem}, x_1, K'_2, y}$ used above is the following. Note that it contains hardcoded values that are computed during the execution of $\mathsf{GenTrigger}$.
    \begin{mdframed}
        {\bf $\underline{Q_{cpk, \mathsf{OPMem}, x_1, K'_2, y}(\mathsf{mode}, w)}$}
        
        {\bf Hardcoded: $cpk, \mathsf{OPMem}, x_1, K'_2, y$}
        \begin{enumerate}
            \item If $\mathsf{mode} = \mathsf{eval}$:
           \begin{enumerate}[label=\arabic*.]
           \item Parse $id, u_1, \dots, u_{\cosettcount} = w$.
           \item Run $\mathsf{OPMem}(id, u_1, \dots, u_{\cosettcount},  \textcolor{blue}{x_1})$. If it outputs $0$, output $\perp$ and terminate.
            \item Output $\mathsf{IBE.Enc}(cpk, id, \textcolor{blue}{y}; F_2( \textcolor{blue}{K'_2}, id))$.
           \end{enumerate}

             \item If $\mathsf{mode} = \mathsf{check}$, parse $sig || 0^{\cosettcount \cdot \lambda} = w$ and check if $f(sig) = f(\textcolor{blue}{y})$. If so, output $1$, and otherwise output $0$.        \end{enumerate}
        
    \end{mdframed} 
\end{definition}
\begin{lemma}\label{lem:hiddentrigger}
    Let $\dss$ be the signature scheme from \cref{sec:sigcons} and let $\granlen$ denote the length of the randomness used by $\io$ to obfuscate $Q$ in \cref{defn:hiddentrigger}. Consider the following experiment, parameterized by $\ell(\lambda)$.
\paragraph{$\underline{\hiddentriggame(\lambda, \adve, \ell(\lambda), b)}$}
\begin{enumerate}
    \item The challenger runs $sk, vk \samp \dss.\mathsf{Setup}(1^\lambda)$ and submits $vk$ to the adversary.
    \item For multiple rounds, $\adve$ makes quantum key queries. For each query, the challenger generates a key as $\reg \samp \dss.\qkeygen(sk)$ and submits $\reg$ to the adversary.
    \item The adversary outputs a register $\regi{adv}$.
    \item Sample $\mathsf{OGenTrigger} \samp \io(\mathsf{GenTrigger})$.
    \item For $i = 1$ to $\ell$:
    \begin{enumerate}[label=\arabic*.]
        \item Sample $r^i_1  \samp \zo^{s_1(\lambda)/2}$. 
        \item Sample $r^i_2 \samp \zo^{\lambda}$.
        \item Sample $r^i_3  \samp \zo^{\granlen}$. 
        \item Set $z^{0,i} = \mathsf{OGenTrigger}(r_1^i, r_2^i, r_3^i)$.
        \item Sample $z^{1,i} \samp \zo^{m(\lambda)}$.
        
    \end{enumerate}
    
    \item Output $((z^{b,i})_{i \in [\ell]}, \mathsf{OGenTrigger}, \regi{adv})$.
\end{enumerate}
    Then, for any polynomial $\ell(\lambda)$,
    \begin{equation*}
   \hiddentriggame(\lambda, \adve, \ell(\lambda), 0) \approx_c \hiddentriggame(\lambda, \adve, \ell(\lambda), 1).
    \end{equation*}
\end{lemma}
For now, assume \cref{lem:hiddentrigger} and we prove it at the end of the section. We will prove anti-piracy security through a series of hybrids. Define $\hyb_0$ to be the original game $\unclonsiggame(\lambda, \adve)$.

\paragraph{$\underline{\hyb_1}$}: We now sample $m^\ell$ for all $\ell \in [k + 1]$ as hidden triggers 
(\cref{defn:hiddentrigger}): $m^\ell \samp \mathsf{OGenTrigger}(r_1^\ell, r_2^\ell, r_3^\ell)$ where $r_1^\ell, r_2^\ell, r_3^\ell$ are sampled uniformly at random. We get $\hyb_0 \approx \hyb_1$ by \cref{lem:hiddentrigger}. Crucially note that at the end of the game in $\hyb_1$, the forged signatures output by the adversary are verified using $\mathsf{OPVer}$, which is in adversary's view. Hence, the adversary can indeed simulate $\hyb_0, \hyb_1$ in the reduction to \cref{lem:hiddentrigger}.

\paragraph{$\underline{\hyb_2}$}: We now sample $m^\ell$ for $\ell \in [k+1]$ as follows.

    \begin{enumerate}
        
        \item Sample $r_1^\ell \samp \zo^{s_1(\lambda)/2}$. 
        \item Sample $r_2^\ell \samp \zo^{\lambda}$.
        \item Sample $r_3^\ell  \samp \zo^{\granlen}$. 
        \item Let $x^\ell = G_2(r_1^\ell)$.
        \item Parse $x_1^\ell || x_2^\ell || x_3^\ell = x^\ell$ with $|x^\ell_i| = s_i$.
        \item Parse $y^\ell || K_2^\ell = G_1(F(K, x^\ell))$ with $|y^\ell| = n(\lambda)$.
        \item $\mathsf{OQ}^\ell \samp \io(Q_{cpk, \mathsf{OPMem}, x_1^\ell, K^\ell_2, y^\ell}; r^\ell_3)$.
        \item $x_2^{'\ell} = F_3(K_3, x_1^\ell|| \mathsf{OQ}^\ell || G_3(r_2^\ell))$.
        \item $x_3^{'\ell} = F_4(K_4, x_2^{'\ell}) \oplus (x_1^\ell|| \mathsf{OQ}^\ell || G_3(r_2^\ell))$.
        \item Set $m^\ell =  x_1^\ell || x_2^{'\ell} || x_3^{'\ell}$.
    \end{enumerate}

Note that this is only a syntactic change, we only unwrapped the sampling of the inputs $m^\ell$. Hence, $\hyb_1 \equiv \hyb_2$.

\paragraph{$\underline{\hyb_3}$}:
We now sample $x^\ell$ for all $\ell \in [k+1]$  uniformly at random. We get  $\hyb_2 \approx \hyb_3$ by the security of the PRG $G_2$.

\paragraph{$\underline{\hyb_4}$}: We now sample all $y^\ell$ and $K_{2}^{\ell}$ uniformly at random. First, observe that with overwhelming probability, the values $x^\ell$ for all $\ell \in [k + 1]$ will be unique. Then, observe that the inputs $x^\ell$ to $F$ have each min-entropy $s_2 + s_3$ given the view of the adversary, since only the first part $x^\ell_1$ is used during the hidden trigger generation and the rest is discarded. Therefore, by our choice of parameters for the extracting PRF $F$ and by the security of PRG $G_1$, we have $\hyb_3 \approx \hyb_4$.

\paragraph{$\underline{\hyb_5}$}: At the end of the game, instead of running $\mathsf{OPVer}$ on the forged signatures $sig^\ell$ output by the adversary, we compare them to $y^\ell$. Observe that since $m^\ell$ are hidden trigger inputs, running $\mathsf{OPVer}$ on $m^\ell$ makes it enter the hidden trigger mode, and it executes $\mathsf{OQ^\ell}$ with $\mathsf{mode}=\mathsf{verify}$, which checks if $f(sig^\ell) = f(y^\ell)$. Since $f$ is injective, this can happen if and only if $y^\ell = sig^\ell$. Hence, $\hyb_4 \equiv \hyb_5$.

Finally, observe that the hidden trigger inputs generated above in $\hyb_4$ 
are special encodings of $(\mathsf{OQ}^\ell, x_1^\ell)$, which are (\emph{almost}) the same as ciphertexts of our PKE scheme (\cref{sec:pkecons}) encrypting the random messages $y^\ell$; and we are comparing the adversary's outputs to $y^\ell$. Therefore, the security follows by the random message anti-piracy security (see \cref{sec:pkeproof}) of our scheme and we have $\Pr[\hyb_4 = 1] \leq \negl(\lambda)$.

The only difference between the programs $\mathsf{OQ}^\ell$ here and the ciphertext programs of our PKE is that $\mathsf{OQ}^\ell$ also has a mode that acts as a point on function on the encrypted message: the adversary can query it on some message $m'$ and can check if it equals the encrypted message. See \cref{appn:dstopke} for a more detailed discussion on how the security follows from our PKE security proof even the ciphertext programs are modified as such.

\subsubsection{Proof of \cref{lem:hiddentrigger}}\label{sec:hiddentriggerproof}
In this section, we prove \cref{lem:hiddentrigger}. We will only prove the case $\ell = 1$ - the general case for any polynomial $\ell(\lambda)$ follows easily by the hybrid lemma since $\mathsf{OGenTrigger}$ is given to the adversary.

We follow an indirect approach to prove the security by first showing the security of another game, $\hiddentriggame'$. That is, we will first show $\Pr[(\lambda, \adve) = 1] \leq 1/2 + \negl(\lambda)$.

\paragraph{$\underline{\hiddentriggame'(\lambda, \adve)}$}
\begin{enumerate}
    \item The challenger runs $sk, vk \samp \dss.\mathsf{Setup}(1^\lambda)$ and submits $vk$ to the adversary.
    \item For multiple rounds, $\adve$ makes quantum key queries. For each query, the challenger generates a key as $\reg \samp \dss.\qkeygen(sk)$ and submits $\reg$ to the adversary.
    \item The adversary outputs a register $\regi{adv}$.
    \item Sample $\mathsf{OGenTrigger} \samp \io(\mathsf{GenTrigger})$.
       
        \item Sample $r_1^* \samp \zo^{s_1(\lambda)/2}$. 
        \item Sample $r_2^* \samp \zo^{\lambda}$.
        \item Sample $r_3^*  \samp \zo^{\granlen}$. 
        \item Let $x^* = G_2(r_1^*)$.
        \item Parse $x_1^* || x_2^* || x_3^* = x^*$ with $|x^*_i| = s_i$.
        \item Parse $y^* || K_2^* = G_1(F(K, x^*))$ with $|y^*| = n(\lambda)$.
        \item $\mathsf{OQ}^* \samp \io(Q_{cpk, \mathsf{OPMem}, x_1^*, K^*_2, y^*}; r^*_3)$.
        \item $\Tilde{r} =  G_3(r_2^*)$.
        \item $x_2^{'*} = F_3(K_3, x_1^*|| \mathsf{OQ}^* ||  \Tilde{r})$.
        \item $x_3^{'*} = F_4(K_4, x_2^{'*}) \oplus (x_1^*|| \mathsf{OQ}^* || \Tilde{r})$.
        \item Set $z^{0} =  x_1^* || x_2^{'*} || x_3^{'*}$.
        \item Set $z^{1} =  x_1^* || x_2^* || x_3^*$.
    \item Sample $b \samp \zo$.
    \item Submit $z^{b}, \mathsf{OGenTrigger}$ to the adversary $\adve$ and the adversary outputs a bit $b'$.
    \item Output $1$ if and only if  $b' =  b$.
    
\end{enumerate}

First, note that the security of $\hiddentriggame'$ implies indistinguishability between $z^0$ and $z^1$ (given the adversary's state and $\mathsf{OGenTrigger}$) by the usual elementary argument. Since $r^*_1$ is random and $G_2$ is a secure PRG, we also obtain indistinguishability between $z^1$ and a random $x$. Combining these two gives us \cref{lem:hiddentrigger} (for $\ell = 1$) as desired.

We now prove security of $\hiddentriggame'$ through a series of hybrids, each of which is constructed by modifying the previous one. In the hybrids below, whenever we have a check in an obfuscated program where a variable is compared to multiple different values, or a PRF key is punctured at various values, we assume that these are (implicitly) coded in lexicographical order of these values\footnote{To re-emphasize, we order them according to the hardcoded value, not the variable name that denotes the hardcoded value.} to have symmetry which will be needed in the last hybrid.

\paragraph{$\underline{\hyb_0}$:} $\hiddentriggame'(\lambda, \adve)$.

\paragraph{$\underline{\hyb_1}$:} We now sample $x^*$ uniformly at random from $\zo^{m(\lambda)}$ instead of setting $x^* = G_2(r_1^*)$. By the security of $G_2$, we get $\hyb_0 \approx \hyb_1$.

\paragraph{$\underline{\hyb_2}$:} We now sample $\mathsf{OGenTrigger}$ as $\mathsf{OGenTrigger} \samp \io(\mathsf{GenTrigger}')$ where $K\{x_1^*||\cdot\}$ is the prefix punctured key sampled as $K\{x_1^*||\cdot\} \samp F.\mathsf{Puncture}(K, x_1^*)$.

\begin{mdframed}
        {\bf $\underline{\mathsf{GenTrigger}'_{\textcolor{red}{K\{x_1^*||\cdot\}}, K_3, K_4, \mathsf{OPMem},cpk}(r_1, r_2, r_3)}$}
        
        {\bf Hardcoded: $\textcolor{red}{K\{x_1^*||\cdot\}}, K_3, K_4, \mathsf{OPMem}, cpk$}
        \begin{enumerate}
        
            \item Parse $x_1 || x_2 || x_3 = G_2(r_1)$ with $|x_i| = s_i$.
            \item Parse $y || K_2' = F(\textcolor{red}{K\{x_1^*||\cdot\}}, x)$ with $|y| = n(\lambda)$.
            \item $\mathsf{OQ} \samp \io(Q_{cpk, \mathsf{OPMem}, x_1, K'_2, y}; r_3)$.
            \item $x_2' = F_3(K_3, x_1|| \mathsf{OQ} || G_3(r_2)  )$.
            \item $x_3' = F_4(K_4, x_2') \oplus (x_1|| \mathsf{OQ} ||G_3(r_2))$.
            \item Output $x_1 || x_2' || x_3'$.
        \end{enumerate}
        
    \end{mdframed} 

    We claim $\hyb_0 \approx \hyb_1$. Observe that, by the prefix punctured key correctness of $F$, the functionality of $\mathsf{GenTrigger}$ can only possible change if the input $x$ is such that the first $s_1(\lambda)$ bits of $G_2(r_1)$ equals $x_1^*$. However, the image set of $G_2$ has size at most $2^{s_1(\lambda)/2}$, and hence the same is true for the truncated version where we only keep the first $s_1(\lambda)$ bits of the outputs. However, $x^1*$ is sampled uniformly at random from $\zo^{s_1(\lambda)}$. Hence, the probability that the condition above occurs is at most $2^{-s_1(\lambda)/2}$ where $s_1(\lambda)$ is polynomial in $\lambda$. Hence, with overwhelming probability, the functionality of $\mathsf{GenTrigger}$ does not change, and $\hyb_1 \approx \hyb_2$ follows by the security of $\io$.

\paragraph{$\underline{\hyb_3}$:} We now sample $\mathsf{OPVer}$ as $\mathsf{OPVer} \samp \io(\mathsf{PVer}')$ and $\mathsf{OPEval}$ as $\mathsf{OPEval} \samp \io(\mathsf{PEval}')$.

\begin{mdframed}
        {\bf $\underline{\mathsf{PVer}'(m, sig)}$}
        
        {\bf Hardcoded: $\textcolor{red}{K\{z^0, z^1\}}, K_3, K_4, K_5$}
        \begin{enumerate}[label=\arabic*.]
        \textbf{Hidden Trigger Check}
        \item \textcolor{red}{If $m = z^0$ or $m = z^1$, output $\mathsf{OQ}^*(\textsf{mode}=\textsf{verify}, sig || 0^{\cosettcount \cdot \lambda})$ and terminate.}
        \item Parse $m_1 || m_2 || m_3 = m$ with $|m_i| = s_i$.
        \item Compute $m_1' || \mathsf{OQ}' || r'  = F_4(K_4, m_2) \oplus m_3$.
        \item Check if $m_1' = m_1$ and $m_2 = F_3(K_3, m_1' || \mathsf{OQ}' || r' )$. If so, treat $\mathsf{OQ}'$ as a classical circuit, output $\mathsf{OQ}'(\textsf{mode}=\textsf{verify}, sig || 0^{\cosettcount \cdot \lambda})$ and terminate.\\
        \textbf{Normal Mode}
        \item Parse $y || K_2' = G_1(F(\textcolor{red}{K\{z^0, z^1\}}, m))$ with $|y| = n(\lambda)$.
        \item Output $1$ if $f(sig) = f(y)$. Otherwise, output $0$.
        \end{enumerate}

    \end{mdframed}
    
    \begin{mdframed}
        {\bf $\underline{\mathsf{PEval}'(m, id, u_1, \dots, u_{\cosettcount})}$}
        
        {\bf Hardcoded: $\mathsf{OPMem}, cpk, \textcolor{red}{K\{z^0, z^1\}}, K_3, K_4, K_5$}
        \begin{enumerate}[label=\arabic*.]
        \textbf{Hidden Trigger Check}
        \textcolor{red}{\item If $m = z^0$ or $m = z^1$, output $\mathsf{OQ}^*(\textsf{mode}=\textsf{eval}, id, u_1, \dots, u_{\cosettcount})$ and terminate.}
        \item Parse $m_1 || m_2 || m_3 = m$ with $|m_i| = s_i$.
    
        \item Compute $m_1'||  \mathsf{OQ}' || r'  = F_4(K_4, m_2) \oplus m_3$.
        \item Check if $m_1' = m_1$ and $m_2 = F_3(K_3, m_1'||\mathsf{OQ}'  || r' )$. If so, treat $\mathsf{OQ}'$ as a classical circuit, output $\mathsf{OQ}'(\textsf{mode}=\textsf{eval}, id, u_1, \dots, u_{\cosettcount})$ and terminate.\\
        \textbf{Normal Mode}
            \item Run $\mathsf{OPMem}(id, u_1, \dots, u_{\cosettcount}, m_1)$. If it outputs $0$, output $\perp$ and terminate.
            \item Parse $y || K_2' = G_1(F(\textcolor{red}{K\{z^0, z^1\}}, m))$ with $|y| = n(\lambda)$.
            \item Output $\mathsf{IBE.Enc}(cpk, id, y; F_2(K_2', id))$.
        \end{enumerate}
    \end{mdframed}

By the punctured key correctness of $F$, the functionality of $\mathsf{PVer}$ can only possibly change if $m = z^0$ or $m = z^1$. If $m = z^0$, then the original program would execute the hidden trigger mode and hence would execute $\mathsf{OQ}^*$ due to the way we prepared $z^0$, hence the functionality would be the same. If $m = z^1$, then the original program would execute the normal mode, since with overwhelming probability a random input would not satisfy the hidden trigger condition\footnote{As discussed in \cref{sec:digsigcorrect}, any sequence of inputs that satisfy the hidden trigger condition with non-negligible probability gives us a way of breaking the security of PRF $F_4$. Since an adversary can easily sample random inputs, we could break the security of $F_4$ if random inputs satisfied the hidden trigger condition with non-negligible probaiblity.}. In that case, the original program's output would again be the same as the output of $\mathsf{OQ}^*$ due to the way $\mathsf{OQ}^*$ is prepared. Hence,  $\hyb_2 \approx \hyb_3$ follows by the security of $\io$.

\paragraph{$\underline{\hyb_4}$:} We now sample $y^*$ and $K_2^*$ uniformly at random instead of computing them as $y^* || K^*_2 = G_1(F(K, x^*))$. Since the adversary only has the punctured keys $K\{z^0, z^1\}$ and $K\{x_1^*||\cdot\}$ where $z^1 = x^*$ and $x^*_1$ is a prefix of $x^*$, we have that $F(K, x^*)$ is pseudorandom given the adversary's view, by the punctured key security of $F$. Then, we invoke the security of $G_1$ and conclude $\hyb_3 \approx \hyb_4$.

\paragraph{$\underline{\hyb_5}$:} We now sample $\mathsf{OPVer}$ as $\mathsf{OPVer} \samp \io(\mathsf{PVer}'')$ and $\mathsf{OPEval}$ as $\mathsf{OPEval} \samp \io(\mathsf{PEval}'')$.

\begin{mdframed}
        {\bf $\underline{\mathsf{PVer}''(m, sig)}$}
        
        {\bf Hardcoded: ${K\{z^0, z^1\}}, K_3, \textcolor{red}{K_4\{x_2^*,x_2^{'*}\}}, K_5$}
        \begin{enumerate}[label=\arabic*.]
        \textbf{Hidden Trigger Check}
        {If $m = z^0$ or $m = z^1$, output $\mathsf{OQ}^*(\textsf{mode}=\textsf{verify}, sig || 0^{\cosettcount \cdot \lambda})$ and terminate.}
        
        \item Parse $m_1 || m_2 || m_3 = m$ with $|m_i| = s_i$.
        \textcolor{red}{\item If $m_2 = x_2^{'*}$ or  $m_2 = x_2^{*}$, jump to Normal Mode.}
        \item Compute $m_1' || \mathsf{OQ}' || r'  = F_4(\textcolor{red}{K_4\{x_2^*,x_2^{'*}\}}, m_2) \oplus m_3$.
        \item Check if $m_1' = m_1$ and $m_2 = F_3(K_3, m_1' || \mathsf{OQ}' || r' )$. If so, treat $\mathsf{OQ}'$ as a classical circuit, output $\mathsf{OQ}'(\textsf{mode}=\textsf{verify}, sig || 0^{\cosettcount \cdot \lambda})$ and terminate.\\
        \textbf{Normal Mode}
        \item Parse $y || K_2' = G_1(F({K\{z^0, z^1\}}, m))$ with $|y| = n(\lambda)$.
        \item Output $1$ if $f(sig) = f(y)$. Otherwise, output $0$.
        \end{enumerate}

    \end{mdframed}
    
    \begin{mdframed}
        {\bf $\underline{\mathsf{PEval}''(m, id, u_1, \dots, u_{\cosettcount})}$}
        
        {\bf Hardcoded: $\mathsf{OPMem}, cpk, {K\{z^0, z^1\}}, K_3, \textcolor{red}{K_4\{x_2^*,x_2^{'*}\}}, K_5$}
        \begin{enumerate}[label=\arabic*.]
        \textbf{Hidden Trigger Check}
        {\item If $m = z^0$ or $m = z^1$, output $\mathsf{OQ}^*(\textsf{mode}=\textsf{eval}, id, u_1, \dots, u_{\cosettcount})$ and terminate.}
        \item Parse $m_1 || m_2 || m_3 = m$ with $|m_i| = s_i$.
    \textcolor{red}{\item If $m_2 = x_2^{'*}$ or  $m_2 = x_2^{*}$, jump to Normal Mode.}
        \item Compute $m_1'||  \mathsf{OQ}' || r'  = F_4(\textcolor{red}{K_4\{x_2^*,x_2^{'*}\}}, m_2) \oplus m_3$.
        \item Check if $m_1' = m_1$ and $m_2 = F_3(K_3, m_1'||\mathsf{OQ}'  || r' )$. If so, treat $\mathsf{OQ}'$ as a classical circuit, output $\mathsf{OQ}'(\textsf{mode}=\textsf{eval}, id, u_1, \dots, u_{\cosettcount})$ and terminate.\\
        \textbf{Normal Mode}
            \item Run $\mathsf{OPMem}(id, u_1, \dots, u_{\cosettcount}, m_1)$. If it outputs $0$, output $\perp$ and terminate.
            \item Parse $y || K_2' = G_1(F({K\{z^0, z^1\}}, m))$ with $|y| = n(\lambda)$.
            \item Output $\mathsf{IBE.Enc}(cpk, id, y; F_2(K_2', id))$.
        \end{enumerate}
    \end{mdframed}

    We will first consider the modified versions of $\mathsf{PEval}''$ and $\mathsf{PVer}''$ where the PRF key $K_4$ is not punctured at $x^*_2,x_2^{'*}$ and argue that these versions have the same functionality as $\mathsf{PEval}'$ and $\mathsf{PVer}'$. Then, it is easy to see that puncturing $K_4$  at $x^*_2,x_2^{'*}$  does not change their functionalities by the punctured key correctness of $F_4$.
    
We argue that the newly added skip conditions $m_2 = x_2^{'*}$ or  $m_2 = x_2^{*}$ does not change the functionalities of $\mathsf{PVer}, \mathsf{PEval}$ except with negligible probability. First, let us consider the check  $m_2 = x_2^{*}$. This new check can only possibly change the functionalities of the programs if we also have $m_2 = F_3(K_3, m_1'||\mathsf{OQ}'  || r' )$ along with $m_2 = x_2^{*}$, since then the new program jumps to the normal mode while the old program would possibly execute the hidden trigger mode. However, note that $x_2^{*}$ is sampled independently uniformly at random from $\zo^{s_2(\lambda)}$, whereas for any fixing of $K_3$, the image set of $F_3(K_3, \cdot)$ has size $2^{s_3(\lambda)} \leq 2^{(s_2(\lambda)-\lambda)/2}$. Hence, even the probability that $x_2^{*}$ is in the image of $F_3(K_3, \cdot)$ is at most $2^{-(s_2(\lambda)-\lambda)/2}$.

Now, let us consider the check  $m_2 = x_2^{'*}$. As above, this new check can only possibly change the functionalities of the programs if we also have $m_2 = F_3(K_3, m_1'||\mathsf{OQ}'  || r' )$ along with $m_2 = x_2^{'*}$ and $m_1 = m_1'$ where $m_1'||  \mathsf{OQ}' || r'  = F_4(K_4, m_2) \oplus m_3$. This implies $F_3(K_3, x^*_1 || \mathsf{OQ}^* || \Tilde{r}) = F_3(K_3, F_4(K_4, m_2) \oplus m_3)$. Assume $F_3(K_3, \cdot)$ is an injective function, which is indeed true with probability $1-2^{-\lambda}$ since $F_3$ is a statistically injective PRF. Then, we get
\begin{equation*}
    m_3 = (x^*_1 || \mathsf{OQ}^* || \Tilde{r}) \oplus F_4(K_4, m_2) = (x^*_1 || \mathsf{OQ}^* || \Tilde{r}) \oplus F_4(K_4, x_2^{'*}) = x_3^{'*}.
\end{equation*}
Further, $m_1'||  \mathsf{OQ}' || r'  = F_4(K_4, m_2) \oplus m_3$ along with $m_3 = x_3^{'*}$ and $x_2^{'*} = m_2$ implies $m_1' = m_1 = x_1^*$. In summary, we get $m_1 = x_1^*, m_2 = x_2^{'*}, m_3 = x^{'*}$, meaning that $m = z^{0}$. However, at the beginning of the program we check if $m = z^{0}$ and jump to normal mode if so. Hence, if $m = z^{0}$, the program would not even come to this newly added check $m_2 = x_2^{'*}$.

By above, we get that except with negligible probability, the functionalities of the obfuscated programs did not change. Thus, $\hyb_4 \approx \hyb_5$ by the security of $\io$.

\paragraph{$\underline{\hyb_6}$:} 
We now sample $\Tilde{r}$ uniformly at random from $\zo^{2\lambda}$. By the security of $G_3$, we get  $\hyb_5 \approx \hyb_6$.
\paragraph{$\underline{\hyb_7}$:} We now sample $\mathsf{OGenTrigger}$ as $\mathsf{OGenTrigger} \samp \io(\mathsf{GenTrigger}'')$.

\begin{mdframed}
        {\bf $\underline{\mathsf{GenTrigger}''_{{K\{x_1^*||\cdot\}}, K_3, K_4\{x_2^{'*}\}, \mathsf{OPMem},cpk}(r_1, r_2, r_3)}$}
        
        {\bf Hardcoded: ${K\{x_1^*||\cdot\}}, K_3,  \textcolor{red}{K_4\{x_2^*,x_2^{'*}\}}, \mathsf{OPMem}, cpk$}
        \begin{enumerate}
        
            \item Parse $x_1 || x_2 || x_3 = G_2(r_1)$ with $|x_i| = s_i$.
            \item Parse $y || K_2' = F({K\{x_1^*||\cdot\}}, x)$ with $|y| = n(\lambda)$.
            \item $\mathsf{OQ} \samp \io(Q_{cpk, \mathsf{OPMem}, x_1, K'_2, y}; r_3)$.
            \item $x_2' = F_3({K_3}, x_1|| \mathsf{OQ} || G_3(r_2)  )$.
            \item $x_3' = F_4( \textcolor{red}{K_4\{x_2^*,x_2^{'*}\}}, x_2') \oplus (x_1|| \mathsf{OQ} ||G_3(r_2))$.
            \item Output $x_1 || x_2' || x_3'$.
        \end{enumerate}
        
    \end{mdframed} 
By the punctured key correctness of $F_4$, the functionality can only possibly change if $x_2' = x_2^{'*}$ or if $x_2' = x_2^{*}$. 
Assume that $F_3(K_3, \cdot)$ is an injective function, which is indeed true except with probability $2^{-\lambda}$ since $F_3$ is a statistically injective PRF. Then, $x_2' = x_2^{'*}$ implies $(x_1^*|| \mathsf{OQ}^* || \Tilde{r}) =  x_1|| \mathsf{OQ} || G_3(r_2)$, and in particular, $\Tilde{r} = G_3(r_2)$. However, observe that the image set of $G_3$ has size at most $2^\lambda$, whereas $\Tilde{r}$ is sampled uniformly at random from $\zo^{2\lambda}$. Thus, with overwhelming probability, $\Tilde{r}$  will be outside the image set of $G_3$, and hence we will not have $x_2' = x_2^{'*}$ for any input to $\mathsf{GenTrigger}$. For the case of $x_2' = x_2^{*}$, observe that $x^*_2$ is independently sampled uniformly at random and is not used anywhere else. Hence,  $x_2' = x_2^{*}$ can only occur with exponentially small probability. Thus, $\hyb_6 \approx \hyb_7$ by the security of $\io$.

\paragraph{$\underline{\hyb_8}$:} 
We now sample $\mathsf{OGenTrigger}$ as $\mathsf{OGenTrigger} \samp \io(\mathsf{GenTrigger}''')$.

\begin{mdframed}
        {\bf $\underline{\mathsf{GenTrigger}'''_{{K\{x_1^*||\cdot\}}, K_3\{x_1^* || \mathsf{OQ}^* || \Tilde{r}\}, K_4\{x_2^{'*}\}, \mathsf{OPMem},cpk}(r_1, r_2, r_3)}$}
        
        {\bf Hardcoded: ${K\{x_1^*||\cdot\}}, \textcolor{red}{K_3\{x_1^* || \mathsf{OQ}^* || \Tilde{r}\}}, K_4\{x_2^*,x_2^{'*}\}, \mathsf{OPMem}, cpk$}
        \begin{enumerate}
        
            \item Parse $x_1 || x_2 || x_3 = G_2(r_1)$ with $|x_i| = s_i$.
            \item Parse $y || K_2' = F({K\{x_1^*||\cdot\}}, x)$ with $|y| = n(\lambda)$.
            \item $\mathsf{OQ} \samp \io(Q_{cpk, \mathsf{OPMem}, x_1, K'_2, y}; r_3)$.
            \item $x_2' = F_3(\textcolor{red}{K_3\{x_1^* || \mathsf{OQ}^* || \Tilde{r}\}}, x_1|| \mathsf{OQ} || G_3(r_2)  )$.
            \item $x_3' = F_4(K_4\{x_2^*,x_2^{'*}\}, x_2') \oplus (x_1|| \mathsf{OQ} ||G_3(r_2))$.
            \item Output $x_1 || x_2' || x_3'$.
        \end{enumerate}
        
    \end{mdframed} 
As above, with overwhelming probability, $\Tilde{r}$  will be outside the image set of $G_3$, and $x_1|| \mathsf{OQ} || G_3(r_2)  = x_1^* || \mathsf{OQ}^* || \Tilde{r}$ will be not be satisfied by any input to $\mathsf{GenTrigger}$. Hence, the functionality of $\mathsf{GenTrigger}$ does not change with overwhelming probability by the punctured key correctness of $F_3$, and thus $\hyb_7 \approx \hyb_8$ follows by the security of $\io$.

\paragraph{$\underline{\hyb_9}$:} We now sample $\mathsf{OPVer}$ as $\mathsf{OPVer} \samp \io(\mathsf{PVer}''')$ and $\mathsf{OPEval}$ as $\mathsf{OPEval} \samp \io(\mathsf{PEval}''')$.

\begin{mdframed}
        {\bf $\underline{\mathsf{PVer}'''(m, sig)}$}
        
        {\bf Hardcoded: ${K\{z^0, z^1\}}, K_3, {K_4\{x_2^{'*}, x_2^*\}}, K_5$}
        \begin{enumerate}[label=\arabic*.]
        \textbf{Hidden Trigger Check}
        {If $m = z^0$ or $m = z^1$, output $\mathsf{OQ}^*(\textsf{mode}=\textsf{verify}, sig || 0^{\cosettcount \cdot \lambda})$ and terminate.}
        
        \item Parse $m_1 || m_2 || m_3 = m$ with $|m_i| = s_i$.
\item If $m_2 = x_2^{'*}$ or  $m_2 = x_2^{*}$, jump to Normal Mode.
        \item Compute $m_1' || \mathsf{OQ}' || r'  = F_4(K_4\{x_2^*,x_2^{'*}\}, m_2) \oplus m_3$.
        \item \textcolor{red}{Check if $ m_1'||\mathsf{OQ}'  || r' = x_1^* || \mathsf{OQ}^* || \Tilde{r}$. If so, jump to Normal Mode.}
        \item Check if $m_1' = m_1$ and $m_2 = F_3(K_3, m_1' || \mathsf{OQ}' || r' )$. If so, treat $\mathsf{OQ}'$ as a classical circuit, output $\mathsf{OQ}'(\textsf{mode}=\textsf{verify}, sig || 0^{\cosettcount \cdot \lambda})$ and terminate.\\
        \textbf{Normal Mode}
        \item Parse $y || K_2' = G_1(F({K\{z^0, z^1\}}, m))$ with $|y| = n(\lambda)$.
        \item Output $1$ if $f(sig) = f(y)$. Otherwise, output $0$.
        \end{enumerate}

    \end{mdframed}
    
    \begin{mdframed}
        {\bf $\underline{\mathsf{PEval}'''(m, id, u_1, \dots, u_{\cosettcount})}$}
        
        {\bf Hardcoded: $\mathsf{OPMem}, cpk, {K\{z^0, z^1\}}, K_3, {K_4\{x_2^{'*},x_2^*\}}, K_5$}
        \begin{enumerate}[label=\arabic*.]
        \textbf{Hidden Trigger Check}
        {\item If $m = z^0$ or $m = z^1$, output $\mathsf{OQ}^*(\textsf{mode}=\textsf{eval}, id, u_1, \dots, u_{\cosettcount})$ and terminate.}
        \item Parse $m_1 || m_2 || m_3 = m$ with $|m_i| = s_i$.
    \item If $m_2 = x_2^{'*}$ or  $m_2 = x_2^{*}$, jump to Normal Mode.
        \item Compute $m_1'||  \mathsf{OQ}' || r'  = F_4(K_4\{x_2^*,x_2^{'*}\}, m_2) \oplus m_3$.
        \item \textcolor{red}{Check if $ m_1'||\mathsf{OQ}'  || r' = x_1^* || \mathsf{OQ}^* || \Tilde{r}$. If so, jump to Normal Mode.}
        \item Check if $m_1' = m_1$ and $m_2 = F_3(K_3, m_1'||\mathsf{OQ}'  || r' )$. If so, treat $\mathsf{OQ}'$ as a classical circuit, output $\mathsf{OQ}'(\textsf{mode}=\textsf{eval}, id, u_1, \dots, u_{\cosettcount})$ and terminate.\\
        \textbf{Normal Mode}
            \item Run $\mathsf{OPMem}(id, u_1, \dots, u_{\cosettcount}, m_1)$. If it outputs $0$, output $\perp$ and terminate.
            \item Parse $y || K_2' = G_1(F({K\{z^0, z^1\}}, m))$ with $|y| = n(\lambda)$.
            \item Output $\mathsf{IBE.Enc}(cpk, id, y; F_2(K_2', id))$.
        \end{enumerate}
    \end{mdframed}

    We claim that the newly added skip condition does not change the functionality of the programs. First, note that the functionality of these programs can only possibly change if $m_1'||\mathsf{OQ}'  || r' = x_1^* || \mathsf{OQ}^* || \Tilde{r}$,  $m_1' = m_1$ and $m_2 = F_3(K_3, m_1'||\mathsf{OQ}'  || r' )$, since then the old programs execute the hidden trigger mode whereas the new program jumps to the normal mode. However, if these conditions are satisfied, then so is $m_2 = x^{'*}_2$ due to the way we prepare $x^{'*}_2$. Observe that if $m_2 = x^{'*}_2$, then the programs will not come to the newly added check since at the beginning we check if $m_2 = x^{'*}_2$ and jump to normal mode if so. Hence,  the functionality of the programs did not change and we have $\hyb_8 \approx \hyb_9$ by the security of $\io$.

\paragraph{$\underline{\hyb_{10}}$:} We now sample $\mathsf{OPVer}$ as $\mathsf{OPVer} \samp \io(\mathsf{PVer}'''')$ and $\mathsf{OPEval}$ as $\mathsf{OPEval} \samp \io(\mathsf{PEval}'''')$.

\begin{mdframed}
        {\bf $\underline{\mathsf{PVer}''''(m, sig)}$}
        
        {\bf Hardcoded: ${K\{z^0, z^1\}}, \textcolor{red}{K_3\{x_1^* || \mathsf{OQ}^* || \Tilde{r}\}}, {K_4\{x_2^{'*}, x_2^*\}}, K_5$}
        \begin{enumerate}[label=\arabic*.]
        \textbf{Hidden Trigger Check}
        {If $m = z^0$ or $m = z^1$, output $\mathsf{OQ}^*(\textsf{mode}=\textsf{verify}, sig || 0^{\cosettcount \cdot \lambda})$ and terminate.}
        
        \item Parse $m_1 || m_2 || m_3 = m$ with $|m_i| = s_i$.
\item If $m_2 = x_2^{'*}$ or  $m_2 = x_2^{*}$, jump to Normal Mode.
        \item Compute $m_1' || \mathsf{OQ}' || r'  = F_4(K_4\{x_2^*,x_2^{'*}\}, m_2) \oplus m_3$.
        \item {Check if $ m_1'||\mathsf{OQ}'  || r' = x_1^* || \mathsf{OQ}^* || \Tilde{r}$. If so, jump to Normal Mode.}
        \item Check if $m_1' = m_1$ and $m_2 = F_3(\textcolor{red}{K_3\{x_1^* || \mathsf{OQ}^* || \Tilde{r}\}}, m_1' || \mathsf{OQ}' || r' )$. If so, treat $\mathsf{OQ}'$ as a classical circuit, output $\mathsf{OQ}'(\textsf{mode}=\textsf{verify}, sig || 0^{\cosettcount \cdot \lambda})$ and terminate.\\
        \textbf{Normal Mode}
        \item Parse $y || K_2' = G_1(F({K\{z^0, z^1\}}, m))$ with $|y| = n(\lambda)$.
        \item Output $1$ if $f(sig) = f(y)$. Otherwise, output $0$.
        \end{enumerate}

    \end{mdframed}
    
    \begin{mdframed}
        {\bf $\underline{\mathsf{PEval}''''(m, id, u_1, \dots, u_{\cosettcount})}$}
        
        {\bf Hardcoded: $\mathsf{OPMem}, cpk, {K\{z^0, z^1\}}, \textcolor{red}{K_3\{x_1^* || \mathsf{OQ}^* || \Tilde{r}\}}, {K_4\{x_2^{'*},x_2^*\}}, K_5$}
        \begin{enumerate}[label=\arabic*.]
        \textbf{Hidden Trigger Check}
        {\item If $m = z^0$ or $m = z^1$, output $\mathsf{OQ}^*(\textsf{mode}=\textsf{eval}, id, u_1, \dots, u_{\cosettcount})$ and terminate.}
        \item Parse $m_1 || m_2 || m_3 = m$ with $|m_i| = s_i$.
    \item If $m_2 = x_2^{'*}$ or  $m_2 = x_2^{*}$, jump to Normal Mode.
        \item Compute $m_1'||  \mathsf{OQ}' || r'  = F_4(K_4\{x_2^*,x_2^{'*}\}, m_2) \oplus m_3$.
        \item {Check if $ m_1'||\mathsf{OQ}'  || r' = x_1^* || \mathsf{OQ}^* || \Tilde{r}$. If so, jump to Normal Mode.}
        \item Check if $m_1' = m_1$ and $m_2 = F_3(\textcolor{red}{K_3\{x_1^* || \mathsf{OQ}^* || \Tilde{r}\}}, m_1'||\mathsf{OQ}'  || r' )$. If so, treat $\mathsf{OQ}'$ as a classical circuit, output $\mathsf{OQ}'(\textsf{mode}=\textsf{eval}, id, u_1, \dots, u_{\cosettcount})$ and terminate.\\
        \textbf{Normal Mode}
            \item Run $\mathsf{OPMem}(id, u_1, \dots, u_{\cosettcount}, m_1)$. If it outputs $0$, output $\perp$ and terminate.
            \item Parse $y || K_2' = G_1(F({K\{z^0, z^1\}}, m))$ with $|y| = n(\lambda)$.
            \item Output $\mathsf{IBE.Enc}(cpk, id, y; F_2(K_2', id))$.
        \end{enumerate}
    \end{mdframed}
The functionality of these programs stay the same by the punctured key correctness of $K_3$. Thus, $\hyb_9 \approx \hyb_{10}$ follows by the security of $\io$.

    \paragraph{$\underline{\hyb_{11}}$:} We now sample $x_2^{'*}$ uniformly at random. $\hyb_{10} \approx \hyb_{11}$ follows by the punctured key security of $F_3$.

    \paragraph{$\underline{\hyb_{12}}$:}We now sample $x_3^{'*}$ uniformly at random. $\hyb_{11} \approx \hyb_{12}$ follows by the punctured key security of $F_4$.

Finally, we claim that $z^0$ and $z^1$ are symmetric in $\hyb_{12}$. First note that they are both sampled uniformly at random. Further, any check we have in the obfuscated programs compare variables to both $z^0$ and $z^1$, which are coded in lexicographical order to keep the symmetry. Finally, any PRF key that is punctured is either punctured at both $z^0, z^1$, or both $x_2^{'*}, x_2^{*}$, or (prefix) punctured at $x^*_1$, or punctured at $x_1^* || \mathsf{OQ}^* || \Tilde{r}$ where $\mathsf{OQ}^*$ depends only on $y^*, K_2^*$ and $x_1^*$. Crucially note that $y^*, K_2^*, \Tilde{r}$ are sampled independently uniformly at random, and  $x^*_1$ is the first $s_1$ bits of both $z^0$ and $z^1$. Hence, $z^0$ and $z^1$ are indeed symmetric in $\hyb_{12}$, and thus we have $\Pr[\hyb_{12} = 1] \leq 1/2$, concluding the proof.
    
\section{Pseudorandom Function Family with Copy-Protected Keys}\label{sec:cpprf}
In this section, we define PRF schemes with copy-protected keys. Then, we give our construction based on coset states and prove it secure.
\subsection{Definitions}
\begin{definition}[PRF Scheme with Copy-Protected Secret Keys]\label{defn:cpprf}
A PRF scheme with copy-protected secret keys consists of the following efficient algorithms.

\begin{itemize}
\item $\keygen(1^\lambda)$: Takes in the security parameter, output a classical key $K$.
    \item $\qkeygen(K)$: Takes as input the classical key and outputs a quantum key.

    \item $\Eval(\reg_{K}, m)$: Takes in a quantum key and an input $x$, outputs a classical value.
\end{itemize}

We require correctness.
\paragraph{Correctness} For all inputs $x$, \begin{equation*}
    \Pr[val = F(K, x) : \begin{array}{c}
         K \samp \mathsf{Setup}(1^\lambda)  \\
         \reg_{K} \samp \qkeygen(K) \\
         val \samp \mathsf{Eval}(\reg_K, x)
    \end{array}] \geq 1 - \negl(\lambda).
\end{equation*}
\end{definition}
As observed by \cite{CLLZ21}, correctness, along with \cref{lem:asgoodasnew}, means that we can implement the evaluation in a way such that the quantum  key is only negligibly disturbed. Thus, we can reuse the key to evaluate the PRF any polynomial number of times.

\begin{definition}[Anti-Piracy Security for PRF Schemes]\label{defn:unclonprf}
Consider the following game between the challenger and an adversary $\adve$.
\paragraph{$\underline{\unclonprfgame(\lambda, \adve)}$}
\begin{enumerate}
    \item The challenger runs $K \samp \mathsf{KeyGen}(1^\lambda)$.
    \item For multiple rounds, $\adve$ makes quantum key queries. For each query, the challenger generates a key as $\reg \samp \qkeygen(K)$ and submits $\reg$ to the adversary.
    \item $\adve$ outputs a $(k + 1)$-partite register $\regi{adv}$ and freeloader unitaries $\{U_{\ell}\}_{\ell \in [k + 1]}$ where $k$ is the number of queries it made.
    \item The challenger executes the following for each $\ell \in [k + 1]$.
    \begin{enumerate}[label*=\arabic*.]
\item Sample $b_\ell \samp \zo$.
        \item $x^\ell \samp \zo^{m(\lambda)}$.
        \item Set $ch^{0, \ell} = F(K, x^\ell)$ and sample $ch^{1, \ell} \samp \zo^{n(\lambda)}$.
        \item $b'_\ell \samp \qunivcla(U_\ell, \regi{adv}[\ell], x^\ell, ch^{b_\ell, \ell})$.
        \item Check if $b'_\ell = b_\ell$.
    \end{enumerate}
    \item The challenger outputs 1 if and only if all the checks pass.
\end{enumerate}
    We say that the PRF scheme satisfies anti-piracy security if for any QPT adversary $\adve$,
    \begin{equation*}
    \Pr[\unclonprfgame(\lambda, \adve) = 1] \leq \frac{1}{2} + \negl(\lambda).
    \end{equation*}
\end{definition}

\subsection{Construction}\label{sec:cpprfcons}
In this section, we present our construction copy-protecting a particular PRF family $F$. Our construction is the same as our copy-protected signature construction (\cref{sec:unclondigsig}), with the verification key removed. We give it in full for completeness.

Assume the existence of following primitives where we set $\nu(\lambda) = 2^{-6\lambda}\cdot2^{-8\lambda^{0.3\constmoecoll}}$.
\begin{itemize}
    \item $F$, prefix puncturable extracting PRF (\cref{defn:prepuncpf})  with error $2^{-\lambda - 1}$ for min-entropy $s_2(\lambda) + s_3(\lambda)$, with input length $m(\lambda)$ and output length $n(\lambda)$, 
    \item $\io$, indistinguishability obfuscation scheme that is $\nu(\lambda)$-secure against $2^{5\lambda}\cdot2^{8\lambda^{0.3\constmoecoll}}$-time adversaries,
    \item $\mathsf{IBE}$, identity-based encryption scheme for the identity space $\mathcal{ID} = \zo^\lambda$ (\cref{defn:ibe}) that is $\nu(\lambda)$-secure against $2^{5\lambda}\cdot2^{8\lambda^{0.3\constmoecoll}}$-time adversaries,
    \item $F_1$, puncturable PRF family with input length $\lambda$ and output length same as the size of the randomness used by $\mathsf{CosetGen}$ (\cref{defn:cosetgen}), that is $\nu(\lambda)$-secure against $2^{5\lambda}\cdot2^{8\lambda^{0.3\constmoecoll}}$-time adversaries,
    \item $F_2$, puncturable PRF family with input length $\lambda$ and output length same as the size of the randomness used by $\mathsf{IBE.Enc}$ that is $\nu(\lambda)$-secure against $2^{5\lambda}\cdot2^{8\lambda^{0.3\constmoecoll}}$-time adversaries,\footnote{We also assume that $F_2$ has uniformly random keys (when not punctured), that is, the key generation algorithm $F_2.\mathsf{KeyGen}$ simply samples and outputs a uniformly random string. This is satisfied by the puncturable PRF constructions based on one-way functions we are using.}
    \item $\mathsf{CCObf}$, compute-and-compare obfuscation for $2^{-\lambda^{0.2\cdot {\constmoecoll}}}$-unpredictable distributions that is $2^{-2\lambda - 1}\cdot 2^{-2\lambda^{0.3\constmoecoll}}$-secure against $2^{3\lambda}\cdot 2^{2\lambda^{0.3\constmoecoll}}$-time adversaries,
    \item $F_3$, puncturable statistically injective PRF with error probability $2^{-\lambda}$ with input length $s_3(\lambda)$ and output length $s_2(\lambda)$,
    \item $F_4$, puncturable PRF with input length $s_2(\lambda)$ and output length $s_3(\lambda)$, 
        \item $G_1$, a pseudorandom generator with input length $n(\lambda)$ and output length $n(\lambda)$ plus the key size of the PRF $F_2$,

    \item $G_2$, a pseudorandom generator with input length $s_1(\lambda)/2$ and output length 
    $s_1(\lambda)$,
    \item $G_3$, a pseudorandom generator with input length $\lambda$ and output length 
    $2\cdot\lambda$,

\item $f$, a subexponentially secure injective one-way function with input space $\zo^{n(\lambda)}$.
\end{itemize}

We also set the parameters from above as follows:
\begin{itemize}
\item $n(\lambda) = \lambda$,
\item $s_1(\lambda) = \cosettcount$,
\item $s_3(\lambda) - s_1(\lambda) - 2\lambda$ to be large enough to contain obfuscations of the program $Q$ defined in \cref{defn:prfhiddentrigger},
\item $s_2(\lambda) \geq 2\cdot s_3(\lambda) + \lambda$,
\item $s_2(\lambda) + s_3(\lambda) \geq n(\lambda) + 2\lambda + 4$,
\item $m(\lambda) = s_1(\lambda) + s_2(\lambda) + s_3(\lambda)$.
\end{itemize}

While some of our security assumptions above are exponential with specific exponents, these assumptions can be based solely on subexponential hardness for any exponent, since we can always scale the security parameter by a polynomial factor when instantiating the underlying primitives.

Set $L(\lambda) = \lambda$ and therefore $c_L(\lambda) = 24\cdot\lambda^3$ (see \cref{defn:strmoecoll}). We also assume that all obfuscated programs in the construction and in the proof are appropriately padded.

We now give our PRF scheme with copy-protected keys, for the input space $\zo^{m(\lambda)}$.
\paragraph{$\underline{\mathsf{KeyGen}(1^{\lambda})}$}
\begin{enumerate}
    \item Sample PRF keys $K_0 \samp F.\mathsf{KeyGen}(1^\lambda)$ and $K_i \samp F_i.\mathsf{KeyGen}(1^\lambda)$ for $i \in \{1, 3, 4\}$.
    \item Sample $cpk, csmk \samp \mathsf{IBE.Setup}(1^\lambda)$.
  
 \item Sample $\mathsf{OPMem} \samp \io(\mathsf{PMem}_{K_1})$, where $\mathsf{PMem}_{K_1}$ is the following program.
    
\begin{mdframed}
        {\bf $\underline{\mathsf{PMem}_{K_1}(id, u_1, \dots, u_{\cosettcount}, x)}$}
        
        {\bf Hardcoded: $K_1$}
        \begin{enumerate}[label=\arabic*.]
            \item $(A_i, s_i, s_i')_{i \in [\cosettcount]} \samp \mathsf{CosetGen}(1^{L(\lambda) + \lambda}; F_1(K_1, id))$.
            \item For each $i \in [\cosettcount]$, check if $u_i \in A_i + s_i$ if $(x)_i = 0$ and check if $u_i \in A^{\perp}_i + s'_i$ if $(x)_i = 1$. If any of the checks fail, output $0$ and terminate.
            \item Output $1$.
        \end{enumerate}

    \end{mdframed}
    
    \item Sample $\mathsf{OPEval} \samp \io(\mathsf{PEval})$, where $\mathsf{PEval}$ is the following program.\footnote{Note that it is also possible to put the coset generation PRF key $K_1$ directly inside $\mathsf{OPEval}$ due to the $\io$ security. However, we elect to use $\mathsf{OPMem}$ to preserve the similarities to our PKE construction.}

    \begin{mdframed}
        {\bf $\underline{\mathsf{PEval}(x, id, u_1, \dots, u_{\cosettcount})}$}
        
        {\bf Hardcoded: $\mathsf{OPMem}, cpk, K_0, K_3, K_4$}
        \begin{enumerate}[label=\arabic*.]
        \textbf{Hidden Trigger Check}
        \item Parse $x_1 || x_2 || x_3 = x$ with $|x_i| = s_i$.
        \item Compute $x_1'||  \mathsf{OQ}' || r'  = F_4(K_4, x_2) \oplus m_3$.
        \item Check if $x_1' = x_1$ and $x_2 = F_3(K_3, x_1'||\mathsf{OQ}'  || r' )$. If so, treat $\mathsf{OQ}'$ as a classical circuit, output $\mathsf{OQ}'(id, u_1, \dots, u_{\cosettcount})$ and terminate.\\
        \textbf{Normal Mode}
            \item Run $\mathsf{OPMem}(id, u_1, \dots, u_{\cosettcount}, x_1)$. If it outputs $0$, output $\perp$ and terminate.
            \item Parse $y || K_2' = G_1(F(K_0, x))$ with $|y| = n(\lambda)$.
            \item Output $\mathsf{IBE.Enc}(cpk, id, y; F_2(K_2', id))$.
        \end{enumerate}
    \end{mdframed}
    \item Set $K = (cmsk, cpk, K_1, \mathsf{OPEval})$.
    \item Output $K$.
\end{enumerate}

\paragraph{$\underline{\mathsf{QKeyGen}(K)}$}
\begin{enumerate}
    \item Parse $ (cmsk, cpk, K_1, \mathsf{OPEval}) = K$.

    \item Sample $id \samp \zo^{\lambda}$.
    \item $(A_i, s_i, s_i')_{i \in [\cosettcount]} = \mathsf{CosetGen}(1^{L(\lambda) + \lambda}; F_1(K_1, id))$.
    \item $ck \samp \mathsf{IBE.KeyGen}(cmsk, id)$.
    \item Output $\left(\ket{A_{i, s_i, s'_i}}\right)_{i \in [\cosettcount]}, ck, id,  \mathsf{OPEval}$.
\end{enumerate}

\paragraph{$\underline{\mathsf{Eval}(\regi{key}, x)}$}
\begin{enumerate}
    \item Parse $((\reg_i)_{i \in [\cosettcount]}, ck, id,  \mathsf{OPEval}) = \regi{key}$.
    \item Parse $x_1 || x_2 || x_3 = x$ with $|x_i| = s_i$.
    \item For indices $i \in [\cosettcount]$ such that $(x_0)_i = 1$, apply $H^{\otimes \kappa(L(\lambda) + \lambda)}$ to $\reg_i$.
    \item Run the program $\mathsf{OPEval}$ coherently on $x, id$ and $(\reg_i)_{i \in [\cosettcount]}$.
    \item Measure the output register and denote the outcome by $cct$.
    \item Output $\mathsf{IBE.Dec}(ck, cct)$.
\end{enumerate}

We claim that the construction is correct and secure.
\begin{theorem}\label{thm:prfcorrect}
The PRF scheme satisfies correctness and hence reusability.   
\end{theorem}
\begin{proof}
    Since our construction is the same as our signature scheme, follows from \cref{sec:digsigcorrect}.
\end{proof}

\begin{theorem}
The PRF scheme satisfies PRF security.
\end{theorem}
\begin{proof}
  Our PRF family is a truncation of $G_1(F(K_0, \cdot))$ where $G_1$ is a PRG and $F$ is a PRF. Therefore, it is easy to see that the resulting function family also satisfies the PRF security game.
\end{proof}

\begin{theorem}\label{thm:prfantipiracy}
 The PRF scheme satisfies anti-piracy security.
\end{theorem}
\begin{proof}
    The proof closely follows the anti-piracy security proof our signature scheme, the major difference being that in the PRF case we have a CPA-style game where the adversary is trying to guess a challenge bit and we require that it wins with probability at most $1/2+\negl(\lambda)$, whereas in the signature game we required negligible winning probability. See \cref{sec:prfcpproof} for the full proof.
\end{proof}

When we instantiate the assumed building blocks with known constructions, we get the following corollary.
\begin{corollary}\label{thm:cpprfexists}
Assuming subexponentially secure $i\mathcal{O}$ and subexponentially secure LWE, there exists a PRF scheme that satisfies anti-piracy security against unbounded collusion.
\end{corollary}

\section{Copy-Protection for All Unlearnable Functionalities}\label{sec:allunlearn}
In this section, we first reproduce the generalized copy-protection definitions from \cite{ALLZZCONF20}, and then we show how to copy-protect any unlearnable functionality with respect to a classical oracle.
\subsection{Definitions}
We now reproduce the relevant definitions from \cite{ALLZZCONF20}.
\begin{definition}[Testing an Oracle-Aided Quantum Program]
    Let $\mathcal{F}$ be a  family of functions  with input length $m(\lambda)$ and output length $n(\lambda)$. Fix some program $f$ from this family, an oracle-aided unitary $U$ and some value $st$ of a classical state maintained by the challenger (which will be defined later). Let $\mathcal{D}$ be an efficient challenge input distribution (over $\zo^{m(\lambda)}$), and let $\mathcal{O}$ be a quantumly accessible classical oracle that can depend on $st$.
    Consider the following mixture $\mathcal{P}$ of binary \emph{projective} \emph{measurements}, induced by $\mathcal{D}$ and $f, U, st$, applied on a state $\rho$.

    \begin{enumerate}
        \item Sample $r \samp \mathcal{R}$.
        \item Run $x \samp \mathcal{D}^{st}(1^\lambda; r)$.
        \item Execute $U$ on $(\rho, ct)$, let $y'$ be the output.
        \item Output 1 if $y' = f(x)$. Otherwise, output $0$.
    \end{enumerate}

    Observe that we can efficiently execute the above measurement\footnote{More formally, we are actually talking about the measurement where $r, b$ are fixed} for arbitrary given superpositions of $r$ and $b$ values. Therefore, by \cref{sec:piti}, there exists both exact and approximated projective and threshold implementations for $\mathcal{P}$. We write $\mathsf{PI}^{\mathcal{O}}_{\mathcal{D}}$ and $\mathsf{API}^{\mathcal{O},\eps, \delta}_{\mathcal{D}}$ to denote the projective implementation and approximate projective implementation of $\mathcal{P}$, respectively. Similarly, let $\mathsf{TI}^{\mathcal{O}}_{\mathcal{D}, \eta}$ and $\mathsf{ATI}^{\mathcal{O},\eps, \delta}_{\mathcal{D}, \eta}$ denote the threshold and efficient approximate threshold implementations of $\mathcal{P}$ for a threshold value $\eta$. 
    
    While the fixed values $f, U, st$ are omitted from the notation, they will be clear from the context.
\end{definition}

\begin{definition}[$\gamma$-Quantum Unlearnability \cite{ALLZZCONF20}]
    Let $\mathcal{F}$ be a  family of functions  with input length $m(\lambda)$, and let $\mathcal{D}$ be an input distribution over $\zo^{m(\lambda)}$. Consider the following game between the challenger and an adversary $\adve$.
    \paragraph{$\underline{\unlearngame(\lambda, \gamma(\lambda), \adve)}$}
\begin{enumerate}
    \item The challenger samples a function $f$ from $\mathcal{F}$.
    \item $\adve$ gets oracle access to $f$.
    \item $\adve$ outputs a quantum register $\regi{adv}$ and a unitary $U$.
    \item The challenger applies $\mathsf{TI}_{\mathcal{D}, \gamma}$ to $\regi{adv}$, outputs the measurement result.
\end{enumerate}
We say that $(\mathcal{F}, \mathcal{D})$ is $\gamma$-unlearnable if for any QPT polynomial $\adve$, \begin{equation*}
    \Pr[\unlearngame(\lambda, \gamma(\lambda), \adve) = 1] \leq \negl(\lambda).
\end{equation*}
\end{definition}

\begin{definition}[Quantum Copy-Protection Scheme \cite{ALLZZCONF20}]\label{defn:gencp}
Let $\mathcal{F}$ be a  family of functions  with input length $m(\lambda)$ and output length $n(\lambda)$. A copy-protection scheme for $\mathcal{F}$ consists of the following efficient algorithms.
\begin{itemize}
    \item $\mathsf{Setup}(1^\lambda)$: Takes as input  a security parameter and outputs a classical secret key $sk$,

    \item $\mathsf{QGen}(sk, f)$: Takes in the secret key and a function $f \in \mathcal{F}$, outputs a copy-protected program as a quantum state.
    
    \item $\mathsf{Eval}(\regi{key}, x)$: Takes in a copy-protected program and an input, outputs a value from $\zo^{n(\lambda)}$.
\end{itemize}
    We require correctness.\paragraph{Correctness} For all functions $f \in \mathcal{F}$ and inputs $x \in \zo^{m(\lambda)}$, \begin{equation*}
    \Pr[\mathsf{Eval}(\regi{key}, x) = f(x) : \begin{array}{c}
        sk \samp \mathsf{Setup}(1^\lambda) \\
         \reg_f \samp \mathsf{QGen}(sk, f)
    \end{array}] = 1.
\end{equation*}
\end{definition}

\begin{definition}[$\gamma$-Anti-Piracy Security \cite{ALLZZCONF20}]
    Let $\mathcal{F}$ be a  family of functions  with input length $m(\lambda)$, and let $\mathcal{D}$ be an input distribution over $\zo^{m(\lambda)}$. Consider a copy-protection scheme (\cref{defn:gencp}) for $\mathcal{F}$ and the following game between the challenger and an adversary $\adve$.
    \paragraph{$\underline{\antipirgame(\lambda, \gamma(\lambda), \adve)}$}
\begin{enumerate}
\item The challenger samples a copy-protection key $sk \samp \mathsf{Setup}(1^\lambda)$.
    \item The challenger samples a function $f$ from $\mathcal{F}$.
    \item For multiple rounds, $\adve$ makes quantum key queries. For each query, the challenger generates a key as $\reg \samp \mathsf{QGen}(sk, f)$ and submits $\reg$ to the adversary.
    \item $\adve$ outputs a $(k+1)$-partite  quantum register $\regi{adv}$ and freeloader unitaries $\{U\ell\}_{\ell \in [k + 1]}$ where $k$ is the number of key queries it made.
    \item The challenger applies the test
    \begin{equation*}
    \bigotimes_{\ell \in [k+1]} \mathsf{TI}_{\mathcal{D}, \gamma}
    \end{equation*}
    to $\regi{adv}$ and outputs $1$ if and only if the measurement result is all $1$.
\end{enumerate}
We say that the copy-protection scheme satisfies $\gamma$-anti-piracy if for any QPT polynomial $\adve$, \begin{equation*}
    \Pr[\antipirgame(\lambda, \gamma(\lambda), \adve) = 1] \leq \negl(\lambda).
\end{equation*}
\end{definition}
\subsection{Construction}
In this section, we present our copy-protection construction for a family of functions $\mathcal{F}$ with input length $m(\lambda)$ and output length $n(\lambda)$. Assume the existence of following primitive.
\begin{itemize}
    \item $F_1$,  PRF family with input length $\lambda$ and output length same as the size of the randomness used by $\mathsf{CosetGen}$ (\cref{defn:cosetgen}) that is $2^{-2\lambda}$-secure against QPT adversaries.
\end{itemize}
While we assume exponential security of the above primitive for specific exponents, this assumption can be based solely on subexponential hardness for any exponent, since we can always scale the security parameter by a polynomial factor when instantiating the underlying primitives.

We now give our construction.
\paragraph{$\underline{\mathsf{Setup}(1^{\lambda})}$}
\begin{enumerate}
    \item Sample PRF key $K_1 \samp F_1.\mathsf{KeyGen}(1^\lambda)$.
    \item Output $K_1$.
\end{enumerate}

\paragraph{$\underline{\mathsf{QGen}(sk, f)}$}
\begin{enumerate}
\item Parse $K_1 = sk$.
\item Generate the oracle $\mathcal{O}_f$.
    \begin{mdframed}
        {\bf $\underline{\mathcal{O}_f(id, x, (v_i)_{i \in [m(\lambda)]})}$}
        
        {\bf Hardcoded: $K_1, f$}
        \begin{enumerate}[label=\arabic*.]
            \item $(A_i, s_i, s_i')_{i \in [m(\lambda)]} \samp \mathsf{CosetGen}(1^{\lambda}, m(\lambda), \lambda; F_1(K_1, id))$.
            \item For each $i \in [m(\lambda)]$, check if $v_i \in A_i + s_i$ if $(x)_i = 0$ and check if $v_i \in A^{\perp}_i + s'_i$ if $(x)_i = 1$. If any of the checks fail, output $\perp$ and terminate.
            \item Output $f(x)$.
        \end{enumerate}
    \end{mdframed}
    \item Sample $id \samp \zo^{\lambda}$.
    \item $(A_i, s_i, s_i')_{i \in [m(\lambda)]} \samp \mathsf{CosetGen}(1^{\lambda},  m(\lambda), \lambda; F_1(K_1, id))$.
        \item Output $\left(\ket{A_{i, s_i, s'_i}}\right)_{i \in [m(\lambda)]}, id, \mathcal{O}$.

    \end{enumerate}
    
\paragraph{$\underline{\mathsf{Eval}(\regi{key}, x)}$}
\begin{enumerate}
    \item Parse $((\reg_i)_{i \in [m(\lambda)]}, id, \mathcal{O}) = \regi{key}$.
    \item For indices $i \in [m(\lambda)]$ such that $(x)_i = 1$, apply $H^{\lambda}$ to $\reg_i$.
    \item Run the oracle $\mathcal{O}$ coherently on $id, x$ and $(\reg_i)_{i \in [m(\lambda)]}$.
    \item Measure the output register and output the measurement outcome.
\end{enumerate}

Correctness with probability $1$ follows in a straightforward manner. We claim that the construction is also secure.
\begin{theorem}\label{thm:allunl}
 For any inverse polynomial $\gamma$ and any function family and challenge input distribution $(\mathcal{F}, \mathcal{D})$ that is $\gamma$-unlearnable , the scheme above satisfies strong $\gamma$-anti-piracy.
\end{theorem}
\begin{proof}
    The proof follows in a similar manner to the anti-piracy security proof of our PKE scheme (\cref{sec:pkeproof}). See \cref{appn:allunlproof} for the full proof.
\end{proof}

\section{Impossibility of Hyperefficient Shadow Tomography}\label{sec:hypertom}
In this section, as a corollary of results, we rule out existence of hyperefficient shadow tomography.
\begin{definition}[Hyperefficient Shadow Tomography \cite{Aar18}]\label{defn:hypertom}
    Let $E$ denote a uniform quantum circuit family with classical binary output that takes as input $i \in [M]$ and an $n$-qubit quantum state $\rho$.\footnote{That is, $E$ on input $i, \rho$ measures $\rho$ with respect to a binary measurement, which we can denote $E_i$} Then, a shadow tomography procedure takes as input $E$ and $\rho^{\otimes k}$ where $k$ denotes the number of copies, and outputs a quantum circuit\footnote{More precisely, classical description of a quantum circuit, since otherwise we can just hardwire the state $\rho$ into the circuit.} $C$ such that $\Pr[\forall i \in [M]~ \left|C(i) - \Pr[E(i, \rho) = 1]\right| < \eps] > 1 - \delta$. The procedure is said to be \emph{hyperefficient} if the number of copies $k$ and the runtime are both $\poly(n, \log M, \frac{1}{\eps})$.
\end{definition}
\cite{Aar18} shows that shadow tomography can be performed using polynomially many copies of $\rho$, however, the procedure takes exponential time. They leave it as an open question to give a hyperefficient shadow tomography procedure or rule out its existence. \cite{aaronson2007quantum, Aar09, Kre21} rule it out in oracle models, where the procedure has only black-box query access to the measurement circuit $E(i, \rho)$. 

\cite{Aar18} also shows that shadow tomography gives a generic attack on copy-protection schemes, and combined with their own sample-efficient shadow tomography procedure, they show that collusion-resistant copy-protection cannot exist without computational assumptions. Later, \cite{sattath2022uncloneable} adapts this attack to the case of unclonable decryptors, i.e., copy-protected secret keys for PKE, to conclude its impossibility without computational assumptions.

By \cref{thm:cppkeexists}, we obtain the following result.
\begin{corollary}
    Assuming post-quantum subexponentially secure indistinguishability obfuscation and LWE, there cannot exist a hyperefficient shadow tomography algorithm.    
\end{corollary}
\begin{proof}
    We prove the result by showing that shadow tomograpghy breaks PKE with copy-protected keys, which we construct in \cref{thm:cppkeexists}. Our attack is exactly the same as the one given by \cite{sattath2022uncloneable}, we merely observe that the attack is efficient when the shadow tomography procedure is efficient. We describe it below for completeness.

    Let $\pke$ be the PKE scheme with collusion-resistant copy-protected secret keys given in \cref{sec:pkecons}, for $1$-bit messages. Define the measurement circuit $E$ to be $\pke.\Dec(ct, \rho)$, where the measurement $E_{ct}$ outputs $1$ if the decryption procedure outputs $1$ when $\rho$ is given as the input to the key register. Note that $E$ is uniform. Suppose there exists a hyperefficient shadow tomography algorithm. Then, consider the following adversary for the anti-piracy game for $\pke$. We obtain $k$ keys where $k$ is the number of copies needed by the shadow tomography procedure, which is $\poly(\lambda)$ by assumption. We perform the procedure with $\eps = 1/8$ and $\delta = \frac{1}{2(k+1)}$ to obtain the estimation circuit $C$. We pick $0, 1$ as our challenge messages, and output $C$ to all $k + 1$ freeloaders. When presented with a challenge ciphertext $ct$, a freeloader runs $C(ct)$ and outputs $1$ if outputs a value $> 3/4$, and outputs $0$ otherwise. Note that if $ct$ is an encryption of $1$, we will have $\Pr[\pke.\Dec(ct, \rho) = 1] \leq 1 - \negl(\lambda)$, hence $C(ct) > 3/4$ with probability $1-\delta$, and we will correctly decrypt. By the same argument, all the freeloaders will simultaneously correctly decrypt with probability $> 1/2$. Note that the whole attack is efficient by assumption. This breaks the security of $\pke$, which is a contradiction.
 \end{proof}
\begin{remark}
    We note that to rule out hyperefficient shadow tomography, unbounded collusion-resistant schemes are needed - bounded collusion-resistance (\cite{LLQZ22}) is not sufficient, for the following reason. Since the number of copies required by hyperefficient shadow tomography procedure can depend on $\log M = \log 2^{|ct|} = |ct|$, if the ciphertext size grows with the collusion bound $k$, so does the number of copies needed. Hence, the hyperefficient shadow tomography procedure might need $k + 1$ or more copies to work, in which case we cannot arrive at a contradiction as we did above.
\end{remark}

We also obtain the following result.
\begin{corollary}
    There exists a quantumly accessible classical oracle relative to which there does not exist a hyperefficient shadow tomography algorithm.
\end{corollary}
\begin{proof}
    In this setting, the set of measurements $\{E_i\}_i$ is given by a quantumly accessible classical oracle $\mathcal{O}$ such that on input $i$, the oracle outputs the description of the measurement $E_i$. 
    
    By the same argument as above, our collusion-resistant copy-protection scheme for all unlearnable functionalities given in \cref{sec:allunlearn} implies the stated result.
\end{proof}

\section{Acknowledgements}
We thank Jiahui Liu for helpful discussions. Alper Çakan was supported by the following grants of Vipul Goyal: NSF award 1916939, DARPA SIEVE program, a gift from Ripple, a DoE NETL award, a JP Morgan Faculty Fellowship, a PNC center for financial services innovation award, and a Cylab seed funding award.

\bibliographystyle{alpha}
\bibliography{refs}

\appendix

\section{Proofs from \cref{sec:prelim}}
\subsection{Proof of \cref{thm:impindep}}\label{appn:impindep}
    We can write $\rho$ as $\rho = \sum_{j,k}\alpha_{j,k}\ketbra{j}{j}\otimes\ketbra{k}{k}$. Then, $(M_i\otimes I)\rho(M_i^\dagger\otimes I) = \sum_{j,k}\alpha_{j,k}(M_i\ketbra{j}{j}M_i^\dagger)\otimes\ketbra{k}{k}$ and therefore $(\Tr\otimes I)(M_i\otimes I)\rho(M_i^\dagger\otimes I) = \sum_{j,k}\alpha_{j,k}\bra{j}M_i^\dagger M_i\ket{j}\otimes\ketbra{k}{k}$. Note that this summation only depends on the POVM element $M_i^\dagger M_i$. The same argument applies to $\Lambda'$. Hence, the result follows by POVM equivalence of $\Lambda, \Lambda'$.
\subsection{Proof of \cref{lem:postmesdistlem}}\label{appn:postmesdistlem}
Let $\norm{\cdot}_1$ be the trace norm, and we have $\trd{\rho}{\sigma} =\frac{1}{2}\norm{\rho - \sigma}_1$. Define $q_i = {\Tr{M_i\sigma M_i^\dagger}}$ and we also have $p_i = {\Tr{M_i\sigma M_i^\dagger}}$. Then,
\begin{align*}
    \norm{p_i\rho' - q_i\sigma'}_1 &= \norm{p_i(\rho' - \sigma') - (q_i - p_i)\sigma'}_1 \\
    &\geq \left|p_i\norm{(\rho' - \sigma')}_1 - |q_i - p_i|\norm{\sigma'}_1\right| \\
    &\geq p_i\norm{(\rho' - \sigma')}_1 - \eps
\end{align*}
Last part follows from $|q_i - p_i| \leq  \trd{\rho}{\sigma} \leq \eps$ ad $\norm{\sigma'}_1 = 1$.
We also have by Schatten norm duality \begin{align*}
    \norm{p_i\rho' - q_i\sigma'}_1 &= \norm{M_i(\rho - \sigma)M_i^\dagger}_1 \\
    &= \sup_{-\mathbb{I} \leq E \leq \mathbb{I}}{\Tr{EM_i(\rho - \sigma)M_i^\dagger}} \\
    &= \sup_{-\mathbb{I} \leq E \leq \mathbb{I}}{\Tr{M_i^\dagger E M_i(\rho - \sigma)}} \\
    &\leq \sup_{-\mathbb{I} \leq E \leq \mathbb{I}}{\Tr{E(\rho - \sigma)}} \\
    &= \norm{\rho - \sigma}_1 \leq 2\eps.
\end{align*}
Above we also used the fact that when $-\mathbb{I} \leq E \leq \mathbb{I}$, we also have $-\mathbb{I} \leq M_i^\dagger E M_i \leq \mathbb{I}$. This is because $M_i^\dagger E M_i$ is positive semidefinite and $\bra{v} (I - M_i^\dagger E M_i) \ket{v} = \bra{v}\ket{v} - \bra{v} M_i^\dagger E M_i \ket{v} \geq \bra{v}\ket{v} - \bra{v} M_i^\dagger M_i \ket{v} \geq 0$  since $\bra{v} M_i^\dagger E M_i \ket{v} \leq \bra{v} M_i^\dagger M_i \ket{v}$ and $\sum_i M_i^\dagger M_i = \mathbb{I}$.

Combining the above yields the result.
    
\subsection{Proof of \cref{thm:simulproj}}\label{appn:simulproj}
First, we will prove the case where $\eps = 0$. We will prove it only for pure states, but the general case follows from purification. Let $\ket{\psi}$ be any state of appropriate dimension. We can write $\ket{\psi} = \sum_{j\in\mathcal{J},k\in \mathcal{K}} \alpha_{j,k}\ket{{v_j}}\otimes\ket{w_k}$ where $\{\ket{v_j}\}_{j\in\mathcal{J}}, \{\ket{w_k}\}_{k \in \mathcal{K}}$ are orthonormal eigenbases of $\Pi_1$ and $\Pi'_1$ respectively. We have $\Pi_1 = \sum_{j \in \mathcal{J}'}\ket{{v_j}}\bra{{v_j}}$ and $\Pi'_1 = \sum_{k \in \mathcal{K}'}\ket{{w_k}}\bra{{w_k}}$ for some subsets $\mathcal{J}' \subseteq \mathcal{J}, \mathcal{K}' \subseteq \mathcal{K}$. Since $\Tr{\Pi_1\otimes \Pi'_1 \rho} = 1$, we get $\alpha_{j,k} = 0$ if $(j,k) \not\in \mathcal{J}'\times\mathcal{K}'$.

    We can write the post-measurement state conditioned on outcome $i$ as $\ket{\phi}/\norm{\ket{\phi}}$ where we define the subnormalized state $\ket{\phi} = \sum_{j\in\mathcal{J}',k\in \mathcal{K}'} \alpha_{j,k}(M_i\ket{{v_j}})\otimes\ket{w_k}$.  When we apply $I \otimes \Pi'_1$ to $\ket{\phi}$, we get $\ket{\phi}$ again. Hence, $\Tr{\Pi'_1\frac{\ketbra{\phi}}{\norm{\ket{\phi}}^2}} = 1$, completing the first part of the proof.

Now, we move onto any $\eps \in (0, 1]$. Let $\rho'$ denote the post-measurement state obtained after applying $\Lambda\otimes\Lambda'$ to $\rho$ and obtaining the outcome $(1, 1)$. Note that $\rho'$ satisfies the claim with $\eps = 0$ since the measurement $\Lambda\otimes\Lambda'$  is projective. Then, by \cref{lem:gentlemes}, $\trd{\rho}{\rho'} \leq \sqrt{\eps}$ since canonical square root implementation of a projective measurement is the original measurement itself. Hence, applying the measurement $M$ on the first register and conditioning the outcome $i$, the post-measurement states of the second registers will have trace distance at most $3\sqrt{\eps}/2p_i$ by \cref{lem:postmesdistlem}. Hence, invoking the sub-claim for $\rho'$ with $\eps = 0$ and using the trace distance bound, we get the result.
\subsection{Proof of \cref{thm:qub}}\label{appn:qubproof}
We will prove the result only for pure states $\rho = \ketbra{\phi}{\phi}$ and the general case follows from convexity.

    We will first prove the case $n = 2$. Since $\Pi_1, \Pi_2$ are commuting projectors, there exists an orthonormal basis $\{\ket{v_i}\}_{i \in \mathcal{I}}$ and $\mathcal{I}_1, \mathcal{I}_2 \subseteq \mathcal{I}$ such that $\Pi_1 = \sum_{i \in \mathcal{I}_1} \ketbra{v_i}{v_i}$ and $\Pi_2 = \sum_{i \in \mathcal{I}_2} \ketbra{v_i}{v_i}$. We also have $\ket{\phi} = \sum_{i \in \mathcal{I}}c_i \ket{v_i}$ for some $\{c_i\}_{i \in \mathcal{I}}$ with $\sum_{i \in \mathcal{I}} |c_i|^2 = 1$. Then, 
    \begin{align*}
    \Tr[\Pi_1\rho] + \Tr[\Pi_2\rho] - \Tr[(I - \Pi_1\Pi_2)\rho] &=  
        \sum_{i \in \mathcal{I}_1 }|c_i|^2 + \sum_{i \in  \mathcal{I}_2}|c_i|^2 - \sum_{i \in \mathcal{I}_1 \cap \mathcal{I}_2}|c_i|^2 \\&=  \sum_{i \in \mathcal{I}_1\cup \mathcal{I}_2}|c_i|^2 \\&\leq \sum_{i \in \mathcal{I}} |c_i|^2  = 1.
    \end{align*}
    Hence, $\Tr[(I - \Pi_1\Pi_2) \rho] \leq \Tr[(I - \Pi_1)\rho]+\Tr[(I - \Pi_2)\rho]$. The general case follows by repeatedly applying this case and observing that $\Pi_i$ commutes with $\Pi_{i+1}\cdots\Pi_n$.

\subsection{Proof of \cref{thm:distti}}\label{appn:proofdistti}
First, we note that the result does not directly follow from the efficiency of $\mathsf{API}^{\eps,\delta}$. The reason is that while it is indeed efficient, it obtains superpositions of (exponentially many) outputs from the underlying distributions.

The proof will closely follow the proof of \cite[Theorem~6.5]{Z20}. We first state some of technical results that will be needed in the proof.

\subsubsection{Technical Lemmata}
\begin{lemma}\label{lem:apiequiv}
    Let $\mathcal{D}_0, \mathcal{D}_1$ be two efficient distributions with the same support and $\mathcal{P}$ be a collection of projective measurements indexed by this support. Suppose $\mathcal{D}_0 \equiv \mathcal{D}_1$. Then, $\mathsf{API}^{\eps,\delta}_{\mathcal{P},\mathcal{D}_0}\rho \equiv \mathsf{API}^{\eps,\delta}_{\mathcal{P},\mathcal{D}_1}\rho$ for any state $\rho$ of appropriate dimension.
\end{lemma}
While the claim might seem obvious, it needs to be proven formally since $\mathsf{API}^{\eps,\delta}_{\mathcal{P},\mathcal{D}}$ does not work by obtaining random samples from $\mathcal{D}$ but instead runs the algorithm $\mathcal{D}$ on various choices on random coins.
\begin{proof}
    When we inspect the actual implementation of $\mathsf{API}^{\eps,\delta}$ and proof of \cite[Theorem~6.2]{Z20}, we see that the output distribution of $\mathsf{API}^{\eps,\delta}_{\mathcal{P},\mathcal{D}}\rho$ is equivalent to the following:
    \begin{enumerate}
        \item Sample $p \samp \mathsf{PI}_{\mathcal{P}_\mathcal{D}}\rho$.
        \item Flip $2T$ independent biased coins, where each coin has expected value $p$.
        \item Output some deterministic function of all coin flips.
    \end{enumerate}
    Since $\mathcal{D}_0 \equiv \mathcal{D}_1$ implies  $\mathsf{PI}({\mathcal{P}}_{\mathcal{D}_0}) = \mathsf{PI}({\mathcal{P}}_{\mathcal{D}_1})$, the result follows.
\end{proof}

\begin{lemma}\label{lem:multsamp}
    Let $\mathcal{D}_0, \mathcal{D}_1$ be two distributions with sampling time $p(\lambda)$ such that $\mathcal{D}_0 \approx^c_{\nu(\lambda)} \mathcal{D}_1$ for all adversaries that run in time $t(\lambda)$. Define $\mathcal{D}^{s(\lambda)}_b$ to be the distribution where we sample $s$ independent samples from $\mathcal{D}_b$. Then, $\mathcal{D}^s_0 \approx^c_{s(\lambda)\cdot\nu(\lambda)} \mathcal{D}^s_1$ for all adversaries that run in time $t(\lambda)/(p(\lambda)\cdot s(\lambda))$.
\end{lemma}
\begin{proof}
The result follows from a standard hybrid argument. We give in full detail for completeness.

For all $i \in \{0, 1, \dots,s(\lambda)\}$, we define the hybrid distribution $\hyb_i$ as the distribution where the first $i$ components are sampled from $\mathcal{D}_1$ and the rest are sampled from $\mathcal{D}_0$. Observe that $\hyb_0$ is $\mathcal{D}^s_0$ and $\hyb_s$ is $\mathcal{D}^s_1$.

Now, we claim $\hyb_{i - 1} \approx^c_{\nu(\lambda)} \hyb_i$ for adversaries that run in $t(\lambda)/s(\lambda)$, for all $i \in [s]$. Suppose otherwise, for a contradiction. Let $\adve$ be the that distinguishes them. Then, we can create an adversary $\adve'$ for distinguishing $\mathcal{D}_0$ versus $\mathcal{D}_1$ as follows.
\paragraph{$\underline{\adve'(a)}$}
\begin{enumerate}
    \item For $j \in [s]$, sample $a_j \samp \mathcal{D}_1$ if $j \leq i - 1$ and $a_j \samp \mathcal{D}_0$ if $j > i$.
    \item Set $a_i = a$.
    \item Output $\adve((a_j)_{j \in [s]})$.
\end{enumerate}
It is easy to see that $\adve'$ runs in time $O(t(\lambda)/(p(\lambda)\cdot s(\lambda))\cdot p(\lambda)\cdot s(\lambda))$ and has advantage $\geq \nu(\lambda)$, which is a contradiction.

Finally, triangle inequality yields the claim.
\end{proof}

\begin{lemma}[\cite{Z20}]\label{lem:smallrange}
Let $A$ be a set.
Sample $\Pi$ to be a random permutation on $A$. Sample random functions $G: [s] \to A$ and $F: A \to [s]$. Then, for any quantum algorithm $B$ making $Q$ quantum queries to its oracle, we have \begin{equation*}
    \left|\Pr[B^{\Pi}() = 1] - \Pr[B^{G \circ F}() = 1] \right| \leq O(Q^3/s + Q^3/|A|).
\end{equation*}
\end{lemma}
\begin{lemma}[\cite{Z12b}]\label{lem:qindep}
    Sample a random function $F: \mathcal{A} \to \mathcal{B}$ and a $2Q$-wise independent function $E: \mathcal{A} \to \mathcal{B}$.
     Then, for any quantum algorithm $B$ making $Q$ quantum queries to its oracle, we have \begin{equation*}
    \Pr[B^{F}() = 1] = \Pr[B^{E}() = 1].
\end{equation*}
\end{lemma}

\subsubsection{Proof of the Theorem}
Now we move onto the proof of the theorem. We will construct a sequence of hybrid distributions, starting with $\Vec{p_0}$ and ending with $\Vec{p_1}$, that are obtained by modifying the previous one. Without loss of generality, assume that all $\mathcal{B}_\ell^b$ have the same random coin set $\mathcal{R}$. Let $s$ be a parameter that we will set later. We assume that $|\mathcal{R}|$ is at least $s$, which is without loss of generality since we can pad the random coins (and later ignore the padding when using the coins).

\paragraph{$\hyb_0$:} Same as $\Vec{p_0}$.
\paragraph{$\hyb_1$:} We now sample $\rho, pp$ as $\mathcal{S}^1(1^\lambda)$.

\paragraph{$\hyb_2$:} For all $\ell \in [k]$, sample a random permutation $\Pi_\ell: \mathcal{R} \to \mathcal{R}$. Then, instead of applying $\bigotimes_{\ell \in [k]}\mathsf{API}^{\eps,\delta}_{\mathcal{P}_{\ell}, \mathcal{B}^0_\ell(pp)}$ to $\rho$, now apply  $\bigotimes_{\ell \in [k]}\mathsf{API}^{\eps,\delta}_{\mathcal{P}_{\ell},\mathcal{B}^{'0}_\ell(pp))}$ where we define $\mathcal{B}^{'0}_\ell(pp; r) = \mathcal{B}^{0}_\ell(pp; \Pi_\ell(r))$.

\paragraph{$\hyb_{3,i}$ for $i \in [k]$:}  Sample random functions $G_\ell: [s] \to \mathcal{R}$ and $F_\ell: \mathcal{R} \to [s]$ for all $\ell \in [i]$. We now apply  $\left(\bigotimes_{\ell \in [i]}\mathsf{API}^{\eps,\delta}_{\mathcal{P}_{\ell},\mathcal{B}^{''0}_\ell(pp))}\right) \otimes \left(\bigotimes_{\ell \in \{i+1,\dots,k\}}\mathsf{API}^{\eps,\delta}_{\mathcal{P}_{\ell},\mathcal{B}^{'0}_\ell(pp))}\right)$ where we define $\mathcal{B}^{''0}_\ell(pp; r) = \mathcal{B}^{0}_\ell(pp; G(F(r)))$.

\paragraph{$\hyb_{4,i}$ for $i \in [k]$:}  Sample $2Q$-wise independent function $E_\ell: \mathcal{R} \to [s]$ for all $\ell \in [i]$. We now apply  $\left(\bigotimes_{\ell \in [i]}\mathsf{API}^{\eps,\delta}_{\mathcal{P}_{\ell},\mathcal{B}^{'''0}_\ell(pp))}\right) \otimes \left(\bigotimes_{\ell \in \{i+1,\dots,k\}}\mathsf{API}^{\eps,\delta}_{\mathcal{P}_{\ell},\mathcal{B}^{''0}_\ell(pp))}\right)$ where we define $\mathcal{B}^{'''0}_\ell(pp; r) = \mathcal{B}^{0}_\ell(pp; G(E(r)))$.

\paragraph{$\hyb_{5,i}$ for $i \in \{0,1,\dots,k-1\}$:}  We now apply  $\left(\bigotimes_{\ell \in [i]}\mathsf{API}^{\eps,\delta}_{\mathcal{P}_{\ell},\mathcal{B}^{'''1}_\ell(pp))}\right) \otimes \left(\bigotimes_{\ell \in \{i+1,\dots,k\}}\mathsf{API}^{\eps,\delta}_{\mathcal{P}_{\ell},\mathcal{B}^{'''0}_\ell(pp))}\right)$.

\paragraph{$\hyb_{6, k - i + 1}$ for $i \in [k]$:}   We now apply  $\left(\bigotimes_{\ell \in [i]}\mathsf{API}^{\eps,\delta}_{\mathcal{P}_{\ell},\mathcal{B}^{'''1}_\ell(pp))}\right) \otimes \left(\bigotimes_{\ell \in \{i+1,\dots,k\}}\mathsf{API}^{\eps,\delta}_{\mathcal{P}_{\ell},\mathcal{B}^{''1}_\ell(pp))}\right)$ where we define $\mathcal{B}^{'''1}_\ell(pp; r) = \mathcal{B}^{1}_\ell(pp; G(E(r)))$ and  $\mathcal{B}^{''1}_\ell(pp; r) = \mathcal{B}^{1}_\ell(pp; G(F(r)))$.

\paragraph{$\hyb_{7,k-i+1}$ for $i \in [k]$:}  We now apply  $\left(\bigotimes_{\ell \in [i]}\mathsf{API}^{\eps,\delta}_{\mathcal{P}_{\ell},\mathcal{B}^{''1}_\ell(pp))}\right) \otimes \left(\bigotimes_{\ell \in \{i+1,\dots,k\}}\mathsf{API}^{\eps,\delta}_{\mathcal{P}_{\ell},\mathcal{B}^{'1}_\ell(pp))}\right)$ where we define $\mathcal{B}^{'1}_\ell(pp; r) = \mathcal{B}^{1}_\ell(pp; \Pi_\ell(r))$.

\paragraph{$\hyb_8$:} We now apply $\bigotimes_{\ell \in [k]}\mathsf{API}^{\eps,\delta}_{\mathcal{P}_{\ell},\mathcal{B}^{'1}_\ell(pp))}$ to $\rho$.
\paragraph{$\hyb_9$:} Same as $\vec{p_1}$.
\begin{lemma}
    $\statdist{\hyb_0}{\hyb_1} \leq \nu(\lambda)$
\end{lemma}
\begin{lemma}
    $\hyb_1 \equiv \hyb_2$ and $\hyb_8 \equiv \hyb_9$.
\end{lemma}
\begin{proof}
    Observe that $\mathcal{B}^{'b}_\ell(pp)$ and $\mathcal{B}^{b}_\ell(pp)$ are exactly the same distribution. The result follows from \cref{lem:apiequiv}.
\end{proof}
\begin{lemma}
\begin{itemize}
    \item $\statdist{\hyb_2}{\hyb_{3,1}} \leq O(Q^3/s).$
    \item $\statdist{\hyb_{3,i}}{\hyb_{3,i+1}} \leq O(Q^3/s)$ for all $i \in [k - 1]$.
    \item   $\statdist{\hyb_{7,i}}{\hyb_{7,i+1}} \leq O(Q^3/s)$ for all $i \in [k - 1]$.
    \item $\statdist{\hyb_{7,k}}{\hyb_{8}} \leq O(Q^3/s).$
\end{itemize}
    
\end{lemma}
\begin{proof}
Follows from \cref{lem:smallrange}.
\end{proof}

\begin{lemma}
    \begin{itemize}
        \item $\hyb_{3,k} \equiv \hyb_{4,1}$.
        \item $\hyb_{4,i} \equiv \hyb_{4,i+1}$ for all $i \in [k -  1]$.
        \item $\hyb_{6,k} \equiv \hyb_{7,1}$.
        \item $\hyb_{6,i} \equiv \hyb_{6,i+1}$ for all $i \in [k -  1]$.

    \end{itemize}
\end{lemma}
\begin{proof}
    Follows from \cref{lem:qindep}.
\end{proof}

\begin{lemma}
\begin{itemize}
\item $\statdist{\hyb_{5,k-1}}{\hyb_{6,0}} \leq \nu(\lambda)\cdot s(\lambda)$.
        \item $\statdist{\hyb_{5,i-1}}{\hyb_{5,i}} \leq \nu(\lambda)\cdot s(\lambda)$ for $i \in [k]$
\end{itemize}

\end{lemma}
\begin{proof}
Observe that  $\mathcal{B}^{b}_\ell(pp; G(E(\cdot)))$ can be interpreted as $s$ samples from $\mathcal{B}^{b}$ with the input selecting which sample to use. Also, both experiments can be computed in time $\poly(\lambda)\cdot k \cdot s$. The result then  follows from \cref{lem:multsamp}.
\end{proof}

Combining the above, we get $\statdist{\Vec{p_0}}{\Vec{p_1}} < O(k\cdot(Q^3/s + \nu(\lambda)\cdot s(\lambda)))$. We set $s = 1/\mu(\lambda)$, which yields the result since $Q = \poly(\lambda)$.
\subsection{Proofs of \cref{thm:multiatiprop}, \cref{thm:multiapiprop},\cref{thm:multiapismallproof}}\label{appn:multiatiprop}
\paragraph{Proof of \cref{thm:multiatiprop}}
    See \cite[Corollary~3]{ALLZZ20} for the proofs of the first two points. Note that while they consider the same threshold value $\eta_\ell$ for all indices $\ell$, an inspection of their proof easily shows that the results still hold for any $\eta_\ell$.

    Combining the first two bullet points yields the third bullet point in a straightforward manner.
    Fourth point follows similarly to arguments below for \cref{thm:multiapiprop}.
\paragraph{Proof of \cref{thm:multiapiprop}}
    We will only prove the first claim, and the second prove follows by the same argument. 

Fix any $\ell \in [k]$. Consider the projective measurement $\mathcal{M}_\ell$ where we apply  $ \mathsf{PI}({\mathcal{P}_\ell}_{\mathcal{D}_\ell})$ to the $\ell$-th register and apply $I$ to the other registers. Then, we have $\Pr[(\mathcal{M}_\ell \rho) \leq (\vec{p})_\ell + 2\eps] \geq  1 - 2\delta$ since $\mathsf{API}^{\eps,\delta}$ is $(\eps,\delta)$-almost projective and since it $\delta$-approximates $\mathsf{PI}$ in $\eps$-shift distance. Note while $\rho'$ is obtained after a measurement on all registers, we can assume that $\mathsf{API}^{\eps,\delta}_{\mathcal{P}_\ell, \mathcal{D}_\ell}$ was applied last since measurements on disjoint registers commute.

Now, observe that $\mathcal{M}_1\mathcal{M}_2\cdots\mathcal{M}_k = \left(\bigotimes_{\ell \in [k]} \mathsf{PI}({\mathcal{P}_\ell}_{\mathcal{D}_\ell})\right)$ and that $\mathcal{M}_\ell$ commute. We define $\mathcal{M}'_\ell$ to be the binary projective measurement where we apply $\mathcal{M}_\ell$ and output $1$ if the outcome is $\leq \vec{p}_\ell + 2\eps$. Then, by above we have $\Pr[\mathcal{M}'_\ell \rho' = 1] \geq 1 - 2\delta$. We get $\Pr[\forall \ell \in [k]~~(\vec{p'})_\ell \leq (\vec{p})_\ell + 2\eps] \geq 1 - 2\cdot k \cdot \delta$ by \cref{thm:qub}.
\subsection{Proof of \cref{thm:multiapismallproof}}\label{appn:multiapismallproof}
Follows similarly to arguments above for \cref{thm:multiapiprop}.

\section{Connections of the Signature Scheme to the Public-key Encryption Scheme}\label{appn:dstopke}
The only difference between the scheme in $\hyb_5$ in the proof of \cref{sec:digsigcpproof} and our PKE scheme is the hidden trigger mechanism and the associated extra programs and keys, $K, K_3, K_4, \mathsf{OPEval}, \mathsf{OPVer}$, which is independent of the actual encryption mechanism and can be generated by the adversary itself during the reduction to the security of PKE scheme; and the fact that the ciphertext programs $\mathsf{OQ}$ now include two branches along with a mode parameter. The mode $\mathsf{mode}=\mathsf{eval}$ is the same as the original ciphertext program of our PKE scheme, and $\mathsf{mode}=\mathsf{verify}$ is a point function that checks if the input has the same image as the encrypted message under a one-way function $f$. However, we can see that the scheme is still secure with this addition.

First argument is as follows. We can first invoke the security argument (\cref{sec:pkeproof}) without $\mathsf{mode}=\mathsf{verify}$, and show that any adversary can win the game (i.e. correctly predict the encrypted messages) with subexponential probability. Then, adding back the $\mathsf{verify}$ mode, we can consider the copy-protected keys obtained by the adversary as quantum auxiliary information, and we can replace the $\mathsf{verify}$ mode with a compute-and-compare obfuscation. Since by the security of the scheme without $\mathsf{verify}$ mode we have that the adversary cannot predict the messages correctly given its auxiliary information, we can finally invoke the security of the compute-and-compare obfuscation to conclude that the adversary still cannot predict the messages correctly given $\mathsf{verify}$ mode.

An alternative way of concluding the security of the scheme is by repeating the security proof in \cref{sec:pkeproof}. It is easy to see that the only part that will be affected by $\mathsf{verify}$ mode is the final extraction argument using compute-and-compare obfuscation. However, it is easy to see that by the security of the injective one-way function $f$ that the compute-and-compare obfuscation arguments still works. In this case, we will modify our compute-and-compare programs so that the compute program also accepts a $\mathsf{mode}$ parameter, and either performs coset vector verification or compares the image of the input to $f(m)$. However, any adversary that can succesfully find an input that passes the compare condition must have correct coset vectors, since it cannot have a correct preimage for $f(m)$ by the security of the one-way function $f$. Thus, the security proof in \cref{sec:pkeproof} still works.
\section{Proof of Anti-Piracy Security of the General Copy-Protection Scheme}\label{appn:allunlproof}
In this section, we prove \cref{thm:allunl}.

Define $\hyb_0$ to the original anti-piracy game. Define $\hyb_1$ by modifying $\hyb_0$ by changing the way we sample the identity strings during each quantum copy-protected function generation as follows. Let the challenger record each sampled identity when answering each query, and when answering a new query, it samples uniformly at random an identity value from the set $\{1, \dots, 2^\lambda - 1\}$ \emph{that has not appeared before}. That is, we sample unique identity strings for each query to $\mathsf{QGen}$. Also, we define the following notation. Let $id_{\alpha(i)}$ be the $i^{th}$ value sampled where $\alpha(\cdot)$ is the permutation $[k] \to [k]$ such that $0 < id_1 < \dots < id_k < 2^{\lambda}$. That is, $id_{\alpha(i)}$ is the identity string that is sampling during the $i^{th}$ query of the adversary. For simplicity of notation, we also set $id_0 = 0$ and $id_{k + 1} = 2^\lambda$.

Define $\hyb_2$ by modifying $\hyb_1$ as follows. At the end of the game, instead of using threshold implementations $\mathsf{TI}^{\mathcal{O}}_{\mathcal{D}, \gamma}$, the challenger uses approximate threshold implementations $\mathsf{ATI}^{\eps, \delta, \mathcal{O}}_{\mathcal{D}, \frac{31\gamma}{32}}$ with $\eps = \frac{\gamma}{32k}$ and $\delta = 2^{-4\lambda}$. It outputs $1$ if and only if all $\mathsf{ATI}$ output $1$.

\begin{claim}
    $\Pr[\hyb_2 = 1] > \Pr[\hyb_0 = 1] - {\exp(-\lambda)}$.
\end{claim}
\begin{proof}
Follows from the same argument as in \cref{sec:pkeproof}.
\end{proof}

Therefore, $\adve$ wins in $\hyb_2$ with probability $\frac{1}{p(\cdot)}$ for some polynomial $p(\cdot)$ and infinitely many values of $\lambda > 0$. Note that in $\hyb_2$, now the challenger is also efficient by \cref{thm:singleatiprop} and our choice of $\eps, \delta$.

We define the following notation and the monogamy-of-entanglement type game.
\begin{definition}
    For all $j \in [k]$, let $(A_i^j, s_i^j, s_i^{'j})_{{i \in [\cosettcount]}}$ denote the tuple of subspaces and vectors sampled during the sampling of the $(\alpha^{-1}(j))$-th key. That is, it is the coset tuple associated with $id_j$.
\end{definition}
\paragraph{$\underline{\mathcal{G}(\lambda, (\adve'_0, \adve'_1, \adve'_2))}$}
\begin{enumerate}
\item The challenger instantiates the copy-protection scheme as $sk \samp \mathsf{Setup}(1^\lambda)$.
\item The challenger samples a function $f$ from $\mathcal{F}$.
    \item For multiple rounds, $\adve$ makes quantum key queries. For each query, the challenger generates a key as $\reg \samp \mathsf{QGen}(sk, f)$ and submits $\reg$ to the adversary.
    \item The adversary outputs a \emph{bipartite} register $\regi{bip}$ and an index $j^* \in [k]$, where $k$ is the number of queries it made.
    \item For $\ell \in \{1, 2\}$, the challenger does the following.
    \begin{enumerate}[label*=\arabic*.]
        \item Sample $r_\ell \samp \mathcal{D}$.
        \item Run $\adve'_\ell$ on $\regi{bip}[\ell]$, $(A^{j^*}_i)_{i \in [m(\lambda)]}$ and $r_\ell$ to obtain a tuple of vectors $(v_{\ell, i})_{i \in [{m(\lambda)}]}$.
        \item For all $i \in [m(\lambda)]$, check if $v_{\ell, i} \in A^{j^*}_{i} + s^{j^*}_i$ if $(r_\ell)_i = 0$ and check if $v_{\ell, i} \in {(A^{j^*})}^\perp_{i} + s^{'{j^*}}_i$ if $(r_\ell)_i = 1$.
    \end{enumerate}
    If all the checks pass, the challenger outputs $1$. Otherwise, it outputs $0$.
\end{enumerate}

 It is easy to see that an adversary $\adve'$ that wins $\mathcal{G}$ gives us a contradiction for \cref{defn:strmoecoll}, since using $\adve'$ we can give an adversary for the monogamy-of-entanglement game, which simply simulates $\adve'$ playing $\mathcal{G}$ by sampling the PRF key $K_1$, the extra oracles and so on efficiently itself. In particular note that the oracle $\mathcal{O}_f$ in $\mathcal{G}$ can be implemented using $\mathsf{PMem}$ from \cref{defn:strmoecoll}\footnote{While the original theorem statement therein uses indistinguishability obfuscation, it is trivial to see that the result still holds when we instead use ideal oracles.}.

Now, we construct a tuple of adversaries $(\adve'_0, \adve'_1, \adve'_2)$ for $\mathcal{G}$, starting with $\adve'_0$. We note that from now on, whenever we are considering the freeloaders, we assume that they no longer have access to the oracle $\mathcal{O}_f$ that $\adve$ had during the query phase, and each will have access to some modified oracle that will be clear from context.

 Let $\mathcal{O}_j$ for $j \in \{0, \dots, k + 1\}$ be efficient oracles, which we will define later. Define $\adve'_0$ as follows.
\begin{mdframed}
        {\bf $\underline{\adve_0'(1^\lambda)}$}
        \begin{enumerate}
            \item Simulate $\adve$ by interacting with $\mathsf{Samp}$ and the challenger, making a query to $\mathsf{QGen}$ whenever $\adve$ makes a query, and forwarding the obtained program to it. Let $f$ be the function determined at the end of the setup phase and let $\regi{adv}$ be the $(k + 1)$-partite register (with state $\sigma$) output by $\adve$ at the output phase.
            \item Uniformly at random sample $x, y, j^*$ such that $1 \leq x < y \leq k + 1$ and $j^* \in \{1, \dots, k\}$.
            \item Apply $\mathsf{API}^{\eps, \delta, \mathcal{O}_0}_{\mathcal{D}}$ to all registers $\regi{adv}[\ell]$ for $\ell \in [k + 1]$, let $b_{\ell, 0}$ be the measurement outcomes.
            \item Apply $\mathsf{API}^{\eps, \delta,\mathcal{O}_i}_{\mathcal{D}}$ in succession for $i = 1$ to $j^*$ to $\regi{adv}[x]$, let $b_{x, i}$ be the measurement outcomes.
            \item Apply $\mathsf{API}^{\eps, \delta,\mathcal{O}_i}_{\mathcal{D}}$ in succession for $i = 1$ to $j^*$ to $\regi{adv}[y]$, let $b_{y, i}$ be the measurement outcomes.
            \item Output  \begin{align*}
                (&(\regi{adv}[x], j^*, x, y, (b_{\ell, 0})_{\ell \in [k + 1]}, (b_{x,i})_{i \in [j^*]}, (b_{y,i})_{i \in [j^*]}),\\ &(\regi{adv}[y], j^*, x, y, (b_{\ell, 0})_{\ell \in [k + 1]}, (b_{x,i})_{i \in [j^*]}, (b_{y,i})_{i \in [j^*]}, ),\\ &j^*).
            \end{align*}
            
        \end{enumerate}
    \end{mdframed}

For $j \in \{1, \dots, k\}$, define $\mathcal{O}_j$ to be the following oracles.
 \begin{mdframed}
        {\bf $\underline{\mathcal{O}_j(id, x, (v_i)_{i \in [m(\lambda)]})}$}
        
        {\bf Hardcoded: $K_1, f, \textcolor{red}{id_j}$}
        \begin{enumerate}[label=\arabic*.]
            \item $(A_i, s_i, s_i')_{i \in [m(\lambda)]} \samp \mathsf{CosetGen}(1^{\lambda}, m(\lambda), \lambda; F_1(K_1, id))$.
            \textcolor{red}{\item If $id < id_j$, output $\perp$ and terminate.}
            \item For each $i \in [m(\lambda)]$, check if $v_i \in A_i + s_i$ if $(x)_i = 0$ and check if $v_i \in A^{\perp}_i + s'_i$ if $(x)_i = 1$. If any of the checks fail, output $\perp$ and terminate.
            \item Output $f(x)$.
        \end{enumerate}
    \end{mdframed}
We define $\mathcal{O}_{0}$ to be the original oracle $\mathcal{O}_f$ and we define $\mathcal{O}_{k + 1}$ to be the empty oracle that always outputs $\perp$. We also define some intermediary oracles. Define the following for all $j \in \{0, 1, \dots, k\}$ and $\Delta \in \{0, 1, \dots, id_{j+1}-id_{j} - 1\}$. For notational convenience, also define $\mathcal{O}_{j, {id_{j+1}-id_{j}}}$ to be $\mathcal{O}_{j+1, 0}$ for all $j \in \{0, 1, \dots, k\}$. Also note that $\mathcal{O}_{0,0}$ is exactly the same as $\mathcal{O}_0$ for $j \in [k]$.
\begin{mdframed}
        {\bf $\underline{\mathcal{O}_{j,\Delta}(id, x, (v_i)_{i \in [m(\lambda)]})}$}
        
        {\bf Hardcoded: $K_1, f, \textcolor{red}{id_j + \Delta}$}
        \begin{enumerate}[label=\arabic*.]
            \item $(A_i, s_i, s_i')_{i \in [m(\lambda)]} \samp \mathsf{CosetGen}(1^{\lambda}, m(\lambda), \lambda; F_1(K_1, id))$.
            \textcolor{red}{\item If $id < id_j  + \Delta$, output $\perp$ and terminate.}
            \item For each $i \in [m(\lambda)]$, check if $v_i \in A_i + s_i$ if $(x)_i = 0$ and check if $v_i \in A^{\perp}_i + s'_i$ if $(x)_i = 1$. If any of the checks fail, output $\perp$ and terminate.
            \item Output $f(x)$.
        \end{enumerate}
    \end{mdframed}

We now define some notation.

\begin{definition}
Consider the following experiment.
\begin{enumerate}
\item Simulate the first two steps of $\adve'_0$ and the challenger:
\begin{enumerate}[label*=\arabic*.]
             \item Simulate $\adve$ and the challenger. Let $f$ be the function determined at the end of the setup phase and let $\regi{adv}$ be the $(k + 1)$-partite register (with state $\sigma$) output by $\adve$ at the output phase.
            \item Uniformly at random sample $x, y, j^*$ such that $1 \leq x < y \leq k + 1$ and $j^* \in \{1, \dots, k\}$.
\end{enumerate}
    \item Set  $pp =  (x, y, j^*, (id_j)_{j\in[k+1]}, f)$.
    \item Output $\regi{adv}, pp$.
\end{enumerate} 
We will write $\expfromdddall{\mathcal{O}}{\ell} \approx_\nu^c \expfromdddall{\mathcal{O}'}{\ell}$ to denote that the advantage of any QPT adversary in distinguishing the oracles $\mathcal{O}, \mathcal{O}'$ (which can depend on $pp$) given the outcome of the above experiment, is $\nu$. We omit $\nu$ when $\nu$ is $\negl(\lambda)$.
\end{definition}

\begin{claim}\label{claim:alld0closed1}
    We have 
\begin{itemize}
\item $\expfromdddall{\mathcal{O}_{j, \Delta}}{\ell} \approx^c_{2^{-2\lambda}} \expfromdddall{\mathcal{O}_{j, \Delta + 1}}{\ell}$ for all $j \in \{0,1,\dots,k\}$ and $\Delta \in \{1, \dots, id_{j+1}-id_{j} - 1\}$,
\item $\expfromdddall{\mathcal{O}_{j, 1}}{\ell} \approx^c \expfromdddall{\mathcal{O}_{j+1}}{\ell}$ for all $j \in \{0,1,\dots,k\}$,

    \item $\expfromdddall{\mathcal{O}_0}{\ell} \approx^c \expfromdddall{\mathcal{O}_1}{\ell} $.
\end{itemize}
\end{claim}
\begin{proof}
We prove the first point and the rest follow by the hybrid lemma through a simple calculation.

First, note that the oracles $\mathcal{O}_{j,\Delta}$ and $\mathcal{O}_{j,\Delta + 1}$ only differ at points such that $id = id_j + \Delta$ and $v_i$ are in the correct cosets (primal or dual) for the coset tuple generated using the randomness $F_1(K_1, id_j + \Delta)$. Let $S$ denote the set of all such inputs and we claim $q_S \leq \negl(\lambda)$, that is, the adversary has negligible query weight on $S$. Suppose otherwise for a contradiction. Then, we can measure a random query of the adversary to obtain vectors as above with non-negligible probability. However, note the following:
\begin{itemize}
    \item The adversary has only oracle access to the PRF key $K$, hence, $F_1(K_1, id_j + \Delta)$ is random given the adversary's view,
    \item The adversary never obtains an actual coset state for this tuple,
    \item The oracles $\mathcal{O}_f, \mathcal{O}_{j,\Delta}$ and $\mathcal{O}_{j,\Delta + 1}$ can be simulated using a membership oracle for this tuple.
\end{itemize}
Since coset membership is unlearnable, by above we obtain a contradiction. Thus,  $q_S \leq \negl(\lambda)$ and $\expfromdddall{\mathcal{O}_{j, \Delta}}{\ell} \approx^c_{2^{-2\lambda}} \expfromdddall{\mathcal{O}_{j, \Delta + 1}}{\ell}$ by \cref{bbbv}.
\end{proof}

\begin{claim}\label{claim:allunlmesres}
Let $\tau$ be the state of the bipartite register $\regi{adv}[x,y]$ output by $\adve'_0$ in $\mathcal{G}$, and also consider the classical values $j^*, x, y, \{b_{\ell, i}\}_{\ell, i}$ contained in the output of $\adve'_0$.
    
    Suppose we apply the measurement $\mathsf{API}^{\eps, \delta, \mathcal{O}_{j^* + 1}}_{\mathcal{D}}\otimes\mathsf{API}^{\eps, \delta, \mathcal{O}_{j^* + 1}}_{\mathcal{D}}$ to $\tau$ and let $b_{x,j^*+1}, b_{y,j^*+1}$ denote the measurement outcomes we obtain. Then,
   
    \begin{equation*}
     \Pr[ b_{x,j^*}  - b_{x,j^*+1} > \frac{29\gamma}{32k}  \wedge  b_{y,j^*} - b_{y,j^*+1} > \frac{29\gamma}{32k}] > \frac{1}{\poly(\lambda)}
    \end{equation*}
    
    where the probability is taken over the randomness of the challenger, the adversary $\adve_0'$ and the measurement outcomes.
\end{claim}
\begin{proof}
First, note that $$
 \Pr[\left(\mathsf{PI}^{\mathcal{O}_{k+1}}_{\mathcal{D}}\right)\cdot\iota \geq \frac{\gamma}{32}] = 0.
 $$
 for any state $\iota$ that can be efficiently obtained during game $\mathcal{G}$, since $\mathcal{O}_{k+1}$ is the empty oracle and $(\mathcal{F}, \mathcal{D})$ is unlearnable.
    The rest follows from the same argument as in the proof \cref{claim:mesresdist} and by the fact that $\expfromdddall{\mathcal{O}_0}{\ell} \approx^c \expfromdddall{\mathcal{O}_1}{\ell}$ (\cref{claim:alld0closed1}).
\end{proof}

\begin{claim}\label{claim:allunlearnjplusone}
Let $\tau$ be the bipartite state output by $\adve'_0$ in $\mathcal{G}$. Let $p'_x, p'_y$ be the outcome of applying $\mathsf{PI}^{\mathcal{O}_{j^*}}_{\mathcal{D}}\otimes \mathsf{PI}^{\mathcal{O}_{j^*}}_{\mathcal{D}}$ to $\tau$. Similarly, let $p''_x, p''_y$ be the outcome of applying $\mathsf{PI}^{\mathcal{O}_{j^*, 1}}_{\mathcal{D}}\otimes \mathsf{PI}^{\mathcal{O}_{j^*, 1}}_{\mathcal{D}}$ to $\tau$. Then,
\begin{itemize}
    \item $\Pr[p'_x > b_{x,j^*} - \frac{3\gamma}{32k} \wedge p'_y > b_{y,j^*} - \frac{3\gamma}{32k}] \geq 1 - 2^{-3\lambda}$.
    \item $\Pr[b_{x,j^*} - p''_x > \frac{28\gamma}{32k} \wedge b_{y,j^*} - p''_y  > \frac{28\gamma}{32k}] > \frac{1}{q(\lambda)}$ for some polynomial $q(\cdot)$.
\end{itemize}
\end{claim}
\begin{proof}
    Follows from the same arguments as in the proof of \cref{claim:dprimdprim}, and by \cref{claim:allunlmesres}.
\end{proof}

Now, we claim that we can extract correct vectors from the output state $\tau$ of the adversary. First, observe that we have $\expfromdddall{\mathcal{O}_{j^*}}{\ell} \not\approx_c \expfromdddall{\mathcal{O}_{j^*1, 1}}{\ell}$, since $\expfromdddall{\mathcal{O}_{j^*}}{\ell} \approx_c \expfromdddall{\mathcal{O}_{j^*1, 1}}{\ell}$ would give us a contradiction to \cref{claim:allunlearnjplusone} by \cref{thm:apiprop}. Then, by the contrapositive of \cref{bbbv}, we get that the freeloader encoded in $\tau[1]$ has a non-negligible query weight on the set of points of where the oracles $\mathcal{O}_{j^*}$ and $\mathcal{O}_{j^*, 1}$ differ. Observe that these oracles only differ on points that satisfy $id = id_{j^*}$ and $v_i \in A_i^{j^*}+{s_i}^{j^*}$ if $(r_1)_i = 0$ and $v_i \in (A_i^{j^*})^\perp + {s^{'j^*}_i}$ if $(r_1)_i = 1$. Hence, by measuring a random query of the freeloader adversary encoded in $\tau[1]$ (which we simulate using a universal quantum circuit), we obtain correct coset vectors wtih non-negligible probability. We set this algorithm as our adversary $\adve'_1$ for the game $\mathcal{G}$.

Finally, we claim that even conditioned a successful coset vector extraction from the first freeloader, we can still extract from the second freeloader. Let $\xi$ denote the post-measurement state of the second register, obtained by extracting from the first register of $\tau$ as above using $\adve'_1$ and conditioning on a successful extraction. We claim that $\xi$ satisfies
\begin{enumerate}
\item \label{item:allunlearnpidprime2} $\Pr[\mathsf{PI}^{\mathcal{O}_{j^*}}_{\mathcal{D}}\cdot\xi \leq b_{y,j^*} - \frac{3\gamma}{32k}] \leq 2^{-\lambda}.$
    \item  \label{item:allunlearnpidprimeprime2} $\Pr[\mathsf{PI}^{\mathcal{O}_{j^*, 1}}_{\mathcal{D}}\cdot\xi < b_{y,j^*} - \frac{28\gamma}{32k}] \geq \frac{1}{\poly(\lambda)}.$
\end{enumerate}

This first claim follows from \cref{claim:allunlearnjplusone}, \cref{thm:simulproj}, and the fact that extraction on the first register succeeds with non-negligible probability. To see the second point, observe that we can imagine $\mathsf{PI}^{\mathcal{O}_{j^*, 1}}_{\mathcal{D}}$ being applied to the second register, before an extraction on the first register, and condition on obtaining a value $< b_{y,j^*} - \frac{28\gamma}{32k}$ (denote this as event $G_2$). Then, by \cref{claim:allunlearnjplusone}, it is easy to see that the first register still has a \emph{gap} between $\mathcal{O}_{j^*}$ and $\mathcal{O}_{j^*, 1}$, satisfying the properties we used to construct $\adve'_1$ for an extraction from the first register. Let $E_1$ denote the event that $\adve'_1$ succesfully extracts from the first register. Hence, by the foregoing discussion, we have $\Pr[E | G] > \frac{1}{\poly}$. We also have  $\Pr[G] > \frac{1}{\poly}$ and $\Pr[E] > \frac{1}{\poly}$ from before. Thus, we get $\Pr[G | E] > \frac{1}{\poly}$, proving the second point (\cref{item:allunlearnpidprimeprime2}).

Given \cref{item:allunlearnpidprime2} and \cref{item:allunlearnpidprimeprime2}, by the same extraction argument  we used for the first register, we conclude that there exists an adversary $\adve'_2$ that extracts correct coset vectors from the second register of the output of $\adve'_0$ with non-negligible probability conditioned on $\adve'_1$ extracting correct vectors from the first register. Hence, we have that the adversary tuple $\adve' = (\adve'_0, \adve'_1, \adve'_2)$ wins the game $\mathcal{G}$ with non-negligible probability. It is easy to see that this gives us a contradiction by \cref{defn:strmoecoll} (see \cref{claim:pkereducetome} for a similar reduction).

\section{Proof of Anti-Piracy Security of the PRF Scheme}\label{sec:prfcpproof}
In this section, we prove \cref{thm:prfantipiracy}.

First, we show that hidden trigger inputs are indistinguishable from uniformly random challenge strings, even when the adversary gets a (obfuscated) program that allows it to generate its own hidden trigger inputs.

\begin{definition}[Hidden Trigger Inputs]\label{defn:prfhiddentrigger}
    Let $\mathsf{GenTrigger}_{K_0, K_3, K_4, \mathsf{OPMem}, cpk}$ be the following program, where the hardcoded values are as in the PRF scheme construction (\cref{sec:cpprfcons}). The input format to the program will be clear from context.
    \begin{mdframed}
        {\bf $\underline{\mathsf{GenTrigger}_{K_0, K_3, K_4, \mathsf{OPMem},cpk}(r_1, r_2, r_3)}$}
        
        {\bf Hardcoded: $K_0, K_3, K_4, \mathsf{OPMem}, cpk$}
        \begin{enumerate}
        
            \item Parse $x_1 || x_2 || x_3 = G_2(r_1)$ with $|x_i| = s_i$.
            \item Parse $y || K_2' = G_1(F(K_0, x))$ with $|y| = n(\lambda)$.
            \item $\mathsf{OQ} \samp \io(Q_{cpk, \mathsf{OPMem}, x_1, K'_2, y}; r_3)$.
            \item $x_2' = F_3(K_3, x_1|| \mathsf{OQ} || G_3(r_2)  )$.
            \item $x_3' = F_4(K_4, x_2') \oplus (x_1|| \mathsf{OQ} ||G_3(r_2))$.
            \item Output $x_1 || x_2' || x_3'$.
        \end{enumerate}
        
    \end{mdframed} 
    The circuit $Q_{cpk, \mathsf{OPMem}, x_1, K'_2, y}$ used above is the following. Note that it contains hardcoded values that are computed during the execution of $\mathsf{GenTrigger}$.
    \begin{mdframed}
        {\bf $\underline{Q_{cpk, \mathsf{OPMem}, x_1, K'_2, y}(w)}$}
        
        {\bf Hardcoded: $cpk, \mathsf{OPMem}, x_1, K'_2, y$}
        \begin{enumerate}
            
           \item Parse $id, u_1, \dots, u_{\cosettcount} = w$.
           \item Run $\mathsf{OPMem}(id, u_1, \dots, u_{\cosettcount},  \textcolor{blue}{x_1})$. If it outputs $0$, output $\perp$ and terminate.
            \item Output $\mathsf{IBE.Enc}(cpk, id, \textcolor{blue}{y}; F_2( \textcolor{blue}{K'_2}, id))$.
                \end{enumerate}
        
    \end{mdframed} 
\end{definition}
\begin{lemma}\label{lem:prfhiddentrigger}
    Consider the following game for the PRF scheme from \cref{sec:cpprfcons}, where we let $\granlen$ denote the length of the randomness used by $\io$ to obfuscate $Q$ in \cref{defn:prfhiddentrigger}. Consider the following experiment, parameterized by $\ell(\lambda)$.
\paragraph{$\underline{\hiddentriggame(\lambda, \adve, \ell(\lambda), b)}$}
\begin{enumerate}
    \item The challenger runs $K \samp \mathsf{KeyGen}(1^\lambda)$.
    \item For multiple rounds, $\adve$ makes quantum key queries. For each query, the challenger generates a key as $\reg \samp \qkeygen(K)$ and submits $\reg$ to the adversary.
    \item The adversary outputs a register $\regi{adv}$.
    \item Sample $\mathsf{OGenTrigger} \samp \io(\mathsf{GenTrigger})$.
    \item For $i = 1$ to $\ell$:
    \begin{enumerate}[label=\arabic*.]
        \item Sample $r^i_1  \samp \zo^{s_1(\lambda)/2}$. 
        \item Sample $r^i_2 \samp \zo^{\lambda}$.
        \item Sample $r^i_3  \samp \zo^{\granlen}$. 
        \item Set $z^{0,i} = \mathsf{OGenTrigger}(r_1^i, r_2^i, r_3^i)$.
        \item Sample $z^{1,i} \samp \zo^{m(\lambda)}$.
        
    \end{enumerate}
    
    \item Output $((z^{b,i})_{i \in [\ell]}, \mathsf{OGenTrigger}, \regi{adv})$.
\end{enumerate}
    Then, for any polynomial $\ell(\lambda)$,
    \begin{equation*}
   \hiddentriggame(\lambda, \adve, \ell(\lambda), 0) \approx_c \hiddentriggame(\lambda, \adve, \ell(\lambda), 1).
    \end{equation*}
\end{lemma}
\begin{proof}
    Follows from \cref{lem:hiddentrigger}. Note that the only difference is that the adversary no longer gets a verification program, which only makes the adversary weaker.
\end{proof}

We will prove anti-piracy security through a series of hybrids. Define $\hyb_0$ to be the original game $\unclonprfgame(\lambda, \adve)$ from \cref{defn:unclonprf}.

\paragraph{$\underline{\hyb_1}$}: The challenger now computes the PRF challenge outputs $ch^{0,\ell}$ by sampling a new key using $\mathsf{QKeyGen}$ for each $\ell$ and using $\mathsf{OPEval}$. By the correctness of the copy-protected PRF scheme, we will have $ch^{0,\ell} = F(K_0, x^\ell)$ for all $\ell \in [k+1]$ with overwhelming probability. Hence,  $\hyb_0 \approx \hyb_1$.

\paragraph{$\underline{\hyb_2}$}: Instead of sampling $ch^{1,\ell}\samp \zo^{m(\lambda)}$, we will now compute $ch^{1,\ell}$ by sampling a new key using $\mathsf{QKeyGen}$ for each $\ell$ and running $\mathsf{OPEval}$ on input $q^\ell$, where we sample $q^\ell \samp \zo^{m(\lambda)}$. With overwhelming probability, we will have $ch^{1,\ell} = F(K_0, q^\ell)$, and this is indistinguishable from a random string since $q^\ell$ is not given to the adversary and $F$ is an extracting PRF. Hence, $\hyb_1 \approx \hyb_2$.

\paragraph{$\underline{\hyb_3}$}: We now sample $x^\ell, q^\ell$ for all $\ell \in [k + 1]$ as hidden triggers (\cref{defn:prfhiddentrigger}). We get $\hyb_2 \approx \hyb_3$ by \cref{lem:prfhiddentrigger}. Crucially note that at the challenger no longer directly uses the PRF key $K$ and instead computes the challenges using $\mathsf{QKeyGen}$ queries and $\mathsf{OPEval}$, which is in adversary's view in \cref{lem:prfhiddentrigger}. Hence, the adversary can indeed simulate $\hyb_2, \hyb_3$ in the reduction to \cref{lem:hiddentrigger}.  Hence,  $\hyb_2 \approx \hyb_3$.

\paragraph{$\underline{\hyb_4}$}: We now sample $x^\ell, q^\ell$ for $\ell \in [k+1]$ as follows.

    \begin{enumerate}
        
        \item Sample $r_1^\ell \samp \zo^{s_1(\lambda)/2}$. 
        \item Sample $r_2^\ell \samp \zo^{\lambda}$.
        \item Sample $r_3^\ell  \samp \zo^{\granlen}$. 
        \item Let $w^\ell = G_2(r_1^\ell)$.
        \item Parse $w_1^\ell || w_2^\ell || w_3^\ell = w^\ell$ with $|w^\ell_i| = s_i$.
        \item Parse $y^\ell || K_2^\ell = G_1(F(K_0, w^\ell))$ with $|y^\ell| = n(\lambda)$.
        \item $\mathsf{OQ}^\ell \samp \io(Q_{cpk, \mathsf{OPMem}, w_1^\ell, K^\ell_2, y^\ell}; r^\ell_3)$.
        \item $w_2^{'\ell} = F_3(K_3, w_1^\ell|| \mathsf{OQ}^\ell || G_3(r_2^\ell))$.
        \item $w_3^{'\ell} = F_4(K_4, w_2^{'\ell}) \oplus (w_1^\ell|| \mathsf{OQ}^\ell || G_3(r_2^\ell))$.
        \item Set $x^\ell =  w_1^\ell || w_2^{'\ell} || w_3^{'\ell}$.

        \item Sample $qr_1^\ell \samp \zo^{s_1(\lambda)/2}$. 
        \item Sample $qr_2^\ell \samp \zo^{\lambda}$.
        \item Sample $qr_3^\ell  \samp \zo^{\granlen}$. 
        \item Let $qw^\ell = G_2(qr_1^\ell)$.
        \item Parse $qw_1^\ell || qw_2^\ell || qw_3^\ell = qw^\ell$ with $|qw^\ell_i| = s_i$.
        \item Parse $qy^\ell || qK_2^\ell = G_1(F(K_0, qw^\ell))$ with $|qy^\ell| = n(\lambda)$.
        \item $\mathsf{qOQ}^\ell \samp \io(Q_{cpk, \mathsf{OPMem}, qw_1^\ell, qK^\ell_2, qy^\ell}; r^\ell_3)$.
        \item $qw_2^{'\ell} = F_3(K_3, qw_1^\ell|| \mathsf{qOQ}^\ell || G_3(qr_2^\ell))$.
        \item $qw_3^{'\ell} = F_4(K_4, qw_2^{'\ell}) \oplus (qw_1^\ell|| \mathsf{qOQ}^\ell || G_3(qr_2^\ell))$.
        \item Set $q^\ell =  qw_1^\ell || qw_2^{'\ell} || qw_3^{'\ell}$.
    \end{enumerate}

Note that this is only a syntactic change, we only unwrapped the sampling of the inputs $x^\ell, q^\ell$. Hence, $\hyb_3 \equiv \hyb_4$.

\paragraph{$\underline{\hyb_5}$}:
We now sample $w^\ell, qw^\ell$ for all $\ell \in [k+1]$  uniformly at random. We get  $\hyb_4 \approx \hyb_5$ by the security of the PRG $G_2$.

\paragraph{$\underline{\hyb_5}$}: We now sample all $y^\ell, qy^\ell$ and $K_{2}^{\ell},qK_{2}^{\ell}$ uniformly at random. First, observe that with overwhelming probability, the values $w^\ell,qw^\ell$ for all $\ell \in [k + 1]$ will be unique. Then, observe that the inputs $w^\ell,qw^\ell$ to $F$ have each min-entropy $s_2 + s_3$ given the view of the adversary, since only the first part $w^\ell_1,qw^\ell_1$ is used during the hidden trigger generation and the rest is discarded. Therefore, by our choice of parameters for the extracting PRF $F$ and by the security of PRG $G_1$, we have $\hyb_5 \approx \hyb_6$.

Finally, observe that the adversary gets either $(x^\ell, ch^{0,\ell} = (\mathsf{OQ}^\ell, w_1^\ell))$ (for $b_\ell = 0$) or $(x^\ell, ch^{1,\ell} = (\mathsf{qOQ}^\ell, qw_1^\ell))$ (for $b_\ell = 1$) since $x^\ell, q^\ell$ are hidden trigger inputs and we obtain $ch^{0,\ell},ch^{1,\ell}$ using $\mathsf{OPEval}$.
Observe that these are the same as ciphertexts of our PKE scheme (\cref{sec:pkecons}) encrypting the random messages $y^\ell$ and $qy^\ell$ respectively. Hence, the security follows by the CPA-style anti-piracy security (see \cref{sec:pkeproof}) of our scheme and we have $\Pr[\hyb_5 = 1] \leq 1/2 + \negl(\lambda)$.

\end{document}